\documentclass[Portugues,Final]{ic-tese-v3}

\usepackage{amsmath}
\usepackage{amsthm}
\usepackage{mathtools}
\usepackage{amsfonts}
\usepackage{booktabs, multirow}
\usepackage{float}
\usepackage[latin1,utf8]{inputenc}
\usepackage{tikz}
\usepackage{microtype}
\usepackage[framemethod=TikZ]{mdframed}
\usetikzlibrary{arrows}
\usetikzlibrary{arrows.meta}
\usetikzlibrary{positioning, shapes.arrows, backgrounds}
\usepackage{xcolor}
\definecolor{ferngreen}{rgb}{0.31, 0.47, 0.26}
\usepackage{subcaption}

\usepackage
 [pdfauthor={Alexsandro Oliveira Alexandrino},
  pdftitle={Variações do Problema de Distância de Rearranjos},
  pdfkeywords={Rearranjo de Genomas, Reversões, Transposições, Algoritmos de Aproximação},
  pdfproducer={Latex with hyperref},
  pdfcreator={pdflatex}]
{hyperref}

\mdfdefinestyle{problembox}{ nobreak=true, outerlinewidth=0pt, roundcorner=5pt, leftmargin=10, rightmargin=10, innertopmargin=10pt, innerbottommargin=10pt, splittopskip=\topskip}
\newenvironment{problembox}{\begin{mdframed}[style=problembox]}{\end{mdframed}}

\theoremstyle{definition}
\newtheorem{definition}{Definição}
\numberwithin{definition}{section}

\numberwithin{problem}{section}

\newtheorem{example}{Exemplo}
\numberwithin{example}{section}

\newtheorem{remark}{Observação}
\numberwithin{remark}{section}

\theoremstyle{plain}
\newtheorem{lemma}{Lema}
\numberwithin{lemma}{section}

\newtheorem{theorem}{Teorema}
\numberwithin{theorem}{section}

\newtheorem{corollary}{Corolário}
\numberwithin{corollary}{section}

\usepackage[portuguese, ruled, linesnumbered, vlined]{algorithm2e}
\usepackage{algorithmicx}
\usepackage{algpseudocode}

\makeatletter
\newcommand{\G}{\mathcal{G}}
\newcommand{\bi}{bi}
\newcommand{\I}{\mathcal{I}}
\newcommand{\Ig}{\mathcal{I}^{ig}}
\newcommand{\deletionElement}{\alpha}
\newcommand{\comp}{\cdot}
\newcommand{\transrev}{\rho\tau}
\newcommand{\revrev}{\rho\rho}
\newcommand{\BI}{\mathcal{BI}}
\newcommand{\insertion}{\phi}
\newcommand{\deletion}{\psi}
\newcommand{\M}{\mathcal{M}}
\newcommand{\Mindel}{\mathcal{M}^{\phi, \psi}}
\newcommand{\cclean}{c_{\mathrm{clean}}}
\newcommand{\cgood}{c_{g}}
\newcommand{\crotulado}{c_{\mathrm{labeled}}}
\newcommand{\pname}[1]{\texttt{\bf{#1}}}
\newcommand{\nprime}{n'}
\newcommand{\Mod}{\,\mathrm{mod}\,}
\newcommand{\anterior}{\texttt{anterior}}
\newcommand{\posterior}{\texttt{posterior}}
\newcommand{\revisao}{\color{black}}
\newcommand{\revisaodois}{\color{black}}
\newcommand{\revisaof}{\color{black}}

\makeatother

\begin{document}

\autor{Alexsandro Oliveira Alexandrino}

\titulo{Variações do Problema de Distância de Rearranjos}


\orientador{Prof. Dr. Zanoni Dias}

\coorientador{Prof. Dr. Ulisses Martins Dias}

\doutorado


\datadadefesa{27}{03}{2024}

\avaliadorA{Prof.~Dr.~Zanoni Dias}{Universidade Estadual de Campinas}
\avaliadorB{Profa.~Dra.~Marie-France Sagot}{Institut National de Recherche en Sciences et Technologies du Numérique}
\avaliadorC{Profa.~Dra.~Maria Emília Machado Telles Walter}{Universidade de Brasília}
\avaliadorD{Prof.~Dr.~Orlando Lee}{Universidade Estadual de Campinas}
\avaliadorE{Prof.~Dr.~Guilherme Pimentel Telles}{Universidade Estadual de Campinas}

\fichacatalografica{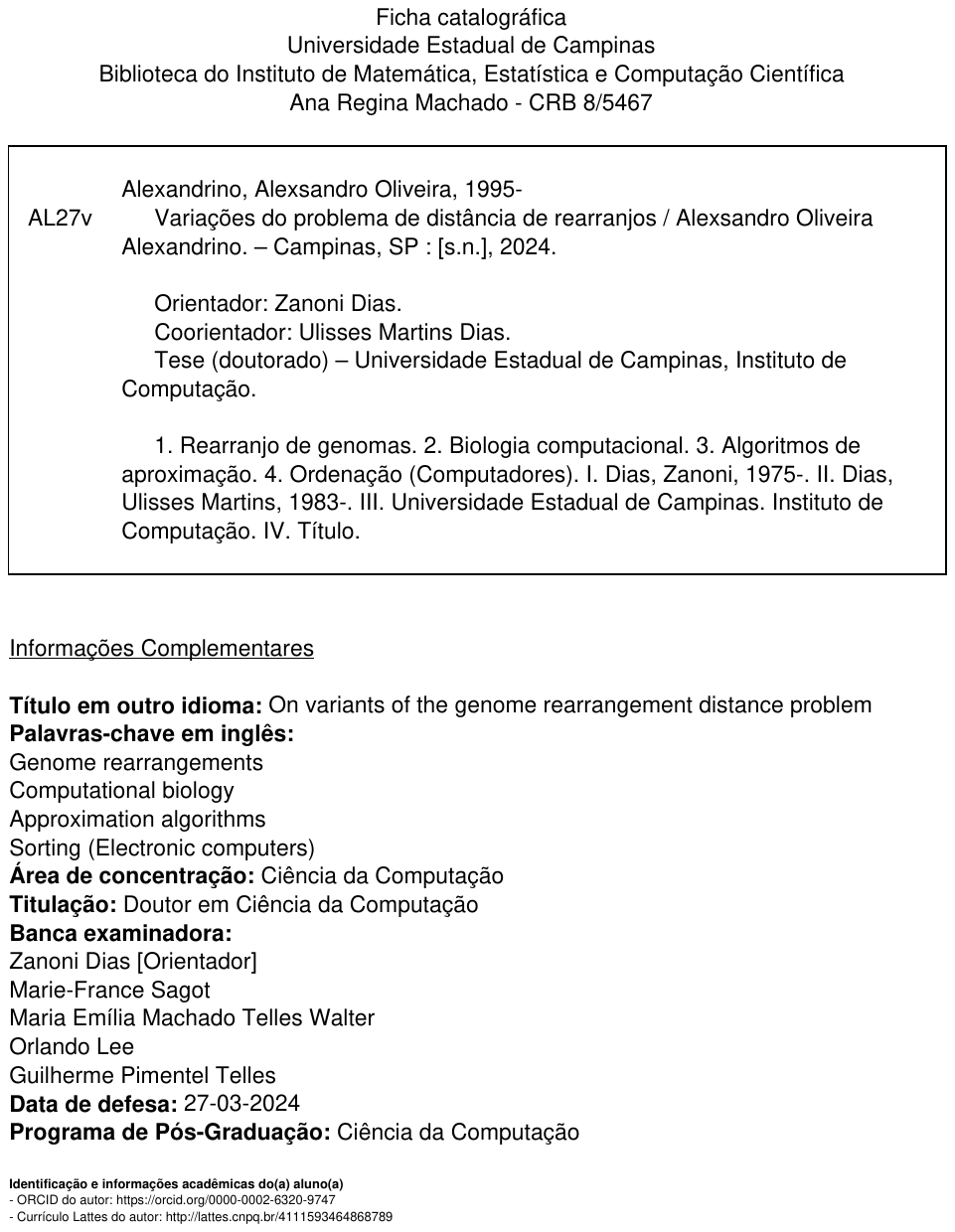}

\paginasiniciais


\begin{epigrafe}
{\it
Quando a educação não é libertadora,\\ o sonho do oprimido é ser o opressor.}

\hfill (Paulo Freire)
\end{epigrafe}

\prefacesection{Agradecimentos}

Esta tese representa o maior marco de uma jornada que começou em 2016, quando comecei o meu trabalho de conclusão de curso da graduação explorando problemas de Rearranjo de Genomas. No início dessa jornada, obtive a ajuda dos professores Criston e Lucas Ismaily, que foram os meus orientadores de TCC. Também gostaria de agradecer às professoras Paulyne e Carla, que me incentivaram a entrar na carreira acadêmica.

Ao ingressar na Unicamp, tive a sorte de ser orientado pelo Zanoni Dias e, além disso, tive a oportunidade de trabalhar como monitor ao seu lado em disciplinas e cursos da Unicamp. Agradeço ao Zanoni por tudo que me ensinou e por sua confiança, e também pela amizade construída ao longo dos anos. Também tive o prazer de ser coorientado pelos professores Carla Lintzmayer e Ulisses Dias no mestrado e doutorado, respectivamente. Gostaria de agradecer a ambos pela colaboração acadêmica e pelo que me ensinaram. 

Gostaria de agradecer ao André, Klairton e Gabriel, pela colaboração na pesquisa e pelas conversas semanais. Além deles, gostaria de agradecer a todos os membros do LOCo e todos os funcionários do IC/Unicamp. 

Alguns amigos fizeram parte dessa jornada desde a graduação na UFC até o doutorado na Unicamp. Gostaria de agradecer a todos e, em especial, à Ana Paula, Daiane, Décio, Italos e Leodécio. 

Durante 2022, tive a oportunidade de viajar para a França e trabalhar com os professores Guillaume Fertin e Géraldine Jean, que me ajudaram muito na minha experiência de intercâmbio. Além disso, fiz grandes amigos em Nantes e gostaria de agradecer a todos eles.

Por fim, e mais importante, gostaria de agradecer a todos da minha família, que sempre estiveram ao meu lado desde 2013, quando saí da minha cidade para ingressar na graduação. Agradeço aos meus pais, Pedro e Marileide, e também à Mara, Alex, Katyeudo e Katyenne.  

Esta tese foi realizada com o apoio da Coordenação de Aperfeiçoamento de Pessoal de Nível Superior - Brasil (CAPES) - Código de Financiamento 001 e do Conselho Nacional de Desenvolvimento Científico e Tecnológico - Brasil (CNPq) - Processo 202292/2020-7.

\begin{resumo}
Considerando um par de genomas de organismos de espécies relacionadas, os problemas de distância de rearranjos têm como objetivo estimar {\revisaof quão distante um deles está em relação ao outro em termos de rearranjos de genomas}, que são eventos mutacionais capazes de modificar o material genético ou a posição relativa de segmentos de um genoma. {\revisaof Considerando o Princípio da Máxima Parcimônia, o termo {\it distância}, ou ainda {\it distância de rearranjos}, é definido como o número mínimo de rearranjos de genomas necessários para transformar um genoma no outro.}

Os primeiros trabalhos que estudaram a distância de rearranjos assumiram que os genomas comparados possuem o mesmo conjunto de genes (genomas balanceados) e, além disso, apenas a ordem relativa dos genes e suas orientações, quando conhecidas, são utilizadas na representação matemática dos genomas. Essas restrições implicam que é possível transformar um genoma em outro usando apenas rearranjos que não alteram a quantidade de material genético no genoma (rearranjos conservativos). Nesse caso, os genomas são representados como permutações, o que deu origem aos problemas de Ordenação de Permutações por Rearranjos. Os principais problemas de Ordenação de Permutações por Rearranjos consideram DCJs, reversões, transposições, ou a combinação de reversões e transposições, sendo que eles possuem complexidade conhecida. Além desses, foram estudados outros problemas que combinam transposições com um ou mais dos seguintes rearranjos: transposições inversas, revrevs e reversões. Apesar de existirem algoritmos de aproximação na literatura para esses problemas, as complexidades deles permaneciam em aberto. Um dos resultados desta tese é a prova de complexidade desses problemas que combinam transposições com transposições inversas, revrevs e reversões. Além disso, apresentamos um novo algoritmo de $1.375$-aproximação para a Ordenação de Permutações por Transposições que possui melhor complexidade de tempo.

Com o avanço da área, novos trabalhos começaram a considerar genomas desbalanceados e a incorporar a distribuição dos tamanhos das regiões intergênicas. Ao considerar genomas desbalanceados, é necessário o uso de inserções e deleções para transformar um genoma em outro. Nesta tese, estudamos tanto o problema de Distância de Rearranjos em genomas desbalanceados considerando apenas a sequência de genes e suas orientações (quando conhecidas), quanto o problema de Distância de Rearranjos Intergênicos em genomas desbalanceados, que incorpora os tamanhos das regiões intergênicas na representação dos genomas, além do uso da sequência de genes e suas orientações (quando conhecidas). Apresentamos novas estruturas e conceitos para problemas que envolvem reversões, transposições e a combinação de reversões e transposições, que são usados em provas de complexidade e algoritmos de aproximação. Além disso, realizamos experimentos em genomas sintéticos e em genomas reais, evidenciando a aplicabilidade dos nossos algoritmos. 
\end{resumo}

\begin{abstract}
Considering a pair of genomes from individuals of related species, the goal of genome rearrangement distance problems is to estimate how distant these genomes are from each other based on genome rearrangements, which are mutational events that modify the genetic material or the relative position from segments of a genome. Using the Principe of Parsimony, the term {\it distance}, or {\it rearrangement distance}, refers to the minimum number of rearrangements necessary to transform one genome into the other.

Seminal works in genome rearrangements assumed that both genomes being compared have the same set of genes (balanced genomes) and, furthermore, only the relative order of genes and their orientations, when they are known, are used in the mathematical representation of the genomes. These restrictions imply that it is possible to transform one genome into the other using only conservative rearrangements, which are rearrangements that do not alter the genetic material from a genome. In this case, the genomes are represented as permutations, originating the Sorting Permutations by Rearrangements problems. The main problems of Sorting Permutations by Rearrangements considered DCJs, reversals, transpositions, or the combination of both reversals and transpositions, and these problems have their complexity known. Besides these problems, other ones were studied involving the combination of transpositions with one or more of the following rearrangements: transreversals, revrevs, and reversals. Although there are approximation results for these problems, their complexity remained open. Some of the results of this thesis are the complexity proofs for these problems. Furthermore, we present a new $1.375$-approximation algorithm, which has better time complexity, for the Sorting Permutations by Transpositions. 

As the field has progressed, new works started to consider unbalanced genomes and to incorporate the size distribution of intergenic regions. When considering unbalanced genomes, it is necessary to use insertions and deletions to transform one genome into another. In this thesis, we studied Rearrangement Distance problems on unbalanced genomes considering only gene order and their orientations (when they are known), as well as Intergenic Rearrangement Distance problems, which incorporate the information regarding the size distribution of intergenic regions, besides the use of gene order and their orientations (when they are known). We present new structures and concepts for problems that include reversals, transpositions, and the combination of reversals and transpositions. These structures and concepts are used in complexity proofs and approximation algorithms. Furthermore, we performed experiments in simulated and real genomes, showing the applicability of our algorithms.
\end{abstract}


\listoffigures

\listoftables




\tableofcontents

\fimdaspaginasiniciais

\chapter{Introdução}\label{cap:label:introducao}

\emph{Rearranjos de genomas} são mutações que alteram grandes trechos de um genoma. Na genômica comparativa, uma das formas mais aceitas para estimar a distância evolutiva é com o uso da \emph{distância de rearranjos de genoma}. Nessa distância, cada rearranjo está associado a um custo, que pode ser unitário, quando simplesmente indica a ocorrência de um rearranjo, ou pode ter um valor que indica uma característica do rearranjo, como a quantidade de {\revisaof quebras na sequência de DNA}. 
{\revisaof
Para os genomas de origem $\G_o$ e destino $\G_d$, a \emph{distância de rearranjos} entre $\G_o$ e $\G_d$ é definida como o menor custo possível de uma sequência de rearranjos que transforma o genoma origem $\G_o$ no genoma destino $\G_d$.
}

Nos problemas de rearranjo de genomas, um genoma é usualmente representado a partir da sua sequência de genes e, {\revisaof dependendo da informação disponível sobre cada gene e das características dos genomas estudados}, diferentes representações matemáticas podem ser utilizadas. Supondo que não existem genes repetidos e que os genomas possuem o mesmo conjunto de genes, um genoma pode ser representado como uma permutação de números inteiros, onde cada elemento da permutação representa um gene. Quando a orientação dos genes é conhecida, essa informação é representada por um sinal ``$+$'' ou ``$-$'' e, nesse caso, as permutações são ditas com sinais. Quando a orientação dos genes é desconhecida, permutações sem sinais são utilizadas para a representação dos genomas. Ao utilizar permutações, o problema de Distância de Rearranjos é equivalente ao problema de Ordenação de Permutações por Rearranjos~\cite{2009-fertin-etal}.

Um \emph{modelo de rearranjo} é o conjunto de rearranjos (ou operações) permitidos em um dado problema de Distância de Rearranjos. Dentre os rearranjos mais estudados na literatura, temos a \emph{reversão}, que inverte um segmento do genoma, e a \emph{transposição}, que troca as posições relativas de dois segmentos adjacentes do genoma. {\revisaof Assim sendo, um modelo de rearranjo pode ser definido de forma a considerar apenas as reversões, outro modelo pode considerar apenas as transposições, enquanto um terceiro modelo pode ainda considerar tanto reversões quanto transposições.}

O estudo dos problemas de Distância de Rearranjo teve início com modelos de rearranjo formados por um único tipo de operação, o que gerou o problema da Ordenação de Permutações por Reversões~\cite{1999a-caprara,1995-kececioglu-sankoff} e o problema da Ordenação de Permutações por Transposições~\cite{1995b-bafna-pevzner}. Com o avanço da área, esses rearranjos foram reunidos em um único modelo, mas sem considerar custos distintos entre as operações~\cite{1998-walter-etal}.
O problema da Ordenação de Permutações com Sinais por Reversões possui algoritmo exato polinomial~\cite{1999-hannenhalli-pevzner}. Já os problemas de Ordenação de Permutações sem Sinais por Reversões ou por Transposições pertencem à classe de problemas NP-difíceis~\cite{2012-bulteau-etal, 1999a-caprara}. Para permutações com ou sem sinais, o problema que considera o modelo de rearranjo contendo tanto reversões quanto transposições também é NP-difícil~\cite{2019b-oliveira-etal}.

Dentre os algoritmos mais conhecidos desenvolvidos para esses problemas, podemos citar o algoritmo exato polinomial para o problema da Ordenação de Permutações com Sinais por Reversões, desenvolvido por Hannenhalli e Pevzner~\cite{1999-hannenhalli-pevzner}, e o algoritmo de $1.375$-aproximação para a Ordenação de Permutações por Transposições, desenvolvido por Elias e Hartman~\cite{2006-elias-hartman}. Recentemente, Silva e coautores~\cite{2022-silva-etal} mostraram que o fator de aproximação do algoritmo de Elias e Hartman~\cite{2006-elias-hartman} ultrapassa o valor de $1.375$ em alguns casos. Silva e coautores~\cite{2022-silva-etal} também mostraram como solucionar o problema e apresentaram um algoritmo com complexidade de tempo de $O(n^6)$. Nesta tese, apresentamos uma nova versão do algoritmo proposto por Elias e Hartman~\cite{2006-elias-hartman} que garante o fator de aproximação de $1.375$ em todos os casos e possui complexidade de tempo de $O(n^5)$. 

Outros rearranjos estudados na literatura são as transposições inversas e as revrevs. Assim como as transposições, essas operações agem em dois segmentos adjacentes de um genoma. Uma \emph{transposição inversa} troca as posições relativas dos dois segmentos adjacentes, assim como as transposições, mas essa operação também inverte um dos segmentos. Já uma \emph{revrev} inverte cada um dos dois segmentos adjacentes, mas não troca as posições relativas desses segmentos. A complexidade dos problemas de Ordenação de Permutações por Rearranjos com modelos que possuem transposições inversas e revrevs eram desconhecidas, apesar de existirem algoritmos de aproximação para esses problemas~\cite{2009-fertin-etal}. Nesta tese, apresentamos provas de dificuldade para o problema de Ordenação de Permutações por Rearranjos considerando modelos de rearranjo contendo transposições junto com combinações de reversões, transposições inversas e revrevs.



Quando os genomas de origem e destino possuem conjuntos de genes distintos, dizemos que esses genomas são desbalanceados e usamos strings para representá-los matematicamente. Assim como nas permutações, quando a orientação dos genes é conhecida, essa informação é representada por um sinal ``$+$'' ou ``$-$''. Já quando a orientação dos genes é desconhecida, apenas usamos strings sem sinais. 

{\revisaof
As reversões, transposições, transposições inversas, e revrevs são operações conservativas, ou seja, são rearranjos que não alteram a quantidade de material genético do genoma. As operações não conservativas mais estudadas são as inserções e deleções de material genético~\cite{2013-willing-etal, 2000-el-mabrouk}. Inserções e deleções são coletivamente chamadas de {\it indels} em trabalhos da área. Nos modelos que possuem {\it indels}, o conjunto de genes do genoma origem e o conjunto de genes do genoma destino podem ser distintos e, por isso, a representação dos genomas é feita usando strings. {\revisaof Nesta tese, apresentamos provas de dificuldade e algoritmos de aproximação para os problemas de Distância de Rearranjos em genomas desbalanceados}, considerando modelos que combinam as operações de reversão, transposição e {\it indel}.

Além disso, estudamos uma outra operação chamada {\it block interchange}, que troca a posição relativa de dois segmentos quaisquer do genoma. Note que uma transposição é um tipo específico de {\it block interchange} em que os segmentos afetados são adjacentes. Nesta tese, também apresentamos algoritmos de aproximação para os problemas de Distância de Rearranjos em genomas desbalanceados, considerando {\it block interchanges} e a combinação de {\it block interchanges} e reversões.
}

Além da ordem relativa em que os genes aparecem no genoma, estudos recentes incorporaram a distribuição do tamanho das \emph{regiões intergênicas} (quantidade de nucleotídeos entre cada par de genes consecutivos) na representação matemática dos genomas. A incorporação dessa informação na representação dos genomas é motivada por evidências de que as regiões intergênicas possibilitam inferir melhores cenários evolucionários~\cite{2016a-biller-etal, 2016b-biller-etal}. Os problemas de Distância de Rearranjos Intergênicos existentes na literatura tratam apenas genomas balanceados. Os modelos de rearranjo considerados possuem operações conservativas e {\it indels}, no entanto, esses {\it indels} podem inserir ou remover apenas nucleotídeos de regiões intergênicas, restringindo os genomas comparados a terem o mesmo conjunto de genes. Nesta tese estudamos os problemas de Distância de Rearranjos Intergênicos em genomas desbalanceados, considerando modelos que combinam as operações de reversão, transposição, e {\it indels}, {\revisaof sendo que os {\it indels} também podem inserir ou remover genes}. Para a maioria dos modelos estudados, apresentamos provas de dificuldade e, além disso, para todos os modelos, apresentamos algoritmos de aproximação. 

Esta tese está organizada da seguinte forma. O Capítulo~\ref{cap:label:fundamentacao} apresenta notações e conceitos gerais para os problemas estudados. O Capítulo~\ref{cap:label:transposicao} apresenta o novo algoritmo de $1.375$-aproximação para a Ordenação de Permutações por Transposições e provas de dificuldade para problemas de Ordenação de Permutações por Transposições e Outros Rearranjos. O Capítulo~\ref{cap:label:indel} considera os problemas de Distância de Rearranjos em genomas desbalanceados e o Capítulo~\ref{cap:label:intergenicos} considera os problemas de Distância de Rearranjos Intergênicos em genomas desbalanceados. Nos capítulos~\ref{cap:label:indel} e \ref{cap:label:intergenicos} apresentamos provas de dificuldade e algoritmos de aproximação para os modelos estudados. {\revisaodois No Capítulo~\ref{cap:label:experimentos}, apresentamos resultados experimentais em genomas sintéticos usando os algoritmos dos capítulos ~\ref{cap:label:indel} e \ref{cap:label:intergenicos}. Além disso, evidenciamos a aplicabilidade de um dos nossos algoritmos em uma base de dados de genomas reais, usando as soluções desse algoritmo na construção de uma árvore filogenética.} Por fim, no Capítulo~\ref{cap:label:conclusao}, apresentamos as considerações finais, sumário dos resultados, e direções de trabalhos futuros.

\chapter{Fundamentação Teórica}\label{cap:label:fundamentacao}


Neste capítulo, apresentamos definições e notações fundamentais para esta tese. As definições apresentadas neste capítulo são gerais para os problemas de distância de rearranjos e são usadas nos capítulos seguintes. No entanto, conceitos específicos para alguma variação de problemas de rearranjos são definidos apenas no capítulo referente a essa variação.

\section{Representação de Genomas}

Nesta seção, apresentamos como os genomas em um problema de Distância de Rearranjos são modelados matematicamente. 

\subsection{Representação da Ordem Relativa dos Genes}

Os problemas clássicos da área de rearranjo de genomas utilizam apenas a ordem relativa dos genes para a representação matemática dos genomas. {\revisao Exceto quando dito expressamente o contrário, nesta tese assumimos que os genomas não possuem genes repetidos, assim como a maior parte da literatura na área.} 

Quando os genomas a serem comparados possuem o mesmo conjunto de genes, dizemos que os genomas são \emph{balanceados} e podemos utilizar permutações de números inteiros para representar a ordem relativa dos genes em cada genoma. Cada elemento da permutação representa um gene. Se a orientação dos genes é considerada, então utilizamos uma \emph{permutação com sinais}, em que cada elemento possui um sinal ``$+$'' ou ``$-$'' que indica a orientação do gene. Caso contrário, utilizamos uma \emph{permutação sem sinais} ou, de forma equivalente, consideramos que todos os elementos possuem sinal ``$+$''. 

\begin{definition}
Denotamos uma \emph{permutação com sinais} de tamanho $n$ por $\pi = (\pi_1\;\pi_2\;\ldots\;$ $\pi_{n-1}\;\pi_n)$, tal que $\pi_i \in (\{+1,+2,\ldots,+({n-1}),+n\} \cup \{{-1},{-2},\ldots,{-({n-1})},$ ${-n}\})$ e $i \neq j$ se, e somente se, $|\pi_i| \neq |\pi_j|$, para quaisquer $i$ e $j$.
\end{definition}

\begin{definition}
Denotamos uma \emph{permutação sem sinais} de tamanho $n$ por $\pi = (\pi_1\;\pi_2\;\ldots\;$ $\pi_{n-1}\;\pi_n)$, tal que $\pi_i \in \{1,2,\ldots,{n-1},n\}$ e $i \neq j$ se, e somente se, $\pi_i \neq \pi_j$, para quaisquer $i$ e $j$.
\end{definition}

Quando os genomas possuem conjuntos distintos de genes, dizemos que os genomas são \emph{desbalanceados} e utilizamos \emph{strings} para representar a ordem relativa dos genes. Cada elemento de uma string representa um gene ou uma sequência contígua de genes exclusiva de um dos genomas. Assim como nas permutações, para o caso em que a orientação dos genes é considerada, cada elemento possui um sinal ``$+$'' ou ``$-$'' que indica a orientação do gene. 

\begin{definition}
Dada uma string $A = (A_1\;A_2\;\ldots\;A_m)$, denotamos por $|A| = m$ o \emph{tamanho} da string $A$ e por $|A_i|$ o \emph{rótulo} que corresponde ao elemento $A_i$ desconsiderando o seu sinal. 
\end{definition}

\begin{definition}
O \emph{alfabeto} $\Sigma_A$ de uma string $A$ é o conjunto de rótulos da string $A$.
\end{definition}

A fim de fazer com que as definições para strings e permutações sejam semelhantes, usamos o alfabeto $\Sigma_A = \{1,2, \ldots, n\} \cup \{\deletionElement\}$, onde $\deletionElement$ é um rótulo usado para elementos que fazem parte do genoma origem $\G_o$, mas não fazem parte do genoma destino $\G_d$.

Uma instância genérica $\I = (\G_o, \G_d)$ para um problema de Distância de Rearranjos consiste em um genoma origem $\G_o$ e um genoma destino $\G_d$. 

A \emph{permutação identidade} ou \emph{string identidade} de $n$ elementos é denotada por $\iota^n = (+1\;+2\;\ldots\;+n)$ ou $\iota^n = (1\;2\;\ldots\;n)$, dependendo se a orientação é considerada ou não na variação do problema. 

Representamos a ordem relativa dos genes do genoma destino $\G_d$ usando a permutação (string) identidade $\iota^n$, onde cada elemento $\iota^n_i$ representa um gene de $\G_d$, que possui correspondência em ambos genomas, ou um segmento contíguo maximal de genes que não possuem correspondência em $\G_o$.

Para a ordem relativa dos genes do genoma origem $\G_o$, usamos uma string $A = (A_1\;A_2\;\ldots\;A_m)$, onde cada elemento $A_i$ representa um gene de $\G_o$, usando o mesmo mapeamento de rótulos e genes usado na representação de $\G_d$, ou um segmento contíguo maximal de genes sem correspondência em $\G_d$. Se $A_i$ mapeia um gene que possui correspondência em ambos os genomas, então $A_i$ tem sinal ``$+$'', se o gene possui a mesma orientação em ambos os genomas, ou $A_i$ tem sinal ``$-$'', caso contrário. Para qualquer elemento $A_i$ que mapeia um segmento contíguo de genes sem correspondência em $\G_d$, definimos $A_i = \deletionElement$, sem qualquer sinal, já que o elemento será removido independentemente do seu conteúdo. 

\begin{example}
	{\revisao
	Considere dois genomas com as sequências de genes $\G_o = (a~b~d~e~i~f~g)$ e $\G_d = (a~c~d~h~f~g)$, onde cada letra representa um gene. Representamos $\G_d$ como a string $\iota^n = (1~2~3~4~5~6)$ e $\G_o$ como a string $A = (1~\deletionElement~3~\deletionElement~5~6)$. Neste exemplo, temos $|A| = 6$, $\Sigma_{A} \cap \Sigma_{\iota^n} = \{1, 3, 5, 6\}$, $\Sigma_A \setminus \Sigma_{\iota^n} = \{\deletionElement\}$ e $\Sigma_{\iota^n} \setminus \Sigma_{A} = \{2, 4\}$. Note que o segmento $(e, i)$ de $\G_o$ é mapeado em um único elemento $\deletionElement$ em $A$.
	}
\end{example}

\begin{definition}
Dada uma permutação $\pi$, denotamos por $-\pi_i$ o elemento $\pi_i$ com orientação invertida. O mesmo é válido para strings.
\end{definition}

\begin{example}
Dado $\pi = ({+4}\;{-3}\;{-2}\;{+1})$, temos que $-\pi_1 = {-4}$, $-\pi_2 = {+3}$, $-\pi_3 = {+2}$ e $-\pi_4 = {-1}$. 
\end{example}

O conjunto de \emph{rótulos comuns} às strings é definido por $\Sigma_A \cap \Sigma_{\iota^n}$, o conjunto de \emph{rótulos exclusivos} ao genoma destino é definido por $\Sigma_{\iota^n} \setminus \Sigma_A$, e o conjunto de \emph{rótulos exclusivos} ao genoma origem é definido por $\Sigma_A \setminus \Sigma_{\iota^n}$, que ou é vazio ou é igual a $\{\deletionElement\}$. Um elemento da string ($A_i$ ou $\iota^n_i$) é classificado como \emph{comum} ou \emph{exclusivo} de acordo com a classificação do seu rótulo. 

A modelagem dos genomas usando exclusivamente a ordem relativa dos genes é chamada de \emph{representação clássica}. Para genomas desbalanceados, uma \emph{instância clássica} é denotada por $\I = (A, \iota^n)$, sendo que as strings possuem sinais ou não, dependendo da informação disponível em relação à orientação dos genes. 

Quando $\G_o$ e $\G_d$ são genomas balanceados, a string $A$ corresponde a uma permutação e, assim, usamos $\pi$ para representar a ordem relativa dos genes de $\G_o$. Exceto quando dito expressamente o contrário, assumimos que a permutação $\pi$ tem tamanho $n$ e, portanto, o genoma destino $\G_d$ é representado pela permutação identidade $\iota^n$. Note que em genomas balanceados ambas as sequências de genes possuem o mesmo tamanho. Quando a orientação dos genes é desconhecida, apenas ignoramos os sinais no mapeamento dos genomas $\G_o$ e $\G_d$. Como $\iota^n$ pode ser inferido a partir do tamanho de $\pi$, uma instância clássica é denotada apenas por $\I = \pi$. 

Note que para genomas balanceados, os problemas de distância de rearranjos entre dois genomas são equivalentes aos problemas de ordenação de uma permutação $\pi$ usando operações de rearranjo.

\subsection{Representação de Regiões Intergênicas}

Regiões intergênicas são sequências de nucleotídeos presentes entre pares de genes contíguos e nas extremidades do genoma. O \emph{tamanho} de uma região intergênica é a quantidade de nucleotídeos que estão nela. Estudos mostram que as regiões intergênicas são mais suscetíveis a mudanças e são quebradas por rearranjos de genomas~\cite{2016a-biller-etal,2016b-biller-etal}. Além disso, ao comparar dois genomas, podemos achar correspondência entre os genes desses genomas, o que não é possível com as regiões intergênicas~\cite{2016a-biller-etal,2016b-biller-etal}. Esses motivos nos levam a representar regiões intergênicas usando os seus tamanhos ao invés de rótulos, o que já é feito em outros trabalhos~\cite{2016b-bulteau-etal,2017-fertin-etal,2018a-oliveira-etal}.

A distribuição do tamanho das regiões intergênicas é modelada matematicamente por uma lista de números inteiros não-negativos. Ao considerar regiões intergênicas, o genoma origem $\G_o$ é modelado pela tupla $(A, \breve{A})$ e o genoma destino $\G_d$ é modelado pela tupla $(\iota^n, \breve{\iota}^n)$. De forma similar ao que é feito nas instâncias em que apenas a ordem relativa dos genes é considerada, ao construir a representação de uma instância intergênica $\Ig = (\G_o, \G_d)$, {\revisaof mapeamos com um único elemento tanto em $A$ quanto em $\iota^n$ os segmentos contíguos maximais de genes que não possuem correspondência no outro genoma. }

A lista $\breve{A}$ tem tamanho $|\breve{A}| = |A| + 1$ e a lista $\breve{\iota}^n$ tem tamanho $|\breve{\iota}^n| = n + 1$. Temos que $\breve{A}_i$ corresponde ao número de nucleotídeos na região intergênica entre $A_{i-1}$ e $A_{i}$, para $2 \leq i \leq |A|$, e os valores $\breve{A}_1$ e $\breve{A}_{|A|+1}$ correspondem ao número de nucleotídeos nas extremidades de $\G_o$. Analogamente, temos que $\breve{\iota}^n_i$ corresponde ao número de nucleotídeos na região intergênica entre $\iota^n_{i-1}$ e $\iota^n_{i}$, para $2 \leq i \leq n$, e os valores $\breve{\iota}^n_{1}$ e $\breve{\iota}^n_{n+1}$ correspondem ao número de nucleotídeos nas extremidades de $\G_d$. 

A modelagem dos genomas que considera a ordem relativa dos genes e a distribuição do tamanho de regiões intergênicas é chamada de \emph{representação intergênica}. Nos problemas intergênicos estudados neste trabalho, sempre tratamos de genomas desbalanceados. Assim, uma \emph{instância intergênica} é denotada por $\Ig = (\G_o, \G_d)$, onde $\G_o = (A, \breve{A})$ e $\G_d = (\iota^n, \breve{\iota}^n)$. 

{\revisao
Nas figuras \ref{cap2:image:example_genome:intergenic} e \ref{cap2:image:example_genome:intergenic2}, apresentamos exemplos de como genomas são modelados matematicamente usando a representação intergênica.
}

\begin{figure*}[t]
\centering

\resizebox{0.8\textwidth}{!}{
\begin{tikzpicture}[scale=0.9, gene/.style={single arrow,
        draw=black,
        fill=#1,
        single arrow head extend=1mm
    }]
    \begin{scope}
        \node(g1) at (-2,0) {$\mathcal{G}_o = $};
    \end{scope}

    \begin{scope}[every node/.style={rectangle, draw=black, fill=white, text width=3mm, align=center}, scale=1]
      \scriptsize
      \node(i1) at (-0.75,0) {0};
      \node(i2) at (0.75,0) {3};
      \node(i3) at (2.25,0) {2};
      \node(i4) at (3.75,0) {3};
      \node(i5) at (5.25,0) {10};
      \node(i6) at (6.75,0) {2};
      \node[dashed](i7) at (8.25,0) {1};
      \node(i8) at (9.75,0) {6};
    \end{scope}

    \begin{scope}[every node/.style={draw=black, fill=white, minimum size=8mm}]
    \node[gene, draw=teal, rotate=180](g1) at (0.1,0) {};
    \node[gene, draw=orange](g2) at (1.4,0) {};
    \node[gene, dashed](g3) at (2.9,0) {};
    \node[gene, draw=blue, rotate=180](g4) at (4.6,0) {};
    \node[gene, draw=ferngreen](g5) at (5.9,0) {};
    \node[gene, dashed](g6) at (7.4,0) {};
    \node[gene, dashed](g7) at (8.9,0) {};
    \end{scope}

    \begin{scope}
      \node(g1) at (0,0) {$h$};
      \node(g2) at (1.5,0) {$d$};
      \node(g3) at (3,0) {$x$};
      \node(g4) at (4.5,0) {$a$};
      \node(g5) at (6,0) {$c$};
      \node(g6) at (7.5,0) {$y$};
      \node(g7) at (9,0) {$z$};
    \end{scope}

    \begin{scope}
        \node(g2) at (-2,-1) {$\mathcal{G}_d = $};
    \end{scope}
    \begin{scope}[every node/.style={circle, draw=black, fill=white, minimum size=8mm}]
    \node[gene, draw=blue](g1) at (-0.1,-1) {};
    \node[gene, draw=ferngreen](g2) at (1.4,-1) {};
    \node[gene, draw=orange](g3) at (2.9,-1) {};
    \node[gene, draw=teal, rotate = 180](g4) at (4.6,-1) {};
    \node[gene, draw=violet](g5) at (5.9,-1) {};
    \end{scope}
    \begin{scope}
      \node(g1) at (-0.1,-1) {$a$};
      \node(g2) at (1.4,-1) {$c$};
      \node(g3) at (2.9,-1) {$d$};
      \node(g4) at (4.6,-1) {$h$};
      \node(g5) at (5.9,-1) {$f$};
    \end{scope}
    \begin{scope}[every node/.style={rectangle, draw=black, fill=white, text width=3mm, align=center}, scale=1]
    \scriptsize
    \node(i1) at (-0.75,-1) {5};
    \node(i2) at (0.75,-1) {2};
    \node(i3) at (2.25,-1) {7};
    \node(i4) at (3.75,-1) {1};
    \node(i5) at (5.25,-1) {0};
    \node(i6) at (6.75,-1) {5};
    \end{scope}
\end{tikzpicture}
}
\caption{\label{cap2:image:example_genome:intergenic}
{\revisao 
Exemplo de dois genomas $\G_o$ e $\G_d$, onde genes são representados por letras dentro de setas, a orientação dos genes é indicada pela orientação das setas, e os tamanhos das regiões intergênicas são representados por números dentro de retângulos. Os genes de $\G_d$ são mapeados da seguinte forma: $a$ é mapeado em $+1$, $c$ é mapeado em $+2$, $d$ é mapeado em $+3$, $h$ é mapeado em $+4$, e $f$ é mapeado em $+5$. Assim, o genoma $\G_d$ é representado por $(\iota^n, \breve{\iota}^n)$, onde $\iota^n = ({+1}~{+2}~{+3}~{+4}~{+5})$ e $\breve{\iota}^n = (5,2,7,1,0,5)$. O gene $x$ e o segmento que vai de $y$ até $z$ em $\G_o$ não estão presentes em $\G_d$. Portanto, ambos são mapeados no elemento $\deletionElement$. O genoma $\G_o$ é representado por $(A, \breve{A})$, onde $A = ({+4}~{+3}~\deletionElement~{-1}~{+2}~\deletionElement)$ e $\breve{A} = (0,3,2,3,10,2,6)$. Os alfabetos $\Sigma_A$ e $\Sigma_{\iota^n}$ são os conjuntos $\{1,2,3,4,\alpha\}$ e $\{1,2,3,4,5\}$, respectivamente.
}
}
\end{figure*}

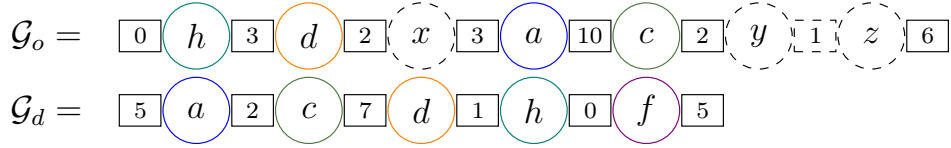
\begin{figure*}[t]
\centering

\resizebox{0.8\textwidth}{!}{
\begin{tikzpicture}[scale=0.9, gene/.style={circle,
        draw=black,
        fill=#1,
    }]
    \begin{scope}
        \node(g1) at (-2,0) {$\mathcal{G}_o = $};
    \end{scope}

    \begin{scope}[every node/.style={rectangle, draw=black, fill=white, text width=3mm, align=center}, scale=1]
      \scriptsize
      \node(i1) at (-0.75,0) {0};
      \node(i2) at (0.75,0) {3};
      \node(i3) at (2.25,0) {2};
      \node(i4) at (3.75,0) {3};
      \node(i5) at (5.25,0) {10};
      \node(i6) at (6.75,0) {2};
      \node[dashed](i7) at (8.25,0) {1};
      \node(i8) at (9.75,0) {6};
    \end{scope}

    \begin{scope}[every node/.style={draw=black, fill=white, minimum size=8mm}]
    \node[gene, draw=teal, rotate=180](g1) at (0.0,0) {};
    \node[gene, draw=orange](g2) at (1.5,0) {};
    \node[gene, dashed](g3) at (3.0,0) {};
    \node[gene, draw=blue, rotate=180](g4) at (4.5,0) {};
    \node[gene, draw=ferngreen](g5) at (6.0,0) {};
    \node[gene, dashed](g6) at (7.5,0) {};
    \node[gene, dashed](g7) at (9.0,0) {};
    \end{scope}

    \begin{scope}
      \node(g1) at (0,0) {$h$};
      \node(g2) at (1.5,0) {$d$};
      \node(g3) at (3,0) {$x$};
      \node(g4) at (4.5,0) {$a$};
      \node(g5) at (6,0) {$c$};
      \node(g6) at (7.5,0) {$y$};
      \node(g7) at (9,0) {$z$};
    \end{scope}

    \begin{scope}
        \node(g2) at (-2,-1) {$\mathcal{G}_d = $};
    \end{scope}
    \begin{scope}[every node/.style={circle, draw=black, fill=white, minimum size=8mm}]
    \node[gene, draw=blue](g1) at (-0.0,-1) {};
    \node[gene, draw=ferngreen](g2) at (1.5,-1) {};
    \node[gene, draw=orange](g3) at (3.0,-1) {};
    \node[gene, draw=teal, rotate = 180](g4) at (4.5,-1) {};
    \node[gene, draw=violet](g5) at (6.0,-1) {};
    \end{scope}
    \begin{scope}
      \node(g1) at (-0.0,-1) {$a$};
      \node(g2) at (1.5,-1) {$c$};
      \node(g3) at (3.0,-1) {$d$};
      \node(g4) at (4.5,-1) {$h$};
      \node(g5) at (6.0,-1) {$f$};
    \end{scope}
    \begin{scope}[every node/.style={rectangle, draw=black, fill=white, text width=3mm, align=center}, scale=1]
    \scriptsize
    \node(i1) at (-0.75,-1) {5};
    \node(i2) at (0.75,-1) {2};
    \node(i3) at (2.25,-1) {7};
    \node(i4) at (3.75,-1) {1};
    \node(i5) at (5.25,-1) {0};
    \node(i6) at (6.75,-1) {5};
    \end{scope}
\end{tikzpicture}
}
\caption{\label{cap2:image:example_genome:intergenic2}
{\revisao
Exemplo de dois genomas $\G_o$ e $\G_d$, onde genes são representados por letras dentro de círculos, a orientação dos genes é desconhecida, e os tamanhos das regiões intergênicas são representados por números dentro de retângulos. Os genes de $\G_d$ são mapeados da seguinte forma: $a$ é mapeado em $1$, $c$ é mapeado em $2$, $d$ é mapeado em $3$, $h$ é mapeado em $4$, e $f$ é mapeado em $5$. Assim, o genoma $\G_d$ é representado por $(\iota^n, \breve{\iota}^n)$, onde $\iota^n = ({1}~{2}~{3}~{4}~{5})$ e $\breve{\iota}^n = (5,2,7,1,0,5)$. O gene $x$ e o segmento que vai de $y$ até $z$ em $\G_o$ não estão presentes em $\G_d$. Portanto, ambos são mapeados no elemento $\deletionElement$. O genoma $\G_o$ é representado por $(A, \breve{A})$, onde $A = ({4}~{3}~\deletionElement~{1}~{2}~\deletionElement)$ e $\breve{A} = (0,3,2,3,10,2,6)$. Os alfabetos $\Sigma_A$ e $\Sigma_{\iota^n}$ são os conjuntos $\{1,2,3,4,\alpha\}$ e $\{1,2,3,4,5\}$, respectivamente.
}
}
\end{figure*}

\section{Rearranjos de Genomas}

Nesta seção, apresentamos os tipos de rearranjos de genomas que serão estudados no decorrer desta tese, assim como variações desses rearranjos dependendo da representação utilizada para os genomas. 

Os rearranjos de genomas podem afetar apenas a ordem relativa de um ou mais segmentos do genoma (\emph{rearranjos conservativos}) ou eles podem alterar o material genético do genoma (\emph{rearranjos não conservativos}). Dentre os rearranjos conservativos, podemos citar as reversões, transposições e \textit{block interchanges}. Os rearranjos não conservativos considerados neste trabalho são as inserções e deleções. Essas últimas duas operações são chamadas coletivamente de \emph{indels}.

Como uma permutação é um tipo específico de string que possui alfabeto $\Sigma = \{1,2, \ldots, n\}$, definimos as operações de rearranjo usando strings. Apenas os {\it indels} não são operações válidas em permutações. 

\begin{definition}
Para uma operação de rearranjo $\beta$, denotamos por $\G \comp \beta = \G'$ o genoma resultante após $\beta$ ser aplicado em $\G$. Também temos $A \comp \beta = A'$ e $\breve{A} \comp \beta = \breve{A}'$ como resultados da aplicação de $\beta$ na string $A$ e na lista $\breve{A}$, respectivamente.
\end{definition}

Dada uma sequência de operações $S = (\beta_1, \beta_2, \ldots, \beta_k)$ com tamanho $|S| = k$, denotamos por $\G \comp S = (((\G \comp \beta_1) \comp \beta_2) \comp \ldots \comp \beta_{k-1}) \comp \beta_k = \G'$ o genoma resultante da aplicação da sequência $S$ em $\G$. A mesma notação é válida para aplicação de uma sequência de operações em strings e listas de regiões intergênicas.

Agora, apresentamos os rearranjos conservativos e os seus respectivos efeitos na ordem relativa dos genes de um genoma.

\begin{definition}
Considerando a representação clássica, dada uma string $A$ de tamanho $m$, uma \emph{reversão} $\rho(i,j)$, com $1 \leq i \leq j \leq m$, inverte o segmento $(A_i~A_{i+1}~\ldots~A_{j-1}~A_j)$ e, caso $A$ seja uma string com sinais, inverte todos os sinais dos elementos afetados. Mostramos a seguir o resultado da aplicação de $\rho(i,j)$ quando $A$ é uma string com sinais (1) e quando $A$ é uma string sem sinais (2).
\begin{flalign*}
	(1)\quad& A \comp \rho(i,j) = (A_1~\ldots~A_{i-1}~\underline{{-A_j}~{-A_{j-1}}~\ldots~{-A_{i+1}}~{-A_i}}~A_{j+1}~\ldots~A_m)\\
	(2)\quad& A \comp \rho(i,j) = (A_1~\ldots~A_{i-1}~\underline{{A_j}~{A_{j-1}}~\ldots~{A_{i+1}}~{A_i}}~A_{j+1}~\ldots~A_m)
\end{flalign*}
\end{definition}

\begin{definition}
Considerando a representação clássica, dada uma string $A$ de tamanho $m$, uma \emph{transposição} $\tau(i,j,k)$, com $1 \leq i < j < k \leq m+1$, troca as posições relativas dos segmentos $(A_i~A_{i+1}~\ldots~A_{j-1})$ e $(A_j~A_{j+1}~\ldots~A_{k-1})$. Por não afetar a orientação dos elementos, o efeito de uma transposição é o mesmo em strings com e sem sinais. Mostramos a seguir o efeito de $\tau(i,j,k)$ em $A$.
\begin{flalign*}
A \comp \tau(i,j,k) = (A_1~\ldots~A_{i-1}~\underline{A_j~A_{j+1}~\ldots~A_{k-1}}~\underline{A_i~A_{i+1}~\ldots~A_{j-1}}~A_{k}~\ldots~A_m)
\end{flalign*}
\end{definition}

\begin{definition}
Considerando a representação clássica, dada uma string $A$ de tamanho $m$, o rearranjo \emph{\textit{block interchange}} $\BI(i,j,k,l)$, com $1 \leq i \leq j < k \leq l \leq m$, troca as posições relativas dos segmentos $(A_i~\ldots~A_j)$ e $(A_k~\ldots~A_l)$. O efeito de um \textit{block interchange} é o mesmo para strings com ou sem sinais. Mostramos a seguir o efeito de $\BI(i,j,k,l)$ em $A$.
\begin{flalign*}
A \comp \BI(i,j,k,l) = (A_1~\ldots~A_{i-1}~\underline{A_k~\ldots~A_{l}}~A_{j+1}~\ldots~A_{k-1}~\underline{A_i~\ldots~A_{j}}~A_{l+1}~\ldots~A_m)
\end{flalign*}
\end{definition}

Note que uma transposição é um tipo especial de \textit{block interchange} em que os segmentos afetados são adjacentes.

Dados dois segmentos adjacentes $\sigma_1 = (A_i~A_{i+1}~\ldots~A_{j-1})$ e $\sigma_2 = (A_j~A_{j+1}~\ldots~A_{k-1})$, uma \emph{transposição inversa} $\transrev$ troca as posições relativas desses dois segmentos adjacentes e inverte os elementos de um dos dois segmentos afetados, sendo que essa operação também inverte todos os sinais dos elementos presentes no segmento invertido, quando aplicada a uma string com sinais.

\begin{definition}
Considerando a representação clássica, dada uma string $A$ de tamanho $m$, uma transposição inversa Tipo 1 $\transrev_1(i,j,k)$ e uma transposição inversa Tipo 2 $\transrev_2(i,j,k)$, com $1 \leq i < j < k \leq m+1$, afetam $A$ da seguinte forma, dependendo se $A$ é uma string com sinais (1) ou se $A$ é uma string sem sinais (2):
\begin{flalign*}
(1)~& A \comp \transrev_1(i,j,k) = (A_1~\ldots~A_{i-1}~\underline{A_j~A_{j+1}~\ldots~A_{k-1}}~\underline{{-A_{j-1}}~\ldots~{-A_{i+1}}~{-A_{i}}}~A_{k}~\ldots~A_m) \\
(2)~& A \comp \transrev_1(i,j,k) = (A_1~\ldots~A_{i-1}~\underline{A_j~A_{j+1}~\ldots~A_{k-1}}~\underline{A_{j-1}~\ldots~A_{i+1}~A_{i}}~A_{k}~\ldots~A_m) \\
(1)~& A \comp \transrev_2(i,j,k) = (A_1~\ldots~A_{i-1}~\underline{{-A_{k-1}}~\ldots~{-A_{j+1}}~{-A_j}}~\underline{A_i~A_{i+1}~\ldots~A_{j-1}}~A_{k}~\ldots~A_m) \\
(2)~& A \comp \transrev_2(i,j,k) = (A_1~\ldots~A_{i-1}~\underline{A_{k-1}~\ldots~A_{j+1}~A_j}~\underline{A_i~A_{i+1}~\ldots~A_{j-1}}~A_{k}~\ldots~A_m)
\end{flalign*}
\end{definition}

Dados dois segmentos adjacentes $\sigma_1 = (A_i~A_{i+1}~\ldots~A_{j-1})$ e $\sigma_2 = (A_j~A_{j+1}~\ldots~A_{k-1})$, uma \emph{revrev} $\revrev$ inverte esses dois segmentos adjacentes, sendo que essa operação também inverte todos os sinais dos elementos afetados, quando aplicada em uma string com sinais. As revrevs não trocam a posição relativa dos segmentos $\sigma_1$ e $\sigma_2$.

\begin{definition}
Considerando a representação clássica, dada uma string $A$ de tamanho $m$, uma \emph{revrev} $\revrev(i,j,k)$, com $1 \leq i < j < k \leq m+1$, afeta $A$ da seguinte forma, dependendo se $A$ é uma string com sinais (1) ou se $A$ é uma string sem sinais (2):
{\small
\begin{flalign*}
(1)~& A \comp \revrev(i,j,k) = (A_1~\ldots~A_{i-1}~\underline{{-A_{j-1}}~\ldots~{-A_{i+1}}~{-A_{i}}}~\underline{{-A_{k-1}}~\ldots~{-A_{j+1}}~{-A_j}}~A_{k}~\ldots~A_m) \\
(2)~& A \comp \revrev(i,j,k) = (A_1~\ldots~A_{i-1}~\underline{{A_{j-1}}~\ldots~{A_{i+1}}~{A_{i}}}~\underline{{A_{k-1}}~\ldots~{A_{j+1}}~{A_j}}~A_{k}~\ldots~A_m)
\end{flalign*}
}
\end{definition}

Nesta tese, estudamos também as inserções e deleções ({\it indels}). Essas operações afetam o alfabeto e o tamanho da string em que são aplicadas, sendo operações essenciais em problemas de rearranjos que comparam genomas desbalanceados, já que rearranjos conservativos não conseguem balancear os genomas comparados. As próximas definições apresentam formalmente os {\it indels} e seus respectivos efeitos em strings.

\begin{definition}
Considerando a representação clássica, dada uma string $A$ de tamanho $m$, uma \emph{inserção} $\insertion(i,\sigma)$ adiciona a string $\sigma$ após o $i$-ésimo elemento de $A$, onde $0 \leq i \leq m$ e $\sigma$ é uma string não vazia sem repetições. Mostramos a seguir o efeito de $\insertion(i,\sigma)$ em $A$.
\begin{flalign*}
A \comp \insertion(i,\sigma) = (A_1~\ldots~A_i~\underline{\sigma_1~\ldots~\sigma_{|\sigma|}}~A_{i+1}~\ldots~A_m)
\end{flalign*}
\end{definition}

Note que, em uma inserção $\insertion(i,\sigma)$ aplicada na string $A$, é necessário que $\sigma$ seja do mesmo tipo (com ou sem sinais) que a string $A$.

\begin{definition}
Considerando a representação clássica, dada uma string $A$ de tamanho $m$, uma \emph{deleção} $\deletion(i,j)$, com $1 \leq i \leq j \leq m$, remove o segmento $(A_i~\ldots~A_j)$. Mostramos a seguir o efeito de $\deletion(i,j)$ em $A$.
\begin{flalign*}
A \comp \deletion(i,j) = (A_1~\ldots~A_{i-1}~A_{j+1}~\ldots~A_m)
\end{flalign*}
\end{definition}

\subsection{Efeito dos Rearranjos de Genomas em Regiões Intergênicas}

Nesta seção, mostramos como os principais rearranjos de genomas (reversões, transposições e {\it indels}) afetam tanto a ordem relativa dos genes quanto as regiões intergênicas de um genoma.

\begin{definition}
Considerando a representação intergênica, dado um genoma $\G = (A, \breve{A})$ com $|A| = m$, uma \emph{reversão intergênica} $\rho^{(i,j)}_{(x,y)}$, com $1 \leq i \leq j \leq m$, $0 \leq x \leq \breve{A}_i$ e $0 \leq y \leq \breve{A}_{j+1}$, quebra as regiões intergênicas $\breve{A}_i$ em $(x, x')$ e $\breve{A}_{j+1}$ em $(y, y')$, onde $x' = \breve{A}_i - x$ e $y' = \breve{A}_{j+1} - y$, e inverte o segmento $(x'~A_i~\breve{A}_{i+1}~A_{i+1}~\ldots~\breve{A}_{j}~A_j~y)$. Além disso, caso $A$ seja uma string com sinais, todos os sinais dos elementos afetados da string $A$ são invertidos. 
Essa reversão transforma $\G$ em $\G \comp \rho^{(i,j)}_{(x,y)} = (A', \breve{A}')$, onde: 
\begin{flalign*}
A' = A \comp \rho^{(i,j)}_{(x,y)} &= (A_1~\ldots~A_{i-1}~\underline{{-A_j}~\ldots~{-A_i}}~A_{j+1}~\ldots~A_m)\\ 
\breve{A}' = \breve{A} \comp \rho^{(i,j)}_{(x,y)} &= (\breve{A}_1, \ldots, \breve{A}_{i-1}, x + y, \underline{\breve{A}_j, \ldots, \breve{A}_{i+1}}, x' + y', \breve{A}_{j+2}, \ldots, \breve{A}_{m+1}) 
\end{flalign*}

O efeito de uma reversão intergênica é similar quando $A$ é uma string sem sinais, sendo necessário apenas ignorar os sinais da string. Mostramos a seguir o efeito de $\rho^{(i,j)}_{(x,y)}$ em $\G = (A, \breve{A})$ quando $A$ é uma string com sinais. 
\begin{align*}
\begin{tikzpicture}[scale=1]
    \begin{scope}[>={Stealth[black]},
                  every edge/.style={draw=black}, every node/.style={inner sep=0pt, minimum size = 0pt}]
    \node[label=\phantom{}] (bi1) at (2.25, -0.5) {};
    \node[label=\phantom{}] (bi2) at (5.25, -0.5) {};
    \path [{Bar}-{Bar}] (bi1) edge node [black, pos=0.5, sloped, below] {} (bi2);
    \end{scope}
    \begin{scope}
        \node(g1) at (-1.6,0) {$\mathcal{G} = $};
        \node(i2) at (0.75,0) {$\ldots$};
        \node(i4) at (3.75,0) {$\ldots$};
    \node(i6) at (6.75,0) {$\ldots$};
    \end{scope}
    \begin{scope}[every node/.style={circle, draw=black, fill=white, text width=6.2mm, inner sep = 0pt, align=center}]
    \node(g1)[draw=black] at (0,0) {\scriptsize $A_{1}$};
    \node[draw=black](g2) at (1.5,0) {\scriptsize $A_{i-1}$};
    \node[draw=black](g3) at (3,0) {\scriptsize $A_{i}$};
    \node[draw=black](g4) at (4.5,0) {\scriptsize $A_{j}$};
    \node[draw=black](g5) at (6,0) {\scriptsize $A_{j+1}$};
    \node[draw=black](g6) at (7.5,0) {\scriptsize$A_m$};
    \end{scope}
    \begin{scope}[every node/.style={rectangle, draw=black, fill=white, text width=6.2mm, align=center}, scale=1]
    \scriptsize
    \node(i1) at (-0.75,0) {$\breve{A}_1$};
    \node[label=above:$\breve{A}_i$](i3) at (2.25,0) {$x|x'$};
    \node[label=above:$\breve{A}_{j+1}$](i5) at (5.25,0) {$y|y'$};
    \node[text width=6.7mm](i7) at (8.28,0) {$\breve{A}_{m+1}$};
    \end{scope}
\end{tikzpicture}\\
\begin{tikzpicture}[scale=1]
    \begin{scope}
        \node(g2) at (-2.2,0) {$\mathcal{G} \comp \rho^{(i,j)}_{(x,y)} = $};
        \node(i2) at (0.75,0) {$\ldots$};
        \node(i4) at (3.75,0) {$\ldots$};
    \node(i6) at (6.75,0) {$\ldots$};
    \end{scope}
    \begin{scope}[every node/.style={circle, draw=black, fill=white, text width=6.2mm, inner sep = 0pt, align=center}]
    \node(g1)[draw=black] at (0,0) {\scriptsize $A_{1}$};
    \node[draw=black](g2) at (1.5,0) {\scriptsize $A_{i-1}$};
    \node[draw=black](g3) at (3,0) {\scriptsize $-A_{j}$};
    \node[draw=black](g4) at (4.5,0) {\scriptsize $-A_{i}$};
    \node[draw=black](g5) at (6,0) {\scriptsize $A_{j+1}$};
    \node[draw=black](g6) at (7.5,0) {\scriptsize $A_m$};
    \end{scope}
    \begin{scope}[every node/.style={rectangle, draw=black, fill=white, text width=6.2mm, align=center}, scale=1]
    \scriptsize
    \node(i1) at (-0.75,0) {$\breve{A}_1$};
    \node(i3) at (2.25,0) {$x|y$};
    \node(i5) at (5.25,0) {$x'|y'$};
    \node[text width=6.7mm](i7) at (8.28,0) {$\breve{A}_{m+1}$};
    \end{scope}
\end{tikzpicture}
\end{align*}
\end{definition}

\begin{definition}
Considerando a representação intergênica, dado um genoma $\G = (A, \breve{A})$ com $|A| = m$, uma \emph{transposição intergênica} $\tau^{(i,j,k)}_{(x,y,z)}$, com $1 \leq i < j < k \leq m+1$, $0 \leq x \leq \breve{A}_i$, $0 \leq y \leq \breve{A}_j$ e $0 \leq z \leq \breve{A}_k$, quebra as regiões intergênicas $\breve{A}_i$ em $(x, x')$, $\breve{A}_j$ em $(y, y')$ e $\breve{A}_k$ em $(z, z')$, onde $x' = \breve{A}_i - x$, $y' = \breve{A}_j - y$ e $z' = \breve{A}_k - z$, e troca as posições relativas dos segmentos adjacentes $(x'~A_i~\breve{A}_{i+1}~A_{i+1}~\ldots~\breve{A}_{j-1}~A_{j-1}~y)$ e $(y'~A_j~\breve{A}_{j+1}~\ldots~\breve{A}_{k-1}~A_{k-1}~z)$. Essa transposição transforma $\G$ em $\G \comp \tau^{(i,j,k)}_{(x,y,z)} = (A', \breve{A}')$, onde:
\begin{flalign*}
\small
A' &= (A_1~\ldots~A_{i-1}~\underline{A_{j}~\ldots~A_{k-1}}~\underline{A_{i}~\ldots~A_{j-1}}~A_k~\ldots~A_m) \\
\breve{A}' &= (\breve{A}_1, \ldots, \breve{A}_{i-1}, x+y', \underline{\breve{A}_{j+1},\ldots,\breve{A}_{k-1}},z+x',\underline{\breve{A}_{i+1}, \ldots, \breve{A}_{j-1}}, y+z', \breve{A}_{k+1}, \ldots,\breve{A}_{m+1})
\end{flalign*}

Mostramos a seguir o efeito de $\tau^{(i,j,k)}_{(x,y,z)}$ em $\G = (A, \breve{A})$.
\begin{align*}
\begin{tikzpicture}[scale=1.05]
    \begin{scope}
        \node(g1) at (-1.6,0) {$\mathcal{G} = $};
        \node(dots1) at (0.75,0) {$\ldots$};
        \node(dots2) at (3.75,0) {$\ldots$};
        \node(dots3) at (6,0) {$\ldots$};
        \node(dots4) at (9.75,0) {$\ldots$};
    \node(i6) at (6.75,0) {$\ldots$};
    \end{scope}
    \begin{scope}[every node/.style={circle, draw=black, fill=white, text width=6.8mm, inner sep = 0pt, align=center}]
    \node(g1)[draw=black] at (0,0) {\scriptsize $A_{1}$};
    \node[draw=black](g2) at (1.5,0) {\scriptsize $A_{i-1}$};
    \node[draw=black](g3) at (3,0) {\scriptsize $A_{i}$};
    \node[draw=black](g4) at (4.5,0) {\scriptsize $A_{j-1}$};
    \node[draw=black](g5) at (6,0) {\scriptsize $A_{j}$};
    \node[draw=black](g6) at (7.5,0) {\scriptsize $A_{k{-}1}$};
    \node[draw=black](g7) at (9,0) {\scriptsize $A_{k}$};
    \node[draw=black](g8) at (10.5,0) {\scriptsize $A_{m}$};
    \end{scope}
    \begin{scope}[every node/.style={rectangle, draw=black, fill=white, text width=6.2mm, align=center}, scale=1]
    \scriptsize
    \node(i1) at (-0.75,0) {$\breve{A}_1$};
    \node[label=above:$\breve{A}_i$](i3) at (2.25,0) {$x|x'$};
    \node[label=above:$\breve{A}_j$](i5) at (5.25,0) {$y|y'$};
    \node[label=above:$\breve{A}_k$](i7) at (8.25,0) {$z|z'$};
    \node[text width=6.7mm](i8) at (11.28,0) {$\breve{A}_{m+1}$};
    \end{scope}
    \begin{scope}[>={Stealth[black]},
                  every edge/.style={draw=black}, every node/.style={inner sep=0pt, minimum size = 0pt}]
    \node[label=\phantom{}] (bi1) at (2.25, -0.5) {};
    \node[label=\phantom{}] (bi2) at (5.25, -0.5) {};
    \node[label=\phantom{}] (bi3) at (8.25, -0.5) {};
    \path [{Bar}-{Bar}] (bi1) edge node [black, pos=0.5, sloped, below] {} (bi2);
    \path [-{Bar}] (5.2, -0.5) edge node [black, pos=0.5, sloped, below] {} (bi3);
    \end{scope}
\end{tikzpicture}\\
\begin{tikzpicture}[scale=1.05]
    \begin{scope}
        \node(g2) at (-2.2,0) {$\mathcal{G} \comp \tau^{(i,j,k)}_{(x,y,z)} = $};
        \node(dots1) at (0.75,0) {$\ldots$};
        \node(dots2) at (3.75,0) {$\ldots$};
        \node(dots3) at (6,0) {$\ldots$};
        \node(dots4) at (9.75,0) {$\ldots$};
    \node(i6) at (6.75,0) {$\ldots$};
    \end{scope}
    \begin{scope}[every node/.style={circle, draw=black, fill=white, text width=6.8mm, inner sep = 0pt, align=center}]
    \node(g1)[draw=black] at (0,0) {\scriptsize $A_{1}$};
    \node[draw=black](g2) at (1.5,0) {\scriptsize $A_{i-1}$};
    \node[draw=black](g3) at (3,0) {\scriptsize $A_{j}$};
    \node[draw=black](g4) at (4.5,0) {\scriptsize $A_{k-1}$};
    \node[draw=black](g5) at (6,0) {\scriptsize $A_{i}$};
    \node[draw=black](g6) at (7.5,0) {\scriptsize $A_{j{-}1}$};
    \node[draw=black](g7) at (9,0) {\scriptsize $A_{k}$};
    \node[draw=black](g8) at (10.5,0) {\scriptsize $A_{m}$};
    \end{scope}
    \begin{scope}[every node/.style={rectangle, draw=black, fill=white, text width=6.2mm, align=center}, scale=1]
    \scriptsize
    \node(i1) at (-0.75,0) {$\breve{A}_1$};
    \node(i3) at (2.25,0) {$x|y'$};
    \node(i5) at (5.25,0) {$z|x'$};
    \node(i7) at (8.25,0) {$y|z'$};
    \node[text width=6.7mm](i8) at (11.28,0) {$\breve{A}_{m+1}$};
    \end{scope}
\end{tikzpicture}
\end{align*}
\end{definition}

\begin{definition}
Considerando a representação intergênica, dado um genoma $\G = (A, \breve{A})$ com $|A| = m$, uma \emph{inserção intergênica} $\phi_{(x)}^{(i, \sigma, \breve{\sigma})}$, tal que $0 \leq i \leq m$, $ 0 \leq x \leq \breve{A}_{i+1}$, $\sigma$ é uma string sem repetições e $\breve{\sigma}$ é uma lista de inteiros de tamanho $|\breve{\sigma}| = |\sigma| + 1$, adiciona o segmento $(\breve{\sigma}_1~\sigma_1~\ldots~\breve{\sigma}_{|\sigma|}~\sigma_{|\sigma|}~\breve{\sigma}_{|\sigma|+1})$ após o $x$-ésimo nucleotídeo da região intergênica $\breve{A}_{i+1}$. A inserção $\phi_{(x)}^{(i, \sigma, \breve{\sigma})}$ transforma $\G$ em $\G \comp \phi_{(x)}^{(i, \sigma, \breve{\sigma})} = (A', \breve{A}')$, onde:
\begin{flalign*}
A' &= A \cdot \phi^{(i,\sigma,\breve{\sigma})}_{x} = (A_1~\ldots~A_{i}~\underline{\sigma_1~\ldots~\sigma_{|\sigma|}}~A_{i+1}~\ldots~A_{m}) \\
\breve{A}' &= \breve{A} \cdot \phi^{(i,\sigma,\breve{\sigma})}_{x} = (\breve{A}_1,\ldots,\breve{A}_i,\underline{x+\breve{\sigma}_1,\breve{\sigma}_2,\ldots,\breve{\sigma}_{|\sigma|},\breve{\sigma}_{|\sigma|+1}+x'},\breve{A}_{i+2},\ldots,\breve{A}_{m+1}) \\
x' &= \breve{A}_{i+1} - x
\end{flalign*} 

Ao contrário de uma inserção considerando a representação clássica, em uma inserção intergênica a string $\sigma$ pode ser vazia. Nesse caso, a inserção intergênica $\phi_{(x)}^{(i, \sigma, \breve{\sigma})}$ altera apenas a região intergênica $\breve{A}_{i+1}$. Mostramos a seguir o efeito de $\phi_{(x)}^{(i, \sigma, \breve{\sigma})}$ em $\G = (A, \breve{A})$.
\begin{align*}
\begin{tikzpicture}[scale=1]
    \begin{scope}
        \node(g1) at (-3.53,0) {$\mathcal{G} = $};
        \node(s2) at (0.75,0) {$\ldots$};
        \node(s4) at (3.75,0) {$\ldots$};
        \node(phantom) at (7.7,0) {};
    \end{scope}
    \begin{scope}[every node/.style={circle, draw=black, fill=white, text width=6.2mm, inner sep = 0pt, align=center}]
    \node(g1)[draw=black] at (0,0) {\scriptsize $A_{1}$};
    \node[draw=black](g2) at (1.5,0) {\scriptsize $A_{i}$};
    \node[draw=black](g3) at (3,0) {\scriptsize $A_{i+1}$};
    \node[draw=black](g4) at (4.5,0) {\scriptsize $A_{m}$};
    \end{scope}
    \begin{scope}[every node/.style={rectangle, draw=black, fill=white, text width=6.2mm, align=center}, scale=1]
    \scriptsize
    \node(i1) at (-0.75,0) {$\breve{A}_1$};
    \node[label=above:$\breve{A}_{i+1}$](i3) at (2.25,0) {$x|x'$};
    \node[text width=6.7mm](i5) at (5.28,0) {$\breve{A}_{m+1}$};
    \end{scope}
\end{tikzpicture}\\
\begin{tikzpicture}[scale=1]
    \begin{scope}[>={Stealth[black]},
                  every edge/.style={draw=black}, every node/.style={inner sep=0pt, minimum size = 0pt}]
    \node[label=\phantom{}] (bi1) at (2.25, 0.5) {};
    \node[label=\phantom{}] (bi2) at (6.15, 0.5) {};
    \path [{Bar}-{Bar}] (bi1) edge node [black, pos=0.5, sloped, below] {} (bi2);
    \end{scope}
    \begin{scope}
        \node(g2) at (-2.3,0) {$\mathcal{G} \comp \phi_{(x)}^{(i, \sigma, \breve{\sigma})} = $};
        \node(i2) at (0.75,0) {$\ldots$};
        \node(i4) at (4.50,0) {$\ldots$};
        \node(i8) at (7.80,0) {$\ldots$};
    \end{scope}
    \begin{scope}[every node/.style={circle, draw=black, fill=white, text width=6.2mm, inner sep = 0pt, align=center}]
    \node(g1)[draw=black] at (0,0) {\scriptsize $A_{1}$};
    \node[draw=black](g2) at (1.5,0) {\scriptsize $A_{i}$};
    \node[draw=black](g3) at (3,0) {\scriptsize $\sigma_{1}$};
    \node[draw=black](g4) at (5.25,0) {\scriptsize $\sigma_{|\sigma|}$};
    \node[draw=black](g5) at (7.05,0) {\scriptsize $A_{i+1}$};
    \node[draw=black](g6) at (8.55,0) {\scriptsize $A_{m}$};
    \end{scope}
    \begin{scope}[every node/.style={rectangle, draw=black, fill=white, text width=6.2mm, align=center}, scale=1]
    \scriptsize
    \node(i1) at (-0.75,0) {$\breve{A}_1$};
    \node(i3) at (2.25,0) {$x|{\breve{\sigma}_1}$};
    \node(i5) at (3.75,0) {$\breve{\sigma}_{2}$};
    \node[text width=9mm](i6) at (6.15,0) {$\breve{\sigma}_{|\breve{\sigma}|}|{x'}$};
    \node[text width=6.7mm](i5) at (9.33,0) {$\breve{A}_{m+1}$};
    \end{scope}
\end{tikzpicture}
\end{align*}
\end{definition}

\begin{definition}
Considerando a representação intergênica, dado um genoma $\G = (A, \breve{A})$ com $|A| = m$, uma \emph{deleção intergênica} $\psi_{(x,y)}^{(i,j)}$, tal que $1 \leq i \leq j \leq m+1$, $0 \leq x \leq \breve{A}_i$ e $0 \leq y \leq \breve{A}_{j}$, remove o segmento que inicia após o $x$-ésimo nucleotídeo de $\breve{A}_i$ e termina no $y$-ésimo nucleotídeo de $\breve{A}_{j}$, transformando $\G$ em $\G \comp \psi_{(x,y)}^{(i,j)} = (A', \breve{A}')$, onde:
\begin{flalign*}
A' &= A \cdot \psi^{(i,j)}_{(x,y)} = (A_1~\ldots~A_{i-1}~A_j~\ldots~A_m)\\ 
\breve{A}' &= \breve{A} \cdot \psi^{(i,j)}_{(x,y)} = (\breve{A}_1,\ldots,\breve{A}_{i-1}, x + y', \breve{A}_{j+1}, \ldots, \breve{A}_{m+1})\\ 
y' &= \breve{A}_j - y
\end{flalign*}

Quando $i = j$, uma deleção $\psi_{(x,y)}^{(i,j)}$ não remove elementos de $A$ e apenas altera a região intergênica $\breve{A}_j$, portanto, essa operação deve atender a restrição $0 \leq x \leq y \leq \breve{A}_j$. Mostramos a seguir o efeito de $\psi_{(x,y)}^{(i,j)}$ em $\G = (A, \breve{A})$.
\begin{gather*}
\begin{tikzpicture}[scale=1]
    \begin{scope}[>={Stealth[black]},
                  every edge/.style={draw=black}, every node/.style={inner sep=0pt, minimum size = 0pt}]
    \node[label=\phantom{}] (bi1) at (2.25, -0.5) {};
    \node[label=\phantom{}] (bi2) at (5.25, -0.5) {};
    \path [{Bar}-{Bar}] (bi1) edge node [black, pos=0.5, sloped, below] {} (bi2);
    \end{scope}
    \begin{scope}
        \node(g1) at (-1.78,0) {$\mathcal{G} = $};
        \node(i2) at (0.75,0) {$\ldots$};
        \node(i4) at (3.75,0) {$\ldots$};
    \node(i6) at (6.75,0) {$\ldots$};
    \end{scope}
    \begin{scope}[every node/.style={circle, draw=black, fill=white, text width=6.2mm, inner sep = 0pt, align=center}]
    \node(g1)[draw=black] at (0,0) {\scriptsize $A_{1}$};
    \node[draw=black](g2) at (1.5,0) {\scriptsize $A_{i-1}$};
    \node[draw=black](g3) at (3,0) {\scriptsize $A_{i}$};
    \node[draw=black](g4) at (4.5,0) {\scriptsize $A_{j-1}$};
    \node[draw=black](g5) at (6,0) {\scriptsize $A_{j}$};
    \node[draw=black](g6) at (7.5,0) {\scriptsize$A_{m}$};
    \end{scope}
    \begin{scope}[every node/.style={rectangle, draw=black, fill=white, text width=6.2mm, align=center}, scale=1]
    \scriptsize
    \node(i1) at (-0.75,0) {$\breve{A}_1$};
    \node[label=above:$\breve{A}_i$](i3) at (2.25,0) {$x|x'$};
    \node[label=above:$\breve{A}_j$](i5) at (5.25,0) {$y|y'$};
    \node[text width=6.7mm](i7) at (8.28,0) {$\breve{A}_{m+1}$};
    \end{scope}
\end{tikzpicture}\\
\begin{tikzpicture}[scale=1]
    \begin{scope}
        \node(g2) at (-3.85,0) {$\mathcal{G} \comp \psi_{(x,y)}^{(i,j)} = $};
        \node(i2) at (0.75,0) {$\ldots$};
        \node(i6) at (3.75,0) {$\ldots$};
        \node(phantom) at (8.35,0) {};
    \end{scope}
    \begin{scope}[every node/.style={circle, draw=black, fill=white, text width=6.2mm, inner sep = 0pt, align=center}]
    \node(g1)[draw=black] at (0,0) {\scriptsize $A_{1}$};
    \node[draw=black](g2) at (1.5,0) {\scriptsize $A_{i-1}$};
    \node[draw=black](g5) at (3,0) {\scriptsize $A_{j}$};
    \node[draw=black](g6) at (4.5,0) {\scriptsize$A_{m}$};
    \end{scope}
    \begin{scope}[every node/.style={rectangle, draw=black, fill=white, text width=6.2mm, align=center}, scale=1]
    \scriptsize
    \node(i1) at (-0.75,0) {$\breve{A}_1$};
    \node(i3) at (2.25,0) {$x|{y'}$};
    \node[text width=6.7mm](i7) at (5.28,0) {$\breve{A}_{m+1}$};
    \end{scope}
\end{tikzpicture}
\end{gather*}
\end{definition}

Neste trabalho, seguimos as restrições apresentadas por Willing e coautores~\cite{2013-willing-etal}: dada uma instância clássica $\I = (A, \iota^n)$ ou uma instância intergênica $\Ig = ((A, \breve{A}), (\iota^n, \breve{\iota}^n))$, temos que um elemento $A_i$ só pode ser removido se $|A_i| \notin \Sigma_{\iota^n}$, e um elemento $x$ só pode ser inserido em $A$ se $x \in \Sigma_{\iota^n} \setminus \Sigma_{A}$. Ou seja, adicionamos restrições para prevenir que os {\it indels} removam e, posteriormente, insiram os mesmos genes (ou vice-versa).

\section{Problemas de Distância de Rearranjos}

Um \emph{modelo de rearranjos} $\M$ é o conjunto de operações permitidas em um problema de Distância de Rearranjos.

Em problemas de distância, podemos considerar que todas as operações utilizadas na solução possuem o mesmo custo (\emph{abordagem não ponderada}) ou que a aplicação de uma determinada operação possui um custo associado a ela (\emph{abordagem ponderada}). A não ser que seja dito expressamente o contrário, assumimos que os problemas de distância de rearranjos consideram uma abordagem não ponderada. A seguir, apresentamos esses problemas considerando cada tipo de instância.

Dados um modelo $\M$ e uma permutação $\pi$ de tamanho $n$, a \emph{distância de ordenação} $d_{\M}(\pi)$ é o número mínimo de rearranjos do modelo $\M$ necessários para ordenar $\pi$ (transformar $\pi$ em $\iota^n$), ou seja, temos $d_{\M}(\pi) = |S|$ tal que $\pi \comp S = \iota^n$, a sequência $S$ tem tamanho mínimo e $\beta \in \M$, para todo $\beta \in S$.

\begin{problembox}
\textbf{Ordenação de Permutações por Rearranjos}\\
\textbf{Entrada:} Uma instância clássica de genomas balanceados $\I = \pi$.\\
\textbf{Objetivo:} Considerando um modelo de rearranjos $\M$, encontrar uma sequência $S = (\beta_1, \beta_2, \ldots, \beta_{|S|})$ de tamanho mínimo (i.e., $|S| = d_{\M}(\pi)$) que ordena $\pi$ e é formada apenas de rearranjos pertencentes a $\M$.
\end{problembox}

Na abordagem ponderada, considerando uma função de custo $w : \M \rightarrow \mathbb{R}$, dizemos que um rearranjo $\beta$ possui custo $w(\beta)$. Para uma sequência $S = (\beta_1, \ldots, \beta_\ell)$, temos que $w(S) = \sum_{i = 1}^{\ell} w(\beta_i)$.

Dados um modelo $\M$, uma função de custo $w$ e uma permutação $\pi$, a distância de ordenação $d^w_{\M}(\pi)$ é o custo mínimo necessário para ordenar $\pi$ usando apenas operações de $\M$, ou seja, $d^w_{\M}(\pi) = w(S)$, com $S = (\beta_1, \ldots, \beta_\ell)$, tal que $\beta_i \in \M$, para todo $\beta_i \in S$, $\pi \comp S = \iota^n$ e o valor de $w(S)$ é mínimo.

\begin{problembox}
\textbf{Ordenação de Permutações por Rearranjos Ponderados}\\
\textbf{Entrada:} Uma instância clássica de genomas balanceados $\I = \pi$.\\
\textbf{Objetivo:} {\revisaof Considerando um modelo de rearranjos $\M$ e uma função de custo $w$, encontrar uma sequência $S$ de custo mínimo que ordena $\pi$ e é formada apenas de rearranjos pertencentes a $\M$, ou seja, $w(S) = d^w_{\M}(\pi)$ e $\beta \in \M$, para todo $\beta \in S$}.
\end{problembox}

Dados um modelo $\M$ e uma instância clássica de genomas desbalanceados $\I = (A, \iota^n)$, a distância de rearranjos $d_{\M}(A, \iota^n)$ (ou $d_{\M}(\I)$) é o número mínimo de rearranjos do modelo $\M$ necessários para transformar a string $A$ na string $\iota^n$, ou seja, temos $d_{\M}(A, \iota^n) = |S|$ tal que $A \comp S = \iota^n$, a sequência $S$ tem tamanho mínimo e $\beta \in \M$, para todo $\beta \in S$.

\begin{problembox}
\textbf{Distância de Rearranjos}\\
\textbf{Entrada:} Uma instância clássica de genomas desbalanceados $\I = (A, \iota^n)$.\\
\textbf{Objetivo:} Considerando um modelo de rearranjos $\M$, encontrar uma sequência $S = (\beta_1, \beta_2, \ldots, \beta_{|S|})$ de tamanho mínimo (i.e., $|S| = d_{\M}(A, \iota^n)$) que transforma $A$ em $\iota^n$ e é formada apenas de rearranjos pertencentes a $\M$.
\end{problembox}

Dados um modelo $\M$ e uma instância intergênica $\Ig = (\G_o, \G_d)$, com $\G_o = (A, \breve{A})$ e $\G_d = (\iota^n, \breve{\iota}^n)$, a distância de rearranjos intergênicos $d_{\M}(\G_o, \G_d)$ (ou $d_{\M}(\Ig)$) é o número mínimo de rearranjos do modelo $\M$ necessários para transformar o genoma origem $\G_o$ no genoma destino $\G_d$, ou seja, temos $d_{\M}(\G_o, \G_d) = |S|$ tal que $A \comp S = \iota^n$, $\breve{A} \comp S = \breve{\iota}^n$, a sequência $S$ tem tamanho mínimo e $\beta \in \M$, para todo $\beta \in S$.

\begin{problembox}
\textbf{Distância de Rearranjos Intergênicos}\\
\textbf{Entrada:} Uma instância intergênica $\Ig = (\G_o, \G_d)$, com $\G_o = (A, \breve{A})$ e $\G_d = (\iota^n, \breve{\iota}^n)$.\\
\textbf{Objetivo:} Considerando um modelo de rearranjos $\M$, encontrar uma sequência $S = (\beta_1, \beta_2, \ldots, \beta_{|S|})$ de tamanho mínimo (i.e., $|S| = d_{\M}(\G_o, \G_d)$) que transforma $\G_o$ em $\G_d$ e é formada apenas de rearranjos pertencentes a $\M$.
\end{problembox}

\section{Breakpoints}

O conceito de \emph{breakpoints} é bastante usado em problemas de distância de rearranjos, sendo útil na definição de limitantes e algoritmos, além de já ter sido utilizado em provas de NP-dificuldade~\cite{2019b-oliveira-etal,2012-bulteau-etal}. De forma geral, existe um {\it breakpoint} entre dois elementos consecutivos do genoma origem se esses dois elementos não são consecutivos no genoma destino. A definição de um {\it breakpoint} depende do modelo de rearranjos utilizado e da representação dos genomas. A seguir, apresentamos cada tipo de {\it breakpoint}.

\subsection{Breakpoints em Permutações}\label{cap2:sec:breakpoints_permutacao}

Nesta seção, apresentamos definições de {\it breakpoints} para uma instância clássica de genomas balanceados $\I = \pi$. 

\begin{definition}
Dada uma permutação $\sigma$ de tamanho $n$, obtemos a \emph{versão estendida} de $\sigma$ ao adicionar os elementos $\sigma_0 = +0$ e $\sigma_{n+1} = +(n+1)$. Caso $\sigma$ seja uma permutação sem sinais, temos que $\sigma_0 = 0$ e $\sigma_{n+1} = n+1$. Esses elementos são adicionados nas permutações apenas para facilitar algumas notações e, portanto, eles nunca são afetados por rearranjos de genomas.
\end{definition}

Nas próximas definições, assumimos que $\pi$ e $\iota^n$ estão nas suas versões estendidas. 

\begin{definition}\label{cap2:def:breakpoint_r}
Um \emph{breakpoint de reversões sem sinais} existe entre um par de elementos consecutivos $(\pi_i, \pi_{i+1})$ se $|\pi_{i+1} - \pi_i| \neq 1$, para $0 \leq i \leq n$.
\end{definition}

\begin{example}
	Para a permutação sem sinais $\pi = (0~4~3~5~1~2~6~7)$, que está na versão estendida, temos os seguintes {\it breakpoints} de reversões sem sinais (representados pelo símbolo $\circ$):
	\begin{align*}
		\pi = (0 \, \circ \, 4~3  \, \circ \, 5 \, \circ \, 1~2 \, \circ \, 6~7)
	\end{align*}
\end{example}

\begin{definition}\label{cap2:def:breakpoint_t}
Um \emph{breakpoint de transposições} existe entre um par de elementos consecutivos $(\pi_i, \pi_{i+1})$ se $\pi_{i+1} - \pi_i \neq 1$, para $0 \leq i \leq n$.
\end{definition}

\begin{example}
	Para a permutação sem sinais $\pi = (0~4~3~5~1~2~6~7)$, que está na versão estendida, temos os seguintes {\it breakpoints} de transposições (representados pelo símbolo $\circ$):
	\begin{align*}
		\pi = (0 \, \circ \, 4 \, \circ \, 3  \, \circ \, 5 \, \circ \, 1~2 \, \circ \, 6~7)
	\end{align*}
\end{example}

\begin{definition}\label{cap2:def:breakpoint_sr}
Um \emph{breakpoint de reversões com sinais} existe entre um par de elementos consecutivos $(\pi_i, \pi_{i+1})$ se $\pi_{i+1} - \pi_i \neq 1$, para $0 \leq i \leq n$.
\end{definition}

\begin{example}
  Para a permutação com sinais $\pi = (+0~{-4}~{-3}~{+5}~{+1}~{+2}~{-6}~{+7})$, que está na versão estendida, temos os seguintes {\it breakpoints} de reversões com sinais (representados pelo símbolo $\circ$):
  \begin{align*}
    \pi = (+0\, \circ \,{-4}~{-3}\, \circ \,{+5}\, \circ \,{+1}~{+2}\, \circ \,{-6}\, \circ \,{+7})
  \end{align*}
\end{example}

Agora, apresentamos qual tipo de {\it breakpoint} é usado em cada modelo de rearranjos. Dessa forma, quando o modelo de rearranjos está explícito, não precisamos indicar qual o tipo de {\it breakpoint} está sendo considerado.

Para permutações com sinais, todos os modelos considerados neste trabalho utilizam a definição de {\it breakpoint} de reversões com sinais e, assim, a permutação identidade com sinais é a única que não possui {\it breakpoints}.   

Note que a \emph{permutação inversa sem sinais} $\eta^n = (n~({n-1})~({n-2})~\ldots~2~1)$ possui apenas $2$ {\it breakpoints} de reversões sem sinais, enquanto essa mesma permutação possui $n+1$ {\it breakpoints} de transposições. Essa permutação pode ser transformada em $\iota^n$ usando apenas uma reversão ou apenas duas operações, considerando modelos que possuem transposições combinadas com transposições inversas ou revrevs. No entanto, para transformar $\eta^n$ em $\iota^n$ usando apenas transposições, precisamos de $\Theta(n)$ operações~\cite{meidanis2002lower}.

Sendo assim, para permutações sem sinais, os modelos que possuem alguma operação que causa inversão de segmento(s) (e.g., reversão, transposição inversa e revrev) utilizam a definição de {\it breakpoint} de reversões sem sinais. Por último, o modelo que possui apenas transposições ou \textit{block interchanges} utiliza a definição de {\it breakpoint} de transposições. Em ambos os casos, a permutação identidade sem sinais é a única que não possui {\it breakpoints}. 

Note que se um modelo de rearranjos possui apenas transposições ou \textit{block interchanges}, então apenas permutações sem sinais podem ser consideradas, já que ambas as operações não alteram sinais de elementos.

\begin{definition}
Dado um modelo $\M$, $b_{\M}(\pi)$ denota o número de {\it breakpoints} em $\pi$.
\end{definition}

\begin{definition}
Dado um modelo $\M$ e um rearranjo (ou sequência de rearranjos) $\beta$, $\Delta b_\M(\pi, \beta) = b_\M(\pi) - b_\M(\pi \comp \beta)$ denota a variação no número de {\it breakpoints} após $\beta$ ser aplicado em $\pi$.
\end{definition}

\begin{definition}
Considerando um tipo de {\it breakpoints}, uma \emph{strip} é uma sequência maximal de elementos sem {\it breakpoints} entre elementos consecutivos dessa sequência.
\end{definition}

Para permutações sem sinais, uma {\it strip} $(\pi_i~\pi_{i+1}~\ldots~\pi_j)$, com $0 \leq i < j \leq n + 1$, é \emph{crescente} se $\pi_{k+1} > \pi_{k}$, para todo $i \leq k < j$. Caso contrário, a {\it strip} é \emph{decrescente} e temos que $\pi_{k+1} < \pi_{k}$, para todo $i \leq k < j$. Um \emph{singleton} é uma {\it strip} de tamanho um. Um {\it singleton} é crescente se ele é igual a $(\pi_0)$ ou $(\pi_{n+1})$, caso contrário, ele é dito decrescente. Note que os elementos $\pi_0$ e $\pi_{n+1}$ sempre pertencem a {\it strips} crescentes. 

Para permutações com sinais, uma {\it strip} é classificada como \emph{positiva} se todos os elementos possuem sinal ``$+$''. Caso contrário, essa {\it strip} é classificada como \emph{negativa}. Note que pelas definições de breakpoints apresentadas nesta seção, todos os elementos de uma {\it strip} devem possuir o mesmo sinal. 

\subsection{Breakpoints em Genomas Desbalanceados}\label{cap2:sec:breakpoints_indels}

Nesta seção, apresentamos definições de {\it breakpoints} para uma instância clássica de genomas desbalanceados $\I = (A, \iota^n)$. 

\begin{definition}
Dada uma instância $(A, \iota^n)$, obtemos as \emph{versões estendidas} de $A$ e $\iota^n$ ao adicionar os elementos $A_0 = +0$, $A_{|A| + 1} = +(n+1)$, $\iota^n_0 = +0$ e $\iota^n_{n+1} = +(n+1)$. Caso as strings sejam sem sinais, então apenas desconsideramos os sinais dos novos elementos. Assim como em permutações estendidas, os elementos adicionados não são afetados por rearranjos. {\revisaof Além disso, não incluímos os elementos $0$ e $n+1$ nos alfabetos $\Sigma_{A}$ e $\Sigma_{\iota^n}$.}
\end{definition}

Nas próximas definições, assumimos que $A$ e $\iota^n$ estão nas suas versões estendidas. 

\begin{definition}
Dado um elemento $x$, seja $\anterior(x, \I)$ e $\posterior(x, \I)$ definidos como a seguir.
{\footnotesize
\begin{align*}
  \anterior(x, \I) =
  \begin{cases}
    \max(\{y \in \Sigma_{A} \cap \Sigma_{\iota^n} | y < x\}), & \text{if $min(\Sigma_{A} \cap \Sigma_{ \iota^n}) < x \leq n + 1$} \\
    0,& \text{if $x = min(\Sigma_{A} \cap \Sigma_{\iota^n})$} \\
    \deletionElement, &\text{if $x = \deletionElement$}
  \end{cases}
\end{align*}
\begin{align*}
  \posterior(x, \I) =
  \begin{cases}
    \min(\{y \in \Sigma_{A} \cap \Sigma_{ \iota^n} | y > x\}), & \text{if $0 \leq x < max(\Sigma_{A} \cap \Sigma_{ \iota^n})$} \\
    n + 1,& \text{if $x = max(\Sigma_{A} \cap \Sigma_{ \iota^n})$} \\
    \deletionElement, &\text{if $x = \deletionElement$}
  \end{cases}
\end{align*}
}
\end{definition}

\begin{example}
  Para a string estendida $A = (0~4~3~\deletionElement~1~2~6)$ e considerando $\iota^n$ com $n = 5$, temos os seguintes valores:
  \begin{align*}
    \anterior&(A_1 = 4, \I) = 3,
    &\qquad
    \posterior&(A_0 = 0, \I) = 1,\\
    \anterior&(A_2 = 3, \I) = 2,
    &\qquad
    \posterior&(A_1 = 4, \I) = 6, \\
    \anterior&(A_3 = \deletionElement, \I) = \deletionElement,
    &\qquad
    \posterior&(A_2 = 3, \I) = 4, \\
    \anterior&(A_4 = 1, \I) = 0,
    &\qquad
    \posterior&(A_3 = \deletionElement, \I) = \deletionElement, \\
    \anterior&(A_5 = 2, \I) = 1, &\qquad
    \posterior&(A_4 = 1, \I) = 2, \\
    \anterior&(A_6 = 6, \I) = 4,
    &\qquad
    \posterior&(A_5 = 2, \I) = 3.
  \end{align*}
\end{example}

\begin{definition}\label{cap2:def:breakpoint_indel_r}
	Considerando que $A$ e $\iota^n$ são strings sem sinais, para $0 \leq i \leq |A|$, um \emph{breakpoint de reversões sem sinais} existe entre um par de elementos consecutivos $(A_i, A_{i+1})$ se $A_{i+1} \neq \posterior(A_i, \I)$ e $A_{i+1} \neq \anterior(A_i, \I)$.
\end{definition}

\begin{example}
	Para $\I = (A, \iota^n)$ tal que $A = (0~4~3~\deletionElement~1~2~6~7~10)$ e $n = 9$, temos os seguintes {\it breakpoints} de reversões sem sinais (representados pelo símbolo $\circ$):
	\begin{align*}
		(0~\circ~4~3~\circ~\deletionElement~\circ~1~2~\circ~6~7~10).
	\end{align*}
\end{example}

\begin{definition}\label{cap2:def:breakpoint_indel_t}
	Considerando que $A$ e $\iota^n$ são strings sem sinais, para $0 \leq i \leq |A|$, um \emph{breakpoint de transposição} existe entre um par de elementos consecutivos $(A_i, A_{i+1})$ se $A_{i+1} \neq \posterior(A_i, \I)$.
\end{definition}

\begin{example}
  Para $\I = (A, \iota^n)$ tal que $A = (0~4~3~\deletionElement~1~2~6~7~10)$ e $n = 9$, temos os seguintes {\it breakpoints} de transposição (representados pelo símbolo $\circ$):
  \begin{align*}
    (0~\circ~4~\circ~3~\circ~\deletionElement~\circ~1~2~\circ~6~7~10).
  \end{align*}
\end{example}

As definições de {\it strips} são análogas às definições apresentadas para {\it breakpoints} em permutações. No entanto, existe o seguinte caso adicional na classificação de {\it strips}: uma {\it strip} formada apenas de elementos iguais a $\deletionElement$ é considerada uma {\it strip} crescente. 

\begin{definition}
	Dado um modelo $\M$, $b_\M(A, \iota^n)$ (ou $b_\M(\I)$) denota o número de {\it breakpoints} para a instância $\I = (A, \iota^n)$.
\end{definition}

\begin{definition}
	Dado um modelo $\M$ e um rearranjo (ou sequência de rearranjos) $\beta$, $\Delta b_\M(\I, \beta) = \Delta b_\M(A, \iota^n, \beta) = b_\M(A, \iota^n) - b_\M(A \comp \beta, \iota^n)$ denota a variação no número de {\it breakpoints} após $\beta$ ser aplicado em $A$.
\end{definition}

\subsection{Breakpoints Intergênicos}\label{cap2:sec:breakpoints_intergenicos}

O conceito de {\it breakpoints}, apresentado na Seção~\ref{cap2:sec:breakpoints_indels}, também é válido para uma instância intergênica $\Ig = (\G_o, \G_d)$, com $\G_o = (A, \breve{A})$ e $\G_d = (\iota^n, \breve{\iota}^n)$. No entanto, os {\it breakpoints} consideram apenas os elementos das strings $A$ e $\iota^n$, ignorando os tamanhos das regiões intergênicas. Dessa forma, apresentamos os {\it breakpoints} intergênicos, que consideram tanto os elementos das strings quanto os tamanhos das regiões intergênicas de pares de elementos adjacentes. Para {\it breakpoints} intergênicos, também consideramos que as strings $A$ e $\iota^n$ estão nas suas versões estendidas.

\begin{definition}
	Dada uma instância intergênica $\Ig = (\G_o, \G_d)$, com $\G_o = (A, \breve{A})$ e $\G_d = (\iota^n, \breve{\iota}^n)$, e um modelo de rearranjo $\M$, dizemos que $(A_i, A_{i+1})$ é um \emph{{\it breakpoint} intergênico}, para $0 \leq i \leq |A|$, se uma dessas duas condições é verdadeira: $(A_i, A_{i+1})$ é um {\it breakpoint}; ou $(A_i, A_{i+1})$ não é um {\it breakpoint}, $A_i \neq \deletionElement$, e $\breve{A}_{i+1} \neq \breve{\iota}^n_j$, onde $\breve{\iota}^n_j$ é a região intergênica do genoma de destino $\G_d$ entre os dois elementos de $\iota^n$ que correspondem aos elementos $A_i$ e $A_{i+1}$.
\end{definition}

Assim como na Seção~\ref{cap2:sec:breakpoints_indels}, consideramos que as strings $A$ e $\iota^n$ estão nas suas versões estendidas. A partir do modelo $\M$, podemos inferir qual a definição de {\it breakpoints} que será considerada ({\it breakpoints} de reversões sem sinais ou {\it breakpoints} de transposições). 

\begin{definition}
	Seja $(A_i, A_{i+1})$ um {\it breakpoint} intergênico, tal que $(A_i, A_{i+1})$ não é um {\it breakpoint} e $\breve{\iota}^n_j$ é a região intergênica com $j = \max(A_i, A_{i+1})$, ou seja, $\breve{\iota}^n_j$ é a região intergênica do genoma de destino $\G_d$ que está entre os dois elementos de $\iota^n$ que correspondem aos elementos $A_i$ e $A_{i+1}$. Dizemos que o {\it breakpoint} intergênico $(A_i, A_{i+1})$ é \emph{sobrecarregado} se $\breve{A}_{i+1} > \breve{\iota}^n_j$. Caso contrário, $(A_i, A_{i+1})$ é \emph{subcarregado}.
\end{definition}

\begin{example}\label{cap2:example:breakpoint_intergenico}
	Considere $\G_o = (A, \breve{A})$ e $\G_d = (\iota^n, \breve{\iota}^n)$, onde $n = 6$, $A$ é a string estendida $(0~5~\deletionElement~3~2~1~6~7)$, $\breve{A} = (2,15,10,8,10,18,15)$, e $\breve{\iota}^n = (2,10,15,9,11,14,10)$. Os valores entre parênteses representam os tamanhos das regiões intergênicas.
	\begin{align*}
		\G_o &= (0~(2)~5~(15)~\deletionElement~(10)~3~(8)~2~(10)~1~(18)~6~(15)~7)\\
		\G_d &= (0~(2)~1~(10)~2~(15)~3~(9)~4~(11)~5~(14)~6~(10)~7)
	\end{align*}
	Considerando {\it breakpoints} de reversão sem sinais, temos os seguintes {\it breakpoints} intergênicos: $(0,5)$, $(5, \deletionElement)$, $(\deletionElement, 3)$, $(3, 2)$, $(1,6)$, e $(6,7)$. O {\it breakpoint} intergênico $(3, 2)$ é subcarregado e o {\it breakpoint} $(6,7)$ é sobrecarregado.
\end{example}

\begin{definition}
	Para um modelo $\M$, o número de {\it breakpoints} intergênicos é denotado por $\bi_\M(\Ig) = \bi_\M(\G_o, \G_d)$. Dado um rearranjo (ou sequência de operações) $\beta$, denotamos por $\Delta \bi_\M(\Ig, \beta) = \Delta \bi_\M(\G_o, \G_d, \beta) = \bi_\M(\G_o, \G_d) - \bi_\M(\G_o \comp \beta, \G_d)$ a variação no número de {\it breakpoints} intergênicos causada por $\beta$. 
\end{definition}

As definições de {\it strips} são análogas às usadas nas seções anteriores.

\section{Grafo de Ciclos}

A estrutura grafo de ciclos, também chamada de grafo de {\it breakpoints}, foi inicialmente apresentada para o problema da Ordenação de Permutações com Sinais por Reversões~\cite{1996-bafna-pevzner}. O grafo de ciclos é utilizado para representar uma instância de um problema de Distância de Rearranjos usando uma única estrutura matemática. No entanto, em sua forma original, o grafo de ciclos pode ser usado apenas em permutações. Como parte deste trabalho, nós propomos duas adaptações chamadas de \emph{grafo de ciclos rotulado} e \emph{grafo de ciclos ponderado e rotulado}. O grafo de ciclos rotulado é usado para uma instância clássica de genomas balanceados ou desbalanceados, enquanto o grafo de ciclos ponderado e rotulado é usado para uma instância intergênica de genomas balanceados ou desbalanceados. A seguir, apresentamos esses três tipos de grafos.

\subsection{Grafo de Ciclos para Permutações}\label{cap2:sec:grafo_ciclos_permutacoes}

Para as próximas definições, consideramos que permutações estão na sua forma estendida. {\revisao O \emph{grafo de ciclos} de uma permutação $\pi$ é o grafo não direcionado $G(\pi) = (V, E_o \cup E_d)$, onde $V = \{+\pi_0, -\pi_1, +\pi_1, -\pi_2, +\pi_2, \ldots, -\pi_n, +\pi_n, -\pi_{n+1}\}$ é o conjunto de vértices, $E_o = \{(-\pi_i, +\pi_{i-1})~|~1 \leq i \leq n+1\}$ é o conjunto de \emph{arestas de origem}, e $E_d = \{(-\iota^n_i, +\iota^n_{i-1})~|~1 \leq i \leq n+1\}$ é o conjunto de \emph{arestas de destino}.}

Arestas de origem conectam vértices correspondentes a elementos adjacentes na permutação $\pi$, enquanto arestas de destino conectam vértices correspondentes a elementos adjacentes na permutação $\iota^n$. Essas arestas também são chamadas de arestas pretas (origem) e arestas cinzas (destino) na literatura de rearranjos de genomas~\cite{2009-fertin-etal}. 

Em um \emph{ciclo alternante}, arestas com uma extremidade em comum (arestas adjacentes) possuem tipos distintos (aresta de origem e aresta de destino). Todo vértice de $G(\pi)$ é extremidade de exatamente uma aresta de origem e uma aresta de destino. Dessa forma, existe uma decomposição única do grafo $G(\pi)$ em ciclos alternantes~\cite{1996-bafna-pevzner}. É importante ressaltar que essas definições de grafo de ciclos funcionam para todos os modelos estudados considerando permutações com sinais. No entanto, no caso de permutações sem sinais, as propriedades do grafo de ciclos apresentadas nesta seção só valem para os modelos que utilizam {\it breakpoints} de transposição (Definição~\ref{cap2:def:breakpoint_t}). Por exemplo, as propriedades apresentadas nesta seção não são válidas para reversões em permutações sem sinais.

Daqui em diante, dizemos que a aresta de origem $(-\pi_i, +\pi_{i-1})$ tem índice $i$ e é denotada por $e_i$. Além disso, a aresta de destino $(-\iota^n_i, +\iota^n_{i-1})$ tem índice $i$ e é denotada por $e'_i$. Normalmente, nos referimos a uma aresta de origem apenas pelo seu índice. 

Utilizamos a seguinte convenção presente na literatura para o desenho de $G(\pi)$: desenhamos os vértices em uma linha horizontal na ordem em que esses elementos aparecem em $\pi$, sempre colocando o vértice $-\pi_i$ à esquerda do vértice $+\pi_i$, {\revisao ou seja, desenhamos os vértices da esquerda para a direita seguindo a ordem da sequência $(+\pi_0, -\pi_1, +\pi_1, -\pi_2, +\pi_2, \ldots, -\pi_n, +\pi_n, -\pi_{n+1})$}; desenhamos arestas de origem como linhas horizontais; e desenhamos arestas de destino como arcos. Devido ao posicionamento dos vértices de $G(\pi)$, as arestas de origem são dispostas da esquerda para direita de forma que os seus índices formam a sequência $1, 2, \ldots, n, n+1$. 

Um $m$-ciclo é um ciclo com $m$ arestas de origem e $m$ arestas de destino. Além disso, dizemos que um $m$-ciclo possui \emph{tamanho} $m$. Um $1$-ciclo também é chamado de ciclo \emph{unitário} ou \emph{trivial}. Uma permutação $\pi$ é uma \emph{permutação simples} se todos os ciclos em $G(\pi)$ possuem tamanho menor ou igual a $3$. Também classificamos um $m$-ciclo de acordo com a paridade de $m$, sendo que um $m$-ciclo é par, quando $m$ é par, ou ímpar, quando $m$ é ímpar. 

Representamos um $m$-ciclo $C$ usando a lista de índices das suas arestas, sendo que essa lista é construída percorrendo as arestas do ciclo $C$, iniciando pelo vértice mais à direita de $C$ e percorrendo a aresta de origem incidente a esse vértice: $C = (o_1, d_1, o_2, d_2, \ldots,$ $o_m, d_m)$, tal que $o_1 > o_j$, para todo $1 < j \leq m$.
Na maioria dos casos, usamos uma representação simplificada de $C$ que possui apenas os índices das arestas de origem, já que as arestas de destino podem ser inferidas a partir dessa informação. 

As figuras~\ref{cap2:fig_exemplo_grafo} e \ref{cap2:fig_exemplo_grafo2} mostram exemplos de grafos de ciclos para uma permutação sem sinais e uma permutação com sinais, respectivamente. Também exemplificamos algumas notações utilizadas nessas figuras.

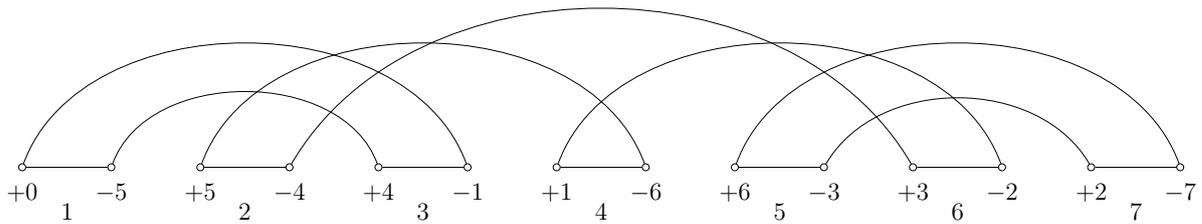
\begin{figure}
  \resizebox{\linewidth}{!}{
  \large
  \begin{tikzpicture}[scale=1.2]
    \begin{scope}[every node/.style={inner sep=0.5mm, draw, circle, minimum size = 0pt}]
      \node[label=below:$+0$] (p0) at (0,0) {};
      \node[label=below:$-5$] (m5) at (1.5,0) {};
      \node[label=below:$+5$] (p5) at (3,0) {};
      \node[label=below:$-4$] (m4) at (4.5,0) {};
      \node[label=below:$+4$] (p4) at (6,0) {};
      \node[label=below:$-1$] (m1) at (7.5,0) {};
      \node[label=below:$+1$] (p1) at (9,0) {};
      \node[label=below:$-6$] (m6) at (10.5,0) {};
      \node[label=below:$+6$] (p6) at (12,0) {};
      \node[label=below:$-3$] (m3) at (13.5,0) {};
      \node[label=below:$+3$] (p3) at (15,0) {};
      \node[label=below:$-2$] (m2) at (16.5,0) {};
      \node[label=below:$+2$] (p2) at (18,0) {};
      \node[label=below:$-7$] (m7) at (19.5,0) {};

    \end{scope}
    \begin{scope}[>={Stealth[black]},
                  every edge/.style={draw=black}]
      \path [-] (p0) edge (m5);
      \node[draw=none, fill=none, align=center, minimum width=1cm, text width=1cm] at (0.75, -0.75) {${1}$};
      \path [-] (p5) edge (m4);
      \node[draw=none, fill=none, align=center, minimum width=1cm, text width=1cm] at (3.75, -0.75) {$2$};
      \path [-] (p4) edge (m1);
      \node[draw=none, fill=none, align=center, minimum width=1cm, text width=1cm] at (6.75, -0.75) {${3}$};
      \path [-] (p1) edge (m6);
      \node[draw=none, fill=none, align=center, minimum width=1cm, text width=1cm] at (9.75, -0.75) {$4$};
      \path [-] (p6) edge (m3);
      \node[draw=none, fill=none, align=center, minimum width=1cm, text width=1cm] at (12.75, -0.75) {$5$};
      \path [-] (p3) edge (m2);
      \node[draw=none, fill=none, align=center, minimum width=1cm, text width=1cm] at (15.75, -0.75) {$6$};
      \path [-] (p2) edge (m7);
      \node[draw=none, fill=none, align=center, minimum width=1cm, text width=1cm] at (18.75, -0.75) {$7$};
    \end{scope}
    \begin{scope}[>={Stealth[black]},
                  every edge/.style={draw=black}]
      \path [-] (p0) edge [bend left=70] (m1);
      \path [-] (p1) edge [bend left=70] (m2);
      \path [-] (p2) edge [bend right=60] (m3);
      \path [-] (p3) edge [bend right=60] (m4);
      \path [-] (p4) edge [bend right=70] (m5);
      \path [-] (p5) edge [bend left=70] (m6);
      \path [-] (p6) edge [bend left=70] (m7);
    \end{scope}
  \end{tikzpicture}
  }
  \caption{\label{cap2:fig_exemplo_grafo}
  {\revisao
  Grafo de ciclos $G(\pi)$ da permutação sem sinais $\pi=(5~4~1~6~3~2)$. Linhas horizontais e arcos representam arestas de origem e arestas de destino, respectivamente. O índice de uma aresta de origem é indicado por um número abaixo dessa aresta. Neste exemplo, temos três ciclos em $G(\pi)$: $C_1 = (3, 1)$, $C_2 = (6, 2, 4)$ e $C_3 = (7,5)$. O ciclo $C_2$ é ímpar e os ciclos $C_1$ e $C_3$ são pares.
  }
  }
  
\end{figure}

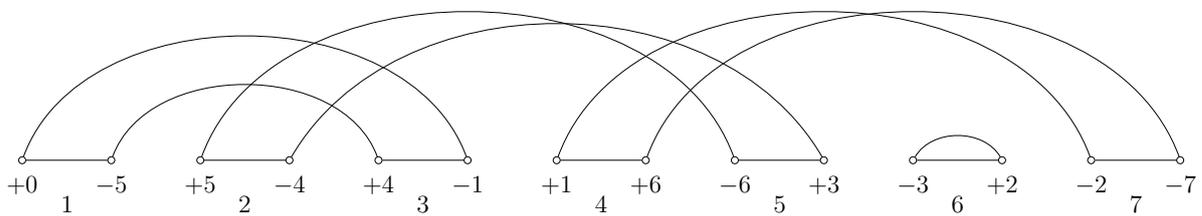
\begin{figure}
  \resizebox{\linewidth}{!}{
  \large
  \begin{tikzpicture}[scale=1.2]
    \begin{scope}[every node/.style={inner sep=0.5mm, draw, circle, minimum size = 0pt}]
      \node[label=below:$+0$] (p0) at (0,0) {};
      \node[label=below:$-5$] (m5) at (1.5,0) {};
      \node[label=below:$+5$] (p5) at (3,0) {};
      \node[label=below:$-4$] (m4) at (4.5,0) {};
      \node[label=below:$+4$] (p4) at (6,0) {};
      \node[label=below:$-1$] (m1) at (7.5,0) {};
      \node[label=below:$+1$] (p1) at (9,0) {};
      \node[label=below:$+6$] (p6) at (10.5,0) {};
      \node[label=below:$-6$] (m6) at (12,0) {};
      \node[label=below:$+3$] (p3) at (13.5,0) {};
      \node[label=below:$-3$] (m3) at (15,0) {};
      \node[label=below:$+2$] (p2) at (16.5,0) {};
      \node[label=below:$-2$] (m2) at (18,0) {};
      \node[label=below:$-7$] (m7) at (19.5,0) {};

    \end{scope}
    \begin{scope}[>={Stealth[black]},
                  every edge/.style={draw=black}]
      \path [-] (p0) edge (m5);
      \node[draw=none, fill=none, align=center, minimum width=1cm, text width=1cm] at (0.75, -0.75) {${1}$};
      \path [-] (p5) edge (m4);
      \node[draw=none, fill=none, align=center, minimum width=1cm, text width=1cm] at (3.75, -0.75) {$2$};
      \path [-] (p4) edge (m1);
      \node[draw=none, fill=none, align=center, minimum width=1cm, text width=1cm] at (6.75, -0.75) {${3}$};
      \path [-] (p1) edge (p6);
      \node[draw=none, fill=none, align=center, minimum width=1cm, text width=1cm] at (9.75, -0.75) {$4$};
      \path [-] (m6) edge (p3);
      \node[draw=none, fill=none, align=center, minimum width=1cm, text width=1cm] at (12.75, -0.75) {$5$};
      \path [-] (m3) edge (p2);
      \node[draw=none, fill=none, align=center, minimum width=1cm, text width=1cm] at (15.75, -0.75) {$6$};
      \path [-] (m2) edge (m7);
      \node[draw=none, fill=none, align=center, minimum width=1cm, text width=1cm] at (18.75, -0.75) {$7$};
    \end{scope}
    \begin{scope}[>={Stealth[black]},
                  every edge/.style={draw=black}]
      \path [-] (p0) edge [bend left=70] (m1);
      \path [-] (p1) edge [bend left=70] (m2);
      \path [-] (p2) edge [bend right=60] (m3);
      \path [-] (p3) edge [bend right=60] (m4);
      \path [-] (p4) edge [bend right=70] (m5);
      \path [-] (p5) edge [bend left=70] (m6);
      \path [-] (p6) edge [bend left=70] (m7);
    \end{scope}
  \end{tikzpicture}
  }
  \caption{Grafo de ciclos $G(\pi)$ da permutação com sinais $\pi = ({+5}~{+4}~{+1}~{-6}~{-3}~{-2})$. Neste exemplo, temos 4 ciclos em $G(\pi)$: $C_1 = (3,1)$, $C_2 = (5, 2)$, $C_3 = (6)$ e $C_4 = (7, 4)$.\label{cap2:fig_exemplo_grafo2}}
  
\end{figure}

O número de ciclos e o número de ciclos ímpares em $G(\pi)$ são denotados por $c(\pi)$ e $c_{odd}(\pi)$, respectivamente. Para um rearranjo (ou sequência de rearranjos) $\beta$, usamos $\Delta c(\pi, \beta)$ para denotar a variação no número de ciclos causado pelo rearranjo $\beta$, ou seja, $\Delta c(\pi, \beta) = c(\pi \cdot \beta) - c(\pi)$. De forma similar, temos $\Delta c_{odd}(\pi, \beta) = c_{odd}(\pi \cdot \beta) - c_{odd}(\pi)$. 

Dado um $m$-ciclo $C = (o_1, o_2, \ldots, o_m)$, classificamos $C$ como \emph{não orientado} se $o_1 > o_2 > \ldots > o_m$. Caso contrário, classificamos $C$ como \emph{orientado}. Na Figura~\ref{cap2:fig_exemplo_grafo}, o ciclo $C_2 = (6, 2, 4)$ é orientado enquanto os ciclos $C_1 = (3, 1)$ e $C_3 = (7,5)$ são não orientados. 

Dizemos que três arestas de origem com índices $o_i$, $o_j$ e $o_k$, com $i < j < k$, que pertencem ao mesmo ciclo $C$ formam uma \emph{tripla orientada} se pelo menos uma das condições seguintes é verdadeira: $o_i > o_k > o_j$; $o_j > o_i > o_k$, ou $o_k > o_j > o_i$. Bafna e Pevzner~\cite{1998-bafna-pevzner} provaram que todo ciclo orientado possui uma tripla orientada $o_i$, $o_j$ e $o_k$, com $i < j < k$, tal que $o_i > o_k > o_j$ e $k = j+1$. Além disso, eles demonstraram que uma transposição $\tau$ tem $\Delta c(\pi, \tau) = 2$ se, e somente se, essa transposição é aplicada em três arestas de origem que formam uma tripla orientada.

Dado um $m$-ciclo $C = (o_1, o_2, \ldots, o_m)$, uma aresta de origem $e_{o_i}$ é dita \emph{convergente} se ela é percorrida da direita para a esquerda. Caso contrário, $e_{o_i}$ é dita \emph{divergente}. Note que a aresta $e_{o_1}$ é sempre convergente seguindo a convenção de como o ciclo é percorrido quando listando os índices. Um par de arestas de origem $(o_i, o_j)$ é divergente se uma das arestas é convergente e a outra é divergente. Um ciclo $C$ é \emph{divergente} se pelo menos uma aresta de $C$ é divergente. Se todas arestas de $C$ são convergentes, então $C$ é \emph{convergente}.

Um grafo de ciclos $G(\pi)$ possui arestas divergentes se, e somente se, existe algum elemento com sinal ``$-$'' em $\pi$. Portanto, para permutações sem sinais, todos os ciclos de $G(\pi)$ são convergentes.

Agora, apresentamos algumas definições que serão usadas apenas em modelos que contém transposições.

Dizemos que uma transposição $\tau$ é uma $m$-transposição se $\Delta c_{odd}(\pi, \tau) = m$. Por exemplo, uma $2$-transposição é uma operação que aumenta o número de ciclos ímpares em $2$. Uma $(2,2)$-sequência é um par de $2$-transposições que podem ser aplicadas consecutivamente em $\pi$, ou seja, $(\tau, \tau')$ é uma $(2,2)$-sequência se $\Delta c_{odd}(\pi, \tau) = \Delta c_{odd}(\pi \cdot \tau, \tau') = 2$.

Dizemos que uma tripla orientada $(o_i, o_j, o_k)$, com $i < j < k$, é \emph{válida} se a transposição $\tau$ que age nessas arestas é uma $2$-transposição. Ou seja, uma transposição aplicada em uma tripla orientada válida aumenta tanto o número de ciclos quanto o número de ciclos ímpares em $2$ unidades.

\subsection{Grafo de Ciclos Rotulado}\label{cap2:sec:grafo_ciclos_strings}

Nesta seção, apresentamos uma adaptação do grafo de ciclos que torna possível o seu uso em instâncias clássicas de genomas desbalanceados. 

\begin{definition}
	{\revisaof Dada uma instância clássica $\I = (A, \iota^n)$, definimos as \emph{strings simplificadas} $\pi^A = (\pi^A_1~\ldots~\pi^A_{\nprime})$ e $\pi^\iota = (\pi^\iota_1~\ldots~\pi^\iota_{\nprime})$ como cópias de $A$ e $\iota^n$, respectivamente, mas removendo elementos que não pertencem ao conjunto $\Sigma_A \cap \Sigma_{\iota^n}$.}
\end{definition}

\begin{example}
{\revisao
Considerando a instância $\I = (A, \iota^n)$, com $A = (4~\deletionElement~2~5~\deletionElement~1)$ e $\iota^n = (1~2~3~4~5)$, temos $\pi^A = (4~2~5~1)$ e $\pi^{\iota} = (1~2~4~5)$, com $\nprime = 4$.
}
\end{example}

Para as próximas definições, consideramos uma instância clássica $\I = (A, \iota^n)$ e assumimos que as strings simplificadas $\pi^A$ e $\pi^\iota$ estão nas suas versões estendidas. Além disso, consideramos que $|\pi^A| = |\pi^\iota| = \nprime$ é a quantidade de elementos das strings simplificadas desconsiderando os elementos da versão estendida ($0$ e $n+1$), ou seja, $\pi^A = (0~\pi^A_1~\pi^A_2~\ldots~\pi^A_{\nprime}~{(n+1)})$ e $\pi^\iota = (0~\pi^\iota_1~\pi^\iota_2~\ldots~\pi^\iota_{\nprime}~{(n+1)})$.

O \emph{grafo de ciclos rotulado} é o grafo não direcionado $G(\I) = G(A, \iota^n) = (V, E_o \cup E_d, \ell)$, onde $V = \{+\pi^A_0, -\pi^A_1, +\pi^A_1, -\pi^A_2, +\pi^A_2, \ldots, -\pi^A_{\nprime}, +\pi^A_{\nprime}, -\pi^A_{\nprime+1} \}$ é o conjunto de vértices, $E_o$ é o conjunto de \emph{arestas de origem}, $E_d$ é o conjunto de \emph{arestas de destino}, e $\ell : (E_o \cup E_d) \rightarrow (\Sigma_{\iota^n} \setminus \Sigma_{A}) \cup \{\alpha, \emptyset\}$ é uma função que atribui um rótulo a cada aresta do grafo.

Arestas de origem conectam vértices que correspondem a elementos adjacentes em $\pi^A$, enquanto arestas de destino conectam vértices que correspondem a elementos adjacentes em $\pi^\iota$. O conjunto de arestas de origem é definido como $E_o = \{e_i = (+\pi^A_{i-1}, -\pi^A_{i}) : 1 \leq i \leq \nprime+1\}$, sendo que a aresta de origem $e_i = (+\pi^A_{i-1}, -\pi^A_{i})$ possui índice $i$. O rótulo $\ell(e_i) = \emptyset$ se $\pi^A_{i-1}$ e $\pi^A_{i}$ são consecutivos em $A$. Caso contrário, temos que $\ell(e_i) = \deletionElement$. O conjunto de arestas de destino é definido como $E_d = \{e'_i = (+\pi^{\iota}_{i-1}, -\pi^{\iota}_i): 1 \leq i \leq \nprime + 1\}$, sendo que a aresta de destino $e'_i = (+\pi^{\iota}_{i-1}, -\pi^{\iota}_i)$ possui índice $i$. O rótulo $\ell(e'_i) = \emptyset$ se $\pi^{\iota}_{i-1}$ e $\pi^{\iota}_{i}$ são consecutivos em $\iota^n$. Caso contrário, temos o rótulo $\ell(e'_i) = \pi^{\iota}_{i-1} + 1$.

Assim como o grafo de ciclos para permutações, no grafo de ciclos rotulado $G(\I)$ cada vértice de $G(\I)$ é extremidade de exatamente uma aresta de origem e uma aresta de destino. Portanto, existe uma decomposição única do grafo $G(\I)$ em ciclos alternantes. Todas as notações e definições relacionadas a ciclos para um grafo de ciclos (Seção~\ref{cap2:sec:grafo_ciclos_permutacoes}) se estendem para um grafo de ciclos rotulado. 

Para uma aresta de destino ou origem $e$, dizemos que essa aresta é \emph{limpa} se o seu rótulo é vazio (i.e., $\ell(e) = \emptyset$). Caso contrário, o seu rótulo é não vazio (i.e., $\ell(e) \neq \emptyset$) e dizemos que a aresta $e$ é \emph{rotulada}.

A forma como desenhamos o grafo é similar ao grafo de ciclos para permutações, sendo que representamos arestas rotuladas como linhas tracejadas e colocamos os rótulos correspondentes acima das arestas. Essa representação visual é apresentada na Figura~\ref{cap2:fig:exemplo_grafo_ciclos_rotulado}. 

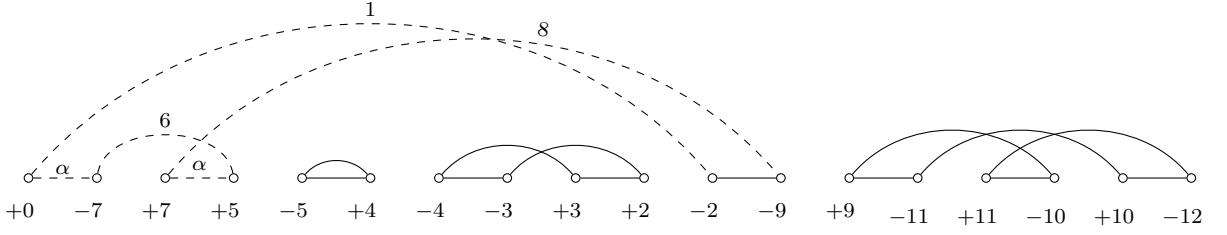
\begin{figure}[t]
\centering
\begin{tikzpicture}[scale=0.9]
\scriptsize
\begin{scope}[every node/.style={inner sep=0.4mm, draw, circle, minimum size = 0pt}]
    \node[label=below:$+0$\phantom{+}] (p0) at (3,0) {};
    \node[label=below:$-7$\phantom{+}] (m7) at (4,0) {};
    \node[label=below:$+7$\phantom{+}] (p7) at (5,0) {};
    \node[label=below:$+5$\phantom{+}] (p5) at (6,0) {};
    \node[label=below:$-5$\phantom{+}] (m5) at (7,0) {};
    \node[label=below:$+4$\phantom{+}] (p4) at (8,0) {};
    \node[label=below:$-4$\phantom{+}] (m4) at (9,0) {};
    \node[label=below:$-3$\phantom{+}] (m3) at (10,0) {};
    \node[label=below:$+3$\phantom{+}] (p3) at (11,0) {};
    \node[label=below:$+2$\phantom{+}] (p2) at (12,0) {};
    \node[label=below:$-2$\phantom{+}] (m2) at (13,0) {};
    \node[label=below:$-9$\phantom{+}] (m9) at (14,0) {};
    \node[label=below:$+9$\phantom{+}] (p9) at (15,0) {};
    \node[label=below:$-11$\phantom{+}] (m11) at (16,0) {};
    \node[label=below:$+11$\phantom{+}] (p11) at (17,0) {};
    \node[label=below:$-10$\phantom{+}] (m10) at (18,0) {};
    \node[label=below:$+10$\phantom{+}] (p10) at (19,0) {};
    \node[label=below:$-12$\phantom{+}] (m12) at (20,0) {};
\end{scope}

\begin{scope}[>={Stealth[black]},
              every edge/.style={draw=black}]
    \path [-] (m5) edge node [black, pos=0.5, sloped, above] {} (p4);
    \path [-] (m4) edge node [black, pos=0.5, sloped, above] {} (m3);
    \path [-] (p3) edge node [black, pos=0.5, sloped, above] {} (p2);
    \path [-] (m2) edge node [black, pos=0.5, sloped, above] {} (m9);

    \path [-] (p9) edge node [black, pos=0.5, sloped, above] {} (m11);
    \path [-] (p11) edge node [black, pos=0.5, sloped, above] {} (m10);
    \path [-] (p10) edge node [black, pos=0.5, sloped, above] {} (m12);
\end{scope}

\begin{scope}[>={Stealth[black]},
              dashed]
    \path [-] (p0) edge node [black, pos=0.5, sloped, above, yshift=-0.05cm] {$\deletionElement$} (m7);
    \path [-] (p7) edge node [black, pos=0.5, sloped, above] {$\deletionElement$} (p5);
\end{scope}

\begin{scope}[>={Stealth[black]},
              every edge/.style={draw=black}]
    \path [-] (p2) edge  [bend right=50] (m3);
    \path [-] (p3) edge  [bend right=50] (m4);
    \path [-] (p4) edge  [bend right=50] (m5);

    \path [-] (p9) edge  [bend left=50] (m10);
    \path [-] (p11) edge  [bend left=50] (m12);
    \path [-] (m11) edge  [bend left=50] (p10);
\end{scope}

\begin{scope}[>={Stealth[black]},
              dashed]
    \path [-] (p0) edge  [bend left=50] node [black, pos=0.5, sloped, above] {$1$} (m2);
    \path [-] (p5) edge  [bend right=80] node [black, pos=0.5, sloped, above] {$6$} (m7);
    \path [-] (p7) edge  [bend left=50] node [black, pos=0.6, sloped, above] {$8$} (m9);
\end{scope}

\end{tikzpicture}

\caption{\label{cap2:fig:exemplo_grafo_ciclos_rotulado}
Grafo de ciclos rotulado $G(\I) = G(A, \iota^n)$ para as strings $\iota^n$, com $n = 11$, e $A = (\deletionElement~{+7}~{\deletionElement}~{-5}~{-4}~{+3}~{-2}~{+9}~{+11}~{+10})$. Existem quatro ciclos nesse grafo. O ciclo $C_1 = (6, 1, 2)$ é um ciclo rotulado divergente. Todos os outros ciclos são ciclos limpos. O ciclo $C_2 = (3)$ é um ciclo unitário, o ciclo $C_3 = (5, 4)$ é um ciclo divergente, e o ciclo $C_4 = (9, 7, 8)$ é um ciclo orientado.
}
\end{figure}

{\revisaof
Um ciclo $C$ é \emph{limpo} se todas as suas arestas de origem são limpas, caso contrário, dizemos que $C$ é \emph{rotulado}. 

\begin{definition}
	O número de ciclos limpos em $G(\I)$ é denotado por $\cclean(A, \iota^n)$ ou $\cclean(\I)$. 
\end{definition}

\begin{definition}
	O número de ciclos rotulados em $G(\I)$ é denotado por $\crotulado(A, \iota^n)$ ou $\crotulado(\I)$. 
\end{definition}
}

{\revisaof Note que o grafo $G(A, \iota^n)$ possui $n+1$ ciclos unitários limpos se, e somente se, $A = \iota^n$. Como o número de arestas de origem pode ser alterado por inserções, precisamos de uma nova definição para $\Delta c(\I, \beta)$, que denota a variação na quantidade de ciclos relativo à quantidade de arestas de origem.}

\begin{definition}
	Dada uma operação (ou uma sequência de operações) $\beta$, definimos $\Delta c(\I, \beta) = \Delta c(A, \iota^n, \beta) = (|\pi^A| + 1 - c(A, \iota^n)) - (|\pi^A \comp \beta| + 1 - c(A \comp \beta, \iota^n))$.
\end{definition}

\begin{definition}
	Dada uma operação (ou uma sequência de operações) $\beta$, definimos $\Delta \cclean(\I, \beta) = \Delta \cclean(A, \iota^n, \beta) = (|\pi^A| + 1 - \cclean(A, \iota^n)) - (|\pi^A \comp \beta| + 1 - \cclean(A \comp \beta, \iota^n))$.
\end{definition}

\subsection{Grafo de Ciclos Rotulado e Ponderado}\label{cap2:sec:grafo_ciclos_intergenico}

\begin{figure*}[t]
\centering
\resizebox{\textwidth}{!}{
\begin{tikzpicture}[scale=1]
\scriptsize
\begin{scope}[every node/.style={inner sep=0.4mm, draw, circle, minimum size = 0pt}]
    \node[label=below:$+0$\phantom{+}] (p0) at (0,0) {};
    \node[label=below:$-4$\phantom{+}] (m4) at (1.5,0) {};
    \node[label=below:$+4$\phantom{+}] (p4) at (3,0) {};
    \node[label=below:$-3$\phantom{+}] (m3) at (4.5,0) {};
    \node[label=below:$+3$\phantom{+}] (p3) at (6,0) {};
    \node[label=below:$+1$\phantom{+}] (p1) at (7.5,0) {};
    \node[label=below:$-1$\phantom{+}] (m1) at (9,0) {};
    \node[label=below:$-2$\phantom{+}] (m2) at (10.5,0) {};
    \node[label=below:$+2$\phantom{+}] (p2) at (12,0) {};
    \node[label=below:$-6$\phantom{+}] (m6) at (13.5,0) {};
    \node[label=below:$+6$\phantom{+}] (p6) at (15,0) {};
    \node[label=below:$-8$\phantom{+}] (m8) at (16.5,0) {};
    \node[label=below:$+8$\phantom{+}] (p8) at (18,0) {};
    \node[label=below:$-7$\phantom{+}] (m7) at (19.5,0) {};
    \node[label=below:$+7$\phantom{+}] (p7) at (21,0) {};
    \node[label=below:$-9$\phantom{+}] (m9) at (22.5,0) {};
\end{scope}

\begin{scope}[>={Stealth[black]},
              every edge/.style={draw=black}]
    \path [-] (p0) edge node [black, pos=0.5, sloped, below, yshift=-0.00cm] {$w = 0$} (m4);
    \path [-] (p4) edge node [black, pos=0.5, sloped, below, yshift=-0.00cm] {$w = 3$} (m3);
    \path [-] (m1) edge node [black, pos=0.5, sloped, below, yshift=-0.00cm] {$w = 10$} (m2);
    
    \path [-] (p6) edge node [black, pos=0.5, sloped, below, yshift=-0.00cm] {$w = 10$} (m8);
    \path [-] (p8) edge node [black, pos=0.5, sloped, below, yshift=-0.00cm] {$w = 10$} (m7);
    \path [-] (p7) edge node [black, pos=0.5, sloped, below, yshift=-0.00cm] {$w = 10$} (m9);
\end{scope}

\begin{scope}[>={Stealth[black]},
              every edge/.style={draw=black}]
    \path [-] (p3) edge [dashed] node [black, pos=0.5, sloped, below, yshift=-0.0cm] {$w = 5$} node [black, pos=0.5, sloped, above, yshift=+0.05cm] {$\ell = \deletionElement$} (p1);
    \path [-] (p2) edge [dashed] node [black, pos=0.5, sloped, below, yshift=-0.0cm] {$w = 8$} node [black, pos=0.5, sloped, above, yshift=+0.05cm] {$\ell = \deletionElement$} (m6);
\end{scope}

\begin{scope}[>={Stealth[black]},
              every edge/.style={draw=black}]
    \path [-] (p0) edge [bend left=50] node [black, pos=0.5, sloped, below, yshift=-0.05cm] {$w=5$} (m1);
    \path [-] (p1) edge [bend left=60] node [black, pos=0.5, sloped, below, yshift=-0.05cm] {$w=2$}  (m2);
    \path [-] (m3) edge [bend left=50] node [black, pos=0.5, sloped, below, yshift=-0.05cm] {$w=7$}  (p2);
    \path [-] (m4) edge [bend left=50] node [black, pos=0.5, sloped, below, yshift=-0.05cm] {$w=1$}  (p3);
    
    \path [-] (p6) edge [bend left=50] node [black, pos=0.3, sloped, below, yshift=-0.05cm] {$w=15$}  (m7);
    \path [-] (p7) edge [bend right=70] node [black, pos=0.5, sloped, below, yshift=-0.05cm] {$w=5$}  (m8);
    \path [-] (p8) edge [bend left=50] node [black, pos=0.7, sloped, below, yshift=-0.05cm] {$w=10$}  (m9);
\end{scope}

\begin{scope}[>={Stealth[black]},
              every edge/.style={draw=black}]
    \path [-] (p4) edge [dashed, bend left=50] node [black, pos=0.5, sloped, below, yshift=-0.05cm] {$w=5$} node [black, pos=0.5, sloped, above] {$\ell = 5$} (m6);
\end{scope}

\end{tikzpicture}
}
\caption{\label{cap2:fig:labeled_weighted_cycle_graph}
Grafo de ciclos rotulado e ponderado para $\Ig = (\G_o, \G_d)$, com $\G_o = (A, \breve{A})$ e $\G_d = (\iota^n, \breve{\iota}^n)$, onde $n = 8$, $A = ({+0}~{+4}~{+3}~\deletionElement~{-1}~{+2}~\deletionElement~{+6}~{+8}~{+7}~{+9})$, $\breve{A} = (0,3,2,3,10,2,6,10,10,10)$, $\iota^n = ({+0}~{+1}~{+2}~{+3}~{+4}~{+5}~{+6}~{+7}~{+8}~{+9})$ e $\breve{\iota}^n = (5,2,7,1,2,3,15,5,10)$. Neste exemplo temos os ciclos $C_1 = (4,1,3)$, $C_2 = (5,2)$, $C_3 = (8,6,7)$. O ciclo $C_1$ é um ciclo divergente, rotulado e desbalanceado (a soma dos custos das arestas de origem é igual a $15$, enquanto a soma dos custos das arestas de destino é igual a $8$). O ciclo $C_2$ é um ciclo convergente, rotulado e desbalanceado. O ciclo $C_3$ é um ciclo convergente, balanceado e limpo.}
\end{figure*}

Nesta seção, apresentamos uma adaptação do grafo de ciclos que torna possível o seu uso em instâncias intergênicas de genomas desbalanceados. 

Para as próximas definições, consideramos uma instância intergênica $\Ig = (\G_o, \G_d)$, onde $\G_o = (A, \breve{A})$ e $\G_d = (\iota^n, \breve{\iota}^n)$. Além disso, assumimos que as strings simplificadas $\pi^A$ e $\pi^\iota$ estão nas suas versões estendidas. Consideramos que $|\pi^A| = |\pi^\iota| = \nprime$ é a quantidade de elementos das strings simplificadas desconsiderando os elementos da versão estendida ($0$ e $n+1$), ou seja, $\pi^A = (0~\pi^A_1~\pi^A_2~\ldots~\pi^A_{\nprime}~{(n+1)})$ e $\pi^\iota = (0~\pi^\iota_1~\pi^\iota_2~\ldots~\pi^\iota_{\nprime}~{(n+1)})$.

O \emph{grafo de ciclos rotulado e ponderado} é o grafo não direcionado $G(\Ig) = G(\G_o, \G_d) = (V, E_o \cup E_d, \ell, w)$, onde $V = \{+\pi^A_0, -\pi^A_1, +\pi^A_1, -\pi^A_2, +\pi^A_2, \ldots, -\pi^A_{\nprime}, +\pi^A_{\nprime}, -\pi^A_{\nprime+1} \}$ é o conjunto de vértices, $E_o$ é o conjunto de \emph{arestas de origem}, $E_d$ é o conjunto de \emph{arestas de destino}, $\ell : (E_o \cup E_d) \rightarrow (\Sigma_{\iota^n} \setminus \Sigma_{A}) \cup \{\alpha, \emptyset\}$ é uma função que atribui um rótulo a cada aresta do grafo, e $w : (E_o \cup E_d) \rightarrow \mathbb{Z}^*$ é uma função de custo que relaciona arestas a tamanhos de regiões intergênicas.

O conjunto de arestas de origem $E_o$, o conjunto de arestas de destino $E_d$ e a função $\ell$ são definidos de forma análoga ao grafo de ciclos rotulado $G(A, \iota^n)$. {\revisaof Além disso, todas as definições usadas para o grafo de ciclos rotulado $G(A, \iota^n)$ são válidas para o grafo de ciclos rotulado e ponderado $G(\Ig)$, exceto pelas definições de ciclos limpos e rotulados.}

{\revisaof
\begin{definition}
Um ciclo $C$ é \emph{limpo} se todas as suas arestas de origem e todas as suas arestas de destino são limpas, caso contrário, dizemos que $C$ é \emph{rotulado}. 
\end{definition}

\begin{definition}
	O número de ciclos limpos em $G(\Ig)$ é denotado por $\cclean(\Ig)$ e o número de ciclos rotulados em $G(\Ig)$ é denotado por $\crotulado(\Ig)$.
\end{definition}
}

A seguir, definimos a função de custo $w$ e as definições adicionais relacionadas ao custo das arestas do grafo. Para $1 \leq i \leq \nprime+1$, uma aresta de origem $e_i = (+\pi^A_{i-1}, -\pi^A_i)$ possui custo $w(e_i) = \sum_{k = {i'} + 1}^{j'} \breve{A}_k$, onde $A_{i'} = \pi^A_{i-1}$ e $A_{j'} = \pi^A_i$, ou seja, o custo de $e_i$ é igual à soma dos tamanhos das regiões intergênicas entre os elementos $\pi^A_{i-1}$ e $\pi^A_i$ em $\G_o$. Para $1 \leq i \leq \nprime + 1$, a aresta de destino $e'_i = (+\pi^{\iota}_{i-1}, -\pi^{\iota}_i)$ possui custo $w(e'_i) = \sum_{k = {i'} + 1}^{j'} \breve{\iota}^n_k$, onde $i' = \pi^{\iota}_{i-1}$ e $j' = \pi^{\iota}_i$, ou seja, o custo de $e'_i$ é igual à soma dos tamanhos das regiões intergênicas entre os elementos $\pi^{\iota}_{i-1}$ e $\pi^{\iota}_i$ em $\G_d$.

\begin{definition}
Um ciclo $C = (o_1, d_1, o_2, d_2, \ldots, o_m, d_m)$ é \emph{balanceado} se a soma dos custos das arestas de destino é igual à soma dos custos das arestas de origem, ou seja, $\sum_{i=1}^{m} w(e'_{d_i}) = \sum_{i=1}^{m} w(e_{o_i})$. Caso contrário, o ciclo $C$ é \emph{desbalanceado}.
\end{definition}

Um ciclo desbalanceado é classificado em \emph{positivo} ou \emph{negativo}. Um ciclo positivo $C = (o_1, d_1, o_2, d_2, \ldots, o_m, d_m)$ é um ciclo tal que a soma dos custos das arestas de destino é maior que a soma dos custos das arestas de origem, ou seja, $\sum_{i=1}^{m} w(e'_{d_i}) > \sum_{i=1}^{m} w(e_{o_i})$. Um ciclo negativo $C = (o_1, d_1, o_2, d_2, \ldots, o_m, d_m)$ é um ciclo tal que a soma dos custos das arestas de destino é menor que a soma dos custos das arestas de origem, ou seja, $\sum_{i=1}^{m} w(e'_{d_i}) < \sum_{i=1}^{m} w(e_{o_i})$. A Figura~\ref{cap2:fig:labeled_weighted_cycle_graph} mostra um exemplo de um grafo de ciclos rotulado e ponderado. 

\begin{definition}
	Dizemos que um ciclo $C$ é \emph{bom} se ele é balanceado e limpo. 
	Caso contrário, $C$ é desbalanceado ou rotulado, e dizemos que $C$ é um ciclo \emph{ruim}. 
	O número de ciclos bons em $G(\Ig)$ é denotado por $\cgood(\G_o, \G_d)$ ou $\cgood(\Ig)$.
\end{definition}

Note que o grafo $G(\Ig)$ possui apenas ciclos unitários bons se, e somente se, $\G_o = \G_d$. Em outra palavras, $\G_o = \G_d$ se, e somente se, $|\pi^A| + 1 - \cgood(\G_o, \G_d) = 0$.

\begin{definition}
	Dada uma operação (ou uma sequência de operações) $\beta$, a variação na quantidade de arestas de origem menos a quantidade de ciclos bons do grafo é denotada por $\Delta \cgood(\Ig, \beta) = \Delta \cgood(\G_o, \G_d, \beta) = (|\pi^A| + 1 - \cgood(\G_o, \G_d)) - (|\pi^{A'}| + 1 - \cgood(\G'_o, \G_d))$, onde $\G_o \comp \beta = \G'_o = (A', \breve{A}')$.
\end{definition}

\chapter{Ordenação de Permutações por Transposições e Outros Rearranjos}\label{cap:label:transposicao}

Muitos algoritmos de aproximação para problemas de distância de rearranjos foram desenvolvidos mesmo sendo que as complexidades desses problemas ainda estavam em aberto~\cite{2009-fertin-etal}. 

Em 1995, Kececioglu e Sankoff~\cite{1995-kececioglu-sankoff} apresentaram uma $2$-aproximação para o problema da Ordenação de Permutações sem Sinais por Reversões. Esses autores também apresentaram algoritmos exatos, porém esses algoritmos não possuem complexidade polinomial. Apenas em 1999, Caprara~\cite{1999a-caprara} provou que esse problema é NP-difícil. Até o momento, o melhor algoritmo de aproximação conhecido para esse problema possui fator de aproximação de $1.375$~\cite{2002-berman-etal}.

Em 1996, Bafna e Pevzner~\cite{1996-bafna-pevzner} também estudaram reversões, mas consideraram que a orientação dos genes era conhecida. Os autores apresentaram uma $1.5$-aproximação para a Ordenação de Permutações com Sinais por Reversões. Em 1999, Hannenhalli e Pevzner~\cite{1999-hannenhalli-pevzner} apresentaram um algoritmo polinomial exato para esse problema, sendo esse um dos trabalhos mais importantes na área de genômica comparativa. 

Em 1995, Bafna e Pevzner~\cite{1995b-bafna-pevzner} apresentaram uma $1.5$-aproximação para a Ordenação de Permutações (sem Sinais) por Transposições. Após isso, em 2006, Elias e Hartman~\cite{2006-elias-hartman} conseguiram melhorar o fator de aproximação, apresentando um algoritmo com fator de aproximação igual a $1.375$. Apenas em 2012, Bulteau e coautores~\cite{2012-bulteau-etal} apresentaram uma prova de que o problema é NP-difícil. 

A $1.375$-aproximação de Elias e Hartman~\cite{2006-elias-hartman} tem complexidade de tempo quadrática e depende de um processo que transforma a permutação de entrada em uma permutação simples. Após essa transformação, o próximo passo do algoritmo é aplicar uma $(2,2)$-sequência na nova permutação para garantir o fator de aproximação nos passos seguintes. No entanto, um estudo recente~\cite{2022-silva-etal} mostrou que essa busca deve ser feita antes do processo de simplificação da permutação, pois existem casos em que há uma $(2,2)$-sequência para a permutação de entrada $\pi$, mas não há uma $(2,2)$-sequência para a permutação simplificada gerada a partir de $\pi$ pelo algoritmo de Elias e Hartman~\cite{2006-elias-hartman}. Portanto, em alguns casos, o algoritmo de Elias e Hartman~\cite{2006-elias-hartman} falha em garantir a $1.375$-aproximação.

Silva e coautores~\cite{2022-silva-etal} apresentaram um novo algoritmo para a Ordenação de Permutações por Transposições que garante o fator de aproximação de $1.375$ em todos os casos usando uma abordagem algébrica. Esse algoritmo usa uma busca exaustiva para achar a $(2,2)$-sequência inicial e possui complexidade de tempo de $O(n^6)$. O algoritmo de Elias e Hartman~\cite{2006-elias-hartman} é quadrático porque faz a busca da $(2,2)$-sequência inicial, caso alguma exista, na permutação simplificada, que possui características que tornam esse procedimento mais simples. No entanto, para uma permutação qualquer, o melhor resultado conhecido até então para achar uma $(2,2)$-sequência, caso alguma exista, tem complexidade de tempo de $O(n^6)$~\cite{2022-silva-etal}. 

Na Seção~\ref{cap3:secao_novo_algoritmo} deste capítulo, apresentamos um algoritmo com complexidade de $O(n^5)$ para achar uma $(2,2)$-sequência, caso alguma exista, em uma permutação qualquer. Esse procedimento é essencial para obter um fator de aproximação igual a $1.375$ para a Ordenação de Permutações por Transposições. A partir dessa melhoria, nós propomos uma nova versão do algoritmo de $1.375$-aproximação com complexidade de tempo de $O(n^5)$.

Outro problema muito estudado na literatura é o da Ordenação de Permutações (com ou sem Sinais) por Reversões e Transposições. Em 1998, Walter e coautores~\cite{1998-walter-etal} apresentaram uma $2$-aproximação para a versão com sinais do problema. Após isso, Rahman e coautores~\cite{2008-rahman-etal} mostraram outro algoritmo que possui fator de aproximação igual a $2$, para permutações com sinais, e mostraram uma $2k$-aproximação, para permutações sem sinais, onde $k$ é o fator de aproximação para o problema da Decomposição Máxima de Ciclos Alternantes em um grafo de {\it breakpoints}. O melhor fator de aproximação conhecido para o problema da Decomposição Máxima de Ciclos Alternantes é de $17/12 + \epsilon$~\cite{2013-chen}, para qualquer $\epsilon$ positivo.

Em 2019, Oliveira e coautores~\cite{2019b-oliveira-etal} provaram que a Ordenação de Permutações com ou sem Sinais por Reversões e Transposições é NP-difícil. Além da abordagem não ponderada, Oliveira e coautores~\cite{2019b-oliveira-etal} também estudaram a versão ponderada do problema onde uma reversão tem custo $w_{\rho}$ e uma transposição tem custo $w_{\tau}$. Os autores provaram que a Ordenação de Permutações (com ou sem Sinais) por Reversões e Transposições Ponderadas é NP-difícil para quaisquer valores de $w_{\rho}$ e $w_{\tau}$ tal que $w_{\tau}/w_{\rho} \leq 1.5$.

Além dos modelos mais estudados de reversões, transposições e a combinação de ambas operações, também foram propostos algoritmos de aproximação e exatos para modelos que possuem transposições inversas e revrevs, apesar da complexidade do problema de distância de rearranjos com esses modelos ser desconhecida~\cite{2009-fertin-etal, hartmann2017genome}. Gu e coautores~\cite{1999-gu-etal} criaram uma $2$-aproximação para o problema da Ordenação de Permutações com Sinais por Reversões, Transposições e Transposições Inversas. Lou e Zhu~\cite{2008-lou-zhu} desenvolveram um algoritmo com fator de aproximação de $2.25$ para o mesmo modelo de rearranjos, mas considerando permutações sem sinais.

Lin e Xue~\cite{2001-lin-xue} estudaram o problema da Ordenação de Permutações com Sinais por Reversões, Transposições, Transposições Inversas e Revrevs e desenvolveram uma $1.75$-aproximação para esse problema. Já para o problema da Ordenação de Permutações com Sinais por Transposições, Transposições Inversas e Revrevs, o melhor resultado é um algoritmo com fator de aproximação de $1.5$~\cite{2005-hartman-sharan}. 

Para a abordagem ponderada, neste trabalho, consideramos uma função de custo baseada no tipo de rearranjo. O custo de uma reversão é indicado por $w_{\rho}$, {\revisaof enquanto os custos de uma transposição, transposição inversa ou revrev são indicados simplesmente por $w_{\tau}$}. Utilizamos o mesmo custo $w_{\tau}$ para essas últimas três operações devido ao fato de que todas elas afetam três adjacências do genoma, enquanto as reversões afetam apenas duas adjacências. 

As reversões são os rearranjos mais observados em muitos cenários evolutivos~\cite{1996-blanchette-etal}. Por esse motivo, o peso de uma reversão ($w_{\rho}$) tende a ser menor do que o peso de uma transposição ou rearranjos similares ($w_{\tau}$). Bader e coautores~\cite{2008-bader-etal} estudaram o problema da Ordenação de Permutações com Sinais por Reversões, Transposições e Transposições Inversas Ponderadas e apresentaram uma $1.5$-aproximação para valores de $w_{\rho}$ e $w_{\tau}$ tal que $1 \leq w_{\tau}/w_{\rho} \leq 2$. Eriksen~\cite{2002-eriksen} já havia considerado o mesmo problema e apresentado uma $7/6$-aproximação, mas esse fator de aproximação só é garantido quando $w_{\tau}/w_{\rho} = 2$.

Na Seção~\ref{cap3:secao_complexidade} deste capítulo, considerando que $w_{\tau}/w_{\rho} \leq 1.5$, provamos que os problemas de Ordenação de Permutações com ou sem Sinais por Rearranjos são NP-difíceis para modelos que possuem transposições juntamente com um ou mais dos seguintes rearranjos: reversões, transposições inversas e revrevs.

\section{Uma $1.375$-Aproximação Mais Eficiente para Transposições}\label{cap3:secao_novo_algoritmo}

{\revisaof 
Nesta seção, apresentamos um novo algoritmo de $1.375$-aproximação para a Ordenação de Permutações por Transposições com complexidade de tempo de $O(n^5)$. 

A $1.375$-aproximação de Elias e Hartman~\cite{2006-elias-hartman} tem complexidade de tempo quadrática, mas possui um erro em um dos passos do algoritmo, a busca de uma $(2,2)$-sequência na permutação inicial, que faz com que o algoritmo não garanta o fator de aproximação de $1.375$ para todas as instâncias~\cite{2022-silva-etal}. Silva e coautores~\cite{2022-silva-etal} mostraram um novo algoritmo de $1.375$-aproximação que corrige esse problema e possui complexidade de $O(n^6)$. O nosso algoritmo é uma modificação do algoritmo de Elias e Hartman~\cite{2006-elias-hartman} que garante uma melhoria na complexidade de tempo em relação ao algoritmo proposto por Silva e coautores~\cite{2022-silva-etal}. Essa melhoria é alcançada pelo uso de um novo procedimento para a busca de $2$-transposições em uma permutação qualquer. Note que um algoritmo de busca exaustiva para a busca de uma $2$-transposição, caso exista, tem complexidade de tempo de $O(n^3)$, já que precisamos testar todas as combinações de triplas $(i,j,k)$, com $1 \leq i < j < k \leq n+1$. Nesta seção, apresentamos um algoritmo para a busca de uma $2$-transposição, caso exista, com complexidade de tempo de $O(n^2)$.}

A seguir, listamos resultados conhecidos na literatura sobre a variação no número de ciclos ímpares causada por uma transposição e, também, sobre um limitante para a distância $d_{\tau}(\pi)$ relacionado com a quantidade de ciclos ímpares em $G(\pi)$.

\begin{lemma}[Bafna e Pevzner~\cite{1998-bafna-pevzner}, Lema 2.3]
  Para qualquer permutação $\pi$ e transposição $\tau$, temos que $\Delta c_{odd}(\pi, \tau) = \{-2,0,2\}$.
\end{lemma}

\begin{lemma}[Bafna e Pevzner~\cite{1998-bafna-pevzner}, Teorema 2.4]
  Para qualquer permutação $\pi$, temos que $d_{\tau}(\pi) \geq \frac{n+1-c_{odd}(\pi)}{2}$.
\end{lemma}

Agora, apresentamos propriedades sobre as $2$-transposições em ciclos pares e ímpares. O próximo lema caracteriza os tipos de ciclos afetados por uma $2$-transposição.

\begin{lemma}[Christie~\cite{1998b-christie}, Lemas 3.2.5 e 3.3.1]\label{cap3:lemma:types_of_2_transpositions}
  Se existe uma $2$-transposição aplicada nas arestas de origem $o_i, o_j,$ e $o_k$, então essas arestas pertencem a dois ciclos pares ou todas as três arestas pertencem ao mesmo ciclo orientado.
\end{lemma}

Quando existem pelo menos dois ciclos pares no grafo, Christie~\cite{1998b-christie} provou que é possível encontrar uma $2$-transposição em tempo linear. Portanto, nos próximos lemas, focamos apenas em $2$-transposições que afetam ciclos orientados, sendo que dividimos a análise em três casos para ciclos orientados ímpares.

\begin{lemma}\label{cap3:lemma:2-transposition-even-cycle}
Se existe um ciclo orientado $C$ em $G(\pi)$ que é um ciclo par, então existe uma tripla orientada válida $(o_i, o_j, o_k)$ em $C = (o_1, o_2, \ldots, o_m)$, com $i < j < k$, tal que $k = j+1$.
\end{lemma}

\begin{proof}
Bafna e Pevzner~\cite[Lema 2.3]{1998-bafna-pevzner} mostraram que todo ciclo orientado $C$ possui tripla orientada $(o_i, o_j, o_k)$, com $i < j < k$, tal que $o_i > o_k > o_j$ e $k = j + 1$. Uma transposição aplicada nessas arestas de origem transforma $C$ em três ciclos $C_1$, $C_2$, $C_3$, tal que pelo menos um deles é um ciclo unitário, sendo que esse ciclo unitário possui a aresta de destino $d_j$ que é adjacente às arestas de origem $o_j$ e $o_k$. Suponha, sem perda de generalidade, que $C_1$ é o ciclo unitário que contém a aresta de destino $d_j$. Como $C$ é um ciclo par, sabemos que dentre os ciclos $C_2$ e $C_3$ temos um ciclo par e um ciclo ímpar. Portanto, essa transposição adiciona dois ciclos ímpares no grafo de ciclos e é uma $2$-transposição, o que implica que $(o_i, o_j, o_k)$ é uma tripla orientada válida.
\end{proof}

\begin{lemma}\label{cap3:lemma:2-transposition-odd-cycle-case1}
  Se existe uma tripla orientada válida $(o_i, o_j, o_k)$ em um ciclo ímpar $C = (o_1, o_2, \ldots, o_m)$, tal que $i < j < k$ e $o_i > o_k > o_j$, então existe uma tripla orientada válida $(o_{i'}, o_{j'}, o_{k'})$ em $C$, com $i' < j' < k'$, tal que $i' \in \{1,2\}$ ou $k' = j' + 1$.
\end{lemma}

\begin{proof}
  Note que uma $2$-transposição que afeta apenas um ciclo também aumenta o número de ciclos no grafo em duas unidades. A $2$-transposição que age nas arestas $(o_i, o_j, o_k)$ transforma o ciclo $C$ em três ciclos $D$, $D'$ e $D''$, tal que $D$ possui as arestas de destino do caminho que vai de $o_i$ até $o_j$, $D'$ possui as arestas de destino do caminho que vai de $o_j$ até $o_k$, e $D''$ possui as arestas de destino do caminho que vai de $o_k$ até $o_i$. Portanto, o tamanho desses ciclos são $|D| = j - i$, $|D'| = k-j$, $|D''| = |C| + i - k$. Como essa $2$-transposição afeta o ciclo ímpar $C$, temos que $|D|, |D'|,$ e $|D''|$ são valores ímpares.

  Quando $i \in \{1,2\}$ ou $k = j+1$, temos que a própria tripla $(o_i, o_j, o_k)$ já atende as condições descritas no enunciado do lema. Caso contrário, temos que $i \geq 3$ e $k \geq j+3$. A seguir, dividimos a prova em casos dependendo da paridade dos valores de $i$ ou $k$. Para cada caso, mostramos que existe uma tripla orientada válida que atende as condições do lema.

  Se $i$ é ímpar, então $j$ deve ser par e $k$ deve ser ímpar. Lembre-se que pela nossa definição de como as arestas de origem de um ciclo são listadas, temos que $o_1 > o_i$ para qualquer ciclo. Portanto, ($o_1$, $o_j$, $o_k$) é uma tripla orientada válida, já que $o_1 > o_i$ e os valores $j-1$, $k-j$, e $|C| + 1 - k$ são todos ímpares. 

  Se $i$ é par, então $j$ deve ser ímpar e $k$ deve ser par. Se $o_2 > o_k$, então ($o_2$, $o_j$, $o_k$) é uma tripla orientada válida, já que $j-2$, $k-j$, e $|C|+2-k$ são valores ímpares. Se $o_2 < o_k$, ainda dividimos a prova nos seguintes casos.

  \begin{itemize}
      \item Se $o_{k-1} > o_k > o_2$, então $(o_1, o_2, o_{k-1})$ é uma tripla orientada válida, pois $o_1 > o_{k-1} > o_2$ e, já que $k$ é par, temos que os valores $1$, $(k-1)-2$ e $|C|+1-(k-1)$, que correspondem aos tamanhos dos ciclos criados por uma transposição que age nessa tripla, são ímpares.
      \item Se $o_{k} > o_{k-1}$, então $(o_i, o_{k-1}, o_{k})$ é uma tripla orientada válida, pois $o_i > o_{k} > o_{k-1}$ e os valores $(k-1)-i$, $1$ e $|C|+i-k$, que correspondem aos tamanhos dos ciclos criados por uma transposição que age nessa tripla, são todos ímpares.
  \end{itemize}
\end{proof}

\begin{lemma}\label{cap3:lemma:2-transposition-odd-cycle-case2}
  Se existe uma tripla orientada válida $(o_i, o_j, o_k)$ em um ciclo ímpar $C = (o_1, o_2, \ldots, o_m)$, tal que $i < j < k$ e $o_k > o_j > o_i$, então existe uma tripla orientada válida $(o_{i'}, o_{j'}, o_{k'})$ em $C$, com $i' < j' < k'$, tal que $i' = 1$.
\end{lemma}

\begin{proof}
  Se $i = 1$, então a tripla $(o_i, o_j, o_k)$ já atende as condições descritas no enunciado do lema. Caso contrário, temos que $i \geq 2$.

  Se $j$ é par, então $k$ é ímpar e $(o_1, o_j, o_k)$ é uma tripla orientada válida, pois $o_1 > o_k > o_j$ e uma transposição aplicada nessa tripla cria ciclos com tamanhos $j-1$, $k-j$ e $|C| + 1 - k$, que são todos ímpares. Caso contrário, temos que $j$ é ímpar, $i$ é par e $(o_1, o_i, o_j)$ é uma tripla orientada válida, pois $o_1 > o_j > o_i$ e uma transposição aplicada nessa tripla cria ciclos com tamanhos $i-1$, $j-i$, $|C| + 1 - j$, que são todos ímpares.
\end{proof}

\begin{lemma}\label{cap3:lemma:2-transposition-odd-cycle-case3}
  Se existe uma tripla orientada válida $(o_i, o_j, o_k)$ em um ciclo ímpar $C = (o_1, o_2, \ldots, o_m)$, tal que $i < j < k$ e $o_j > o_i > o_k$, então existe uma tripla orientada válida $(o_{i'}, o_{j'}, o_{k'})$ em $C$, com $i' < j' < k'$, tal que pelo menos uma dessas condições é verdadeira: $i' \in \{1,2\}$; $j' = i' + 1$; ou $k' = j' + 1$.
\end{lemma}

\begin{proof}
  Note que se $i \in \{1,2\}$, ou $j = i+1$, ou $k = j+1$, então a tripla $(o_i, o_j, o_k)$ já atende as condições do enunciado do lema. Caso contrário, temos que $i \geq 3$, $j \geq i+3$ e $k \geq j+3$.

  Se $j$ é ímpar, então $i$ e $k$ são pares. Consequentemente, temos a tripla orientada válida $(o_1, o_i, o_j)$, pois $o_1 > o_j > o_i$ e os valores $i-1$, $j-i$ e $|C| + 1 - j$ são ímpares.

  Como temos muitos casos quando $j$ é par, iremos apenas listar as condições e a tripla orientada válida correspondente que atende as condições do enunciado do lema, sendo fácil conferir que ela é de fato uma tripla orientada válida. De agora em diante, assuma que $j$ é par e ambos os valores de $i$ e $k$ são ímpares.  

  Primeiramente, consideramos o valor da aresta de origem $o_2$:
  \begin{itemize}
    \item Se $o_2 > o_j$, então $(o_2, o_i, o_j)$ é uma tripla orientada válida.
    \item Se $o_2 < o_i$, então $(o_1, o_2, o_i)$ é uma tripla orientada válida.
  \end{itemize}

  Caso as condições acima sejam ambas falsas, temos que $o_j > o_2 > o_i$. Agora, considere a aresta de origem $o_{i+1}$, que é a aresta de origem subsequente a $o_i$ no ciclo $C$. Lembre-se que $i+1 < j$ e ambos $i+1$ e $j$ são pares.
  \begin{itemize}
      \item Se $o_{i+1} < o_k$, então $(o_1, o_{i+1}, o_k)$ é uma tripla orientada válida.
      \item Se $o_{i+1} > o_2$, então $(o_i, o_{i+1}, o_k)$ é uma tripla orientada válida.
      \item Se $o_2 > o_{i+1} > o_i$, então $(o_2, o_i, o_{i+1})$ é uma tripla orientada válida.
  \end{itemize}

  Caso as condições anteriores sejam todas falsas, temos que $o_1 > o_j > o_2 > o_i > o_{i+1} > o_k$. Agora, considere a aresta de origem $o_{j+1}$, a aresta de origem subsequente a $o_j$ no ciclo $C$. Lembre-se que $j+1 < k$ e $j+1$ é ímpar.

  \begin{itemize}
    \item Se $o_{j+1} > o_{i}$, então $(o_1, o_{i+1}, o_{j+1})$ é uma tripla orientada válida.
    \item Caso contrário, temos $o_{j+1} < o_{i}$ e, portanto, $(o_i, o_{j}, o_{j+1})$ é uma tripla orientada válida. 
  \end{itemize}

  Dessa forma, cobrimos todos os casos possíveis e, para cada um dos casos, mostramos que existe uma tripla orientada válida que atende as condições do enunciado do lema.
\end{proof}

Os lemas~\ref{cap3:lemma:2-transposition-even-cycle} ao \ref{cap3:lemma:2-transposition-odd-cycle-case3} implicam no seguinte resultado.

\begin{corollary}\label{cap3:corollary_search}
Se existe uma $2$-transposição afetando um ciclo orientado $C$ de $G(\pi)$, então existe uma $2$-transposição aplicada em três arestas de origem $o_i$, $o_j$ e $o_k$ do ciclo $C$, com $i < j < k$, tal que pelo menos uma destas condições é verdadeira: $i \in \{1,2\}$; $j = i + 1$; ou $k = j + 1$.
\end{corollary}

Dado uma tripla orientada válida $(o_i, o_j, o_k)$, com $i < j < k$, definimos os parâmetros de uma $2$-transposição $\tau$ que age nessa tripla da seguinte forma:
\begin{itemize}
  \item Se $o_i > o_k > o_j$, então $\tau = \tau(o_j, o_k, o_i)$;
  \item Se $o_j > o_i > o_k$, então $\tau = \tau(o_k, o_i, o_j)$;
  \item Se $o_k > o_j > o_i$, então $\tau = \tau(o_i, o_j, o_k)$.
\end{itemize}

A partir desses resultados, apresentamos o Algoritmo~\ref{cap3:search_2transposition}, que recebe uma permutação $\pi$ de entrada e retorna uma $2$-transposição, caso exista, ou indica que não existe $2$-transposição para essa permutação. 

\begin{lemma}\label{cap3:lemma:2transposition_quadratic}
Dada uma permutação $\pi$, se existe pelo menos uma $2$-transposição que pode ser aplicada em $G(\pi)$, então o Algoritmo~\ref{cap3:search_2transposition} retorna uma $2$-transposição. Caso contrário, o algoritmo retorna que não existe $2$-transposição para $G(\pi)$.
\end{lemma}

\begin{proof}
Pelo Lema~\ref{cap3:lemma:types_of_2_transpositions}, qualquer $2$-transposição $\tau$ é aplicada em dois ciclos pares ou em um único ciclo orientado.

Se existe um par de ciclos pares em $G(\pi)$, então o algoritmo sempre acha uma $2$-transposição de acordo com o Lema~{3.2.5} de Christie~\cite{1998b-christie}.

Se a $2$-transposição é aplicada em um ciclo orientado $C = (o_1, o_2, \ldots, o_m)$, então existe uma $2$-transposição aplicada em uma tripla orientada válida $(o_{i'}, o_{j'}, o_{k'})$, com $i' < j' < k'$, tal que $i' \in \{1, 2\}$, ou $j' = i' + 1$, ou $k' = j'+1$ (Corolário~\ref{cap3:corollary_search}). Como o algoritmo faz uma busca exaustiva verificando todas as triplas $(o_i, o_j, o_k)$, com $i < j < k$, tal que $i \in \{1,2\}$, ou $j = i + 1$, ou $k = j + 1$, então o algoritmo irá encontrar uma $2$-transposição nesse caso.

Se não existe $2$-transposição, então o Algoritmo~\ref{cap3:search_2transposition} retorna vazio na sua última linha.
\end{proof}

{\revisaof Para analisar a complexidade do Algoritmo~\ref{cap3:search_2transposition}, verificamos primeiramente que a construção do grafo de ciclos $G(\pi)$ leva tempo linear.} A complexidade de tempo para as linhas~\ref{cap3:transp_ciclos_pares_inicio} e~\ref{cap3:transp_ciclos_pares_fim} também é linear, como mostrado na prova do Lema~{3.2.5} de Christie~\cite{1998b-christie}. A complexidade de tempo da busca nas linhas \ref{cap3:beggining_search_1}--\ref{cap3:end_search_1} é a seguinte, onde $c$ é uma constante relacionada às operações das linhas \ref{cap3:constant_search_1}--\ref{cap3:end_search_1}.

$$
\sum_{C \in G(\pi)} \sum_{j = 2}^{|C|-1} \sum_{k=j+1}^{|C|} c = \sum_{C \in G(\pi)} \sum_{j = 2}^{|C|-1} (|C|-j) c < c \sum_{C \in G(\pi)} |C|^2 = O(n^2),
$$
já que $\sum_{C \in G(\pi)} |C| \leq n + 1$.

De forma similar, podemos provar que a complexidade de tempo das linhas \ref{cap3:beggining_search_2}--\ref{cap3:end_search_2} também é quadrática. Dessa forma, temos que o Algoritmo~\ref{cap3:search_2transposition} possui complexidade de tempo de $O(n^2)$.

\begin{algorithm}[tb]
    \caption{Busca por uma $2$-transposição\label{cap3:search_2transposition}}
    \DontPrintSemicolon
    \Entrada{Uma permutação $\pi$}
    \Saida{Uma $2$-transposição $\tau$, caso exista, ou $\emptyset$}

    Construa o grafo $G(\pi)$\;
    \Se{existem pelo menos dois ciclos pares em $G(\pi)$}{\label{cap3:transp_ciclos_pares_inicio}
        {\bf retorne} a $2$-transposição do Lema~{3.2.5} de Christie~\cite{1998b-christie}\;\label{cap3:transp_ciclos_pares_fim}
    }\Senao{
        \Para{todo ciclo orientado $C = (o_1, o_2, \ldots, o_m)$ em $G(\pi)$\label{cap3:beggining_search_1}}{
             \Para{todo $j \in \{2, \ldots, m-1\}$}{
                \Para{todo $k \in \{j+1, \ldots, m\}$}{
                    \uSe{$(o_1, o_j, o_k)$ é uma tripla orientada válida\label{cap3:constant_search_1}}{
                        {\bf retorne} a $2$-transposição que age em $(o_1, o_j, o_k)$\;
                    }\SenaoSe{$(o_2, o_j, o_k)$ é uma tripla orientada válida}{
                        {\bf retorne} a $2$-transposição que age em $(o_2, o_j, o_k)$\label{cap3:end_search_1}\;
                    }
                }
             }   
        }
        \Para{todo ciclo orientado $C = (o_1, o_2, \ldots, o_m)$ em $G(\pi)$\label{cap3:beggining_search_2}}{
            \Para{todo $i \in \{3, \ldots, m-2\}$}{
               \Para{todo $j \in \{i+1, \ldots, m-1\}$}{
                  \Se{$(o_i, o_j, o_{j+1})$ é uma tripla orientada válida}{
                    {\bf retorne} a $2$-transposição que age em $(o_i, o_j, o_{j+1})$\;
                  }
                }
               \Para{todo $k \in \{i+2, \ldots, m\}$}{
                \Se{$(o_i, o_{i+1}, o_k)$ é uma tripla orientada válida}{
                    {\bf retorne} a $2$-transposição que age em $(o_i, o_{i+1}, o_k)$\label{cap3:end_search_2}\;
                }
             }
         }
        }
        {\bf retorne} $\emptyset$ \Comment{não existe $2$-transposição para $G(\pi)$}
    }
\end{algorithm}

Agora, apresentamos como alcançar uma complexidade de tempo de $O(n^5)$ para o algoritmo de $1.375$-aproximação proposto por Elias e Hartman~\cite{2006-elias-hartman}. Primeiro, notamos que o Algoritmo~\ref{cap3:search_2transposition} não pode ser usado para listar todas as possíveis $2$-transposições aplicadas em uma permutação, apesar de que o algoritmo sempre retorna uma $2$-transposição se $G(\pi)$ admite pelo menos uma $2$-transposição.

{\revisaof
O Algoritmo~\ref{cap3:alg:2-2transp} retorna uma $(2,2)$-sequência, caso exista, ou indica que não existe uma $(2,2)$-sequência para $G(\pi)$. Esse algoritmo testa todas as combinações de triplas $(i,j,k)$, com $1 \leq i < j < k \leq n+1$, a fim de encontrar $2$-transposições para $G(\pi)$. Para toda $2$-transposição $\tau(i,j,k)$ encontrada na linha~\ref{cap3:alg2_inside_loop}, o Algoritmo~\ref{cap3:alg:2-2transp} utiliza o Algoritmo~\ref{cap3:search_2transposition} para verificar se existe uma $2$-transposição para o grafo $G(\pi \comp \tau(i,j,k))$.

Por fim, o Algoritmo~\ref{cap3:alg:1375-o5} usa o Algoritmo~\ref{cap3:alg:2-2transp} como sub-rotina para encontrar uma $(2,2)$-sequência, caso exista, e também usa o algoritmo de Elias e Hartman~\cite{2006-elias-hartman}, o qual chamamos de \texttt{Algoritmo\_EH}, que é aplicado após a busca de uma $(2,2)$-sequência para a permutação $\pi$.
}

\begin{algorithm}[h]
    \caption{Retorna uma $(2,2)$-sequência, caso exista, ou uma sequência vazia}\label{cap3:alg:2-2transp}
    \DontPrintSemicolon
    
    \SetKwFunction{FEH}{Algoritmo\_EH}
    \SetKwFunction{FAlgoUm}{Algoritmo\_\ref{cap3:search_2transposition}}
    \Entrada{Uma permutação $\pi$}
    \Saida{Uma $(2,2)$-sequência, caso exista, ou $(\emptyset, \emptyset)$}
    
    Construa $G(\pi)$\;
    Seja $z$ o número de arestas de origem em $G(\pi)$\;
    \Para{todo $i \in \{1, \ldots, z-2\}$}{\label{cap3:alg2_first_loop_transp}
        \Para{todo $j \in \{i+1, \ldots, z-1\}$}{
            \Para{todo $k \in \{j+1, \ldots, z\}$}{\label{cap3:alg2_last_loop_transp}
                \Se{$\tau(i,j,k)$ é uma $2$-transposição}{\label{cap3:alg2_inside_loop}
                    $\pi' \gets \pi \cdot \tau(i,j,k)$\;
                    $\tau' \gets \FAlgoUm(\pi')$\;
                    \Se{$\tau' \neq \emptyset$}{
                        {\bf retorne} ($\tau,\tau'$)\;\label{cap3:return-22sequence}
                    }
                }
            }
        }
    }
    {\bf retorne} $(\emptyset, \emptyset)$\; \Comment{Não existe uma $(2,2)$-sequência para $\pi$}
\end{algorithm}

\begin{algorithm}[h]
    \caption{Uma $1.375$-Aproximação Mais Eficiente para Transposições\label{cap3:alg:1375-o5}}
    \DontPrintSemicolon
    
    \SetKwFunction{FAlgoDois}{Algoritmo\_\ref{cap3:alg:2-2transp}}
    \SetKwFunction{FEH}{Algoritmo\_EH}
    \SetKwFunction{FAlgoUm}{Algoritmo\_\ref{cap3:search_2transposition}}

    \Entrada{Uma permutação $\pi$}
    \Saida{Uma sequência de transposições $\tau_1,\ldots,\tau_r$ que ordena $\pi$}
 
    $(\tau_1,\tau_2) \gets$ \FAlgoDois{$\pi$}\;\label{cap3:alg3:busca_sequencia}
    \uSe{$\tau_1 \neq \emptyset$}{
        $\pi' \gets \pi \cdot \tau_1 \cdot \tau_2$\Comment{existe uma $(2,2)$-sequência para $\pi$}\;
        {\bf retorne} $(\tau_1,\tau_2)$ + \FEH{$\pi'$}\;\label{cap3:alg3:alg_eh1}
    }
    \Senao() {
        {\bf retorne} \FEH{$\pi$}\;\label{cap3:alg3:alg_eh2}
    }    
\end{algorithm}

Lembre-se que o principal problema do \texttt{Algoritmo\_EH} em garantir o fator de aproximação de $1.375$ é possivelmente negligenciar uma $(2,2)$-sequência após transformar $\pi$ em uma permutação simples $\hat{\pi}$.

Primeiramente, iremos analisar o Algoritmo~\ref{cap3:alg:2-2transp} que busca por uma $(2,2)$-sequência para uma permutação $\pi$. Para cada transposição $\tau$ gerada usando todos os possíveis valores para $i$, $j$ e $k$, se $\tau$ é uma $2$-transposição, então o algoritmo aplica essa operação em $\pi$, resultando em $\pi' = \pi \comp \tau$. Após isso, o Algoritmo~\ref{cap3:alg:2-2transp} usa o Algoritmo~\ref{cap3:search_2transposition} com $\pi'$ como sub-rotina. Se o Algoritmo~\ref{cap3:search_2transposition} retorna uma segunda $2$-transposição, então o Algoritmo~\ref{cap3:alg:2-2transp} retorna uma $(2,2)$-sequência na linha~\ref{cap3:return-22sequence}. Se não existe uma $(2,2)$-sequência para $\pi$, então o algoritmo retorna uma sequência vazia.

O Algoritmo~\ref{cap3:alg:2-2transp} usa laços aninhados nas linhas~\ref{cap3:alg2_first_loop_transp}--\ref{cap3:alg2_last_loop_transp} para a busca da primeira $2$-transposição $\tau(i,j,k)$, usando todas as combinações para os índices $i$, $j$ e $k$. Para executar as instruções dentro desses laços aninhados (linhas~\ref{cap3:alg2_inside_loop}--\ref{cap3:return-22sequence}), é necessário complexidade de tempo de $O(n^2)$. Portanto, o Algoritmo~\ref{cap3:alg:2-2transp} executa em tempo $O(n^5)$.

Após usar o Algoritmo~\ref{cap3:alg:2-2transp} na linha~\ref{cap3:alg3:busca_sequencia}, o Algoritmo~\ref{cap3:alg:1375-o5} usa a $1.375$-aproximação de Elias e Hartman~\cite{2006-elias-hartman} (\texttt{Algoritmo\_EH}) na linha \ref{cap3:alg3:alg_eh1} ou na linha \ref{cap3:alg3:alg_eh2}.
Como o \texttt{Algoritmo\_EH} tem complexidade de tempo quadrática, concluímos que o Algoritmo~\ref{cap3:alg:1375-o5} possui complexidade de tempo de $O(n^5)$.

É garantido que o Algoritmo~\ref{cap3:alg:1375-o5} possui fator de aproximação de $1.375$ para qualquer permutação, pois ele garante a aplicação de uma $(2,2)$-sequência na permutação de entrada, caso ela exista, e apenas depois disso usa o \texttt{Algoritmo\_EH}, que transforma a permutação de entrada em uma permutação simples.

\subsection{Resultados Experimentais}

Nesta seção, apresentamos resultados experimentais para o Algoritmo~\ref{cap3:alg:1375-o5}. Comparamos os resultados do nosso algoritmo com os resultados apresentados para os algoritmos propostos por Elias e Hartman~\cite{2006-elias-hartman} e Silva e coautores~\cite{2022-silva-etal}. 

Testamos os algoritmos para todas as permutações de tamanho $n \leq 12$ e comparamos os tamanhos das sequências de ordenação retornadas pelos algoritmos com as distâncias exatas disponíveis no sistema GRAAu~\cite{2014a-galvao-dias}. O GRAAu é uma ferramenta de auditoria para algoritmos de ordenação de permutações por rearranjos, sendo que essa ferramenta possui um banco de dados com os valores das distâncias exatas para permutações pequenas.  

A Tabela~\ref{cap3:tab:all-results} sumariza os resultados dos algoritmos testados, os quais são identificados por:
\begin{itemize}
  \item {\bf ALG3}: Algoritmo~\ref{cap3:alg:1375-o5} apresentado nesta seção;
  \item {\bf EH}: Algoritmo apresentado por Elias e Hartman~\cite{2006-elias-hartman};
  \item {\bf SKRW}: Algoritmo apresentado por Silva e coautores~\cite{2022-silva-etal}. 
\end{itemize}

As permutações são agrupadas por tamanho e cada linha apresenta resultados para o grupo de permutações de tamanho $n$, indicado na primeira coluna da tabela. As colunas {\bf APROX MAX} e {\bf APROX MED} representam o valor máximo e o valor médio do fator de aproximação observado, respectivamente. A coluna {\bf DIST MED} representa a média dos tamanhos das sequências de ordenação retornadas pelos algoritmos e a coluna {\bf \% SOLUÇÕES ÓTIMAS} indica o percentual de instâncias em que uma solução ótima foi encontrada.

O algoritmo de Elias e Hartman~\cite{2006-elias-hartman} retornou sequências de ordenação com um fator de aproximação acima de $1.375$ (comparado com a distância exata) em $2$ instâncias de tamanho $8$, $20$ instâncias de tamanho $9$, $110$ instâncias de tamanho $10$, $440$ instâncias de tamanho $11$ e, por último, $1448$ instâncias de tamanho $12$. Por outro lado, o algoritmo proposto nesta seção e o algoritmo proposto por Silva e coautores~\cite{2022-silva-etal} não retornaram nenhuma sequência de ordenação com aproximação acima de $1.333$.

{\revisaof
Na Tabela~\ref{cap3:tab:all-results}, os valores para o fator de aproximação máximo dos algoritmos {\bf ALG3} e {\bf SKRW} foram os mesmos, considerando todos os grupos de permutações. Ao analisar os valores médios dos fatores de aproximação para os algoritmos {\bf ALG3} e {\bf SKRW}, concluímos que nos grupos de tamanho menor ou igual $6$, o valor de $1.0$ foi mantido, o que indica que, para todas as instâncias nesses grupos, ambos os algoritmos encontraram sequências de ordenação ótimas (a coluna {\bf \% SOLUÇÕES ÓTIMAS} também indica a mesma informação). Para grupos de permutações com tamanhos maiores que $6$, ao comparar os três algoritmos, concluímos que o algoritmo apresentado nesta seção obteve melhores resultados de aproximação máxima, aproximação média, médias das distâncias, e percentual de solução ótimas encontradas.
}

É importante destacar que o percentual de instâncias em que o algoritmo {\bf ALG3} encontra uma solução ótima foi maior que $88$\% para todos os grupos de permutações. O Algoritmo {\bf SKRW} mantém esse comportamento apenas para grupos de permutações com tamanho menor ou igual a $9$. Assim, concluímos que o algoritmo proposto nesta seção traz uma melhoria do ponto de vista da qualidade da solução, além da melhoria da complexidade do algoritmo já demonstrada anteriormente.

{\revisaof
Optamos por não comparar os tempos de execução das implementações dos algoritmos por não considerarmos que a análise seria justa, dados os seguintes motivos: (i) o algoritmo {\bf SKRW} é implementado em uma linguagem de programação distinta da utilizada na implementação do algoritmo {\bf ALG3}; (ii) o algoritmo {\bf EH} possui complexidade de tempo consideravelmente menor do que os outros algoritmos, apesar de não ser um algoritmo de aproximação válido, como mostrado na Tabela~\ref{cap3:tab:all-results}.
}

\begin{table}[tbh]
\caption{\revisaof
Comparação entre os resultados experimentais do Algoritmo~\ref{cap3:alg:1375-o5} e os resultados dos algoritmos propostos por Elias e Hartman~\cite{2006-elias-hartman} ({\bf EH}) e Silva e coautores~\cite{2022-silva-etal} ({\bf SKRW}), em todas permutações de tamanho $n \leq 12$, excluindo a permutação identidade $\iota^n$.}
\label{cap3:tab:all-results}
\centering
\setlength\tabcolsep{5pt}
\resizebox{0.99\textwidth}{!}{
\centering
\begin{tabular}{ccccccccccccc}
\toprule
 & \multicolumn{3}{c}{{\bf APROX MAX}} & \multicolumn{3}{c}{{\bf APROX MED}} & \multicolumn{3}{c}{{\bf DIST MED}} & \multicolumn{3}{c}{{\bf \% SOLUÇÕES ÓTIMAS}} \\
\cmidrule(r){2-4} \cmidrule(r){5-7} \cmidrule(r){8-10} \cmidrule(r){11-13}
n & {\bf EH}  & {\bf SKRW} & {\bf ALG3}  & {\bf EH} & {\bf SKRW} & {\bf ALG3}  & {\bf EH} & {\bf SKRW} & {\bf ALG3}  & {\bf EH} & {\bf SKRW} & {\bf ALG3} \\
\midrule
2  & 1.00 & 1.00 & 1.00 & 1.00   & 1.00    & 1.00    & 1.00   & 1.00   & 1.00   & 100.00 & 100.00 & 100.00 \\
3  & 1.00 & 1.00 & 1.00 & 1.00   & 1.00    & 1.00    & 1.20   & 1.20   & 1.20   & 100.00 & 100.00 & 100.00 \\
4  & 1.00 & 1.00 & 1.00 & 1.00   & 1.00    & 1.00    & 1.6086 & 1.6086 & 1.6086 & 100.00 & 100.00 & 100.00 \\
5  & 1.00 & 1.00 & 1.00 & 1.00   & 1.00    & 1.00    & 2.0924 & 2.0924 & 2.0924 & 100.00 & 100.00 & 100.00 \\
6  & 1.33 & 1.00 & 1.00 & 1.0004 & 1.00    & 1.00    & 2.6063 & 2.6050 & 2.6050 & 99.86  & 100.00 & 100.00 \\
7  & 1.33 & 1.25 & 1.25 & 1.0129 & 1.0113 & 1.0014 & 3.1762 & 3.1704 & 3.1311 & 94.90  & 95.47  & 99.40 \\
8  & 1.50 & 1.25 & 1.25 & 1.0210 & 1.0183 & 1.0042 & 3.7178 & 3.7076 & 3.6512 & 91.64  & 92.65  & 98.29 \\
9  & 1.50 & 1.25 & 1.25 & 1.0301 & 1.0256 & 1.0085 & 4.2796 & 4.2603 & 4.1846 & 86.62  & 88.54  & 96.10 \\
10 & 1.50 & 1.25 & 1.25 & 1.0341 & 1.0282 & 1.0125 & 4.8051 & 4.7772 & 4.7032 & 83.80  & 86.53  & 93.94 \\
11 & 1.50 & 1.33 & 1.33 & 1.0392 & 1.0321 & 1.0170 & 5.3526 & 5.3157 & 5.2367 & 79.40  & 82.98  & 90.88 \\
12 & 1.50 & 1.33 & 1.33 & 1.0415 & 1.0336 & 1.0206 & 5.8694 & 5.8248 & 5.7514 & 76.67  & 80.91  & 88.27 \\
\bottomrule
\end{tabular}
}
\end{table}

\section{Complexidade de Problemas com Transposições e Outros Rearranjos}\label{cap3:secao_complexidade}

Nesta seção, apresentamos provas de NP-dificuldade para problemas de Ordenação de Permutações por Rearranjos considerando os seguintes modelos de rearranjos:

\begin{itemize}
  \item $\M_1 = \{\tau, \transrev \}$: Transposições e Transposições Inversas;
  \item $\M_2 = \{\rho, \tau, \transrev \}$: Reversões, Transposições e Transposições Inversas;
  \item $\M_3 = \{\tau, \revrev \}$: Transposições e Revrevs;
  \item $\M_4 = \{\rho, \tau, \revrev \}$: Reversões, Transposições e Revrevs;
  \item $\M_5 = \{\tau, \transrev, \revrev \}$: Transposições, Transposições Inversas e Revrevs.
  \item $\M_6 = \{\rho, \tau, \transrev, \revrev \}$: Reversões, Transposições, Transposições Inversas e Revrevs.
\end{itemize}

Além disso, denotamos por $\tau$ o modelo que possui apenas transposições. A seguir, apresentamos formalmente a versão de decisão dos problemas estudados.

\begin{problembox}
\textbf{Ordenação de Permutações por Rearranjos (\pname{SbR})}\\
\textbf{Entrada:} Uma permutação $\pi$ e um inteiro $k$.\\
\textbf{Objetivo:} Considerando um modelo de rearranjos $\M$, decidir se é possível ordenar $\pi$ com uma sequência de rearranjos $S$, tal que $|S| \leq k$ e $\beta \in \M$, para todo $\beta \in S$. Ou seja, determinar se $d_{\M}(\pi) \leq k$.
\end{problembox}

\begin{problembox}
\textbf{Ordenação de Permutações por Rearranjos Ponderados (\pname{SbWR})}\\
\textbf{Entrada:} Uma permutação $\pi$ e um inteiro $k$.\\
\textbf{Objetivo:} Considerando um modelo de rearranjos $\M$ e pesos $\{w_\rho, w_\tau\}$, decidir se é possível ordenar $\pi$ com uma sequência de rearranjos $S$, tal que $w(S) \leq k$ e $\beta \in \M$, para todo $\beta \in S$. Ou seja, determinar se $d^w_{\M}(\pi) \leq k$.
\end{problembox}

Quando $w_{\rho} = w_{\tau}$, esse problema é equivalente à abordagem não ponderada. Se um modelo não contém reversões, então consideramos que $w_{\rho} = \infty$.

Seja $\pi$ uma permutação sem sinais e $k = b_{\tau}(\pi)/3$. Um limitante bastante conhecido para a distância de transposições é $d_t(\pi) \geq b_{\tau}(\pi)/3$~\cite{2012-bulteau-etal}. 
As provas de dificuldade apresentadas nesta seção são baseadas em reduções do problema NP-difícil \pname{B3T}~\cite{2012-bulteau-etal}. 

\begin{problembox}
\textbf{Ordenação de Permutações por Transposições Ótimas (\pname{B3T})}\\
\textbf{Entrada:} Uma permutação sem sinais $\pi$.\\ 
\textbf{Objetivo:} Decidir se é possível ordenar a permutação $\pi$ usando exatamente $b_{\tau}(\pi)/3$ transposições, ou seja, determinar se $d_{\tau}(\pi) = b_{\tau}(\pi)/3$.
\end{problembox}

Como a permutação identidade $\iota^n$ é a única com zero {\it breakpoints}, podemos interpretar o processo de ordenar $\pi$ como o de remover todos os {\it breakpoints} de $\pi$. Agora, mostramos como os rearranjos podem afetar o número de {\it breakpoints} em uma permutação e também definir limitantes para a distância. Também apresentamos como os rearranjos estudados afetam o número de {\it breakpoints} em algumas famílias de permutações, que serão úteis nas provas de complexidade.

\begin{lemma}\label{cap3:lemma:lb_signed}
  Para qualquer permutação com sinais $\pi$, pesos $w_{\rho}$ e $w_{\tau}$, e modelo $\M \in \{{\M_1}, {\M_2}, {\M_3}, {\M_4}, {\M_5}, {\M_6}\}$, temos que
  $$d^w_{\M}(\pi) \geq \min\left\{\frac{w_\rho}{2}, \frac{w_\tau}{3}\right\} b_{\M}(\pi).$$
\end{lemma}

\begin{proof}
  Como uma reversão $\rho(i,j)$ quebra apenas as adjacências $(\pi_{i-1}, \pi_i)$ e $(\pi_j, \pi_{j+1})$, então no máximo dois {\it breakpoints} podem ser removidos ou adicionados e, portanto, $-2 \leq \Delta b_\M(\pi, \rho) \leq 2$. De forma similar, se uma operação $\beta$ é uma transposição, transposição inversa ou revrev, então apenas três adjacências de $\pi$ são quebradas e, portanto, no máximo três {\it breakpoints} podem ser removidos ou adicionados (i.e., $-3 \leq \Delta b_\M(\pi, \beta) \leq 3$). 

  Sendo assim, o custo mínimo para remover um {\it breakpoint} é igual a $\min\left\{\frac{w_\rho}{2}, \frac{w_\tau}{3}\right\}$. Como qualquer sequência $S$ que ordena $\pi$ remove $b_\M(\pi)$ {\it breakpoints}, concluímos que $S$ possui custo maior ou igual a $\min\left\{\frac{w_\rho}{2}, \frac{w_\tau}{3}\right\} b_{\M}(\pi)$.
\end{proof}

\begin{lemma}\label{cap3:lemma:signed_var_b}
  Para qualquer permutação com sinais $\pi$ tal que $\pi$ possui apenas {\it strips} positivas, temos que:
  \begin{itemize}
     \item $\Delta b_{\rho}(\pi, \rho) \leq 0$, para qualquer reversão ${\rho}$;
     \item $\Delta b_{\transrev}(\pi, \transrev) \leq 1$, para qualquer transposição inversa $\transrev$;
     \item $\Delta b_{\revrev}(\pi, \revrev) \leq 1$, para qualquer revrev ${\revrev}$.
   \end{itemize}
\end{lemma}

\begin{proof}
  Note que $\pi$ possuir apenas {\it strips} positivas implica que todos elementos de $\pi$ possuem sinal ``$+$''.

  Considere uma reversão $\rho$ e seja $\pi' = \pi \comp \rho = (\pi_1~\ldots~\pi_{i-1}~\underline{{-\pi_j}~\ldots~{-\pi_i}}~\pi_{j+1}$ $\ldots~\pi_n)$, com $1 \leq i \leq j \leq n$. Provaremos, por contradição, que não existe reversão que remove {\it breakpoints} de $\pi$. Suponha que $\Delta b_{\rho}(\pi, \rho) > 0$, o que implica que $(\pi_{i-1}, -\pi_{j})$ não é um {\it breakpoint} ou $(-\pi_{i}, \pi_{j+1})$ não é um {\it breakpoint}. Se $(\pi_{i-1}, -\pi_{j})$ não é um {\it breakpoint}, então $\pi_{i-1}$ e $-\pi_{j}$ devem ter o mesmo sinal. De forma similar, se $(-\pi_{i}, \pi_{j+1})$ não é um {\it breakpoint}, então $-\pi_{i}$ e $\pi_{j+1}$ devem ter o mesmo sinal. Em ambos os casos chegamos a uma contradição, pois $\pi$ possui apenas elementos com sinal ``$+$'' e, portanto, os elementos desses pares possuem sinais distintos. Sendo assim, $\Delta b_{\rho}(\pi, \rho) \leq 0$.

  Considere uma transposição inversa tipo 1 $\transrev_1$ e seja $\pi' =$ $\pi \comp \transrev_1 =$  $(\pi_1~\ldots~\pi_{i-1}$ $\underline{\pi_j~\ldots~\pi_{k-1}}$ $\underline{-\pi_{j-1}~\ldots~-\pi_{i}}$ $\pi_{k}~\ldots~\pi_n)$, com $1 \leq i < j < k \leq n+1$. Usando um argumento similar ao usado para reversões, temos que os pares $(\pi_{k-1}, -\pi_{j-1})$ e $(-\pi_{i}, \pi_k)$ devem ser {\it breakpoints} e apenas o par $(\pi_{i-1}, \pi_j)$ pode não ser um {\it breakpoint}. Portanto, $\Delta b_{\transrev}(\pi, \transrev_1) \leq 1$. Usamos um argumento análogo para uma transposição inversa tipo 2. Agora, considere uma revrev $\revrev$ e seja $\pi' = \pi \comp \revrev = (\pi_1~\ldots~\pi_{i-1}$ $\underline{-\pi_{j-1}~\ldots~-\pi_{i}}$ $\underline{-\pi_{k-1}~\ldots~-\pi_{j}}$ $\pi_{k}~\ldots~\pi_n)$, com $1 \leq i < j < k \leq n+1$. Usando um argumento similar ao usado para reversões, temos que os pares $(\pi_{i-1}, -\pi_{j-1})$ e $(-\pi_j, \pi_k)$ devem ser {\it breakpoints} e apenas o par $(-\pi_{i}, -\pi_{k-1})$ pode não ser um {\it breakpoint}. Portanto, $\Delta b_{\revrev}(\pi, \revrev) \leq 1$.
\end{proof}

\begin{theorem}\label{cap3:theorem:complexity_signed}
Para os modelos $\M = \{\M_1, \M_2, \M_3, \M_4, \M_5,$ $ \M_6\}$ e $w_{\tau}/w_{\rho} \leq 1.5$, temos que o problema da Ordenação de Permutações por Rearranjos Ponderados \pname{SbWR} é NP-difícil para permutações com sinais.
\end{theorem}

\begin{proof}
  Considere o modelo de rearranjos $\M = \M_6$ nesta prova. A demonstração é similar para os outros modelos, já que a nossa estratégia é mostrar que em uma instância satisfeita (i.e., uma instância em que é possível ordenar $\pi$ com custo menor ou igual a $k$) apenas transposições são usadas para ordenar a permutação $\pi$. Portanto, um argumento similar pode ser usado para os outros modelos, pois eles possuem um subconjunto das operações permitidas em $\M_6$.

  Agora, apresentamos uma redução do problema \pname{B3T} para o problema \pname{SbWR} considerando permutações com sinais. Dada uma instância $\pi = (\pi_1~\ldots~\pi_n)$ para o problema \pname{B3T}, construímos a instância $(\pi', k)$ para \pname{SbWR}, onde $\pi'$ é a permutação com sinais $({+\pi_1}~{+\pi_2}~\ldots~{+\pi_n})$ e $k = w_{\tau} b_{\tau}(\pi)/3$.

  Mostramos que a instância $\pi$ é ordenada usando $b_{\tau}(\pi)/3$ transposições se, e somente se, $d^w_{\M}(\pi') \leq w_{\tau} b_{\tau}(\pi)/3$.

  ($\rightarrow$) Se a permutação $\pi$ é ordenada por uma sequência $S$ de tamanho $b_{\tau}(\pi)/3$, então $S$ também ordena a permutação $\pi'$, já que $\pi'$ possui apenas elementos positivos, e $w(S) = w_\tau b_\tau(\pi)/3$. Note que $S$ possui apenas transposições e cada transposição tem custo $w_\tau$.

  ($\leftarrow$) Se $\pi'$ é ordenada por uma sequência $S$ de custo igual ou menor a $w_\tau b_\tau(\pi)/3$, então afirmamos que $S$ possui apenas transposições e, portanto, $S$ também ordena $\pi$ e $S$ tem tamanho $b_\tau(\pi)/3$.

  Note que o custo mínimo para remover um {\it breakpoint} em $\pi'$ é igual a $\min \{ w_\rho/2,$ $w_\tau/3 \} = w_\tau/3$, já que $w_{\tau}/w_{\rho} \leq 1.5$. Como $\pi'$ possui apenas elementos positivos, $\pi_{i+1} - \pi_i = 1$ se, e somente se, $\pi'_{i+1} - \pi'_i = 1$ e, portanto, $b_{\tau}(\pi) = b_{\M}(\pi')$. Sendo assim, o limitante inferior do Lema~\ref{cap3:lemma:lb_signed} se torna $w_\tau b_{\M}(\pi')/3 = w_\tau b_{\tau}(\pi)/3$.
  
  Temos que $w(S)$ é igual ao limitante inferior $w_\tau b_{\tau}(\pi)/3$ e, consequentemente, todo rearranjo de $S$ deve remover exatamente $w' \times 3/w_\tau$ {\it breakpoints}, onde $w'$ é o custo do rearranjo. Note que uma sequência que ordena $\pi'$ remove $b_{\tau}(\pi)$ {\it breakpoints}. Agora, suponha que $S$ tem algum rearranjo que não é uma transposição. Considere que $S = (\beta_1, \beta_2, \ldots, \beta_{|S|})$. Seja $\beta_i$ o primeiro rearranjo de $S$ a ser aplicado que não é uma transposição, ou seja, todos rearranjos em $(\beta_1, \beta_2, \ldots, \beta_{i-1})$ são transposições, $\beta_i$ não é uma transposição e o valor de $i$ é mínimo. Antes de $\beta_i$ ser aplicado, todas {\it strips} na permutação são positivas, já que apenas transposições foram aplicadas anteriormente e elas não alteram o sinal dos elementos. Pelo Lema~\ref{cap3:lemma:signed_var_b}, uma reversão não remove {\it breakpoints} e uma transposição inversa ou revrev removem no máximo um {\it breakpoint} em permutações que possuem apenas {\it strips} positivas, o que contradiz o fato de que todo rearranjo de $S$ remove $w' \times 3/w_\tau$ {\it breakpoints}. Portanto, $S$ possui apenas transposições e $S$ possui tamanho $b_{\tau}(\pi)/3$.
\end{proof}

\begin{corollary}\label{cap3:theorem:complexity_signed_unweighted}
Para os modelos $\M = \{\M_1, \M_2, \M_3, \M_4, \M_5,$ $ \M_6\}$, temos que o problema da Ordenação de Permutações por Rearranjos \pname{SbR} é NP-difícil para permutações com sinais.
\end{corollary}

\begin{lemma}\label{cap3:lemma:lb_unsigned}
  Para qualquer permutação sem sinais $\pi$, pesos $w_{\rho}$ e $w_{\tau}$, e modelo $\M \in \{{\M_1}, {\M_2}, {\M_3}, {\M_4}, {\M_5}, {\M_6}\}$, temos que
  $$d^w_{\M}(\pi) \geq \min\left\{\frac{w_\rho}{2}, \frac{w_\tau}{3}\right\} b_{\M}(\pi).$$
\end{lemma}

\begin{proof}
  Similar à prova do Lema~\ref{cap3:lemma:lb_signed}.
\end{proof}

\begin{lemma}\label{cap3:lemma:unsigned_var_b}
  Para qualquer permutação sem sinais $\pi$ tal que $\pi$ possui apenas {\it strips} crescentes, temos que: 
  \begin{itemize}
     \item $\Delta b_{\rho}(\pi, \rho) \leq 0$, para qualquer reversão ${\rho}$;
     \item $\Delta b_{\transrev}(\pi, \transrev) \leq 1$, para qualquer transposição inversa $\transrev$;
     \item $\Delta b_{\revrev}(\pi, \revrev) \leq 1$, para qualquer revrev ${\revrev}$.
  \end{itemize}
\end{lemma}

\begin{proof}
  Considere uma reversão $\rho$ e seja $\pi' = \pi \comp \rho(i,j) = (\pi_1~\ldots~\pi_{i-1}~\underline{{\pi_j}~\ldots~{\pi_i}}~\pi_{j+1}$ $\ldots~\pi_n)$. {\revisaof Suponha, por contradição, que $\Delta b_{\rho}(\pi, \rho) > 0$, o que indica que ou $(\pi_{i-1}, \pi_j)$ não é um {\it breakpoint} ou $(\pi_i, \pi_{j+1})$ não é um {\it breakpoint}.}

  Note que todas as {\it strips} em $(\pi'_i, \ldots, \pi'_j)$ são decrescentes, pois $\pi$ possui apenas {\it strips} crescentes. Seja $\sigma = (\pi_{i'}, \ldots, \pi_{i-1})$ a {\it strip} crescente contendo o elemento $\pi_{i-1}$ em $\pi$. Se $(\pi'_{i-1}, \pi'_{i})$ não é um {\it breakpoint}, então a {\it strip} $\sigma$ se torna a {\it strip} $\sigma' = (\pi_{i'}, \ldots, \pi_{i-1}, \pi_{j}, \ldots, \pi_{j'})$ em $\pi'$. Como $\sigma$ é crescente, temos que $\pi_{i-1} = \pi_0$ ou $i' < i - 1$. Então, a {\it strip} $\sigma'$ deve ser uma {\it strip} crescente, o que contradiz o fato de que as {\it strips} em $(\pi'_i, \ldots, \pi'_j)$ são todas decrescentes. Alcançamos uma contradição semelhante se $(\pi'_{j}, \pi'_{j+1})$ não é um {\it breakpoint}. Portanto, temos que $\Delta b_{\rho}(\pi, \rho) \leq 0$.

  Considere uma transposição inversa tipo 1 $\transrev_1$ e seja $\pi' = \pi \comp \transrev_1(i,j,k) =$ $(\pi_1~\ldots~\pi_{i-1}$ $\underline{\pi_j~\ldots~\pi_{k-1}}~\underline{\pi_{j-1}~\ldots~\pi_{i}}$ $\pi_{k}~\ldots~\pi_n)$. Usando um argumento similar ao apresentado para reversões, temos que os pares $(\pi_{k-1}, \pi_{j-1})$ e $(\pi_i, \pi_k)$ são {\it breakpoints} e apenas o par $(\pi_{i-1}, \pi_j)$ pode não ser um {\it breakpoint}. Portanto, temos que $\Delta b_{\transrev}(\pi, \transrev) \leq 1$. Usamos um argumento similar para uma transposição inversa tipo 2.

  Considere uma revrev $\revrev$ e seja $\pi' = \pi \comp \revrev(i,j,k) = (\pi_1~\ldots~\pi_{i-1}$ $\underline{\pi_{j-1}~\ldots~\pi_{i}}~\underline{\pi_{k-1}~\ldots~\pi_{j}}$ $\pi_{k}~\ldots~\pi_n)$. Usando um argumento similar ao apresentado para reversões, temos que os pares $(\pi_{i-1}, \pi_{j-1})$ e $(\pi_j, \pi_k)$ devem ser {\it breakpoints} e apenas o par $(\pi_{i}, \pi_{k-1})$ pode não ser um {\it breakpoint}. Portanto, temos que $\Delta b_{\revrev}(\pi, \revrev) \leq 1$.
\end{proof}

\begin{theorem}\label{cap3:theorem:wsr_unsigned}
Para os modelos $\M = \{\M_1, \M_2, \M_3, \M_4, \M_5,$ $ \M_6\}$ e $w_{\tau}/w_{\rho} \leq 1.5$, temos que o problema da Ordenação de Permutações por Rearranjos Ponderados \pname{SbWR} é NP-difícil para permutações sem sinais.
\end{theorem}

\begin{proof}
  Considere o modelo de rearranjos $\M = \M_6$ nesta prova. A demonstração é similar para os outros modelos, já que a nossa estratégia é mostrar que em uma instância satisfeita (i.e., uma instância em que é possível ordenar $\pi$ com custo menor ou igual a $k$) apenas transposições são usadas para ordenar a permutação $\pi$. Portanto, um argumento similar pode ser usado para os outros modelos, pois eles possuem um subconjunto das operações permitidas em $\M_6$.

  Para esta prova, também apresentamos uma redução do problema \pname{B3T} para o problema \pname{SbWR}, mas considerando permutações sem sinais. Dada uma instância $\pi = (\pi_1~\ldots~\pi_n)$ para o problema \pname{B3T}, construímos a instância $(\pi', k)$ para \pname{SbWR}, onde $\pi'$ é uma permutação sem sinais com $2n$ elementos tal que $\pi'_{2i-1} = 2\pi_{i} - 1$ e $\pi'_{2i} = 2\pi_{i}$, com $1 \leq i \leq n$, e $k = w_{\tau} b_{\tau}(\pi)/3$.

  Mostramos que a instância $\pi$ é ordenada por $b_\tau(\pi)/3$ transposições se, e somente se, $d^w_{\M}(\pi') \leq w_{\tau} b_{\tau}(\pi)/3$.

  ($\rightarrow$) Se $\pi$ é ordenada por uma sequência $S = (\tau_1, \tau_2, \ldots, \tau_{|S|})$ com $|S| = b_\tau(\pi)/3$, então construímos uma sequência $S' = (\tau'_1, \tau'_2, \ldots, \tau'_{|S|})$ tal que, para toda transposição $\tau_i = \tau(x,y,z)$ em $S$, temos que $\tau'_i = \tau(2x-1, 2y-1, 2z-1)$. Como todo elemento de $\pi$ foi mapeado em dois elementos consecutivos em $\pi'$, a sequência $S'$ ordena $\pi'$ e $w(S') = w_\tau b_{\tau}(\pi)/3$.

  ($\leftarrow$) Se $\pi'$ é ordenada por uma sequência $S'$ com $w(S') \leq w_{\tau} b_{\tau}(\pi)/3$, então afirmamos que existe uma sequência $S$ tal que $S$ possui apenas transposições, $S$ ordena $\pi$, e $|S| = b_{\tau}(\pi)/3$. 

  Como $\pi'_{2i-1}$ e $\pi'_{2i}$ são elementos consecutivos, os pares $(\pi'_{2i-1},\pi'_{2i})$ não são {\it breakpoints}, para qualquer $1 \leq i \leq n$. Assim, temos que todas as {\it strips} de $\pi'$ são crescentes pois $\pi'_{2i} > \pi'_{2i-1}$. Além disso, para $0 \leq i \leq n$, temos que:
  \begin{itemize}
      \item se $\pi_{i+1} - \pi_i = 1$, então $\pi'_{2i+1} - \pi'_{2i} = 2\pi_{i+1} - 1 - 2\pi_{i} = 1$;
      \item se $\pi_{i+1} - \pi_i > 1$, então $\pi'_{2i+1} - \pi'_{2i} = 2\pi_{i+1} - 1 - 2\pi_{i} = 2(\pi_{i+1} - \pi_i) - 1 > 1$;
      \item se $\pi_{i+1} - \pi_i < 1$, então $\pi'_{2i+1} - \pi'_{2i} = 2\pi_{i+1} - 1 - 2\pi_{i} = 2(\pi_{i+1} - \pi_i) - 1 < 1$.
  \end{itemize}

  Dessa forma, $b_{\tau}(\pi) = b_{\M}(\pi')$. Note que o custo mínimo para remover um {\it breakpoint} em $\pi'$ é igual a $\min \{ w_\rho/2, w_\tau/3 \} = w_\tau/3$, pois $w_\tau/w_\rho \leq 1.5$. Então, o limitante inferior do Lema~\ref{cap3:lemma:lb_unsigned} se torna $w_\tau b_{\M}(\pi')/3 = w_\tau b_{\tau}(\pi)/3$.

  Temos que $w(S')$ é igual ao limitante inferior $w_\tau b_{\tau}(\pi)/3$ e, consequentemente, todo rearranjo de $S'$ deve remover exatamente $w' \times 3/w_\tau$ {\it breakpoints}, onde $w'$ é o custo do rearranjo. 

  Suponha, por contradição, que existe rearranjo em $S'$ que não seja uma transposição. Considere $S' = (\beta_1, \beta_2, \ldots, \beta_\ell)$ e seja $\beta_i$ o primeiro rearranjo de $S'$ a ser aplicado que não é uma transposição, ou seja, todos rearranjos em $(\beta_1, \beta_2, \ldots, \beta_{i-1})$ são transposições, $\beta_i$ não é uma transposição e o valor de $i$ é mínimo. Antes de $\beta_i$ ser aplicado, todas as {\it strips} na permutação são crescentes pois transposições que removem três {\it breakpoints} não criam {\it strips} decrescentes. 

  Pelo Lema~\ref{cap3:lemma:unsigned_var_b}, $\beta_i$ não remove $w' \times 3/w_\tau$ {\it breakpoints}, onde $w'$ é o custo de $\beta_i$, o que é uma contradição. Portanto, $S'$ possui apenas transposições.

  Como $w(S') = w_\tau b_{\tau}(\pi)/3$, a sequência $S'$ tem $b_{\tau}(\pi)/3$ transposições. Note que as transposições de $S'$ não quebram os pares $(\pi'_{2i-1}, \pi'_{2i})$, para $1 \leq i \leq n$, pois cada transposição remove três {\it breakpoints}.

  Considere $S' = (\tau'_1, \tau'_2, \ldots, \tau'_{|S'|})$, com $|S'| = b_{\tau}(\pi)/3$. Agora, construímos uma sequência $S = (\tau_1, \tau_2, \ldots, \tau_{|S'|})$ que ordena $\pi$, tal que $\tau_i = \tau((x+1)/2, (y+1)/2, (z+1)/2)$ para $\tau'_i = \tau(x,y,z)$, com $1 \leq i \leq |S'|$.
\end{proof}

\begin{corollary}\label{cap3:theorem:complexity_unsigned_unweighted}
Para os modelos $\M = \{\M_1, \M_2, \M_3, \M_4, \M_5,$ $ \M_6\}$, temos que o problema da Ordenação de Permutações por Rearranjos \pname{SbR} é NP-difícil para permutações sem sinais.
\end{corollary}

\section{Conclusões}

Neste capítulo estudamos problemas de distância de rearranjos utilizando a representação clássica de genomas que envolvem transposições, ou seja, estudamos problemas de Ordenação de Permutações por Transposições e Outros Rearranjos. 

Na Seção~\ref{cap3:secao_novo_algoritmo}, nós apresentamos um novo algoritmo de aproximação com fator de $1.375$ para a Ordenação de Permutações por Transposições. Esse algoritmo tem complexidade de tempo de $O(n^5)$, o que representa uma melhoria em relação ao algoritmo recém publicado por Silva e coautores~\cite{2022-silva-etal} que corrige um problema no algoritmo de Elias e Hartman~\cite{2006-elias-hartman}, sendo esse último um dos algoritmos mais conhecidos da área de rearranjos de genomas. Além da melhoria na complexidade de tempo, nossos testes em permutações pequenas mostraram que o nosso algoritmo também apresenta melhoria na qualidade das soluções encontradas, sendo que o nosso algoritmo encontrou soluções ótimas em mais casos do que os outros dois algoritmos, além de apresentar um valor médio do fator de aproximação menor do que os outros dois algoritmos.

Na Seção~\ref{cap3:secao_complexidade}, nós demonstramos que os problemas de Ordenação de Permutações (com ou sem Sinais) por Rearranjos Ponderados são NP-difíceis para $12$ modelos de rearranjos que incluem transposições junto com a combinação de reversões, transposições inversas e revrevs, considerando que o custo de uma reversão é igual a $w_{\rho}$, os custos de uma transposição, de uma transposição inversa ou de uma revrev são os mesmos e iguais a $w_{\tau}$, e atendendo a restrição de que $w_{\tau}/w_{\rho} \leq 1.5$. Note que quando $w_{\tau}/w_{\rho} = 1$, a versão ponderada é equivalente a não ponderada. Portanto, a prova de NP-dificuldade também é válida para a Ordenação de Permutações (com ou sem Sinais) por Rearranjos considerando esses $12$ modelos.

Quando $w_{\tau}/w_{\rho} > 1.5$, a complexidade desses problemas permanece aberta, assim como a complexidade do problema de Ordenação de Permutações por Reversões e Transposições Ponderadas com a mesma restrição de pesos. Uma direção para trabalhos futuros é o estudo da complexidade nesses casos.

\chapter{Distância em Genomas Desbalanceados}\label{cap:label:indel}

{\revisaof Os primeiros trabalhos da área de rearranjo de genomas envolveram apenas permutações. A partir do ano 1999, trabalhos considerando genomas com conjuntos distintos de genes foram introduzidos~\cite{sankoff1999genome}.} Em 2000, El-Mabrouk~\cite{2000-el-mabrouk} estudou o problema da Distância de Reversões e Indels em Strings com Sinais, apresentando heurísticas baseadas no algoritmo exato de Hannenhalli e Pevzner~\cite{1999-hannenhalli-pevzner} para a versão do problema com permutações. 

Yancopoulos e coautores~\cite{yancopoulos2005efficient} estudaram uma operação de rearranjo chamada DCJ (\textit{Double-Cut-and-Join}). Um DCJ consegue simular reversões e outras operações de rearranjos, mas algumas operações de rearranjo, como as transposições, transposições inversas, revrevs e {\it block interchanges}, não podem ser reproduzidas com uma única operação de DCJ, sendo necessário o uso de duas ou mais operações de DCJ~\cite{2009-fertin-etal}. Tanto a Ordenação de Permutações por DCJs~\cite{bergeron2006unifying} quanto a Distância de DCJs e Indels~\cite{braga2011double} possuem algoritmos polinomiais exatos.  

Usando como base os resultados para a Distância de DCJs e Indels~\cite{braga2011double}, Willing e coautores~\cite{2013-willing-etal} criaram algoritmos polinomiais exatos para a Distância de Reversões e Indels em Strings com Sinais considerando classes específicas de grafos de {\it breakpoints}. Apenas recentemente, em 2020, Willing e coautores~\cite{2021-willing-etal} estenderam o algoritmo anterior e desenvolveram um algoritmo exato polinomial que funciona para qualquer instância do problema da Distância de Reversões e Indels em Strings com Sinais.

Neste capítulo, apresentamos algoritmos de aproximação para a Distância de Rearranjos em Strings com ou sem Sinais para reversões, transposições, a combinação de reversões e transposições, {\it block interchanges}, e a combinação de reversões e {\it block interchanges}. Além disso, exceto para os modelos com {\it block interchanges} e o modelo com reversões para strings com sinais, apresentamos demonstrações que esses problemas são NP-difíceis.

\section{Complexidade dos Problemas}\label{cap4:section:complexidade}

Nesta seção, apresentamos provas de NP-dificuldade para problemas de Distância de Rearranjos em Genomas Desbalanceados considerando os seguintes modelos de rearranjos:

\begin{itemize}
	\item $\Mindel_{\rho}=\{\rho, \psi, \phi\}$: reversões e {\it indels} em strings sem sinais;
	\item $\Mindel_{\tau}=\{\tau, \psi, \phi\}$: transposições e {\it indels} em strings sem sinais;
	\item $\Mindel_{\rho,\tau}=\{\rho, \tau, \psi, \phi\}$: reversões, transposições, e {\it indels} em strings com ou sem sinais.
\end{itemize}

\begin{lemma}\label{cap4:lemma:complexidade}
	O problema de Distância de Rearranjos é NP-difícil para os modelos $\Mindel_{\rho}$ e $\Mindel_{\tau}$, considerando strings sem sinais, e para o modelo $\Mindel_{\rho,\tau}$, considerando strings com ou sem sinais.
\end{lemma}

\begin{proof}
	Considere o modelo $\Mindel_{\rho}$. O problema da Ordenação de Permutações sem Sinais por Reversões (\pname{SbR}) já foi provado ser NP-difícil~\cite{1999a-caprara}. A versão de decisão desse problema tem como entrada uma permutação $\pi$ e um inteiro positivo $k$, consistindo em decidir se a permutação $\pi$ pode ser ordenada por no máximo $k$ reversões.

	De forma similar, a versão de decisão do problema da Distância de Reversões e Indels em Strings sem Sinais (\pname{RID}) recebe como entrada uma instância $\I = (A, \iota^n, k)$, e consiste em decidir se a string $A$ pode ser transformada em $\iota^n$ usando no máximo $k$ operações de reversões ou {\it indels}.

	Nesta demonstração, apresentamos uma redução do problema \pname{SbR} para o problema \pname{RID}. Dada uma instância $(\pi, k)$ para \pname{SbR}, tal que $\pi$ tem tamanho $n$, criamos a instância $(A, \iota^n, k)$, com $A = \pi$, para o problema \pname{RID}. Note que $\pi$ pode ser ordenada por uma sequência de reversões $S$ com $|S| \leq k$ se, e somente se, $A$ pode ser transformada em $\iota^n$ usando no máximo $k$ operações do modelo $\Mindel_{\rho}$, já que $\Sigma_{A} \setminus \Sigma_{\iota^n} = \Sigma_{\iota^n} \setminus \Sigma_{A} = \emptyset$ e {\it indels} não podem ser usados nas sequências de rearranjos.

	A prova é similar para os outros modelos, já que a Ordenação de Permutações por Transposições~\cite{2012-bulteau-etal} e a Ordenação de Permutações com ou sem Sinais por Reversões e Transposições~\cite{2019b-oliveira-etal} são NP-difíceis.
\end{proof}

\section{Algoritmos de Aproximação Usando Breakpoints}\label{cap4:section:breakpoints}

Nesta seção, apresentamos algoritmos de aproximação para o problema de Distância de Rearranjos em Strings sem Sinais considerando os modelos $\Mindel_{\rho}$, $\Mindel_{\tau}$ e $\Mindel_{\rho,\tau}$. Sempre consideramos que as strings de uma instância $\I = (A, \iota^n)$ estão nas suas versões estendidas.

Usamos o conceito de {\it breakpoints}, apresentado na Seção~\ref{cap2:sec:breakpoints_indels}, para a definição de limitantes para a distância e a criação dos algoritmos de aproximação. A seguir, apresentamos uma outra definição utilizada nos limitantes para a distância:

\begin{definition}\label{cap4:def:delta_phi}
Dada uma operação (ou sequência de rearranjos) $\beta$ e uma instância $\I = (A, \iota^n)$, definimos $\Delta \Phi(\I, \beta) = \Delta \Phi(A, \iota^n, \beta) = |\Sigma_{\iota^n} \setminus \Sigma_{A}| - |\Sigma_{\iota^n} \setminus \Sigma_{A'}|$, onde $A' = A \comp \beta$.
\end{definition}

Para uma operação ou sequência $\beta$, se $\Delta \Phi(\I, \beta) > 0$, então $\beta$ diminui a quantidade de elementos que precisam ser adicionados para que as strings se tornem balanceadas.

Assim como nos problemas de Ordenação de Permutações por Rearranjos, usamos o conceito de {\it breakpoints} de reversões sem sinais para o modelo $\Mindel_{\rho,\tau}$. {\revisaof Para simplificar a notação, usamos $b_{\rho}$ e $b_{\tau}$ para indicar {\it breakpoints} de reversões sem sinais e {\it breakpoints} de transposições, respectivamente.} Os próximos lemas mostram como um rearranjo afeta o valor de $b_{\M}(\I) + |\Sigma_{\iota^n} \setminus \Sigma_{A}|$.

\begin{lemma}\label{cap4:lemma:breakpoint_insertion}
	Para qualquer inserção $\insertion$ e $\I = (A, \iota^n)$, temos que:
	\begin{align*}
		\Delta \Phi(\I, \insertion) + \Delta b_{\rho}(\I, \insertion) \leq 2,\\
		\Delta \Phi(\I, \insertion) + \Delta b_{\tau}(\I, \insertion) \leq 2.
	\end{align*}
\end{lemma}

\begin{proof}
	Considere a inserção $\insertion(i,\sigma)$ da string $\sigma$ após o $i$-ésimo elemento de $A$ e seja $A' = A \comp \insertion(i,\sigma)$. 

	Note que $|\sigma| = \Delta \Phi(\I, \insertion)$. Além disso, o único {\it breakpoint} que pode ser removido é o {\it breakpoint} entre os elementos $A_i$ e $A_{i+1}$, caso exista. 

	Lembramos que, como cada par de elementos adjacentes em $\sigma$ representam segmentos maximais contíguos do genoma de destino, todo par de elementos distintos do conjunto $\Sigma_{\iota^n} \setminus \Sigma_{A}$ são não adjacentes em $\iota^n$.

	Se $|\sigma| = 1$, então o limitante é válido. Caso contrário, pelo menos $|\sigma| - 1$ {\it breakpoints} foram adicionados na string, já que existe um {\it breakpoint} entre cada par $(\sigma_k, \sigma_{k+1})$, com $1 \leq k < |\sigma|$. Portanto, temos que $\Delta \Phi(\I, \insertion) + \Delta b_{\rho}(\I, \insertion) \leq 2$ e $\Delta \Phi(\I, \insertion) + \Delta b_{\tau}(\I, \insertion) \leq 2$.
\end{proof}

\begin{lemma}\label{cap4:lemma:breakpoint_deletion}
	Para qualquer deleção $\deletion$ e $\I = (A, \iota^n)$, temos que:
	\begin{align*}
		\Delta \Phi(\I, \deletion) + \Delta b_{\rho}(\I, \deletion) \leq 2,\\
		\Delta \Phi(\I, \deletion) + \Delta b_{\tau}(\I, \deletion) \leq 2.
	\end{align*}
\end{lemma}

\begin{proof}
	Considere a deleção $\deletion(i,j)$ que remove o segmento $(A_i, \ldots, A_j)$. Note que todos os elementos de $(A_i, \ldots, A_j)$ possuem valor $\deletionElement$. De acordo com as definições de {\it breakpoints} apresentadas na Seção~\ref{cap2:sec:breakpoints_indels}, não existe {\it breakpoint} entre um par de elementos quando ambos são iguais a $\deletionElement$. Dessa forma, não existem {\it breakpoints} entre elementos de $(A_i, \ldots, A_j)$. Assim, os únicos {\it breakpoints} que podem ser removidos, caso existam, estão nas posições $(A_{i-1}, A_{i})$ e $(A_j, A_{j+1})$. Como uma deleção não afeta o conjunto $|\Sigma_{\iota^n} \setminus \Sigma_{A}|$, os limitantes são válidos.
\end{proof}

\begin{lemma}\label{cap4:lemma:breakpoint_reversals}
	Para qualquer reversão $\rho$ e $\I = (A, \iota^n)$, temos que:
	\begin{align*}
		\Delta \Phi(\I, \rho) + \Delta b_{\rho}(\I, \rho) \leq 2.
	\end{align*}
\end{lemma}

\begin{proof}
	Considere a reversão $\rho(i,j)$ e seja $A' = A \comp \rho(i,j) = (A_1~\ldots~A_{i-1}~A_j~\ldots$ $A_i~A_{j+1}~\ldots~A_n)$.

	Para $i \leq k < j$, o par $(A_k, A_{k+1})$ é um {\it breakpoint} se, e somente se, o par $(A'_{k'}, A'_{k'+1})$ é um {\it breakpoint}, tal que $A_k = A'_{k'+1}$ e $A_{k+1} = A'_{k'}$. Dessa forma, essa reversão só pode remover {\it breakpoints} entre os pares de elementos $(A_{i-1}, A_i)$ e $(A_j, A_{j+1})$, caso existam. Note que $\Delta \Phi(\I, \rho) = 0$, pois uma reversão não adiciona elementos. Portanto, no melhor cenário, dois {\it breakpoints} são removidos e o limitante é válido.
\end{proof}

\begin{lemma}\label{cap4:lemma:breakpoint_transposition}
  Para qualquer transposição $\tau$ e $\I = (A, \iota^n)$, temos que:
	\begin{align*}
		\Delta \Phi(\I, \tau) + \Delta b_{\rho}(\I, \tau) \leq 3,\\
		\Delta \Phi(\I, \tau) + \Delta b_{\tau}(\I, \tau) \leq 3.
	\end{align*}
\end{lemma}

\begin{proof}
	Similar à prova do Lema~\ref{cap4:lemma:breakpoint_reversals}, mas devemos considerar que uma transposição afeta três adjacências de $A$.
\end{proof}

Com esses lemas, podemos apresentar limitantes para a distância de rearranjos considerando os modelos $\Mindel_{\rho}$, $\Mindel_{\tau}$ e $\Mindel_{\rho,\tau}$.

\begin{lemma}\label{cap4:lemma:breakpoints_lower_bound}
	Para qualquer instância $\I = (A, \iota^n)$ de strings sem sinais, temos que:
	\begin{align*}
		d_{\Mindel_{\rho}}(A, \iota^n) \geq \frac{b_{\rho}(\I) + |\Sigma_{\iota^n} \setminus \Sigma_{A}|}{2},\\
		d_{\Mindel_{\tau}}(A, \iota^n) \geq \frac{b_{\tau}(\I) + |\Sigma_{\iota^n} \setminus \Sigma_{A}|}{3},\\
		d_{\Mindel_{\rho,\tau}}(A, \iota^n) \geq \frac{b_{\rho}(\I) + |\Sigma_{\iota^n} \setminus \Sigma_{A}|}{3}.
	\end{align*}
\end{lemma}

\begin{proof}
	Considere o modelo $\Mindel_{\rho}$. Uma instância $\I' = (A', \iota^n)$ possui $b_{\rho}(\I') + |\Sigma_{\iota^n} \setminus \Sigma_{A'}| = 0$ se, e somente se, $A' = \iota^n$. Portanto, toda sequência de rearranjos que transforma $A$ em $\iota^n$ deve diminuir o valor de $b_{\rho}(\I) + |\Sigma_{\iota^n} \setminus \Sigma_{A}|$ para $0$. Pelos lemas~\ref{cap4:lemma:breakpoint_insertion}, \ref{cap4:lemma:breakpoint_deletion} e \ref{cap4:lemma:breakpoint_reversals}, qualquer reversão ou {\it indel} diminui esse valor em no máximo $2$ e, portanto, o limitante para $d_{\Mindel_{\rho}}(A, \iota^n)$ é válido.

	A prova é similar para os modelos que contém transposições, mas considerando o limitante do Lema~\ref{cap4:lemma:breakpoint_transposition}.
\end{proof}

Os algoritmos apresentados nesta seção são algoritmos gulosos que priorizam os rearranjos com maior valor de $\Delta \Phi(\I, \beta) + \Delta b_{\M}(\I, \beta)$. 
Os próximos lemas apresentam casos em que sempre é possível achar um {\it indel} com $\Delta \Phi(\I, \beta) + \Delta b_{\M}(\I, \beta) > 0$. Essas operações serão úteis nos três algoritmos de aproximação.

\begin{lemma}\label{cap4:lemma:deletion_breakpoint}
	Para qualquer instância $\I = (A, \iota^n)$, tal que $|\Sigma_{A} \setminus \Sigma_{\iota^n}| > 0$, existe uma deleção $\deletion$ tal que $\Delta \Phi(\I, \deletion) + \Delta b_{\M}(\I, \deletion) \geq 1$.
\end{lemma}

\begin{proof}
	Seja $(A_i, \ldots, A_j)$ uma {\it strip} em $A$ com $A_k = \deletionElement$, para $i \leq k \leq j$. Tal {\it strip} deve existir em $A$ já que $|\Sigma_{A} \setminus \Sigma_{\iota^n}| > 0$. Por definição, uma {\it strip} é uma sequência maximal e, portanto, existem {\it breakpoints} entre ambos os pares $(A_{i-1}, A_i)$ e $(A_j, A_{j+1})$.

	A deleção $\deletion(i,j)$ tem o seguinte efeito em $A$:
  \begin{align*}
    A &= (A_1~\ldots~A_{i-1}~A_i~\ldots~A_j~A_{j+1}~\ldots~A_n), \\
    A' &= A \comp \deletion(i,j) = (A_1~\ldots~A_{i-1}~A_{j+1}~\ldots~A_n).
  \end{align*}

  {\revisaof 
  Como o par ($A_{i-1}$, $A_{j+1})$ pode formar um {\it breakpoint}, o número de {\it breakpoints} diminui em pelo menos $1$, enquanto o tamanho de $\Sigma_{\iota^n} \setminus \Sigma_{A}$ permanece o mesmo. Portanto, temos que $\Delta \Phi(\I, \deletion) + \Delta b_{\M}(\I, \deletion) \geq 1$.
  }
\end{proof}

\begin{lemma}\label{cap4:lemma:insertion_no_breakpoint}
	Para qualquer instância $\I = (A, \iota^n)$, tal que $A$ não possui {\it breakpoints} e $|\Sigma_{\iota^n} \setminus \Sigma_{A}| > 0$, existe uma inserção $\insertion$ tal que $\Delta \Phi(\I, \insertion) + \Delta b_{\M}(\I, \insertion) = 1$.
\end{lemma}

\begin{proof}
	Como $A$ não possui {\it breakpoints}, todos os elementos de $A$ estão em ordem crescente. Seja $(A_i, A_{i+1})$ um par de elementos tal que $A_i \neq A_{i+1} - 1$. Como $|\Sigma_{\iota^n} \setminus \Sigma_{A}| > 0$, deve existir pelo menos um par de elementos no qual essa condição é válida. A inserção $\insertion(i, \sigma)$, com $\sigma = (A_i + 1)$, diminui o tamanho de $\Sigma_{\iota^n} \setminus \Sigma_{A}$ em $1$ e não altera o número de {\it breakpoints}. Portanto, $\Delta \Phi(\I, \insertion) + \Delta b_{\M}(\I, \insertion) = 1$.
\end{proof}

\subsection{Algoritmo de $2$-Aproximação para Modelo com Reversões e Indels}

Apresentamos um algoritmo guloso com fator de aproximação igual a $2$ para a Distância de Reversões e Indels em Strings sem Sinais. Nesta seção, consideramos que $\M = \Mindel_{\rho}$. Os próximos lemas apresentam casos em que sempre é possível achar uma reversão que remove {\it breakpoints}. 


\begin{lemma}\label{cap4:lemma:reversal_breakpoint_decreasing}
	Para qualquer instância $\I = (A, \iota^n)$, se $b_{\M}(\I, \beta) > 0$, $|\Sigma_{A} \setminus \Sigma_{\iota^n}| = 0$, e $A$ possui pelo menos uma {\it strip} decrescente, então existe uma reversão que remove pelo menos um {\it breakpoint} de $A$. 
\end{lemma}

\begin{proof}
	Seja $(A_i, \ldots, A_j)$ uma {\it strip} decrescente tal que $A_j$ é mínimo. Seja $(A_{i'}, \ldots, A_{j'})$ a {\it strip} contendo o elemento $\anterior(A_j, \I)$. Pela nossa escolha de $A_j$, a {\it strip} $(A_{i'}, \ldots, A_{j'})$ é crescente e $A_{j'} = \anterior(A_j, \I)$.

	Se $i' < i$, então a reversão $\rho(j'+1, j)$ remove o {\it breakpoint} entre as posições $j'$ e $j'+1$.
  {\small
	  \begin{align*}
	    A &= (0~A_1~\ldots~A_{i'}~\ldots~A_{j'}~A_{j'+1}~\ldots~A_i~\ldots~A_j~A_{j+1}~\ldots~A_m~{n+1}),\\
	    A \comp \rho(j'+1, j) &= (0~A_1~\ldots~A_{i'}~\ldots~A_{j'}~\underline{A_j~\ldots~A_i~\ldots~A_{j'+1}}~A_{j+1}~\ldots~A_m~{n+1}).
	  \end{align*}
  }

  Caso contrário, temos que $i' > i$ e a reversão $\rho(j+1, j')$ remove o {\it breakpoint} entre as posições $j$ e $j+1$.
  {\small
	  \begin{align*}
	    A &= (0~A_1~\ldots~A_{i}~\ldots~A_{j}~A_{j+1}~\ldots~A_{i'}~\ldots~A_{j'}~A_{j'+1}~\ldots~A_m~{n+1}),\\
	    A \comp \rho(j+1, j') &=
	    (0~A_1~\ldots~A_{i}~\ldots~A_{j}~\underline{A_{j'}~\ldots~A_{i'}~\ldots~A_{j+1}}~A_{j'+1}~\ldots~A_m~{n+1}).
	  \end{align*}
  }
\end{proof}

\begin{lemma}\label{cap4:lemma:reversal_breakpoint_no_decreasing}
	Para qualquer instância $\I = (A, \iota^n)$ tal que $b_{\M}(\I, \beta) > 0$, $|\Sigma_{A} \setminus \Sigma_{\iota^n}| = 0$ e $A$ possui pelo menos uma {\it strip} decrescente, se qualquer reversão que remove pelo menos um {\it breakpoint} de $A$ resulta em uma string com apenas {\it strips} crescentes, então existe apenas uma reversão que remove {\it breakpoints} de $A$ e essa reversão remove dois {\it breakpoints}.
\end{lemma}

\begin{proof}
	Suponha, por contradição, que existem duas reversões distintas $\rho(i,j)$ e $\rho(x,y)$ que removem {\it breakpoints} de $A$. Assuma, sem perda de generalidade, que $i < x$. Note que se $\rho$ é uma reversão que deixa $A$ sem {\it strips} decrescentes, então as {\it strips} no intervalo afetado por $\rho$ devem ser todas {\it strips} decrescentes, já que uma {\it strip} crescente nesse intervalo seria transformada em uma {\it strip} decrescente em $A \comp \rho$. As reversões $\rho(i,j)$ e $\rho(x,y)$ são distintas e, portanto, o intervalo $[i, j] - [x, y]$ é não vazio. Além disso, as {\it strips} contidas em $[i, j] - [x, y]$ são decrescentes, já que essas {\it strips} são afetadas por $\rho(i,j)$. Portanto, após aplicar a reversão $\rho(x,y)$, as {\it strips} decrescentes contidas em $[i, j] - [x, y]$ continuam sendo {\it strips} decrescentes em $A \comp \rho(x,y)$, o que é uma contradição. 

	Nesse ponto, considere que $\rho(i,j)$ é a reversão que remove {\it breakpoints} de $A$ e que $A' = A \comp \rho(i,j)$. Todas as {\it strips} fora do intervalo $[i, j]$ são crescentes e, como mencionado, as {\it strips} em $[i, j]$ são decrescentes. Portanto, existem {\it breakpoints} nos pares $(A_{i-1}, A_{i})$ e $(A_{j-1}, A_{j})$.

	Suponha, por contradição, que $\rho(i,j)$ remove apenas um único {\it breakpoint} de $A$. Sem perda de generalidade, assuma que o {\it breakpoint} entre as posições $i-1$ e $i$ não é removido por $\rho(i,j)$. Seja $(A_{x}, \ldots, A_{y})$ a {\it strip} em $[i,j]$ tal que $A_{y}$ é mínimo, e seja $(A_{x'}, \ldots, A_{y'})$ a {\it strip} crescente que contém o elemento $\anterior(A_y, \I)$. 

	Se $x' > y$, então existe reversão que remove um {\it breakpoint} e deixa uma {\it strip} decrescente na string resultante, o que é uma contradição.
	{\small
    \begin{align*}
      A &= (0~A_1~\ldots~A_{x}~\ldots~A_{y}~A_{y+1}~\ldots~A_{x'}~\ldots~A_{y'}~A_{y'+1}~\ldots~A_m~{n+1}),\\
      A \comp \rho(y+1, y') &=
      (0~A_1~\ldots~A_{x}~\ldots~A_{y}~\underline{A_{y'}~\ldots~A_{x'}~\ldots~A_{y+1}}~A_{y'+1}~\ldots~A_m~{n+1})
    \end{align*}
	}

	Portanto, temos que $x' < x$ e $\rho(y'+1, y)$ é uma reversão que remove um {\it breakpoint} entre as posições $y'$ e $y'+1$. Já que $\rho(i,j)$ é a única reversão que remove {\it breakpoints} de $A$, temos que $\rho(i,j) = \rho(y'+1, y)$ e o {\it breakpoint} entre as posições $y' = i-1$ e $y' + 1 = i$ é removido, o que também é uma contradição. Portanto, concluímos que $\rho(i,j)$ remove dois {\it breakpoints} e o resultado enunciado neste lema é válido.
\end{proof}

O Algoritmo $\ref{cap4:algorithm_reversals}$ é um algoritmo guloso que, a cada iteração, escolhe a operação $\beta$ com valor máximo de $\Delta \Phi(\I, \beta) + \Delta b_{\M}(\I, \beta)$. Se o algoritmo chega em um ponto onde não existem operações com $\Delta \Phi(\I, \beta) + \Delta b_{\M}(\I, \beta) > 0$, então a string não possui {\it strips} decrescentes (Lema~\ref{cap4:lemma:reversal_breakpoint_decreasing}). Nesse caso, o algoritmo inverte a {\it strip} $(A_i, \ldots, A_j)$ tal que $i > 0$ e $i$ é mínimo. Dessa forma, na próxima iteração existirá pelo menos uma {\it strip} decrescente na instância e, portanto, existe pelo menos uma operação $\beta$ com $\Delta \Phi(\I, \beta) + \Delta b_{\M}(\I, \beta) > 0$. 
{\revisaof O algoritmo termina quando todos os {\it breakpoints} forem removidos e o conjunto de rótulos das duas strings forem iguais. }

Note que $b_{\M}(\I) + |\Sigma_{\iota^n} \setminus \Sigma_{A}| \in O(n)$. Portanto, o loop do algoritmo executa $O(n)$ vezes. Para encontrar a operação com valor máximo de $\Delta \Phi(\I, \beta) + \Delta b_{\M}(\I, \beta)$ temos que testar todas as possíveis combinações de reversões e {\it indels}, o que leva tempo $O(n^2)$. Portanto, o Algoritmo $\ref{cap4:algorithm_reversals}$ possui complexidade de tempo de $O(n^3)$. Os próximos lemas e teoremas apresentam um limitante superior no número de operações usadas pelo algoritmo e uma prova para o fator de aproximação.

\begin{algorithm}[h]
	\caption{$2$-Aproximação para a Distância de Reversões e Indels em Strings sem Sinais\label{cap4:algorithm_reversals}}
	\DontPrintSemicolon
  \Entrada{Uma instância $\I = (A, \iota^n)$}
  \Saida{Uma sequência de rearranjos que transforma $A$ em $\iota^n$}
  Seja $S \gets \emptyset$\;
  \Enqto{$b_{\M}(\I) + |\Sigma_{\iota^n} \setminus \Sigma_{A}| > 0$}{
  	Seja $\beta$ a operação em $\Mindel_{\rho}$ com valor máximo de $\Delta \Phi(\I, \beta) + \Delta b_{\M}(\I, \beta)$\;
  	\uSe{$\Delta \Phi(\I, \beta) + \Delta b_{\M}(\I, \beta) > 0$}{
  		$A \gets A \comp \beta$\;
  		Adicione $\beta$ na sequência $S$\;
  	}\Senao(){
  		Seja $(A_i, \ldots, A_j)$ uma {\it strip} tal que $i > 0$ e $i$ é mínimo\;
  		$A \gets A \comp \rho(i,j)$\;
  		Adicione $\rho(i,j)$ na sequência $S$\;
  	}
  }
  {\bf retorne} a sequência $S$\;
\end{algorithm}

\begin{lemma}\label{cap4:lemma:reversal_breakpoint_sorting_decreasing}
	Para qualquer instância $\I = (A, \iota^n)$, tal que $A$ possui pelo menos uma {\it strip} decrescente, o Algoritmo~\ref{cap4:algorithm_reversals} transforma $A$ em $\iota^n$ usando no máximo $b_{\M}(\I) + |\Sigma_{\iota^n} \setminus \Sigma_{A}| - 1$ operações.
\end{lemma}

\begin{proof}
	Provaremos por indução no valor de $b_{\M}(\I) + |\Sigma_{\iota^n} \setminus \Sigma_{A}|$ que o resultado é válido.

	Note que não existe instância com $b_{\M}(\I) = 1$ de acordo com a definição de {\it breakpoints}. Além disso, como existe {\it strip} decrescente em $A$, então temos $b_{\M}(\I) > 0$. Portanto, como caso base, considere que $b_{\M}(\I) = 2$ e $|\Sigma_{\iota^n} \setminus \Sigma_{A}| \geq 0$. Sejam $(A_i, A_{i+1})$ e $(A_j, A_{j+1})$ os dois {\it breakpoints} em $A$. Essa string possui exatamente três {\it strips}: $(A_0, \ldots, A_i)$; $(A_{i+1}, \ldots, A_j)$; $(A_{j+1}, \ldots, A_{m+1})$, onde $|A| = m$. Como $A$ possui uma {\it strip} decrescente, temos que a {\it strip} $(A_{i+1}, \ldots, A_j)$ é decrescente e a reversão $\rho(i + 1, j)$ remove os dois {\it breakpoints} de $A$, já que após essa reversão todas as {\it strips} são crescentes (Lema~\ref{cap4:lemma:reversal_breakpoint_no_decreasing}). Quando $b_{\M}(\I) = 0$, pelo Lema~\ref{cap4:lemma:insertion_no_breakpoint}, sempre existe uma inserção $\phi$ com $\Delta \Phi(\I, \beta) = 1$. Portanto, para transformar $A$ em $\iota^n$, o algoritmo usa $1 + |\Sigma_{\iota^n} \setminus \Sigma_{A}| = b_{\M}(\I) + |\Sigma_{\iota^n} \setminus \Sigma_{A}| - 1$ operações.

	Suponha que o resultado é válido para qualquer instância $\I = (A, \iota^n)$ tal que $A$ contém uma {\it strip} decrescente e $b_{\M}(\I) + |\Sigma_{\iota^n} \setminus \Sigma_{A}| \leq k-1$.

	Para uma instância $\I = (A, \iota^n)$ tal que $A$ contém uma {\it strip} decrescente e $b_{\M}(\I) + |\Sigma_{\iota^n} \setminus \Sigma_{A}| = k$, temos os seguintes casos dependendo da operação $\beta$ escolhida pelo algoritmo.

	{\bf Caso 1}: Suponha que a operação escolhida $\beta$ é um {\it indel}. Uma operação de {\it indel} não torna {\it strips} decrescentes em crescentes e só é escolhida pelo algoritmo se $\Delta \Phi(\I, \beta) + \Delta b_{\M}(\I, \beta) > 0$. Seja $\I' = (A', \iota^n)$ com $A' = A \comp \beta$. Essa instância possui  $b_{\M}(\I') + |\Sigma_{\iota^n} \setminus \Sigma_{A'}| = k' \leq k - 1$ e, pela hipótese de indução, o algoritmo transforma a string $A'$ em $\iota^n$ usando no máximo $k'-1$ operações. Portanto, o algoritmo usa no máximo $1 + (k' - 1) = k' \leq k-1 = b_{\M}(\I) + |\Sigma_{\iota^n} \setminus \Sigma_{A}| - 1$ operações para transformar $A$ em $\iota^n$. 

	{\bf Caso 2}: Suponha que $\beta$ é uma reversão: Como $A$ possui uma {\it strip} decrescente, pelo Lema~\ref{cap4:lemma:reversal_breakpoint_decreasing}, temos que $\Delta \Phi(\I, \beta) + \Delta b_{\M}(\I, \beta) > 0$. Seja $\I' = (A', \iota^n)$ com $A' = A \comp \beta$. Essa instância possui  $b_{\M}(\I') + |\Sigma_{\iota^n} \setminus \Sigma_{A'}| = k' \leq k - 1$. 

	Se $A'$ possui pelo menos uma {\it strip} decrescente, então o algoritmo transforma $A'$ em $\iota^n$ usando no máximo $k'-1$ operações (hipótese de indução). Portanto, o algoritmo usa no máximo $1 + (k' - 1) = k' \leq k-1 = b_{\M}(\I) + |\Sigma_{\iota^n} \setminus \Sigma_{A}| - 1$ operações para transformar $A$ em $\iota^n$. 

	Se $A'$ não possui {\it strips} decrescentes, então $b_{\M}(\I') + |\Sigma_{\iota^n} \setminus \Sigma_{A'}| = b_{\M}(\I) + |\Sigma_{\iota^n} \setminus \Sigma_{A}| - 2 = k - 2 = k'$, já que dois {\it breakpoints} foram removidos por $\beta$ (Lema~\ref{cap4:lemma:reversal_breakpoint_no_decreasing}). Na próxima iteração, o algoritmo escolhe uma reversão $\beta'$ que inverte uma {\it strip} de $A'$, gerando a string $A'' = A' \comp \beta'$ que possui uma {\it strip} decrescente. Seja $\I'' = (A'', \iota^n)$ e note que $b_{\M}(\I') + |\Sigma_{\iota^n} \setminus \Sigma_{A'}| = b_{\M}(\I'') + |\Sigma_{\iota^n} \setminus \Sigma_{A''}| = k'$. Pela hipótese de indução, o algoritmo transforma $A''$ em $\iota^n$ usando no máximo $k' - 1$ operações. Portanto, o algoritmo transforma $A$ em $\iota^n$ usando no máximo  $2 + k' - 1 = 1 + (k - 2) = k - 1$ operações. 
\end{proof}

\begin{lemma}\label{cap4:lemma:reversal_breakpoint_sorting}
	Para qualquer instância $\I = (A, \iota^n)$, o Algoritmo~\ref{cap4:algorithm_reversals} transforma $A$ em $\iota^n$ usando no máximo $b_{\M}(\I) + |\Sigma_{\iota^n} \setminus \Sigma_{A}|$ operações.
\end{lemma}

\begin{proof}
	Se $A$ possui uma {\it strip} decrescente, então o resultado desse lema segue diretamente do Lema~\ref{cap4:lemma:reversal_breakpoint_sorting_decreasing}. Caso contrário, usando um argumento similar ao usado no Caso $2$ do Lema~\ref{cap4:lemma:reversal_breakpoint_sorting_decreasing}, temos que o algoritmo eventualmente usa uma reversão para tornar uma {\it strip} crescente em uma {\it strip} decrescente e, além disso, a string resultante $A'$ é transformada em $\iota^n$ usando no máximo $b_{\M}(\I') + |\Sigma_{\iota^n} \setminus \Sigma_{A'}| - 1$ operações (Lema~\ref{cap4:lemma:reversal_breakpoint_sorting_decreasing}). Como todas as operações usadas antes dessa reversão satisfazem a condição $\Delta \Phi(\I, \beta) + \Delta b_{\M}(\I, \beta) \geq 1$, concluímos que o algoritmo usa no máximo $b_{\M}(\I) + |\Sigma_{\iota^n} \setminus \Sigma_{A}|$ operações para transformar $A$ em $\iota^n$.
\end{proof}

\begin{theorem}\label{cap4:theorem:reversal_breakpoint_sorting}
	O Algoritmo~\ref{cap4:algorithm_reversals} é uma $2$-aproximação para o problema da Distância de Reversões e Indels em Strings sem Sinais.
\end{theorem}

\begin{proof}
	Segue diretamente dos lemas~\ref{cap4:lemma:breakpoints_lower_bound} e \ref{cap4:lemma:reversal_breakpoint_sorting}.
\end{proof}

\subsection{Algoritmos de $3$-Aproximação para Modelos com Transposições}

Nesta seção, consideramos o problema da Distância de Transposições em Strings sem Sinais e o problema da Distância de Reversões e Transposições em Strings sem Sinais. Os próximos lemas apresentam casos em que sempre existe uma transposição que remove {\it breakpoints}.

\begin{lemma}\label{cap4:lemma:transposition_breakpoint}
	Para qualquer instância $\I = (A, \iota^n)$, tal que $b_{\tau}(\I) > 0$ e $|\Sigma_{A} \setminus \Sigma_{\iota^n}| = 0$, existe uma transposição que remove pelo menos um {\it breakpoint} de $A$. 
\end{lemma}

\begin{proof}
  {\revisaof 
	  Seja $(A_0, \ldots, A_i)$ a primeira {\it strip} de $A$. Seja $(A_j,\ldots, A_{k-1})$ a {\it strip} que contém o elemento $\posterior(A_i, \I)$. Note que como existem apenas {\it strips} crescentes ao considerar {\it breakpoints} de transposição, temos que $A_j = \posterior(A_i, \I)$. Portanto, a transposição $\tau(i+1,j,k)$ remove o {\it breakpoint} que existe entre $A_i$ e $A_{i+1}$.
	  \begin{align*}
	    A &= (A_1~\ldots~A_{i}~A_{i+1}~\ldots~A_{j-1}~A_{j}~\ldots~A_{k-1}~A_k~\ldots~A_m),\\
	    A \comp \tau(i+1,j,k) &= (A_1~\ldots~A_i~\underline{A_j~\ldots~A_{k-1}}~\underline{A_{i+1}~\ldots~A_{j-1}}~A_k~\ldots~A_m).
	  \end{align*}
  }
\end{proof}



\begin{algorithm}[h]
	\caption{$3$-Aproximação para a Distância de Transposições e Indels em Strings sem Sinais\label{cap4:algorithm_transpositions}}
	\DontPrintSemicolon
  \Entrada{Uma instância $\I = (A, \iota^n)$}
  \Saida{Uma sequência de rearranjos que transforma $A$ em $\iota^n$}
  Seja $S \gets \emptyset$\;
  \Enqto{$b_{\tau}(\I) + |\Sigma_{\iota^n} \setminus \Sigma_{A}| > 0$\label{cap4:breakpoints_alg_transp_loop}}{
  	Seja $\beta$ a operação em $\Mindel_{\tau}$ com valor máximo de $\Delta \Phi(\I, \beta) + \Delta b_{\tau}(\I, \beta)$\;
		$A \gets A \comp \beta$\;
		Adicione $\beta$ na sequência $S$\;
  }
	{\bf retorne} a sequência $S$\;
\end{algorithm}

O Algoritmo~\ref{cap4:algorithm_transpositions} também é um algoritmo guloso que escolhe a operação $\beta$ com valor máximo de $\Delta \Phi(\I, \beta) + \Delta b_{\tau}(\I, \beta)$. Para transposições e {\it indels}, os lemas~\ref{cap4:lemma:deletion_breakpoint}, \ref{cap4:lemma:insertion_no_breakpoint} e \ref{cap4:lemma:transposition_breakpoint} garantem que sempre existe uma operação que satisfaz a condição $\Delta \Phi(\I, \beta) + \Delta b_{\tau}(\I, \beta) > 0$. A complexidade de tempo do Algoritmo~\ref{cap4:algorithm_transpositions} é de $O(n^4)$, já que testar todas as possíveis transposições possui complexidade de tempo de $O(n^3)$ e o loop da linha~\ref{cap4:breakpoints_alg_transp_loop} executa no máximo $O(n)$ vezes.

\begin{lemma}\label{cap4:lemma:ub_sorting_transpositions}
		Para qualquer instância $\I = (A, \iota^n)$, o Algoritmo~\ref{cap4:algorithm_transpositions} transforma $A$ em $\iota^n$ usando no máximo $b_{\tau}(\I) + |\Sigma_{\iota^n} \setminus \Sigma_{A}|$ operações.
\end{lemma}

\begin{proof}
	Segue diretamente do fato de que sempre existe uma operação que satisfaz a condição $\Delta \Phi(\I, \beta) + \Delta b_{\tau}(\I, \beta) > 0$ (lemas~\ref{cap4:lemma:deletion_breakpoint}, \ref{cap4:lemma:insertion_no_breakpoint} e \ref{cap4:lemma:transposition_breakpoint}) e do fato de que $b_{\tau}(\I) + |\Sigma_{\iota^n} \setminus \Sigma_{A}| = 0$ se, e somente se, $A = \iota^n$.
\end{proof}

\begin{theorem}
	O Algoritmo~\ref{cap4:algorithm_transpositions} é uma $3$-aproximação para o problema da Distância de Transposições e Indels em Strings sem Sinais.
\end{theorem}

\begin{proof}
	Segue diretamente dos lemas~\ref{cap4:lemma:breakpoints_lower_bound} e \ref{cap4:lemma:ub_sorting_transpositions}.
\end{proof}

Agora, consideramos o modelo $\Mindel_{\rho,\tau}$. No próximo lema, mostramos que sempre é possível achar uma reversão ou transposição que remove {\it breakpoints} quando $A$ não possui nenhum elemento com rótulo igual a $\deletionElement$.

\begin{lemma}\label{cap4:lemma:breakpoints_unsigned_rt}
	Para qualquer instância $\I = (A, \iota^n)$, tal que $b_{\rho}(\I) > 0$ e $|\Sigma_{A} \setminus \Sigma_{\iota^n}| = 0$, existe uma reversão ou transposição que remove pelo menos um {\it breakpoint} de $A$.
\end{lemma}

\begin{proof}
	Se $A$ possui uma {\it strip} decrescente, então existe uma reversão que remove pelo menos um {\it breakpoint} em $A$ (Lema~\ref{cap4:lemma:reversal_breakpoint_decreasing}). Caso contrário, a string $A$ possui apenas {\it strips} crescentes e, pelo Lema~\ref{cap4:lemma:transposition_breakpoint}, existe uma transposição que remove pelo menos um {\it breakpoint} em $A$.
\end{proof}

\begin{algorithm}[h]
	\caption{$3$-Aproximação para a Distância de Reversões, Transposições e Indels em Strings sem Sinais\label{cap4:algorithm_reversals_transpositions}}
	\DontPrintSemicolon
  \Entrada{Uma instância $\I = (A, \iota^n)$}
  \Saida{Uma sequência de rearranjos que transforma $A$ em $\iota^n$}
  Seja $S \gets \emptyset$\;
  \Enqto{$b_{\rho}(\I) + |\Sigma_{\iota^n} \setminus \Sigma_{A}| > 0$\label{cap4:breakpoints_alg_rev_transp_loop}}{
  	Seja $\beta$ a operação em $\Mindel_{\rho,\tau}$ com valor máximo de $\Delta \Phi(\I, \beta) + \Delta b_{\rho}(\I, \beta)$\;
		$A \gets A \comp \beta$\;
		Adicione $\beta$ na sequência $S$\;
  }
	{\bf retorne} a sequência $S$\;
\end{algorithm}

O Algoritmo~\ref{cap4:algorithm_reversals_transpositions} é similar ao Algoritmo~\ref{cap4:algorithm_transpositions}, exceto pelo fato de que consideramos o modelo $\Mindel_{\rho,\tau}$ e a definição de {\it breakpoints} de reversões sem sinais. A complexidade desse algoritmo também é $O(n^4)$.

\begin{lemma}\label{cap4:lemma:ub_sorting_reversals_transpositions}
		Para qualquer instância $\I = (A, \iota^n)$, o Algoritmo~\ref{cap4:algorithm_reversals_transpositions} transforma $A$ em $\iota^n$ usando no máximo $b_{\rho}(\I) + |\Sigma_{\iota^n} \setminus \Sigma_{A}|$ operações.
\end{lemma}

\begin{proof}
	Similar à prova do Lema~\ref{cap4:lemma:ub_sorting_transpositions}.
\end{proof}

\begin{theorem}
	O Algoritmo~\ref{cap4:algorithm_reversals_transpositions} é uma $3$-aproximação para o problema da Distância de Reversões, Transposições e Indels em Strings sem Sinais.
\end{theorem}

\begin{proof}
	Segue diretamente dos lemas~\ref{cap4:lemma:breakpoints_lower_bound} e \ref{cap4:lemma:ub_sorting_reversals_transpositions}.
\end{proof}

\section{Algoritmos de Aproximação Usando Grafo de Ciclos Rotulado}\label{cap4:section:grafo_ciclos}

Nesta seção, apresentamos algoritmos de $2$-aproximação para os seguintes modelos:

\begin{itemize}
	\item $\Mindel_{\tau}=\{\tau, \psi, \phi\}$: transposições e {\it indels} em strings sem sinais;
	\item $\Mindel_{\BI}=\{\BI, \psi, \phi\}$: {\it block interchanges} e {\it indels} em strings sem sinais;
	\item $\Mindel_{\rho,\tau}=\{\rho, \tau, \psi, \phi\}$: reversões, transposições, e {\it indels} em strings com sinais;
	\item $\Mindel_{\rho,\BI}=\{\rho, \BI, \psi, \phi\}$: reversões, {\it block interchanges}, e {\it indels} em strings com sinais.
\end{itemize}

Esses algoritmos e os limitantes inferiores apresentados nesta seção utilizam os conceitos relacionados ao grafo de ciclos rotulado (Seção~\ref{cap2:sec:grafo_ciclos_strings}). Sempre consideramos que as strings de uma instância $\I = (A, \iota^n)$ e suas formas simplificadas ($\pi^A$ e $\pi^\iota$) estão nas suas versões estendidas.

Além das definições de ciclos limpos e rotulados, que são conceitos específicos de grafo de ciclos rotulados, a seguir introduzimos o conceito de \emph{runs}. Os \emph{runs} são úteis na definição de limitantes inferiores e na criação de algoritmos. 

Um \emph{run de inserção} é um caminho maximal que inicia e termina com arestas de destino rotuladas e toda aresta de origem nesse caminho é uma aresta limpa. De forma similar, um \emph{run de deleção} é um caminho maximal que inicia e termina com arestas de origem rotuladas e toda aresta de destino nesse caminho é uma aresta limpa. O número de {\it runs} (inserção ou deleção) de um ciclo $C$ é denotado por $\Lambda(C)$.

\begin{lemma}\label{cap4:lemma:even_runs}
	Para toda instância $\I = (A, \iota^n)$ e grafo $G(\I)$, qualquer ciclo $C \in G(\I)$ possui zero, um, ou um número par de {\it runs}.
\end{lemma}

\begin{proof}
	Suponha, por contradição, que um ciclo $C$ em $G(\I)$ possui $x$ {\it runs}, tal que $x > 1$ e $x$ é ímpar. Sejam $r_1, r_2, \ldots, r_x$ os {\it runs} de $C$ na ordem que eles são percorridos no ciclo, considerando que começamos da aresta de origem mais à direita, ou seja, da aresta de origem com maior valor de índice. Pela definição de um {\it run}, para qualquer $1 \leq i \leq x$, os {\it runs} $r_i$ e $r_{j}$ devem ser de tipos diferentes, onde $j = (i~\Mod~{x}) + 1$. No entanto, como $x$ é um número ímpar maior que $1$ e existem apenas dois tipos de {\it runs}, $r_1$ e $r_x$ devem ser do mesmo tipo, o que é uma contradição.
\end{proof}

Dizemos que uma inserção $\insertion$ remove um {\it run} de inserção $r$ se todas as arestas de destino rotuladas de $r$ são transformadas em arestas limpas após a aplicação da inserção $\insertion$. De forma similar, dizemos que uma deleção $\deletion$ remove um {\it run} de deleção $r$ se todas as arestas de origem rotuladas de $r$ se tornam arestas limpas após a aplicação da deleção $\deletion$. 

{\revisaof
Note que para cada elemento inserido na string $A$, temos que um par de vértices, uma aresta de origem e uma aresta de destino são inseridos no grafo de ciclos rotulado. Portanto, uma inserção $\insertion(i,\sigma)$ adiciona $2|\sigma|$ vértices, $|\sigma|$ arestas de origem e $|\sigma|$ arestas de destino no grafo. No entanto, uma deleção afeta apenas o rótulo de uma única aresta de origem, já que uma deleção só pode ser aplicada em uma sequência contígua de elementos iguais a $\deletionElement$. Dessa forma, um {\it run} de deleção pode ser removido somente se esse {\it run} corresponde a uma única aresta de origem rotulada.
}

O \emph{potencial de indels} de um ciclo $C$ relaciona o número de {\it runs} em um ciclo e o número de {\it indels} necessários para remover todos os {\it runs} desse ciclo. 

\begin{definition}
	Dado um ciclo $C$, o \emph{potencial de indels} de $C$ é igual a:
	\begin{align*}
		\lambda(C) =
		\begin{cases}
			\Bigl\lceil \frac{\Lambda(C) + 1}{2} \Bigr\rceil, &\text{se $\Lambda(C) > 0$} \\
			0, &\text{caso contrário}.
		\end{cases}
	\end{align*}
\end{definition}

Definimos $\lambda(\I) = \lambda(A, \iota^n) = \sum_{C \in G(A, \iota^n)} \lambda(C)$. Dado um rearranjo (ou sequência de rearranjos) $\beta$, denotamos $\Delta \lambda(\I, \beta) = \Delta \lambda(A, \iota^n, \beta) = \lambda(A, \iota^n) - \lambda(A \comp \beta, \iota^n)$.

Para um ciclo $C$ com {\it runs} $r_1, r_2, \ldots, r_x$, com $x > 2$, ao remover um {\it run} $r_i$ de $C$, os seus {\it runs} adjacentes $r_{j}$ e $r_{k}$ formam um único {\it run} no novo ciclo, onde $k = (i~\Mod~{x}) + 1$ e $j = i-1$, se $i > 1$, ou $j = x$, se $i = 1$. Dessa forma, uma operação $\beta$ que remove um {\it run} de um ciclo tem $\Delta \lambda(\I, \beta) = 1$.

\begin{definition}
	O \emph{potencial de inserções} $\lambda_\insertion(C)$ de um ciclo $C$ é denotado por:
	\begin{align*}
		\lambda_\insertion(C) =
		\begin{cases}
			\lambda(C) - 1, &\text{se $\Lambda(C) > 1$} \\ 
			1, &\text{se $C$ possui apenas um {\it run} de inserção} \\
			0, &\text{se $C$ não possui {\it runs} de inserção}.
		\end{cases}
	\end{align*}
\end{definition}

Além disso, definimos $\lambda_\insertion(\I) = \lambda_\insertion(A, \iota^n) = \sum_{C \in G(A, \iota^n)} \lambda_\insertion(C)$. Dado um rearranjo (ou sequência de rearranjos) $\beta$, denotamos $\Delta \lambda_\insertion(\I, \beta) = \Delta \lambda_\insertion(A, \iota^n, \beta) = \lambda_\insertion(A, \iota^n) - \lambda_\insertion(A \comp \beta, \iota^n)$. 

{\revisaof
Lembre que, para um grafo de ciclos rotulado $G(\I)$, as definições de ciclos limpos e ciclos rotulados consideram apenas as arestas de origem. Dessa forma, chegamos aos resultados das observações~\ref{cap4:remark:equivalencia_lambda_ciclo} e \ref{cap4:remark:equivalencia_lambda_soma}.
}

\begin{remark}\label{cap4:remark:equivalencia_lambda_ciclo}
{\revisaof
Para todo ciclo limpo $C$, temos que $\lambda_\insertion(C) = \lambda({C})$ e, para qualquer ciclo rotulado $C'$, temos que $\lambda_\insertion(C') = \lambda(C') - 1$.
}
\end{remark}

\begin{remark}\label{cap4:remark:equivalencia_lambda_soma}
Sejam $H$ e $H'$ os conjuntos de ciclos limpos e de ciclos rotulados de $G(\I)$, respectivamente. Temos que:
\begin{align*}
    \lambda(\I) &= \sum_{C \in H} \lambda(C) +  \sum_{C \in H'} \lambda(C) \\
    &= \sum_{C \in H} \lambda_\insertion(C) +  \sum_{C \in H'} (\lambda_\insertion(C) + 1) \\
    &= \left(\sum_{C \in H} \lambda_\insertion(C) +  \sum_{C \in H'} \lambda_\insertion(C) \right) + \crotulado(\I) \\
    &= \lambda_\insertion(\I) + \crotulado(\I).
\end{align*}

Dessa forma,

\begin{align*}
	&|\pi^A| + 1 - c(\I) + \lambda(\I) = \\
	&|\pi^A| + 1 - (\cclean(\I) + \crotulado(\I)) + (\lambda_\insertion(\I) + \crotulado(\I)) = \\
	&|\pi^A| + 1 - \cclean(\I) + \lambda_\insertion(\I).
\end{align*}
\end{remark}

O grafo $G(\I)$ só possui ciclos unitários e potencial de {\it indel} igual a zero se, e somente se, $A = \iota^n$. Em outras palavras, $A = \iota^n$ se, e somente se, $n + 1 - c(\I) + \lambda(\I) = n + 1 - \cclean(\I) + \lambda_\insertion(\I) = 0$. Dessa forma, transformar $A$ em $\iota^n$ pode ser interpretado como tornar $|\pi^A| + 1 - c(\I) + \lambda(\I) = |\pi^A| + 1 - \cclean(\I) + \lambda_\insertion(\I) = 0$. Observe que quando $A = \iota^n$, temos que $|\pi^A| = n$.  

Agora, apresentamos limitantes para o valor de $\Delta c(\I, \beta) + \Delta \lambda(\I, \beta)$ dependendo do tipo do rearranjo $\beta$.

\begin{lemma}\label{cap4:lemma:lb_graph_deletion}
Para qualquer deleção $\deletion$ e instância $\I = (A, \iota^n)$, temos que $\Delta c(\I, \deletion) + \Delta \lambda(\I, \deletion) \leq 1$.
\end{lemma}

\begin{proof}
Uma deleção pode apenas remover o rótulo de uma aresta de origem, já que deve ser aplicada a uma sequência contígua de elementos $\deletionElement$. No melhor cenário, essa operação remove um {\it run} de deleção que é formado por apenas uma aresta de origem rotulada e diminui o potencial de {\it indel} em uma unidade, ou seja, $\Delta c(\I, \deletion) = 1$. Os ciclos e vértices de $G(\I)$ não são afetados e, portanto, $\Delta c(\I, \deletion) = 0$.
\end{proof}

\begin{lemma}\label{cap4:lemma:lb_graph_insertion}
Para qualquer inserção $\insertion$ e instância $\I = (A, \iota^n)$, temos que $\Delta c(\I, \insertion) + \Delta \lambda(\I, \insertion) \leq 1$.
\end{lemma}

\begin{proof}
	Considere uma inserção que remove o rótulo de arestas de destino de um mesmo ciclo $C$. Se as arestas afetadas pertencem ao mesmo {\it run} de inserção, no melhor cenário, esse {\it run} é removido e $\Delta \lambda(\I, \insertion) = 1$.
	No melhor cenário, um novo ciclo é criado para cada aresta de origem inserida no grafo e $\Delta c(\I, \insertion) = 0$.

	Se as arestas afetadas pertencem a $x$ diferentes {\it runs} de um mesmo ciclo $C$, no melhor cenário, $x$ {\it runs} de inserção são removidos, mas essa inserção cria pelo menos $x$ ciclos que possuem {\it runs} de deleção, como mostrado na Figura~\ref{cap4:fig:insertion_run_multiple}. 
	Sejam $D_1, D_2, \ldots, D_y$ os ciclos criados pela inserção que possuem exatamente um {\it run}, e sejam $D'_1, D'_2, \ldots, D'_z$ os ciclos criados pela inserção que possuem mais que um {\it run}. Pelo Lema~\ref{cap4:lemma:even_runs}, $\Lambda(D'_i)$ é par para todo $1 \leq i \leq z$. Note que $y + z \geq x$ e que $\Lambda(C) = x + \sum_{1 \leq i \leq y} \Lambda(D_i) + \sum_{1 \leq i \leq z} \Lambda(D'_i) = x + y + \sum_{1 \leq i \leq z} \Lambda(D'_i)$. A soma do potencial de {\it indel} dos novos ciclos é:
	\begin{align*}
		&~\sum_{1 \leq i \leq y} \left \lceil \frac{\Lambda(D_i) + 1}{2} \right \rceil + \sum_{1 \leq i \leq z} \left \lceil \frac{\Lambda(D'_i) + 1}{2} \right \rceil \\
		=&~ y + \sum_{1 \leq i \leq z} \left (\frac{\Lambda(D'_i)}{2} + 1 \right)
		~\,=~ y + z + \frac{\sum_{1 \leq i \leq z} \Lambda(D'_i)}{2}
		\\
		=&~ \frac{2y + 2z + \sum_{1 \leq i \leq z} \Lambda(D'_i)}{2}
		\geq~ \frac{y + z + x + \sum_{1 \leq i \leq z} \Lambda(D'_i)}{2} \\
		=&~ \frac{\Lambda(C) + z}{2}
		\geq~ \frac{\Lambda(C)}{2} \geq \lambda(C) - 1.
	\end{align*}

	Portanto, temos que $\Delta \lambda(\I, \insertion) \leq 1$.
	Como anteriormente, no melhor cenário, um novo ciclo é criado para cada elemento inserido no grafo e $\Delta c(\I, \insertion) = 0$.

	Agora, considere uma inserção que remove o rótulo de arestas de destino de $k$ ciclos distintos. Se um {\it run} de inserção é removido de cada ciclo, então $\Delta \lambda(\I, \insertion) = k$. No entanto, esses ciclos são unidos pelos novos elementos adicionados no grafo e, consequentemente, $\Delta c(\I, \insertion) \leq -(k-1)$. Portanto, temos que $\Delta c(\I, \insertion) + \Delta \lambda(\I, \insertion) \leq 1$. A Figura~\ref{cap4:fig:insertion_multi_cycles} mostra um exemplo de tal operação. Quando mais de um {\it run} é removido de cada ciclo, um argumento similar ao usado anteriormente pode ser usado.
\end{proof}

\begin{figure}[t]
\centering
\resizebox{0.9\textwidth}{!}{
\begin{tikzpicture}[scale=1]
\scriptsize
\begin{scope}[every node/.style={inner sep=0.3mm, draw, circle, minimum size = 0pt}]
    \node[label=below:$+0$\phantom{1}] (p0) at (0,0) {};
    \node[label=below:$-11$] (m6) at (1,0) {};
    \node[label=below:$+11$] (p6) at (2,0) {};
    \node[label=below:$-9$\phantom{1}] (m5) at (3,0) {};
    \node[label=below:$+9$\phantom{1}] (p5) at (4,0) {};
    \node[label=below:$-7$\phantom{1}] (m4) at (5,0) {};
    \node[label=below:$+7$\phantom{1}] (p4) at (6,0) {};
    \node[label=below:$-5$\phantom{1}] (m3) at (7,0) {};
    \node[label=below:$+5$\phantom{1}] (p3) at (8,0) {};
    \node[label=below:$-3$\phantom{1}] (m2) at (9,0) {};
    \node[label=below:$+3$\phantom{1}] (p2) at (10,0) {};
    \node[label=below:$-2$\phantom{1}] (m1) at (11,0) {};
    \node[label=below:$+2$\phantom{1}] (p1) at (12,0) {};
    \node[label=below:$-13$] (m8) at (13,0) {};
\end{scope}

\begin{scope}[every edge/.style={draw=black}]
    \path [<-] (0.5, -0.1) edge (0.5, -1.5);
    \node[] (insertion) at (0.5, -1.75) {(1~4~8)};
\end{scope}

\begin{scope}[every edge/.style={draw=black}]
    \path [-] (p5) edge node [black, pos=0.5, sloped, above] {} (m4);
    \path [-] (p3) edge node [black, pos=0.5, sloped, above] {} (m2);
    \path [-] (p1) edge node [black, pos=0.5, sloped, above] {} (m8);
    \path [-] (p4) edge node [black, pos=0.5, sloped, above] {} (m3);
    \path [-] (p6) edge node [black, pos=0.5, sloped, above] {} (m5);
\end{scope}

\begin{scope}[dashed]
    \path [-] (p0) edge node [black, pos=0.5, sloped, above, yshift=-0.05cm] {$\deletionElement$} (m6);
    \path [-] (p2) edge node [black, pos=0.5, sloped, above, yshift=-0.05cm] {$\deletionElement$} (m1);
\end{scope}

\begin{scope}[every edge/.style={draw=black}]
    \path [-] (p1) edge  [bend right=80] (m2);
\end{scope}

\begin{scope}[dashed]
    \path [-] (p5) edge  [bend right=80] node [black, pos=0.5, sloped, above] {$10$} (m6);
    \path [-] (p0) edge  [bend left=50] node [black, pos=0.5, sloped, above] {$1$} (m1);
    \path [-] (p2) edge  [bend right=80] node [black, pos=0.5, sloped, above] {$4$} (m3);
    \path [-] (p4) edge  [bend right=80] node [black, pos=0.5, sloped, above] {$8$} (m5);
    \path [-] (p3) edge  [bend right=80] node [black, pos=0.5, sloped, above] {$6$} (m4);
    \path [-] (p6) edge  [bend left=50] node [black, pos=0.5, sloped, above] {$12$} (m8);
\end{scope}

\begin{scope}[every node/.style={inner sep=0.3mm, draw, circle, minimum size = 0pt}]
    \node[label=below:$+0$\phantom{1}] (2p0) at (-3,-5) {};
    \node[label=below:$-1$\phantom{1}] (2m1) at (-2,-5) {};
    \node[label=below:$+1$\phantom{1}] (2p1) at (-1,-5) {};
    \node[label=below:$-4$\phantom{1}] (2m4) at (0,-5) {};
    \node[label=below:$+4$\phantom{1}] (2p4) at (1,-5) {};
    \node[label=below:$-8$\phantom{1}] (2m8) at (2,-5) {};
    \node[label=below:$+8$\phantom{1}] (2p8) at (3,-5) {};
    \node[label=below:$-11$] (2m11) at (4,-5) {};
    \node[label=below:$+11$] (2p11) at (5,-5) {};
    \node[label=below:$-9$\phantom{1}] (2m9) at (6,-5) {};
    \node[label=below:$+9$\phantom{1}] (2p9) at (7,-5) {};
    \node[label=below:$-7$\phantom{1}] (2m7) at (8,-5) {};
    \node[label=below:$+7$\phantom{1}] (2p7) at (9,-5) {};
    \node[label=below:$-5$\phantom{1}] (2m5) at (10,-5) {};
    \node[label=below:$+5$\phantom{1}] (2p5) at (11,-5) {};
    \node[label=below:$-3$\phantom{1}] (2m3) at (12,-5) {};
    \node[label=below:$+3$\phantom{1}] (2p3) at (13,-5) {};
    \node[label=below:$-2$\phantom{1}] (2m2) at (14,-5) {};
    \node[label=below:$+2$\phantom{1}] (2p2) at (15,-5) {};
    \node[label=below:$-13$] (2m13) at (16,-5) {};
\end{scope}

\begin{scope}[every edge/.style={draw=black}]
    \path [-] (2p1) edge node [black, pos=0.5, sloped, above] {} (2m4);
    \path [-] (2p4) edge node [black, pos=0.5, sloped, above] {} (2m8);
    \path [-] (2p9) edge node [black, pos=0.5, sloped, above] {} (2m7);
    \path [-] (2p7) edge node [black, pos=0.5, sloped, above] {} (2m5);
    \path [-] (2p5) edge node [black, pos=0.5, sloped, above] {} (2m3);
    \path [-] (2p2) edge node [black, pos=0.5, sloped, above] {} (2m13);
    \path [-] (2p0) edge node [black, pos=0.5, sloped, above] {} (2m1);
    \path [-] (2p11) edge node [black, pos=0.5, sloped, above] {} (2m9);
\end{scope}

\begin{scope}[dashed]
  
  \path [-] (2p3) edge node [black, pos=0.5, sloped, above, yshift=-0.05cm] {$\deletionElement$} (2m2);
  \path [-] (2p8) edge node [black, pos=0.5, sloped, above, yshift=-0.05cm] {$\deletionElement$} (2m11);

\end{scope}

\begin{scope}[every edge/.style={draw=black}]
    \path [-] (2p0) edge  [bend left=90] (2m1);
    \path [-] (2p1) edge  [bend left=50] (2m2);
    \path [-] (2p2) edge  [bend right=70] (2m3);
    \path [-] (2p3) edge  [bend right=50] (2m4);
    \path [-] (2p4) edge  [bend left=50] (2m5);
    \path [-] (2p7) edge  [bend right=50] (2m8);
    \path [-] (2p8) edge  [bend left=50] (2m9);
\end{scope}

\begin{scope}[dashed]
    \path [-] (2p5) edge  [bend right=50] node [black, pos=0.5, sloped, above] {$6$} (2m7);
    \path [-] (2p9) edge  [bend right=50] node [black, pos=0.5, sloped, above] {$10$} (2m11);
    \path [-] (2p11) edge  [bend left=50] node [black, pos=0.5, sloped, above] {$12$} (2m13);
\end{scope}
\end{tikzpicture}
}
\caption{
Exemplo de uma inserção que remove o rótulo de arestas de destino de diferentes {\it runs} de um mesmo ciclo. Neste exemplo, temos $A = (0~\alpha~11~9~7~5~3~\alpha~2~13)$ e $\iota^n$ com $n = 12$. A operação aplicada em $A$ é a inserção $\insertion(0, (1~4~8))$.
\label{cap4:fig:insertion_run_multiple}}
\end{figure}
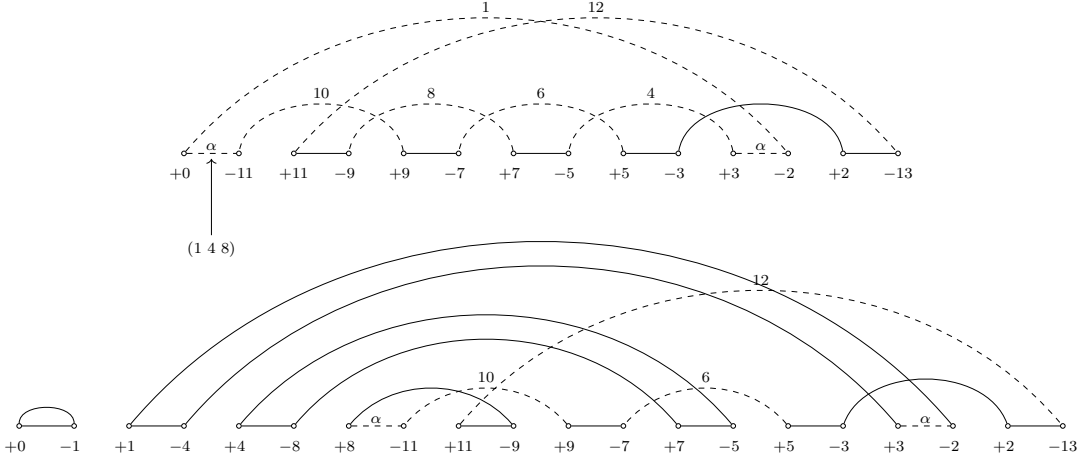

\begin{figure}[t]
\centering
\begin{tikzpicture}[scale=0.75]
\scriptsize
\begin{scope}[every node/.style={inner sep=0.4mm, draw, circle, minimum size = 0pt}]
    \node[label=below:$+0$\phantom{+}] (p0) at (0,0) {};
    \node[label=below:$-3$\phantom{+}] (m3) at (1,0) {};
    \node[label=below:$+3$\phantom{+}] (p3) at (2,0) {};
    \node[label=below:$-2$\phantom{+}] (m2) at (3,0) {};
    \node[label=below:$+2$\phantom{+}] (p2) at (4,0) {};
    \node[label=below:$-4$\phantom{+}] (m4) at (5,0) {};
    \node[label=below:$+4$\phantom{+}] (p4) at (6,0) {};
    \node[label=below:$-6$\phantom{+}] (m6) at (7,0) {};
    \node[label=below:$+6$\phantom{+}] (p6) at (8,0) {};
    \node[label=below:$-5$\phantom{+}] (m5) at (9,0) {};
    \node[label=below:$+5$\phantom{+}] (p5) at (10,0) {};
    \node[label=below:$-8$\phantom{+}] (m8) at (11,0) {};
\end{scope}

\begin{scope}[every edge/.style={draw=black}]
    \path [<-] (0.5, -0.1) edge (0.5, -1.5);
    \node[] (insertion) at (0.5, -1.75) {(1~7)};
\end{scope}

\begin{scope}[every edge/.style={draw=black}]
    \path [-] (p0) edge node [black, pos=0.5, sloped, above] {} (m3);
    \path [-] (p3) edge node [black, pos=0.5, sloped, above] {} (m2);
    \path [-] (p2) edge node [black, pos=0.5, sloped, above] {} (m4);
    \path [-] (p4) edge node [black, pos=0.5, sloped, above] {} (m6);
    \path [-] (p6) edge node [black, pos=0.5, sloped, above] {} (m5);
    \path [-] (p5) edge node [black, pos=0.5, sloped, above] {} (m8);
\end{scope}

\begin{scope}[every edge/.style={draw=black}]
    \path [-] (p2) edge  [bend right=50] (m3);
    \path [-] (p3) edge  [bend left=50] (m4);
    \path [-] (p4) edge  [bend left=50] (m5);
    \path [-] (p5) edge  [bend right=50] (m6);
\end{scope}

\begin{scope}[dashed]
    \path [-] (p0) edge  [bend left=50] node [black, pos=0.5, sloped, above] {$1$} (m2);
    \path [-] (p6) edge  [bend left=50] node [black, pos=0.5, sloped, above] {$7$} (m8);
\end{scope}

\begin{scope}[every node/.style={inner sep=0.4mm, draw, circle, minimum size = 0pt}]
    \node[label=below:$+0$\phantom{+}] (2p0) at (-2,-4) {};
    \node[label=below:$-1$\phantom{+}] (2m1) at (-1,-4) {};
    \node[label=below:$+1$\phantom{+}] (2p1) at (0,-4) {};
    \node[label=below:$-7$\phantom{+}] (2m7) at (1,-4) {};
    \node[label=below:$+7$\phantom{+}] (2p7) at (2,-4) {};
    \node[label=below:$-3$\phantom{+}] (2m3) at (3,-4) {};
    \node[label=below:$+3$\phantom{+}] (2p3) at (4,-4) {};
    \node[label=below:$-2$\phantom{+}] (2m2) at (5,-4) {};
    \node[label=below:$+2$\phantom{+}] (2p2) at (6,-4) {};
    \node[label=below:$-4$\phantom{+}] (2m4) at (7,-4) {};
    \node[label=below:$+4$\phantom{+}] (2p4) at (8,-4) {};
    \node[label=below:$-6$\phantom{+}] (2m6) at (9,-4) {};
    \node[label=below:$+6$\phantom{+}] (2p6) at (10,-4) {};
    \node[label=below:$-5$\phantom{+}] (2m5) at (11,-4) {};
    \node[label=below:$+5$\phantom{+}] (2p5) at (12,-4) {};
    \node[label=below:$-8$\phantom{+}] (2m8) at (13,-4) {};
\end{scope}

\begin{scope}[every edge/.style={draw=black}]
    \path [-] (2p0) edge node [black, pos=0.5, sloped, above] {} (2m1);
    \path [-] (2p1) edge node [black, pos=0.5, sloped, above] {} (2m7);
    \path [-] (2p7) edge node [black, pos=0.5, sloped, above] {} (2m3);
    \path [-] (2p3) edge node [black, pos=0.5, sloped, above] {} (2m2);
    \path [-] (2p2) edge node [black, pos=0.5, sloped, above] {} (2m4);
    \path [-] (2p4) edge node [black, pos=0.5, sloped, above] {} (2m6);
    \path [-] (2p6) edge node [black, pos=0.5, sloped, above] {} (2m5);
    \path [-] (2p5) edge node [black, pos=0.5, sloped, above] {} (2m8);
\end{scope}

\begin{scope}[every edge/.style={draw=black}]
    \path [-] (2p0) edge  [bend left=80] (2m1);
    \path [-] (2p1) edge  [bend left=50] (2m2);
    \path [-] (2p2) edge  [bend right=50] (2m3);
    \path [-] (2p3) edge  [bend left=50] (2m4);
    \path [-] (2p4) edge  [bend left=50] (2m5);
    \path [-] (2p5) edge  [bend right=50] (2m6);
    \path [-] (2p6) edge  [bend right=50] (2m7);
    \path [-] (2p7) edge  [bend left=50] (2m8);
\end{scope}

\begin{scope}[dashed]
\end{scope}
\end{tikzpicture}

\caption{
Exemplo de uma inserção que remove o rótulo de arestas de destino de diferentes ciclos. Neste exemplos, temos $A = (0~3~2~4~6~5~8)$ e $\iota^n$ com $n = 7$. A operação aplicada em $A$ é a inserção $\insertion(0, (1~7))$.
\label{cap4:fig:insertion_multi_cycles}}
\end{figure}
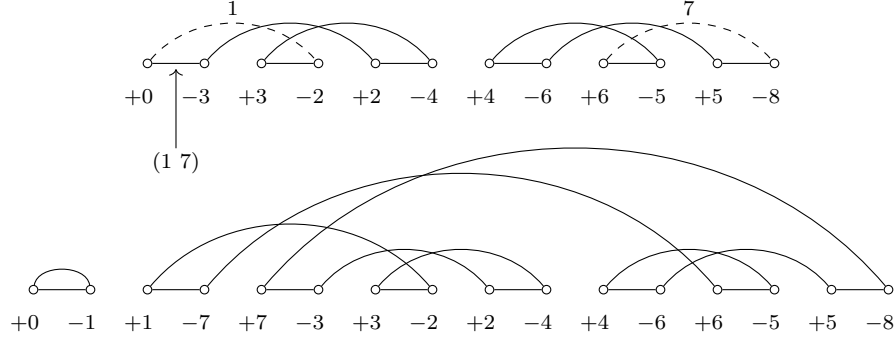

\begin{lemma}\label{cap4:lemma:lb_graph_bi}
	Para qualquer {\it block interchange} $\BI$ e instância $\I = (A, \iota^n)$, temos que $\Delta c(\I, \BI) + \Delta \lambda(\I, \BI) \leq 2$.
\end{lemma}

\begin{proof}
	Dividimos essa demonstração de acordo com o número de ciclos afetados por $\BI$~\cite{1996-christie}.

	Se $\BI$ afeta quatro ciclos $C_1, C_2, C_3$ e $C_4$, então essa operação transforma esses quatro ciclos em dois novos ciclos $C'_1$ e $C'_2$. Assim, temos que $\Delta c(\I, \BI) = -2$. No melhor cenário, dois {\it runs} de deleção e dois {\it runs} de inserção de $C_1$ e $C_3$ são unidos em $C'_1$. Analogamente, dois {\it runs} de deleção e dois {\it runs} de inserção de $C_2$ e $C_4$ são unidos em $C'_2$. Neste caso, $\Lambda(C'_1) = \Lambda(C_1) + \Lambda(C_3) - 2$ e $\Lambda(C'_2) = \Lambda(C_2) + \Lambda(C_4) - 2$. Portanto, $\Delta \lambda(A, \iota^n, \BI) = 4$ e $\Delta c(\I, \BI) + \Delta \lambda(\I, \BI) = 2$. Um exemplo é apresentado na Figura~\ref{cap4:fig:bound_bi_case1}.

	Se $\BI$ afeta três ciclos $C_1, C_2$ e $C_3$, então essa operação une esses três ciclos em um novo ciclo $C'$. Assim, temos que $\Delta c(\I, \BI) = -2$. De forma similar ao caso anterior, no melhor cenário, o número de {\it runs} diminui em quatro e $\Lambda(C') = \Lambda(C_1) + \Lambda(C_2) + \Lambda(C_3) - 4$. Portanto, $\Delta \lambda(A, \iota^n, \BI) = 4$ e $\Delta c(\I, \BI) + \Delta \lambda(\I, \BI) = 2$.

	Se $\BI$ afeta dois ciclos $C_1$ e $C_2$, então essa operação transforma esses dois ciclos em dois novos ciclos ou em quatro novos ciclos. Se $\BI$ transforma $C_1$ e $C_2$ em dois novos ciclos $C'_1$ e $C'_2$, então, no melhor cenário, o número de {\it runs} diminui em quatro com $\Lambda(C'_1) = \Lambda(C_1) - 2$ e $\Lambda(C'_2) = \Lambda(C_1) - 2$. Dessa forma, o potencial de {\it indel} diminui em uma unidade para cada ciclo e $\Delta \lambda(\I, \BI) = 2$. Como $\Delta c(\I, \BI) = 0$, temos que $\Delta c(\I, \BI) + \Delta \lambda(\I, \BI) = 2$.

	Se $\BI$ afeta dois ciclos $C_1$ e $C_2$, transformando esses ciclos em quatro novos ciclos $C'_1$, $C'_2$, $C'_3$ e $C'_4$, então, no melhor cenário, dois pares de {\it runs} de deleção são unidos, porém note que cada ciclo possui pelo menos um {\it run} de inserção, como mostrado na Figura~\ref{cap4:fig:bound_bi_case3}. Portanto, $\Lambda(C'_1) = X$, tal que $1 \leq X < \Lambda(C_1)$, $\Lambda(C'_2) = \min(\Lambda(C_1) - X - 2, 1)$, $\Lambda(C'_3) = Y$, tal que $1 \leq Y < \Lambda(C_2)$, e $\Lambda(C'_4) = \min(\Lambda(C_2) - Y - 2, 1)$. Portanto, o potencial de {\it indel} do grafo permanece o mesmo ($\Delta \lambda(\I, \BI) = 0$) e $\Delta c(\I, \BI) = 2$.

	Se $\BI$ afeta um único ciclo $C_1$, então essa operação transforma esse ciclo em um novo ciclo ou em três novos ciclos. Se $\BI$ não aumenta o número de ciclos do grafo, então, no melhor cenário, essa operação pode diminuir o número de {\it runs} do ciclo em quatro e $\Delta c(\I, \BI) + \Delta \lambda(\I, \BI) = 2$. Se $\BI$ transforma $C_1$ em três novos ciclos $C'_1, C'_2$ e $C'_3$, então, no melhor cenário, dois pares de {\it runs} de deleção são unidos, mas note que cada ciclo possui pelo menos um {\it run} de inserção. De forma similar ao caso anterior, o potencial de {\it indel} do grafo permanece o mesmo e $\Delta c(\I, \BI) + \Delta \lambda(\I, \BI) = 2$.
\end{proof}

\begin{figure}[t]
\centering
\begin{tikzpicture}[scale=0.9]
\scriptsize

\begin{scope}[every node/.style={inner sep=0.4mm, draw, circle, minimum size = 0pt}]
    \node[label=below:$+0$] (2p0) at (0,0) {};
    \node[label=below:$-2$] (2m2) at (1,0) {};
    \node[label=below:$+2$] (2p2) at (2,0) {};
    \node[label=below:$-4$] (2m4) at (3,0) {};
    \node[label=below:$+4$] (2p4) at (4,0) {};
    \node[label=below:$-6$] (2m6) at (5,0) {};
    \node[label=below:$+6$] (2p6) at (6,0) {};
    \node[label=below:$-8$] (2m8) at (7,0) {};
\end{scope}

\begin{scope}[>={Stealth[black]},
              dashed]
    \path [-] (2p0) edge node [black, pos=0.5, sloped, above, yshift=-0.05cm] {$\deletionElement$} (2m2);
    \path [-] (2p6) edge node [black, pos=0.5, sloped, above, yshift=-0.05cm] {$\deletionElement$} (2m8);
    \path [-] (2p2) edge node [black, pos=0.5, sloped, above, yshift=-0.05cm] {$\deletionElement$} (2m4);
    \path [-] (2p4) edge node [black, pos=0.5, sloped, above, yshift=-0.05cm] {$\deletionElement$} (2m6);
\end{scope}

\begin{scope}[>={Stealth[black]},
              every edge/.style={draw=black}, every node/.style={inner sep=0pt, minimum size = 0pt}]
\node[label=\phantom{}] (bi1) at (0.5, -0.8) {};
\node[label=\phantom{}] (bi2) at (2.5, -0.8) {};
\path [{Bar}-{Bar}] (bi1) edge node [black, pos=0.5, sloped, below] {$(2~\deletionElement)$} (bi2);

\node[label=\phantom{}] (bi3) at (4.5, -0.8) {};
\node[label=\phantom{}] (bi4) at (6.5, -0.8) {};
\path [{Bar}-{Bar}] (bi3) edge node [black, pos=0.5, sloped, below] {$( \deletionElement~6)$} (bi4);
\end{scope}

\begin{scope}[>={Stealth[black]},
              every edge/.style={draw=black}]
\end{scope}

\begin{scope}[>={Stealth[black]},
              dashed]
    \path [-] (2p0) edge  [bend left=80] node [black, pos=0.5, sloped, above] {$1$} (2m2);
    \path [-] (2p2) edge  [bend left=80] node [black, pos=0.5, sloped, above] {$3$} (2m4);
    \path [-] (2p4) edge  [bend left=80] node [black, pos=0.5, sloped, above] {$5$} (2m6);
    \path [-] (2p6) edge  [bend left=80] node [black, pos=0.5, sloped, above] {$7$} (2m8);
\end{scope}

\begin{scope}[every node/.style={inner sep=0.4mm, draw, circle, minimum size = 0pt}]
    \node[label=below:$+0$] (p0) at (0,-3) {};
    \node[label=below:$-6$] (m6) at (1,-3) {};
    \node[label=below:$+6$] (p6) at (2,-3) {};
    \node[label=below:$-4$] (m4) at (3,-3) {};
    \node[label=below:$+4$] (p4) at (4,-3) {};
    \node[label=below:$-2$] (m2) at (5,-3) {};
    \node[label=below:$+2$] (p2) at (6,-3) {};
    \node[label=below:$-8$] (m8) at (7,-3) {};
\end{scope}

\begin{scope}[>={Stealth[black]},
              dashed]
    \path [-] (p0) edge node [black, pos=0.5, sloped, above, yshift=-0.05cm] {$\deletionElement$} (m6);
    \path [-] (p2) edge node [black, pos=0.5, sloped, above, yshift=-0.05cm] {$\deletionElement$} (m8);
\end{scope}

\begin{scope}[>={Stealth[black]},
              every edge/.style={draw=black}]
              
    \path [-] (p6) edge node [black, pos=0.5, sloped, above] {} (m4);
    \path [-] (p4) edge node [black, pos=0.5, sloped, above] {} (m2);
\end{scope}

\begin{scope}[>={Stealth[black]},
              dashed]
    \path [-] (p0) edge  [bend left=50] node [black, pos=0.5, sloped, above] {$1$} (m2);
    \path [-] (p2) edge  [bend right=50] node [black, pos=0.5, sloped, above] {$3$} (m4);
    \path [-] (p4) edge  [bend right=50] node [black, pos=0.5, sloped, above] {$5$} (m6);
    \path [-] (p6) edge  [bend left=50] node [black, pos=0.5, sloped, above] {$7$} (m8);
\end{scope}

\end{tikzpicture}

\caption{
Exemplo de uma operação de {\it block interchange} que age em quatro ciclos. Neste exemplo, temos $A = (0~\deletionElement~2~\deletionElement~4~\deletionElement~6~\deletionElement~8)$ e $\iota^n$ com $n = 7$. O potencial de {\it indel} do grafo original é igual a $4 \times \lceil{(2 + 1)/2}\rceil = 8$ e o potencial de {\it indel} do novo grafo é igual a $\lceil{(2 + 1)/2}\rceil + \lceil{(2 + 1)/2}\rceil = 4$.
\label{cap4:fig:bound_bi_case1}
}

\end{figure}
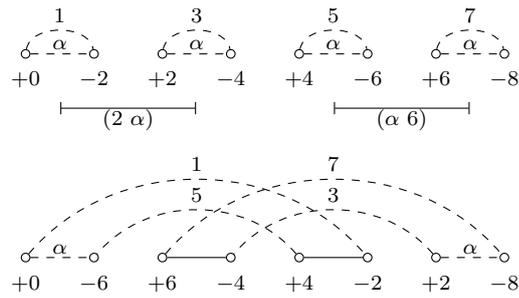
\begin{figure}[t]
\centering
\begin{tikzpicture}[scale=0.9]
\scriptsize
\begin{scope}[every node/.style={inner sep=0.4mm, draw, circle, minimum size = 0pt}]
    \node[label=below:$+0$] (p0) at (0,0) {};
    \node[label=below:$-6$] (m6) at (1,0) {};
    \node[label=below:$+6$] (p6) at (2,0) {};
    \node[label=below:$-4$] (m4) at (3,0) {};
    \node[label=below:$+4$] (p4) at (4,0) {};
    \node[label=below:$-2$] (m2) at (5,0) {};
    \node[label=below:$+2$] (p2) at (6,0) {};
    \node[label=below:$-8$] (m8) at (7,0) {};
\end{scope}

\begin{scope}[>={Stealth[black]},
              every edge/.style={draw=black}, every node/.style={inner sep=0pt, minimum size = 0pt}]
\node[label=\phantom{}] (bi1) at (0.5, -0.8) {};
\node[label=\phantom{}] (bi2) at (2.5, -0.8) {};
\path [{Bar}-{Bar}] (bi1) edge node [black, pos=0.5, sloped, below] {$(6~\deletionElement)$} (bi2);

\node[label=\phantom{}] (bi3) at (4.5, -0.8) {};
\node[label=\phantom{}] (bi4) at (6.5, -0.8) {};
\path [{Bar}-{Bar}] (bi3) edge node [black, pos=0.5, sloped, below] {$(\deletionElement~2)$} (bi4);
\end{scope}

\begin{scope}[>={Stealth[black]},
              dashed]
    \path [-] (p0) edge node [black, pos=0.5, sloped, above, yshift=-0.05cm] {$\deletionElement$} (m6);
    \path [-] (p6) edge node [black, pos=0.5, sloped, above, yshift=-0.05cm] {$\deletionElement$} (m4);
    \path [-] (p4) edge node [black, pos=0.5, sloped, above, yshift=-0.05cm] {$\deletionElement$} (m2);
    \path [-] (p2) edge node [black, pos=0.5, sloped, above, yshift=-0.05cm] {$\deletionElement$} (m8);
\end{scope}

\begin{scope}[>={Stealth[black]},
              every edge/.style={draw=black}]
\end{scope}

\begin{scope}[>={Stealth[black]},
              dashed]
    \path [-] (p0) edge  [bend left=70] node [black, pos=0.5, sloped, above] {$1$} (m2);
    \path [-] (p2) edge  [bend right=70] node [black, pos=0.5, sloped, above] {$3$} (m4);
    \path [-] (p4) edge  [bend right=70] node [black, pos=0.5, sloped, above] {$5$} (m6);
    \path [-] (p6) edge  [bend left=70] node [black, pos=0.5, sloped, above] {$7$} (m8);
\end{scope}

\begin{scope}[every node/.style={inner sep=0.4mm, draw, circle, minimum size = 0pt}]
    \node[label=below:$+0$] (2p0) at (0,-2.5) {};
    \node[label=below:$-2$] (2m2) at (1,-2.5) {};
    \node[label=below:$+2$] (2p2) at (2,-2.5) {};
    \node[label=below:$-4$] (2m4) at (3,-2.5) {};
    \node[label=below:$+4$] (2p4) at (4,-2.5) {};
    \node[label=below:$-6$] (2m6) at (5,-2.5) {};
    \node[label=below:$+6$] (2p6) at (6,-2.5) {};
    \node[label=below:$-8$] (2m8) at (7,-2.5) {};
\end{scope}

\begin{scope}[>={Stealth[black]},
              dashed]
    \path [-] (2p0) edge node [black, pos=0.5, sloped, above, yshift=-0.05cm] {$\deletionElement$} (2m2);
    \path [-] (2p6) edge node [black, pos=0.5, sloped, above, yshift=-0.05cm] {$\deletionElement$} (2m8);
\end{scope}

\begin{scope}[>={Stealth[black]},
              every edge/.style={draw=black}]
    \path [-] (2p2) edge node [black, pos=0.5, sloped, above] {} (2m4);
    \path [-] (2p4) edge node [black, pos=0.5, sloped, above] {} (2m6);
\end{scope}

\begin{scope}[>={Stealth[black]},
              every edge/.style={draw=black}]
\end{scope}

\begin{scope}[>={Stealth[black]},
              dashed]
    \path [-] (2p0) edge  [bend left=80] node [black, pos=0.5, sloped, above] {$1$} (2m2);
    \path [-] (2p2) edge  [bend left=80] node [black, pos=0.5, sloped, above] {$3$} (2m4);
    \path [-] (2p4) edge  [bend left=80] node [black, pos=0.5, sloped, above] {$5$} (2m6);
    \path [-] (2p6) edge  [bend left=80] node [black, pos=0.5, sloped, above] {$7$} (2m8);
\end{scope}

\end{tikzpicture}

\caption{
Exemplo de uma operação de {\it block interchange} que age em dois ciclos e cria quatro novos ciclos. Neste exemplo, temos $A = (0~\deletionElement~6~\deletionElement~4~\deletionElement~2~\deletionElement~8)$ e $\iota^n$ com $n = 7$. O potencial de {\it indel} do grafo original é igual a 
$\lceil{(4 + 1)/ 2}\rceil + \lceil{(4 + 1)/ 2}\rceil = 6$ e o potencial de indel no novo grafo é igual a $\lceil{(2 + 1)/ 2}\rceil + \lceil{(1 + 1)/ 2}\rceil + \lceil{(1 + 1)/ 2}\rceil + \lceil{(2 + 1)/ 2}\rceil = 6$.
\label{cap4:fig:bound_bi_case3}
}
\end{figure}
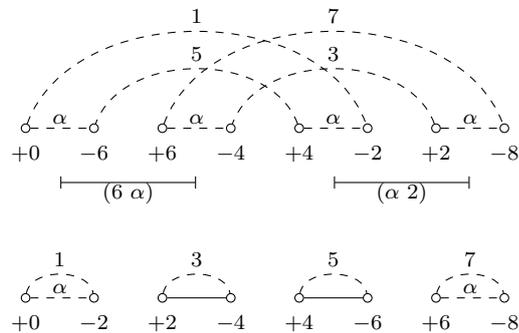

\begin{lemma}\label{cap4:lemma:change_score_transp}
	Para qualquer transposição $\tau$ e instância $\I = (A, \iota^n)$, temos que $\Delta c(\I, \tau) + \Delta \lambda(\I, \tau) \leq 2$.
\end{lemma}

\begin{proof}
	Segue diretamente do Lema~\ref{cap4:lemma:lb_graph_bi} e do fato de que uma transposição é um caso particular de {\it block interchange} que troca a posição relativa de dois segmentos adjacentes.
\end{proof}

\begin{lemma}\label{cap4:lemma:lb_graph_reversal}
	Para qualquer reversão $\rho$ e instância $\I = (A, \iota^n)$, temos que $\Delta c(\I, \rho) + \Delta \lambda(\I, \rho) \leq 1$.
\end{lemma}

\begin{proof}
	Bafna e Pevzner~\cite{1996-bafna-pevzner} mostraram que uma reversão $\rho$ modifica um grafo de ciclos das seguintes formas:
	
	\begin{itemize}
		\item Se a reversão $\rho$ afeta duas arestas de origem de dois ciclos $C_1$ e $C_2$, então $\rho$ une esses dois ciclos em um ciclo $C'$ e, portanto, $\Delta c(\I, \rho) = -1$.
		\item Se a reversão $\rho$ afeta duas arestas de um mesmo ciclo $C$, então ou (i) $\rho$ cria um novo ciclo $C'$ com os mesmos vértices que $C$ ou (ii) $\rho$ transforma $C$ em dois ciclos $C'_1$ e $C'_2$. No subcaso (i) temos $\Delta c(\I, \rho) = 0$ e no subcaso (ii) temos $\Delta c(\I, \rho) = 1$.
	\end{itemize}

	Suponha que $\rho$ afeta arestas de origem de dois ciclos $C_1$ e $C_2$ unindo esses dois ciclos em um ciclo $C'$. No melhor cenário, dois {\it runs} de deleção e dois {\it runs} de inserção de $C_1$ e $C_2$ são unidos em $C'$. Assim, $\Lambda(C') = \Lambda(C_1) + \Lambda(C_2) - 2$. Como nesse cenário existem pelo menos um {\it run} de inserção e um {\it run} de deleção em cada um dos ciclos $C_1, C_2,$ e $C'$, temos que:
	\begin{align*}
		\lambda(C') &= \Bigl\lceil \frac{\Lambda(C_1) + \Lambda(C_2) - 2 + 1}{2} \Bigr\rceil \\
								&= \Bigl\lceil \frac{\Lambda(C_1) + \Lambda(C_2) - 1}{2} \Bigr\rceil \\
								&= \frac{\Lambda(C_1) + \Lambda(C_2)}{2} \\
								&= \Bigl\lceil \frac{\Lambda(C_1) + 1}{2} \Bigr\rceil - 1 + \Bigl\lceil \frac{\Lambda(C_2) + 1}{2} \Bigr\rceil - 1 \\
								&= \lambda(C_1) + \lambda(C_2) - 2.
	\end{align*}

	Lembre-se que $\Lambda(C)$ é par sempre que $\Lambda(C) \geq 2$ (Lema~\ref{cap4:lemma:even_runs}). Sendo assim, temos $\Delta \lambda(\I, \rho) = 2$ e $\Delta c(\I, \rho) + \Delta \lambda(\I, \rho) = 1$.

	Suponha que $\rho$ afeta duas arestas de um mesmo ciclo $C$ e transforma esse ciclo em um ciclo $C'$. No melhor cenário, dois {\it runs} de deleção e dois {\it runs} de inserção são unidos e, portanto, $\Lambda(C') = \Lambda(C) - 2$ e $\lambda(C') = \lambda(C) - 1$. Sendo assim, $\Delta c(\I, \rho) + \Delta \lambda(\I, \rho) = 1$.

	Por último, suponha que $\rho$ afeta duas arestas de um mesmo ciclo $C$ e transforma esse ciclo em dois ciclos $C'_1$ e $C'_2$. No melhor cenário, dois {\it runs} de deleção são unidos, mas note que cada ciclo $C'_1$ e $C'_2$ possui pelo menos um {\it run} de inserção. Suponha, sem perda de generalidade, que os {\it runs} de deleção unidos estão em $C'_1$. Sendo assim, $\Lambda(C'_1) = X$, tal que $2 \leq X < \Lambda(C)$, e $\Lambda(C'_2) = \min(\Lambda(C) - X - 2, 1)$. 

	Note que $X$ é par. Suponha que $\min(\Lambda(C) - X - 2, 1) = \Lambda(C) - X - 2$. Como nesse cenário $C$ possui pelo menos dois {\it runs} de inserção e dois {\it runs} de deleção, temos que:
	\begin{align*}
		\lambda(C'_1) + \lambda(C'_2) &= \Bigl\lceil \frac{\Lambda(C'_1) + 1}{2} \Bigr\rceil + \Bigl\lceil \frac{\Lambda(C'_2) + 1}{2} \Bigr\rceil \\
		&= \Bigl\lceil \frac{X + 1}{2} \Bigr\rceil + \Bigl\lceil \frac{\Lambda(C) - X - 2 + 1}{2} \Bigr\rceil \\
		&=  \frac{X + 2}{2} + \Bigl\lceil \frac{\Lambda(C) - X - 1}{2} \Bigr\rceil \\
		&=  \Bigl\lceil \frac{X + 2 + \Lambda(C) - X - 1}{2} \Bigr\rceil \\
		&=  \Bigl\lceil \frac{X + 2 + \Lambda(C) - X - 1}{2} \Bigr\rceil \\
		&=  \Bigl\lceil \frac{\Lambda(C) + 1}{2} \Bigr\rceil\\
		&= \lambda(C).
	\end{align*}

	Se $\min(\Lambda(C) - X - 2, 1) = 1$, então $X = \Lambda(C) - 2$ e temos que: 
	\begin{align*}
		\lambda(C'_1) + \lambda(C'_2) &= \Bigl\lceil \frac{\Lambda(C'_1) + 1}{2} \Bigr\rceil + \Bigl\lceil \frac{\Lambda(C'_2) + 1}{2} \Bigr\rceil \\
		&= \frac{X + 2}{2} + \Bigl\lceil \frac{1 + 1}{2} \Bigr\rceil \\
		&= \frac{(\Lambda(C) - 2) + 2}{2} + 1 \\
		&= \frac{\Lambda(C)}{2} + 1 \\
		&= \lambda(C).
	\end{align*}

	Portanto, temos que $\Delta c(\I, \rho) + \Delta \lambda(\I, \rho) = 1$.
\end{proof}

Com esses resultados, podemos definir limitantes inferiores para o valor de $d_{\M}(\I)$, considerando $\M = \{\Mindel_{\tau}, \Mindel_{\BI}, \Mindel_{\rho,\tau}, \Mindel_{\rho,\BI}\}$. 

\begin{lemma}\label{cap4:lemma:lb_distance_graph}
	Para qualquer instância $\I = (A, \iota^n)$, temos que
	\begin{align*}
	d_{\Mindel_{\BI}}(\I) \geq \left\lceil \frac{|\pi^A| + 1 - c(\I) + \lambda(\I)}{2} \right\rceil, \\
	d_{\Mindel_{\tau}}(\I) \geq \left\lceil \frac{|\pi^A| + 1 - c(\I) + \lambda(\I)}{2} \right\rceil, \\
	d_{\Mindel_{\rho,\BI}}(\I) \geq \left\lceil \frac{|\pi^A| + 1 - c(\I) + \lambda(\I)}{2} \right\rceil, \\
	d_{\Mindel_{\rho,\tau}}(\I) \geq \left\lceil \frac{|\pi^A| + 1 - c(\I) + \lambda(\I)}{2} \right\rceil.
	\end{align*}
\end{lemma}

\begin{proof}
Considere o modelo $\Mindel_{\BI}$. Como $|\pi^{A'}| + 1 - c(A', \iota^n) + \lambda(A', \iota^n) = 0$ somente se $A' = \iota^n$, qualquer sequência de rearranjos $S$ que transforma $A$ em $\iota^n$ deve tornar o valor $|\pi^A| + 1 - c(A, \iota^n) + \lambda(A, \iota^n)$ igual a zero. Pelos lemas~\ref{cap4:lemma:lb_graph_deletion}, \ref{cap4:lemma:lb_graph_insertion} e \ref{cap4:lemma:lb_graph_bi}, qualquer rearranjo $\beta$ em $\Mindel_{\BI}$ satisfaz $\Delta c(\I, \beta) + \Delta \lambda(\I, \beta) \leq 2$ e, portanto, $|S| \geq \left\lceil \frac{|\pi^A| + 1 - c(\I) + \lambda(\I)}{2} \right\rceil$. 

A demonstração é similar para os outros modelos e também considera os resultados dos lemas~\ref{cap4:lemma:change_score_transp} e \ref{cap4:lemma:lb_graph_reversal}.
\end{proof}

\begin{lemma}\label{cap4:lemma:lb_graph_transp}
	Para qualquer instância $\I = (A, \iota^n)$, temos que
	\begin{align*}
		d_{\Mindel_{\tau}}(\I) \geq \left \lceil \frac{|\pi^A| + 1 - \cclean(\I) + \lambda_{\insertion}(\I)}{2} \right \rceil, \\
		d_{\Mindel_{\rho,\tau}}(\I) \geq \left \lceil \frac{|\pi^A| + 1 - \cclean(\I) + \lambda_{\insertion}(\I)}{2} \right \rceil.
	\end{align*}
\end{lemma}

\begin{proof}
	Segue diretamente do Lema~\ref{cap4:lemma:lb_distance_graph} e das observações \ref{cap4:remark:equivalencia_lambda_ciclo} e \ref{cap4:remark:equivalencia_lambda_soma}.
\end{proof}

\subsection{Algoritmos de $2$-Aproximação para Modelos com Block Interchanges}

Nesta seção, apresentamos algoritmos de $2$-aproximação para o problema de Distância de Rearranjos considerando os modelos $\Mindel_{\BI}$ e $\Mindel_{\rho,\BI}$. Os próximos lemas apresentam operações que diminuem o valor de $\lambda(\I)$ ou que aumentam o número de ciclos, dependendo de características do grafo de ciclos rotulado $G(\I)$. 

\begin{lemma}\label{cap4:lemma:insertion_remove_run}
	Para qualquer instância $\I = (A, \iota^n)$, se $G(\I)$ possui pelo menos um {\it run} de inserção, então existe uma inserção $\insertion$ com  $\Delta c(\I, \insertion) + \Delta \lambda(\I, \insertion) = 1$.
\end{lemma}

\begin{proof}
	Considere o {\it run} de inserção $(v_1, v_2, \ldots, v_{j})$ de um ciclo $C$, tal que $v_1$ possui o mesmo sinal que o seu elemento correspondente em $A$, ou seja, existe um elemento $A_i$ em $A$ tal que $v_1 = A_i$. Note que $(v_1, v_2)$ é uma aresta de destino rotulada pela definição de {\it runs}. Sejam $o_1, o_2, \ldots, o_k$ valores inteiros tal que $(v_{o_i}, v_{o_i+1})$ é a $i$-ésima aresta de destino rotulada desse {\it run}.

	Construímos $\sigma = (x_1, x_2, \ldots, x_k)$ da seguinte forma: para cada $1 \leq i \leq k$, se $v_{o_i+1}$ possui sinal ``$-$'', então $x_i = \ell((v_{o_i}, v_{o_i+1}))$; caso contrário, $x_i = -\ell((v_{o_i}, v_{o_i+1}))$. A inserção de $\sigma$ após o elemento de $A$ correspondente ao vértice $v_1$ remove o {\it run} e adiciona $k$ ciclos no grafo. 

	Um ciclo unitário é criado com os vértices $(v_1, -x_1)$. Para cada elemento $x_i$, com $1 \leq i < k$, existe um ciclo $(+x_i, v_{o_i+1}, v_{o_i+2}, \ldots, v_{o_{i+1}}, -x_{i+1}, +x_{i})$. O último vértice $+x_k$ pertence ao que sobrou do ciclo $C$ ou $+x_k$ pertence a um ciclo unitário, no caso em que todas as arestas de destino de $C$ pertencem ao {\it run} que foi removido pela inserção. Um exemplo dessa operação é apresentado na Figura~\ref{cap4:fig:insertion_remove_run}. 

	Se $\Lambda(C) \leq 2$, então a remoção de um {\it run} de $C$ reduz tanto o número de {\it runs} quanto o potencial de {\it indel} em um. Caso contrário, ao remover o {\it run} de inserção, dois {\it runs} de deleção são unidos. Nesse caso, o número de {\it runs} de $C$ diminui em dois e o potencial de {\it indel} diminui em um. Como a inserção adiciona $k$ elementos em $A$ e $k$ ciclos limpos no grafo, temos que $\Delta c(\I, \insertion) = 0$. Portanto, $\Delta c(\I, \insertion) + \Delta \lambda(\I, \insertion) = 1$.
\end{proof}

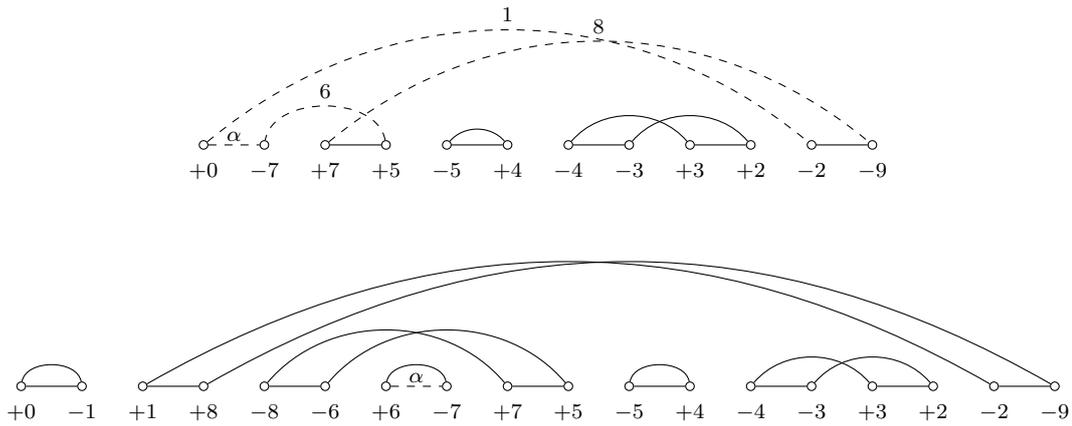
\begin{figure}[t]
\centering
\begin{tikzpicture}[scale=0.8]
\scriptsize
\begin{scope}[every node/.style={inner sep=0.4mm, draw, circle, minimum size = 0pt}]
    \node[label=below:$+0$] (p0) at (3,0) {};
    \node[label=below:$-7$] (m7) at (4,0) {};
    \node[label=below:$+7$] (p7) at (5,0) {};
    \node[label=below:$+5$] (p5) at (6,0) {};
    \node[label=below:$-5$] (m5) at (7,0) {};
    \node[label=below:$+4$] (p4) at (8,0) {};
    \node[label=below:$-4$] (m4) at (9,0) {};
    \node[label=below:$-3$] (m3) at (10,0) {};
    \node[label=below:$+3$] (p3) at (11,0) {};
    \node[label=below:$+2$] (p2) at (12,0) {};
    \node[label=below:$-2$] (m2) at (13,0) {};
    \node[label=below:$-9$] (m9) at (14,0) {};
\end{scope}

\begin{scope}[>={Stealth[black]},
              every edge/.style={draw=black}]
    \path [-] (p7) edge node [black, pos=0.5, sloped, above] {} (p5);
    \path [-] (m5) edge node [black, pos=0.5, sloped, above] {} (p4);
    \path [-] (m4) edge node [black, pos=0.5, sloped, above] {} (m3);
    \path [-] (p3) edge node [black, pos=0.5, sloped, above] {} (p2);
    \path [-] (m2) edge node [black, pos=0.5, sloped, above] {} (m9);
\end{scope}

\begin{scope}[>={Stealth[black]},
              dashed]
    \path [-] (p0) edge node [black, pos=0.5, sloped, above, yshift=-0.05cm] {$\deletionElement$} (m7);

\end{scope}

\begin{scope}[>={Stealth[black]},
              every edge/.style={draw=black}]
    \path [-] (p2) edge  [bend right=50] (m3);
    \path [-] (p3) edge  [bend right=50] (m4);
    \path [-] (p4) edge  [bend right=50] (m5);
\end{scope}

\begin{scope}[>={Stealth[black]},
              dashed]
    \path [-] (p0) edge  [bend left=40] node [black, pos=0.5, sloped, above] {$1$} (m2);
    \path [-] (p5) edge  [bend right=80] node [black, pos=0.5, sloped, above] {$6$} (m7);
    \path [-] (p7) edge  [bend left=40] node [black, pos=0.5, sloped, above] {$8$} (m9);
\end{scope}

\begin{scope}[every node/.style={inner sep=0.4mm, draw, circle, minimum size = 0pt}]
    \node[label=below:$+0$] (p0) at (0,-4) {};
    \node[label=below:$-1$] (m1) at (1,-4) {};
    \node[label=below:$+1$] (p1) at (2,-4) {};
    \node[label=below:$+8$] (p8) at (3,-4) {};
    \node[label=below:$-8$] (m8) at (4,-4) {};
    \node[label=below:$-6$] (m6) at (5,-4) {};
    \node[label=below:$+6$] (p6) at (6,-4) {};
    \node[label=below:$-7$] (m7) at (7,-4) {};
    \node[label=below:$+7$] (p7) at (8,-4) {};
    \node[label=below:$+5$] (p5) at (9,-4) {};
    \node[label=below:$-5$] (m5) at (10,-4) {};
    \node[label=below:$+4$] (p4) at (11,-4) {};
    \node[label=below:$-4$] (m4) at (12,-4) {};
    \node[label=below:$-3$] (m3) at (13,-4) {};
    \node[label=below:$+3$] (p3) at (14,-4) {};
    \node[label=below:$+2$] (p2) at (15,-4) {};
    \node[label=below:$-2$] (m2) at (16,-4) {};
    \node[label=below:$-9$] (m9) at (17,-4) {};
\end{scope}

\begin{scope}[>={Stealth[black]},
              every edge/.style={draw=black}]
    \path [-] (p0) edge node [black, pos=0.5, sloped, above] {} (m1);
    \path [-] (p1) edge node [black, pos=0.5, sloped, above] {} (p8);
    \path [-] (m8) edge node [black, pos=0.5, sloped, above] {} (m6);
    \path [-] (p7) edge node [black, pos=0.5, sloped, above] {} (p5);
    \path [-] (m5) edge node [black, pos=0.5, sloped, above] {} (p4);
    \path [-] (m4) edge node [black, pos=0.5, sloped, above] {} (m3);
    \path [-] (p3) edge node [black, pos=0.5, sloped, above] {} (p2);
    \path [-] (m2) edge node [black, pos=0.5, sloped, above] {} (m9);
\end{scope}

\begin{scope}[>={Stealth[black]},
              dashed]
    \path [-] (p6) edge node [black, pos=0.5, sloped, above, yshift=-0.05cm] {$\deletionElement$} (m7);
\end{scope}

\begin{scope}[>={Stealth[black]},
              every edge/.style={draw=black}]
    \path [-] (p2) edge  [bend right=50] (m3);
    \path [-] (p3) edge  [bend right=50] (m4);
    \path [-] (p4) edge  [bend right=80] (m5);
\end{scope}

\begin{scope}[>={Stealth[black]},
              every edge/.style={draw=black}]
    \path [-] (p0) edge  [bend left=80] node [black, pos=0.5, sloped, above] {} (m1);
    \path [-] (p1) edge  [bend left=30] node [black, pos=0.5, sloped, above] {} (m2);
    \path [-] (p5) edge  [bend right=50] node [black, pos=0.5, sloped, above] {} (m6);
    \path [-] (p7) edge  [bend right=50] node [black, pos=0.5, sloped, above] {} (m8);
    \path [-] (p8) edge  [bend left=30] node [black, pos=0.5, sloped, above] {} (m9);
    \path [-] (p6) edge  [bend left=80] node [black, pos=0.5, sloped, above] {} (m7);
\end{scope}
\end{tikzpicture}

\caption{Exemplo de uma inserção que remove um {\it run} de um ciclo. Neste exemplo, temos o {\it run} de inserção $(+0, -2, -9, +7, +5, -7)$. A inserção de $\sigma = ({+1}~{-8}~{+6})$ no início da string do genoma origem remove esse {\it run} e cria três novos ciclos.
\label{cap4:fig:insertion_remove_run}}
\end{figure}

\begin{lemma}\label{cap4:lemma:block_interchange_add_cycles}
	Para qualquer instância $\I = (A, \iota^n)$, se $|\pi^A| + 1 - c(\I) > 0$, $G(\I)$ só possui arestas de destino limpas e $G(\I)$ não possui ciclos divergentes, então existe {\it block interchange} $\BI$ tal que $\Delta c(\I, \BI) + \Delta \lambda(\I, \BI) \geq 1$.
\end{lemma}

\begin{proof}

Considere que $G(\I)$ possui um ciclo orientado $C = (o_1, \ldots, o_\ell)$. Nesse caso, sempre existe tripla orientada $o_i$, $o_j$ e $o_k$, com $i < j < k$, tal que $o_i > o_k > o_j$ e $k = j+1$~\cite{1998-bafna-pevzner}. 
Uma operação de {\it block interchange} aplicada nessas três arestas de origem cria três ciclos $C'$, $C''$ e $C'''$~\cite{1998-bafna-pevzner}. Sejam $\sigma_1$ e $\sigma_2$ os dois segmentos afetados por esse {\it block interchange}. Se todas as arestas de origem afetadas são rotuladas, podemos mover os elementos $\deletionElement$ de forma que esses rótulos fiquem todos no mesmo ciclo. Para fazer isso, incluímos em $\sigma_1$ qualquer elemento $\deletionElement$ correspondente ao rótulo da aresta de origem mais à esquerda (aresta de origem com índice $o_j$), e incluímos em $\sigma_2$ qualquer elemento $\deletionElement$ correspondente ao rótulo da aresta de origem mais à direita (aresta de origem com índice $o_i$). Dessa forma, os elementos $\deletionElement$ das três arestas de origem afetadas são acumulados em uma só aresta de origem do novo grafo, como mostrado na Figura~\ref{cap4:fig:bi_oriented_cycles}. Uma operação análoga é usada se apenas duas das três arestas de origem afetadas são rotuladas. 

Como $k = j+1$, podemos garantir que o ciclo mais à esquerda é um ciclo unitário limpo. Esse ciclo possui a aresta de destino que é adjacente às arestas de origem $e_{o_j}$ e $e_{o_k}$. No entanto, não podemos garantir que o ciclo mais à direita é sempre unitário. Note que não existe aresta de destino rotulada e, portanto, $\lambda(C') + \lambda(C'') + \lambda(C''') \leq \lambda(C) + 1$, pois um dos ciclos é unitário e limpo. Dessa forma, $\Delta \lambda(\I, \BI) \geq -1$, $\Delta c(\I, \BI) = 2$ e $\Delta c(\I, \BI) + \Delta \lambda(\I, \BI) \geq 1$.

\begin{figure}[t]
\centering
\begin{tikzpicture}[scale=0.7]
\scriptsize
\begin{scope}[every node/.style={inner sep=0.4mm, draw, circle, minimum size = 0pt}]
    \node[label=below:$+0$] (p0) at (0,0) {};
    \node[label=below:$-2$] (m2) at (1.5,0) {};
    \node[label=below:$+2$] (p2) at (3,0) {};
    \node[label=below:$-1$] (m1) at (4.5,0) {};
    \node[label=below:$+1$] (p1) at (6,0) {};
    \node[label=below:$-3$] (m3) at (7.5,0) {};
\end{scope}

\begin{scope}[>={Stealth[black]},
              every edge/.style={draw=black}, every node/.style={inner sep=0.1pt, minimum size = 0pt}]
\node[label=\phantom{}] (bi1) at (.7, -0.8) {};
\node[label=\phantom{}] (bi2) at (3.7, -0.8) {};
\path [{Bar}-{Bar}] (bi1) edge node [black, pos=0.5, sloped, below] {$(\deletionElement~2)$} (bi2);

\node[label=\phantom{}] (bi3) at (3.8, -0.8) {};
\node[label=\phantom{}] (bi4) at (6.7, -0.8) {};
\path [{Bar}-{Bar}] (bi3) edge node [black, pos=0.5, sloped, below] {$(1~\deletionElement)$} (bi4);
\end{scope}

\begin{scope}[>={Stealth[black]},
              dashed]
    \path [-] (p0) edge node [black, pos=0.5, sloped, above, yshift=-0.05cm] {$\deletionElement$} (m2);
    \path [-] (p2) edge node [black, pos=0.5, sloped, above, yshift=-0.05cm] {$\deletionElement$} (m1);
    \path [-] (p1) edge node [black, pos=0.5, sloped, above, yshift=-0.05cm] {$\deletionElement$} (m3);
\end{scope}

\begin{scope}[>={Stealth[black]},
              every edge/.style={draw=black}]
    \path [-] (p0) edge  [bend left=50] (m1);
    \path [-] (p1) edge  [bend right=50] (m2);
    \path [-] (p2) edge  [bend left=50] (m3);
\end{scope}

\begin{scope}[every node/.style={inner sep=0.4mm, draw, circle, minimum size = 0pt}]
    \node[label=below:$+0$] (2p0) at (0,-2) {};
    \node[label=below:$-1$] (2m1) at (1.5,-2) {};
    \node[label=below:$+1$] (2p1) at (3,-2) {};
    \node[label=below:$-2$] (2m2) at (4.5,-2) {};
    \node[label=below:$+2$] (2p2) at (6,-2) {};
    \node[label=below:$-3$] (2m3) at (7.5,-2) {};
\end{scope}

\begin{scope}[>={Stealth[black]},
              dashed]
    \path [-] (2p1) edge node [black, pos=0.5, sloped, above, yshift=-0.05cm] {$\deletionElement$} (2m2);
\end{scope}

\begin{scope}[>={Stealth[black]},
              every edge/.style={draw=black}]
    \path [-] (2p0) edge node [black, pos=0.5, sloped, above] {} (2m1);
    \path [-] (2p2) edge node [black, pos=0.5, sloped, above] {} (2m3);
\end{scope}

\begin{scope}[>={Stealth[black]},
              every edge/.style={draw=black}]
    \path [-] (2p0) edge  [bend left=80] (2m1);
    \path [-] (2p1) edge  [bend left=80] (2m2);
    \path [-] (2p2) edge  [bend left=80] (2m3);
\end{scope}

\end{tikzpicture}

\caption{Exemplo de uma operação de {\it block interchange} que afeta uma tripla orientada de um ciclo orientado e cria três novos ciclos. Os elementos $\deletionElement$ são movidos de forma que apenas uma das arestas de origem afetadas permanece rotulada. Neste exemplo, temos $A = (0~\deletionElement~2~\deletionElement~1~\deletionElement~3)$ e $\iota^n$ com $n = 2$. 
\label{cap4:fig:bi_oriented_cycles}}
\end{figure}
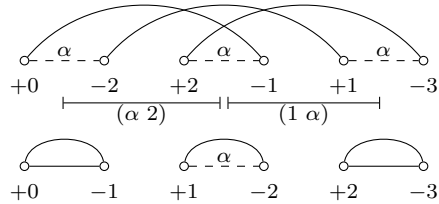

Agora, considere que $G(\I)$ possui apenas ciclos não orientados. Seja $C = (o_1, \ldots, o_\ell)$ um ciclo de $G(\I)$. Bafna e Pevzner~\cite{1998-bafna-pevzner} mostraram que para todo par de arestas de origem $(e_{o_i}, e_{o_j}$) de $C$, com $o_i > o_j$, existe um ciclo $D = (o'_1, \ldots, o'_{\ell'})$ com arestas de origem $e_{o'_x}$ e $e_{o'_y}$ tal que ou $o_i > o'_x > o_j > o'_y$ ou $o'_x > o_i > o'_y > o_j$. Assuma, sem perda de generalidade, que $o_i > o'_x > o_j > o'_y$, $o_i = o_1$ e $o_j = o_\ell$. Uma operação de {\it block interchange} que age nessas quatro arestas de origem cria quatro novos ciclos $C'$, $C''$, $D'$ e $D''$: $C'$ é formado pelo caminho que vai de $e_{o_i}$ até $e_{o_j}$ com uma aresta de origem incidente ao primeiro vértice e ao último vértice desse caminho; $C''$ é formado pelo caminho que vai de $e_{o_j}$ até $e_{o_i}$ com uma aresta de origem incidente ao primeiro vértice e ao último vértice desse caminho; $D'$ e $D''$ seguem um padrão similar aos dos ciclos $C'$ e $C''$. Como $o_i = o_1$ e $o_j = o_\ell$, existe apenas uma aresta de destino no caminho de $e_{o_j}$ até $e_{o_i}$ e portanto, $C''$ é um ciclo unitário.

O primeiro segmento afetado pela operação começa na aresta de origem $o'_y$, incluindo qualquer elemento $\deletionElement$ correspondente ao rótulo da aresta $e_{o'_y}$, e termina na aresta de origem $o_j$ sem incluir qualquer elemento $\deletionElement$. O segundo segmento afetado pela operação começa com a aresta de origem $o'_x$, sem incluir qualquer elemento $\deletionElement$, e termina com a aresta de origem $o_i$, incluindo qualquer elemento $\deletionElement$ correspondente ao rótulo da aresta $e_{o_i}$. Um exemplo é mostrado na Figura~\ref{cap4:fig:bi_non_oriented_cycles}. Dessa forma, garantimos que $D''$ é um ciclo unitário limpo. Note que não existe aresta de destino rotulada e, portanto, $\lambda(C') + \lambda(C'') + \lambda(D') + \lambda(D'') \leq \lambda(C) + \lambda(D) + 1$. Sendo assim, temos $\Delta \lambda(\I, \BI) \geq -1$, $\Delta c(\I, \BI) = 2$ e $\Delta c(\I, \BI) + \Delta \lambda(\I, \BI) \geq 1$.
\end{proof}

\begin{figure}[t]
\centering
\begin{tikzpicture}[scale=0.7]
\scriptsize
\begin{scope}[every node/.style={inner sep=0.4mm, draw, circle, minimum size = 0pt}]
    \node[label=below:$+0$] (p0) at (0,0) {};
    \node[label=below:$-3$] (m3) at (1.5,0) {};
    \node[label=below:$+3$] (p3) at (3,0) {};
    \node[label=below:$-2$] (m2) at (4.5,0) {};
    \node[label=below:$+2$] (p2) at (6,0) {};
    \node[label=below:$-1$] (m1) at (7.5,0) {};
    \node[label=below:$+1$] (p1) at (9,0) {};
    \node[label=below:$-4)$] (m4) at (10.5,0) {};
\end{scope}

\begin{scope}[>={Stealth[black]},
              every edge/.style={draw=black}, every node/.style={inner sep=0.1pt, minimum size = 0pt}]
\node[label=\phantom{}] (bi1) at (.7, -0.8) {};
\node[label=\phantom{}] (bi2) at (3.7, -0.8) {};
\path [{Bar}-{Bar}] (bi1) edge node [black, pos=0.5, sloped, below] {$(\deletionElement~3)$} (bi2);

\node[label=\phantom{}] (bi3) at (6.7, -0.8) {};
\node[label=\phantom{}] (bi4) at (9.7, -0.8) {};
\path [{Bar}-{Bar}] (bi3) edge node [black, pos=0.5, sloped, below] {$(1~\deletionElement)$} (bi4);
\end{scope}

\begin{scope}[>={Stealth[black]},
              dashed]
    \path [-] (p0) edge node [black, pos=0.5, sloped, above, yshift=-0.05cm] {$\deletionElement$} (m3);
    \path [-] (p3) edge node [black, pos=0.5, sloped, above, yshift=-0.05cm] {$\deletionElement$} (m2);
    \path [-] (p2) edge node [black, pos=0.5, sloped, above, yshift=-0.05cm] {$\deletionElement$} (m1);
    \path [-] (p1) edge node [black, pos=0.5, sloped, above, yshift=-0.05cm] {$\deletionElement$} (m4);
\end{scope}

\begin{scope}[>={Stealth[black]},
              every edge/.style={draw=black}]
    \path [-] (p0) edge  [bend left=40] (m1);
    \path [-] (p1) edge  [bend right=40] (m2);
    \path [-] (p2) edge  [bend right=40] (m3);
    \path [-] (p3) edge  [bend left=40] (m4);
\end{scope}

\begin{scope}[every node/.style={inner sep=0.4mm, draw, circle, minimum size = 0pt}]
    \node[label=below:$+0$] (2p0) at (0,-2) {};
    \node[label=below:$-1$] (2m1) at (1.5,-2) {};
    \node[label=below:$+1$] (2p1) at (3,-2) {};
    \node[label=below:$-2$] (2m2) at (4.5,-2) {};
    \node[label=below:$+2$] (2p2) at (6,-2) {};
    \node[label=below:$-3$] (2m3) at (7.5,-2) {};
    \node[label=below:$+3$] (2p3) at (9,-2) {};
    \node[label=below:$-4$] (2m4) at (10.5,-2) {};
\end{scope}

\begin{scope}[>={Stealth[black]},
              dashed]
    \path [-] (2p1) edge node [black, pos=0.5, sloped, above, yshift=-0.05cm] {$\deletionElement$} (2m2);
    \path [-] (2p2) edge node [black, pos=0.5, sloped, above, yshift=-0.05cm] {$\deletionElement$} (2m3);
\end{scope}

\begin{scope}[>={Stealth[black]},
              every edge/.style={draw=black}]
    \path [-] (2p0) edge node [black, pos=0.5, sloped, above] {} (2m1);
    \path [-] (2p3) edge node [black, pos=0.5, sloped, above] {} (2m4);
\end{scope}

\begin{scope}[>={Stealth[black]},
              every edge/.style={draw=black}]
    \path [-] (2p0) edge  [bend left=80] (2m1);
    \path [-] (2p1) edge  [bend left=80] (2m2);
    \path [-] (2p2) edge  [bend left=80] (2m3);
    \path [-] (2p3) edge  [bend left=80] (2m4);
\end{scope}

\end{tikzpicture}

\caption{Exemplo de uma operação de {\it block interchange} que age em dois ciclos não orientados e cria quatro novos ciclos. Neste exemplo, temos $A = (0~\deletionElement~3~\deletionElement~2~\deletionElement~1~\deletionElement~4)$ e $\iota^n$ com $n = 3$. Os elementos $\deletionElement$ são movidos de forma que apenas duas das quatro arestas de origem afetadas permanecem rotuladas.
\label{cap4:fig:bi_non_oriented_cycles}}
\end{figure}
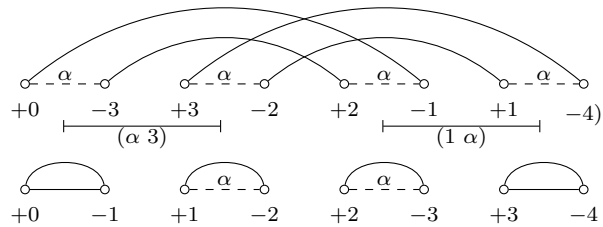

Lembramos que em uma instância de strings sem sinais, o grafo de ciclos rotulado $G(\I)$ possui apenas ciclos convergentes. Já em uma instância de strings com sinais, o grafo de ciclos rotulado $G(\I)$ contém ciclos divergentes se, e somente se, existe pelo menos um elemento com sinal ``$-$''. No próximo lema, apresentamos como usar reversões para lidar com ciclos divergentes.

\begin{lemma}\label{cap4:lemma:reversal_divergent}
	Para qualquer instância $\I = (A, \iota^n)$, se $|\pi^A| + 1 - c(\I) > 0$, $G(\I)$ só possui arestas de destino limpas e $G(\I)$ possui um ciclo divergente, então existe reversão $\rho$ com $\Delta c(\I, \rho) + \Delta \lambda(\I, \rho) = 1$. 
\end{lemma}

\begin{proof}
	Seja $C = (o_1, o_2, \ldots, o_k)$ um ciclo divergente em $G(\I)$ e seja $(e_{o_x}, e_{o_{x+1}})$ um par de arestas de origem divergentes com $x$ mínimo. Uma reversão aplicada nessas arestas de origem transforma o ciclo $C$ em um ciclo unitário $C'$ e um outro ciclo $C''$~\cite{1996-bafna-pevzner}. Essa reversão pode ser aplicada de uma forma que qualquer elemento $\deletionElement$ é movido para o ciclo $C''$, o que faz $\lambda(C') = 0$ e $\lambda(C'') = \lambda(C)$, já que $G(\I)$ só possui arestas de destino limpas. Dessa forma, temos $\Delta c(\I, \rho) + \Delta \lambda(\I, \rho) = 1$. Um exemplo dessa operação é apresentado na Figura~\ref{cap4:fig:reversal_divergent_cycle}.    
\end{proof}

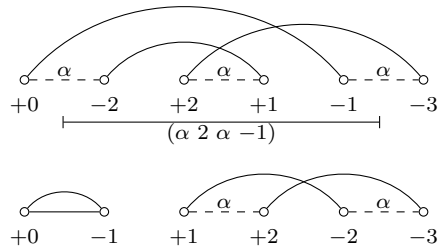
\begin{figure}[t]
\centering
\begin{tikzpicture}[scale=0.7]
\scriptsize
\begin{scope}[every node/.style={inner sep=0.4mm, draw, circle, minimum size = 0pt}]
    \node[label=below:$+0$] (p0) at (0,0) {};
    \node[label=below:$-2$] (m2) at (1.5,0) {};
    \node[label=below:$+2$] (p2) at (3,0) {};
    \node[label=below:$+1$] (p1) at (4.5,0) {};
    \node[label=below:$-1$] (m1) at (6,0) {};
    \node[label=below:$-3$] (m3) at (7.5,0) {};
\end{scope}

\begin{scope}[>={Stealth[black]},
              every edge/.style={draw=black}, every node/.style={inner sep=0.1pt, minimum size = 0pt}]
\node[label=\phantom{}] (bi1) at (.7, -0.8) {};
\node[label=\phantom{}] (bi2) at (6.7, -0.8) {};
\path [{Bar}-{Bar}] (bi1) edge node [black, pos=0.5, sloped, below] {$(\deletionElement~2~\deletionElement~{-1})$} (bi2);

\end{scope}

\begin{scope}[>={Stealth[black]},
              dashed]
    \path [-] (p0) edge node [black, pos=0.5, sloped, above, yshift=-0.05cm] {$\deletionElement$} (m2);
    \path [-] (p2) edge node [black, pos=0.5, sloped, above, yshift=-0.05cm] {$\deletionElement$} (p1);
    \path [-] (m1) edge node [black, pos=0.5, sloped, above, yshift=-0.05cm] {$\deletionElement$} (m3);
\end{scope}

\begin{scope}[>={Stealth[black]},
              every edge/.style={draw=black}]
    \path [-] (p0) edge  [bend left=50] (m1);
    \path [-] (p1) edge  [bend right=50] (m2);
    \path [-] (p2) edge  [bend left=50] (m3);
\end{scope}

\begin{scope}[every node/.style={inner sep=0.4mm, draw, circle, minimum size = 0pt}]
    \node[label=below:$+0$] (p0) at (0,-2.5) {};
    \node[label=below:$-1$] (m1) at (1.5,-2.5) {};
    \node[label=below:$+1$] (p1) at (3,-2.5) {};
    \node[label=below:$+2$] (p2) at (4.5,-2.5) {};
    \node[label=below:$-2$] (m2) at (6,-2.5) {};
    \node[label=below:$-3$] (m3) at (7.5,-2.5) {};
\end{scope}



\begin{scope}[>={Stealth[black]},
              dashed]
    \path [-] (p1) edge node [black, pos=0.5, sloped, above, yshift=-0.05cm] {$\deletionElement$} (p2);
    \path [-] (m2) edge node [black, pos=0.5, sloped, above, yshift=-0.05cm] {$\deletionElement$} (m3);
\end{scope}

\begin{scope}[>={Stealth[black]},
              every edge/.style={draw=black}]
    \path [-] (p0) edge node [black, pos=0.5, sloped, above, yshift=-0.05cm] {} (m1);
    \path [-] (p0) edge  [bend left=50] (m1);
    \path [-] (p1) edge  [bend left=50] (m2);
    \path [-] (p2) edge  [bend left=50] (m3);
\end{scope}

\end{tikzpicture}

\caption{\label{cap4:fig:reversal_divergent_cycle}
Exemplo de uma reversão que age em um ciclo divergente e cria dois novos ciclos. Neste exemplo, temos $A = (0~\deletionElement~{2}~\deletionElement~{-1}~\deletionElement~3)$ e $\iota^n$ com $n = 2$. Essa reversão move qualquer elemento $\deletionElement$ de forma que o ciclo unitário criado possui apenas arestas limpas.
}
\end{figure}

\begin{lemma}\label{cap4:lemma:deletion_remove_runs}
	Para qualquer instância $\I = (A, \iota^n)$, se $|\pi^A| + 1 - c(\I) = 0$ e $G(\I)$ possui pelo menos um {\it run} de deleção, então existe deleção $\deletion$ com $\Delta c(\I, \deletion) + \Delta \lambda(\I, \deletion) = 1$.
\end{lemma}

\begin{proof}
	Já que $|\pi^A| + 1 - c(\I) = 0$, todo ciclo do grafo $G(\I)$ é um ciclo unitário. Cada ciclo possui no máximo um {\it run} de inserção e um {\it run} de deleção. Seja $C$ um ciclo unitário de $G(\I)$ com uma aresta de origem rotulada. Note que esse ciclo possui um {\it run} de deleção que é formado apenas por essa aresta de origem. Seja $\deletion$ uma operação que remove o segmento correspondente ao rótulo dessa aresta de origem. Dessa forma, o {\it run} de deleção é removido e $\Delta c(\I, \deletion) + \Delta \lambda(\I, \deletion) = 1$.
\end{proof}

Com esses lemas, podemos apresentar os algoritmos~\ref{cap4:algorithm:bi_cycle_graph} e \ref{cap4:algorithm:bi_rev_cycle_graph}. Esses algoritmos possuem complexidade de tempo de $O(n^2)$, já que o grafo de ciclos rotulado $G(\I)$ pode ser criado em tempo linear, todos os laços de repetição executam $O(n)$ vezes, e qualquer operação dentro dos laços pode ser realizada em tempo linear.

No Teorema~\ref{cap4:theorem:2_approx_bi_cycle_graph} demonstramos que esses algoritmos possuem fator de aproximação igual a $2$ para os problemas de Distância de Rearranjos considerando os modelos $\Mindel_{\BI}$ e $\Mindel_{\rho,\BI}$.

\begin{theorem}\label{cap4:theorem:2_approx_bi_cycle_graph}
	Os algoritmos~\ref{cap4:algorithm:bi_cycle_graph} e \ref{cap4:algorithm:bi_rev_cycle_graph} são algoritmos de $2$-aproximação para os problemas da Distância de Block Interchanges e Indels e da Distância de Block Interchanges, Reversões e Indels, respectivamente.
\end{theorem}

\begin{proof}
	Considere o Algoritmo~\ref{cap4:algorithm:bi_cycle_graph}. Qualquer operação $\beta$ aplicada pelo algoritmo satisfaz $\Delta c(\I, \beta) + \Delta \lambda(\I, \beta) \geq 1$ e, portanto, a sequência retornada pelo algoritmo possui tamanho de no máximo ${|\pi^A| + 1 - c(\I) + \lambda(\I)}$. Pelo Lema~\ref{cap4:lemma:lb_distance_graph}, esse algoritmo é uma $2$-aproximação para a Distância de Block Interchanges e Indels.

	A prova é similar para o Algoritmo \ref{cap4:algorithm:bi_rev_cycle_graph}.
\end{proof}

\begin{algorithm}[h]
	\caption{Algoritmo de $2$-Aproximação para a Distância de Block Interchanges e Indels\label{cap4:algorithm:bi_cycle_graph}}
	\DontPrintSemicolon
	\Entrada{Uma instância $\I = (A, \iota^n)$}
	\Saida{Uma sequência de rearranjos que transforma $A$ em $\iota^n$}
	Seja $S \gets \emptyset$\;

	\Enqto{$G(\I)$ possui {\it runs} de inserção}{
		Seja $\phi$ uma inserção com $\Delta c(\I, \insertion) + \Delta \lambda(\I, \insertion) = 1$ (Lema~\ref{cap4:lemma:insertion_remove_run})\;
		$A \gets A \comp \phi$\;
		Adicione $\phi$ na sequência $S$\;
	}

	\Enqto{$|\pi^A| + 1 - c(\I) > 0$}{
		Seja $\BI$ uma operação com $\Delta c(\I, \BI) + \Delta \lambda(\I, \BI) \geq 1$ (Lema~\ref{cap4:lemma:block_interchange_add_cycles})\;
		$A \gets A \comp \BI$\;
		Adicione $\BI$ na sequência $S$\;
	}

	\Enqto{$G(\I)$ possui {\it runs} de deleção}{
		Seja $\psi$ uma deleção que remove um {\it run} de deleção (Lema~\ref{cap4:lemma:deletion_remove_runs})\;
		$A \gets A \comp \psi$\;
		Adicione $\psi$ na sequência $S$\;
	}

	{\bf retorne} a sequência $S$\;
\end{algorithm}

\begin{algorithm}[h]
	\caption{Algoritmo de $2$-Aproximação para a Distância de Block Interchanges, Reversões e Indels\label{cap4:algorithm:bi_rev_cycle_graph}}
	\DontPrintSemicolon
	\Entrada{Uma instância $\I = (A, \iota^n)$}
	\Saida{Uma sequência de rearranjos que transforma $A$ em $\iota^n$}
	Seja $S \gets \emptyset$\;

	\Enqto{$G(\I)$ possui {\it runs} de inserção}{
		Seja $\phi$ uma inserção com $\Delta c(\I, \insertion) + \Delta \lambda(\I, \insertion) = 1$ (Lema~\ref{cap4:lemma:insertion_remove_run})\;
		$A \gets A \comp \phi$\;
		Adicione $\phi$ na sequência $S$\;
	}

	\Enqto{$|\pi^A| + 1 - c(\I) > 0$}{
		\uSe{$G(\I)$ possui ciclos divergentes}{
			Seja $\rho$ uma operação com $\Delta c(\I, \rho) + \Delta \lambda(\I, \rho) = 1$ (Lema~\ref{cap4:lemma:reversal_divergent})\;
			$A \gets A \comp \rho$\;
			Adicione $\rho$ na sequência $S$\;
		}\Senao(){
			Seja $\BI$ uma operação com $\Delta c(\I, \BI) + \Delta \lambda(\I, \BI) \geq 1$ (Lema~\ref{cap4:lemma:block_interchange_add_cycles})\;
			$A \gets A \comp \BI$\;
			Adicione $\BI$ na sequência $S$\;
		}
	}

	\Enqto{$G(\I)$ possui {\it runs} de deleção}{
		Seja $\psi$ uma deleção que remove um {\it run} de deleção (Lema~\ref{cap4:lemma:deletion_remove_runs})\;
		$A \gets A \comp \psi$\;
		Adicione $\psi$ na sequência $S$\;
	}

	{\bf retorne} a sequência $S$\;
\end{algorithm}

\subsection{Algoritmos de $2$-Aproximação para Modelos com Transposições}

Nesta seção, apresentamos algoritmos de $2$-aproximação para o problema de Distância de Rearranjos considerando os modelos $\Mindel_{\tau}$ e $\Mindel_{\rho,\tau}$. 

Temos os seguintes corolários que seguem diretamente dos lemas~\ref{cap4:lemma:insertion_remove_run}, \ref{cap4:lemma:reversal_divergent} e \ref{cap4:lemma:deletion_remove_runs}.

\begin{corollary}\label{cap4:corollary:insertion_remove_run}
	Para qualquer instância $\I = (A, \iota^n)$, se $G(\I)$ possui pelo menos um {\it run} de inserção, então existe uma inserção $\insertion$ com  $\Delta \cclean(\I, \insertion) + \Delta \lambda_{\insertion}(\I, \insertion) = 1$.
\end{corollary}

\begin{corollary}\label{cap4:corollary:reversal_divergent}
	Para qualquer instância $\I = (A, \iota^n)$, se $|\pi^A| + 1 - c(\I) > 0$, $G(\I)$ só possui arestas de destino limpas e $G(\I)$ possui um ciclo divergente, então existe uma reversão $\rho$ com $\Delta \cclean(\I, \rho) + \Delta \lambda_{\insertion}(\I, \rho) = 1$. 
\end{corollary}

\begin{corollary}\label{cap4:corollary:deletion_remove_runs}
	Para qualquer instância $\I = (A, \iota^n)$, se $|\pi^A| + 1 - c(\I) = 0$ e $G(\I)$ possui pelo menos um {\it run} de deleção, então existe uma deleção $\deletion$ com $\Delta \cclean(\I, \deletion) + \Delta \lambda_{\insertion}(\I, \deletion) = 1$.
\end{corollary}

Se $G(\I)$ só possui arestas de destino limpas, então qualquer transposição $\tau$ possui $\lambda_{\insertion}(\I, \tau) = 0$, já que uma transposição não altera o conjunto $\Sigma_A \cap \Sigma_{\iota^n}$. Os próximos quatro lemas mostram como podemos usar transposições para criar novos ciclos limpos quando $G(\I)$ não possui arestas de destino rotuladas. O primeiro deles é usado quando existem ciclos orientados em $G(\I)$. Os outros lemas podem ser usados quando o grafo $G(\I)$ satisfaz uma dessas condições: (i) existem dois ciclos rotulados e um deles é não unitário; (ii) existe um único ciclo não unitário rotulado; (iii) todos os ciclos não unitários são limpos.

\begin{lemma}\label{cap4:lemma:transposition_oriented_cycles}
	Para qualquer instância $\I = (A, \iota^n)$, tal que $|\pi^A| + 1 - c(\I) > 0$, $G(\I)$ só possui arestas de destino limpas e $G(\I)$ não possui ciclos divergentes, se existe um ciclo orientado $C \in G(\I)$, então existe uma transposição $\tau$ com $\Delta \cclean(\I, \tau) + \Delta \lambda_{\insertion}(\I, \tau) \geq 1$.
\end{lemma}

\begin{proof}
Como existe um ciclo orientado $C = (o_1, \ldots, o_\ell)$ em $G(\I)$, sempre existe tripla orientada $o_i$, $o_j$ e $o_k$, com $i < j < k$, tal que $o_i > o_k > o_j$ e $k = j+1$~\cite{1998-bafna-pevzner}. Uma transposição aplicada nessas arestas de origem transforma $C$ em três ciclos $C'$, $C''$ e $C'''$ tal que pelo menos um deles é unitário. Dessa forma, temos $\Delta c(\I, \tau) = 2$. 

Se $C$ é um ciclo limpo, então os três novos ciclos também são limpos e, portanto, $\Delta \cclean(\I, \tau) = 2$. Caso contrário, $C$ é rotulado e, portanto, pelo menos um dos três ciclos também é rotulado. No entanto, podemos garantir que o ciclo unitário é sempre limpo ao mover qualquer elemento $\deletionElement$ para um dos outros dois ciclos, como mostrado no exemplo da Figura~\ref{cap4:fig:transposition_oriented_cycles}. Portanto, temos que $\Delta \cclean(\I, \tau) + \Delta \lambda_{\insertion}(\I, \tau) \geq 1$.
\end{proof}

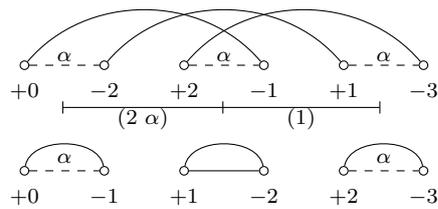
\begin{figure}[t]
\centering
\begin{tikzpicture}[scale=0.7]
\scriptsize
\begin{scope}[every node/.style={inner sep=0.4mm, draw, circle, minimum size = 0pt}]
    \node[label=below:$+0$] (p0) at (0,0) {};
    \node[label=below:$-2$] (m2) at (1.5,0) {};
    \node[label=below:$+2$] (p2) at (3,0) {};
    \node[label=below:$-1$] (m1) at (4.5,0) {};
    \node[label=below:$+1$] (p1) at (6,0) {};
    \node[label=below:$-3$] (m3) at (7.5,0) {};
\end{scope}

\begin{scope}[every edge/.style={draw=black}, every node/.style={inner sep=0pt, minimum size = 0pt}]
\node[label=\phantom{}] (bi1) at (.7, -0.8) {};
\node[label=\phantom{}] (bi2) at (3.75, -0.8) {};
\path [{Bar}-{Bar}] (bi1) edge node [black, pos=0.5, sloped, below] {$(2~\deletionElement)$} (bi2);

\node[label=\phantom{}] (bi3) at (3.7, -0.8) {};
\node[label=\phantom{}] (bi4) at (6.7, -0.8) {};
\path [-{Bar}] (bi3) edge node [black, pos=0.5, sloped, below] {$(1)$} (bi4);
\end{scope}

\begin{scope}[dashed]
    \path [-] (p0) edge node [black, pos=0.5, sloped, above, yshift=-0.05cm] {$\deletionElement$} (m2);
    \path [-] (p2) edge node [black, pos=0.5, sloped, above, yshift=-0.05cm] {$\deletionElement$} (m1);
    \path [-] (p1) edge node [black, pos=0.5, sloped, above, yshift=-0.05cm] {$\deletionElement$} (m3);
\end{scope}

\begin{scope}[every edge/.style={draw=black}]
    \path [-] (p0) edge  [bend left=50] (m1);
    \path [-] (p1) edge  [bend right=50] (m2);
    \path [-] (p2) edge  [bend left=50] (m3);
\end{scope}

\begin{scope}[every node/.style={inner sep=0.4mm, draw, circle, minimum size = 0pt}]
    \node[label=below:$+0$] (2p0) at (0,-2) {};
    \node[label=below:$-1$] (2m1) at (1.5,-2) {};
    \node[label=below:$+1$] (2p1) at (3,-2) {};
    \node[label=below:$-2$] (2m2) at (4.5,-2) {};
    \node[label=below:$+2$] (2p2) at (6,-2) {};
    \node[label=below:$-3$] (2m3) at (7.5,-2) {};
\end{scope}

\begin{scope}[every edge/.style={draw=black}]
    \path [-] (2p1) edge node [black, pos=0.5, sloped, above, yshift=-0.05cm] {} (2m2);
\end{scope}

\begin{scope}[dashed]
    \path [-] (2p0) edge node [black, pos=0.5, sloped, above] {$\deletionElement$} (2m1);
    \path [-] (2p2) edge node [black, pos=0.5, sloped, above] {$\deletionElement$} (2m3);
\end{scope}

\begin{scope}[every edge/.style={draw=black}]
    \path [-] (2p0) edge  [bend left=80] (2m1);
    \path [-] (2p1) edge  [bend left=80] (2m2);
    \path [-] (2p2) edge  [bend left=80] (2m3);
\end{scope}

\end{tikzpicture}

\caption{Exemplo de uma transposição que age em uma tripla orientada e cria três novos ciclos. A transposição move os elementos $\deletionElement$ de forma que o ciclo unitário criado é sempre limpo. Neste exemplo, temos $A = (0~\deletionElement~2~\deletionElement~1~\deletionElement~3)$ e $\iota^n$ com $n = 2$.
\label{cap4:fig:transposition_oriented_cycles}}
\end{figure}

\begin{lemma}\label{cap4:lemma:transposition_nonoriented_multiple_cycles}
	Para qualquer instância $\I = (A, \iota^n)$, tal que $|\pi^A| + 1 - c(\I) > 0$, $G(\I)$ só possui arestas de destino limpas e $G(\I)$ não possui ciclos divergentes ou ciclos orientados, se existem pelo menos dois ciclos $C$ e $D$ que são não unitários e rotulados em $G(\I)$, então existe uma transposição $\tau$ com $\Delta \cclean(\I, \tau) + \Delta \lambda_{\insertion}(\I, \tau) \geq 1$.
\end{lemma}

\begin{proof}
Sejam $C = (o_1,\ldots,o_\ell)$ e $D = ({o'}_1,\ldots,{o'}_k)$, tal que $\ell \geq 2$ e $k \geq 2$, dois ciclos rotulados em $G(\I)$. Suponha, sem perda de generalidade, que $o_\ell > {o'}_k$. Uma transposição $\tau$ aplicada nas arestas de origem $e_{o_1}$, $e_{o_\ell}$ e $e_{{o'}_k}$ cria dois ciclos $C'$ e $D'$, tal que $C'$ é unitário e $D'$ é um $(\ell+k-1)$-ciclo. Podemos escolher a transposição de forma que $C'$ seja um ciclo limpo, para isso basta mover qualquer $\deletionElement$ das arestas $e_{o_1}$ e $e_{o_\ell}$ para o ciclo $D'$, como mostrado na Figura~\ref{cap4:fig:transp_two_non_oriented_cycles}. Portanto, temos $\Delta \cclean(\I, \tau) + \Delta \lambda_{\insertion}(\I, \tau) \geq 1$.
\end{proof}

\begin{figure}[tbh]
\centering
\begin{tikzpicture}[scale=0.7]
\scriptsize
\begin{scope}[every node/.style={inner sep=0.4mm, draw, circle, minimum size = 0pt}]
    \node[label=below:$+0$] (p0) at (0,0) {};
    \node[label=below:$-3$] (m3) at (1.5,0) {};
    \node[label=below:$+3$] (p3) at (3,0) {};
    \node[label=below:$-2$] (m2) at (4.5,0) {};
    \node[label=below:$+2$] (p2) at (6,0) {};
    \node[label=below:$-1$] (m1) at (7.5,0) {};
    \node[label=below:$+1$] (p1) at (9,0) {};
    \node[label=below:$-4)$] (m4) at (10.5,0) {};
\end{scope}

\begin{scope}[every edge/.style={draw=black}, every node/.style={inner sep=0pt, minimum size = 0pt}]
\node[label=\phantom{}] (bi1) at (.7, -0.8) {};
\node[label=\phantom{}] (bi2) at (3.75, -0.8) {};
\path [{Bar}-{Bar}] (bi1) edge node [black, pos=0.5, sloped, below] {$(3)$} (bi2);

\node[label=\phantom{}] (bi3) at (3.7, -0.8) {};
\node[label=\phantom{}] (bi4) at (9.7, -0.8) {};
\path [-{Bar}] (bi3) edge node [black, pos=0.5, sloped, below] {$(\deletionElement~2~\deletionElement~1~\deletionElement)$} (bi4);
\end{scope}

\begin{scope}[dashed]
    \path [-] (p0) edge node [black, pos=0.5, sloped, above, yshift=-0.05cm] {$\deletionElement$} (m3);
    \path [-] (p3) edge node [black, pos=0.5, sloped, above, yshift=-0.05cm] {$\deletionElement$} (m2);
    \path [-] (p2) edge node [black, pos=0.5, sloped, above, yshift=-0.05cm] {$\deletionElement$} (m1);
    \path [-] (p1) edge node [black, pos=0.5, sloped, above, yshift=-0.05cm] {$\deletionElement$} (m4);
\end{scope}

\begin{scope}[every edge/.style={draw=black}]
    \path [-] (p0) edge  [bend left=50] (m1);
    \path [-] (p1) edge  [bend right=50] (m2);
    \path [-] (p2) edge  [bend right=50] (m3);
    \path [-] (p3) edge  [bend left=50] (m4);
\end{scope}

\begin{scope}[every node/.style={inner sep=0.4mm, draw, circle, minimum size = 0pt}]
    \node[label=below:$+0$] (2p0) at (0,-3) {};
    \node[label=below:$-2$] (2m2) at (1.5,-3) {};
    \node[label=below:$+2$] (2p2) at (3,-3) {};
    \node[label=below:$-1$] (2m1) at (4.5,-3) {};
    \node[label=below:$+1$] (2p1) at (6,-3) {};
    \node[label=below:$-3$] (2m3) at (7.5,-3) {};
    \node[label=below:$+3$] (2p3) at (9,-3) {};
    \node[label=below:$-4$] (2m4) at (10.5,-3) {};
\end{scope}

\begin{scope}[dashed]
    \path [-] (2p0) edge node [black, pos=0.5, sloped, above, yshift=-0.05cm] {$\deletionElement$} (2m2);
    \path [-] (2p1) edge node [black, pos=0.5, sloped, above, yshift=-0.05cm] {$\deletionElement$} (2m3);
    \path [-] (2p2) edge node [black, pos=0.5, sloped, above, yshift=-0.05cm] {$\deletionElement$} (2m1);
\end{scope}

\begin{scope}[every edge/.style={draw=black}]
    \path [-] (2p3) edge node [black, pos=0.5, sloped, above] {} (2m4);
\end{scope}

\begin{scope}[every edge/.style={draw=black}]
    \path [-] (2p0) edge  [bend left=70] (2m1);
    \path [-] (2p1) edge  [bend right=70] (2m2);
    \path [-] (2p2) edge  [bend left=70] (2m3);
    \path [-] (2p3) edge  [bend left=70] (2m4);
\end{scope}

\end{tikzpicture}

\caption{Exemplo de uma transposição que age em dois ciclos rotulados não orientados e cria um novo ciclo limpo. Neste exemplo, temos $A = (0~\deletionElement~3~\deletionElement~2~\deletionElement~1~\deletionElement~4)$ e $\iota^n$ com $n = 3$.
\label{cap4:fig:transp_two_non_oriented_cycles}}
\end{figure}

\begin{lemma}\label{cap4:lemma:transposition_nonoriented_single_cycle}
Para qualquer instância $\I = (A, \iota^n)$, tal que $|\pi^A| + 1 - c(\I) > 0$, $G(\I)$ só possui arestas de destino limpas e $G(\I)$ não possui ciclos divergentes ou ciclos orientados, se existe apenas um único ciclo rotulado $C$ em $G(\I)$ e $C$ é não unitário, então existe uma sequência $S_\tau$ de $k$ transposições tal que $\Delta \cclean(\I, S_\tau) + \Delta \lambda_{\insertion}(\I, S_\tau) = k$, para $k \in \{2,3\}$.
\end{lemma}

\begin{proof}
Dividimos a prova de acordo com o número de arestas de origem do ciclo rotulado $C$. 
Suponha que $C = (o_1, o_2)$ é o único ciclo rotulado de $G(\I)$. Bafna e Pevzner \cite[Lema~4.6]{1998-bafna-pevzner} mostraram que existe outro ciclo não orientado $D = (o'_1, \ldots, o'_k)$, com $k \geq 2$, tal que ou $o_1 > o'_x > o_2 > o'_{x+1}$ ou $o'_x > o_1 > o'_{x+1} > o_2$, para algum $1 \leq x \leq k-1$. Como $C$ é o único ciclo rotulado, $D$ deve ser um ciclo limpo. Suponha, sem perda de generalidade, que $o_1 > o'_x > o_2 > o'_{x+1}$. Podemos aplicar duas transposições nesses ciclos que criam dois novos ciclos limpos no grafo. Primeiramente, aplicamos uma transposição $\tau_1$ nas arestas de origem  $e_{o_1}$, $e_{o_2}$ e $e_{o'_{x+1}}$ que gera um ciclo unitário rotulado $C'$ e um ciclo limpo $D'$ que deve ser orientado~\cite{1998-bafna-pevzner}. Para que $D'$ seja limpo, basta que a transposição mova qualquer $\deletionElement$ para a aresta do ciclo unitário $C'$. Agora, usando o Lema~\ref{cap4:lemma:transposition_oriented_cycles}, aplicamos uma transposição $\tau_2$ no ciclo limpo $D'$ que o transforma em três ciclos limpos. Portanto, temos $\Delta \cclean(\I, S_\tau) + \Delta \lambda_{\insertion}(\I, S_\tau) = 2$ com $S = (\tau_1, \tau_2)$. A Figura~\ref{cap4:fig:transp_long_non_oriented_1} mostra um exemplo dessa sequência de transposições.

Suponha que $C = (o_1, \ldots, o_\ell)$, com $\ell \geq 3$, é o único ciclo rotulado de $G(\I)$. De forma similar ao caso anterior, aplicamos três transposições que geram três novos ciclos limpos. Primeiramente, aplicamos uma transposição $\tau_1$ em um ou dois ciclos não unitários limpos $D$ e $E$, que devem sempre existir~\cite[Teorema~4.7]{1998-bafna-pevzner}, que torna $C$ em um ciclo orientado rotulado $C'$, tal que $\Delta c(\I, \tau_1) = \Delta \cclean(\I, \tau_1) = 0$. Agora, usando o Lema~\ref{cap4:lemma:transposition_oriented_cycles}, aplicamos uma transposição $\tau_2$ em $C'$ que o transforma em três ciclos tal que pelo menos um dele é um ciclo unitário limpo. De acordo com o Teorema~4.7 de Bafna e Pevzner~\cite{1998-bafna-pevzner}, a sequência $(\tau_1, \tau_2)$ transforma outro ciclo do grafo em um ciclo orientado, sendo que esse ciclo possui arestas dos ciclos $D$ e $E$ que foram afetados pela primeira transposição. Esse ciclo deve ser limpo uma vez que $D$ e $E$ são ciclos limpos. Usando novamente o Lema~\ref{cap4:lemma:transposition_oriented_cycles}, aplicamos uma transposição $\tau_3$ nesse ciclo orientado limpo que o transforma em três novos ciclos limpos. Portanto, temos $\Delta \cclean(\I, S_\tau) + \Delta \lambda_{\insertion}(\I, S_\tau) = 3$ com $S = (\tau_1, \tau_2, \tau_3)$. As figuras~\ref{cap4:fig:transp_long_non_oriented_2} e \ref{cap4:fig:transp_long_non_oriented_3} mostram exemplos desse caso.
\end{proof}

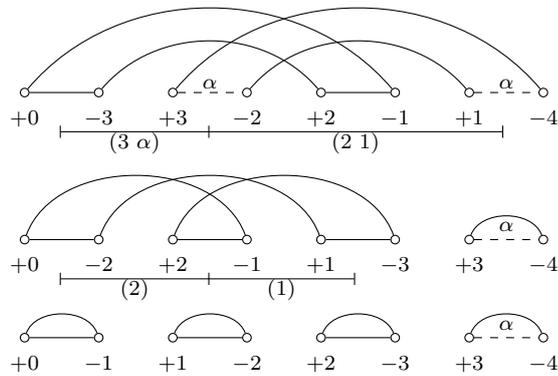
\begin{figure}[H]
\centering
\begin{tikzpicture}[scale=0.65]
\scriptsize
\begin{scope}[every node/.style={inner sep=0.4mm, draw, circle, minimum size = 0pt}]
    \node[label=below:$+0$] (p0) at (0,0) {};
    \node[label=below:$-3$] (m3) at (1.5,0) {};
    \node[label=below:$+3$] (p3) at (3,0) {};
    \node[label=below:$-2$] (m2) at (4.5,0) {};
    \node[label=below:$+2$] (p2) at (6,0) {};
    \node[label=below:$-1$] (m1) at (7.5,0) {};
    \node[label=below:$+1$] (p1) at (9,0) {};
    \node[label=below:$-4$] (m4) at (10.5,0) {};
\end{scope}

\begin{scope}[every edge/.style={draw=black}, every node/.style={inner sep=0pt, minimum size = 0pt}]
\node[label=\phantom{}] (bi1) at (.7, -0.8) {};
\node[label=\phantom{}] (bi2) at (3.75, -0.8) {};
\path [{Bar}-{Bar}] (bi1) edge node [black, pos=0.5, sloped, below] {$(3~\deletionElement)$} (bi2);

\node[label=\phantom{}] (bi3) at (3.7, -0.8) {};
\node[label=\phantom{}] (bi4) at (9.7, -0.8) {};
\path [-{Bar}] (bi3) edge node [black, pos=0.5, sloped, below] {$(2~1)$} (bi4);
\end{scope}
\begin{scope}[dashed]

    \path [-] (p3) edge node [black, pos=0.5, sloped, above, yshift=-0.05cm] {$\deletionElement$} (m2);
    \path [-] (p1) edge node [black, pos=0.5, sloped, above, yshift=-0.05cm] {$\deletionElement$} (m4);
\end{scope}

\begin{scope}[every edge/.style={draw=black}]
    \path [-] (p0) edge node [black, pos=0.5, sloped, above, yshift=-0.05cm] {} (m3);
    \path [-] (p2) edge node [black, pos=0.5, sloped, above, yshift=-0.05cm] {} (m1);
    \path [-] (p0) edge  [bend left=50] (m1);
    \path [-] (p1) edge  [bend right=50] (m2);
    \path [-] (p2) edge  [bend right=50] (m3);
    \path [-] (p3) edge  [bend left=50] (m4);
\end{scope}

\begin{scope}[every node/.style={inner sep=0.4mm, draw, circle, minimum size = 0pt}]
    \node[label=below:$+0$] (2p0) at (0,-3) {};
    \node[label=below:$-2$] (2m2) at (1.5,-3) {};
    \node[label=below:$+2$] (2p2) at (3,-3) {};
    \node[label=below:$-1$] (2m1) at (4.5,-3) {};
    \node[label=below:$+1$] (2p1) at (6,-3) {};
    \node[label=below:$-3$] (2m3) at (7.5,-3) {};
    \node[label=below:$+3$] (2p3) at (9,-3) {};
    \node[label=below:$-4$] (2m4) at (10.5,-3) {};
\end{scope}

\begin{scope}[every edge/.style={draw=black}, every node/.style={inner sep=0pt, minimum size = 0pt}]
\node[label=\phantom{}] (bi1) at (.7, -3.8) {};
\node[label=\phantom{}] (bi2) at (3.75, -3.8) {};
\path [{Bar}-{Bar}] (bi1) edge node [black, pos=0.5, sloped, below] {$(2)$} (bi2);

\node[label=\phantom{}] (bi3) at (3.7, -3.8) {};
\node[label=\phantom{}] (bi4) at (6.7, -3.8) {};
\path [-{Bar}] (bi3) edge node [black, pos=0.5, sloped, below] {$(1)$} (bi4);
\end{scope}

\begin{scope}[every edge/.style={draw=black}]
    \path [-] (2p0) edge node [black, pos=0.5, sloped, above, yshift=-0.05cm] {} (2m2);
    \path [-] (2p1) edge node [black, pos=0.5, sloped, above, yshift=-0.05cm] {} (2m3);
    \path [-] (2p2) edge node [black, pos=0.5, sloped, above, yshift=-0.05cm] {} (2m1);
\end{scope}

\begin{scope}[dashed]
    \path [-] (2p3) edge node [black, pos=0.5, sloped, above] {$\deletionElement$} (2m4);
\end{scope}

\begin{scope}[every edge/.style={draw=black}]
    \path [-] (2p0) edge  [bend left=70] (2m1);
    \path [-] (2p1) edge  [bend right=70] (2m2);
    \path [-] (2p2) edge  [bend left=70] (2m3);
    \path [-] (2p3) edge  [bend left=70] (2m4);
\end{scope}

\begin{scope}[every node/.style={inner sep=0.4mm, draw, circle, minimum size = 0pt}]
    \node[label=below:$+0$] (3p0) at (0,-5) {};
    \node[label=below:$-1$] (3m1) at (1.5,-5) {};
    \node[label=below:$+1$] (3p1) at (3,-5) {};
    \node[label=below:$-2$] (3m2) at (4.5,-5) {};
    \node[label=below:$+2$] (3p2) at (6,-5) {};
    \node[label=below:$-3$] (3m3) at (7.5,-5) {};
    \node[label=below:$+3$] (3p3) at (9,-5) {};
    \node[label=below:$-4$] (3m4) at (10.5,-5) {};
\end{scope}

\begin{scope}[every edge/.style={draw=black}]
    \path [-] (3p0) edge node [black, pos=0.5, sloped, above, yshift=-0.05cm] {} (3m1);
    \path [-] (3p1) edge node [black, pos=0.5, sloped, above, yshift=-0.05cm] {} (3m2);
    \path [-] (3p2) edge node [black, pos=0.5, sloped, above, yshift=-0.05cm] {} (3m3);
\end{scope}

\begin{scope}[dashed]
    \path [-] (3p3) edge node [black, pos=0.5, sloped, above] {$\deletionElement$} (3m4);
\end{scope}

\begin{scope}[every edge/.style={draw=black}]
    \path [-] (3p0) edge  [bend left=70] (3m1);
    \path [-] (3p1) edge  [bend left=70] (3m2);
    \path [-] (3p2) edge  [bend left=70] (3m3);
    \path [-] (3p3) edge  [bend left=70] (3m4);
\end{scope}

\end{tikzpicture}

\caption{Exemplo de uma sequência de transposições agindo em dois ciclos não orientados $C = (4, 2)$ e $D = (3, 1)$, tal que $C$ é rotulado e $D$ é limpo. Essas operações criam dois novos ciclos limpos. Neste exemplo, temos $A = (0~3~\deletionElement~2~1~\deletionElement~4)$ e $\iota^n$ com $n = 3$.
\label{cap4:fig:transp_long_non_oriented_1}}
\end{figure}
\begin{figure}[H]
\centering
\begin{tikzpicture}[scale=0.75]
\scriptsize
\begin{scope}[every node/.style={inner sep=0.4mm, draw, circle, minimum size = 0pt}]
    \node[label=below:$+0$] (p0) at (0,0) {};
    \node[label=below:$-5$] (m5) at (1,0) {};
    \node[label=below:$+5$] (p5) at (2,0) {};
    \node[label=below:$-4$] (m4) at (3,0) {};
    \node[label=below:$+4$] (p4) at (4,0) {};
    \node[label=below:$-3$] (m3) at (5,0) {};
    \node[label=below:$+3$] (p3) at (6,0) {};
    \node[label=below:$-2$] (m2) at (7,0) {};
    \node[label=below:$+2$] (p2) at (8,0) {};
    \node[label=below:$-1$] (m1) at (9,0) {};
    \node[label=below:$+1$] (p1) at (10,0) {};
    \node[label=below:$-6$] (m6) at (11,0) {};
\end{scope}

\begin{scope}[dashed]
    \path [-] (p0) edge node [black, pos=0.5, sloped, above, yshift=-0.05cm] {$\deletionElement$} (m5);
    \path [-] (p4) edge node [black, pos=0.5, sloped, above, yshift=-0.05cm] {$\deletionElement$} (m3);
    \path [-] (p2) edge node [black, pos=0.5, sloped, above, yshift=-0.05cm] {$\deletionElement$} (m1);
\end{scope}

\begin{scope}[every edge/.style={draw=black}]
    \path [-] (p5) edge node [black, pos=0.5, sloped, above, yshift=-0.05cm] {} (m4);
    \path [-] (p3) edge node [black, pos=0.5, sloped, above, yshift=-0.05cm] {} (m2);
    \path [-] (p1) edge node [black, pos=0.5, sloped, above, yshift=-0.05cm] {} (m6);
\end{scope}

\begin{scope}[every edge/.style={draw=black}]
    \path [-] (p0) edge  [bend left=50] (m1);
    \path [-] (p1) edge  [bend right=50] (m2);
    \path [-] (p2) edge  [bend right=50] (m3);
    \path [-] (p3) edge  [bend right=50] (m4);
    \path [-] (p4) edge  [bend right=50] (m5);
    \path [-] (p5) edge  [bend left=50] (m6);
\end{scope}

\begin{scope}[every edge/.style={draw=black}, every node/.style={inner sep=0pt, minimum size = 0pt}]
\node[label=\phantom{}] (bi1) at (2.7, -0.8) {};
\node[label=\phantom{}] (bi2) at (6.75, -0.8) {};
\path [{Bar}-{Bar}] (bi1) edge node [black, pos=0.5, sloped, below] {$(4~\deletionElement~3)$} (bi2);

\node[label=\phantom{}] (bi3) at (6.7, -0.8) {};
\node[label=\phantom{}] (bi4) at (10.7, -0.8) {};
\path [-{Bar}] (bi3) edge node [black, pos=0.5, sloped, below] {$(2~\deletionElement~1)$} (bi4);
\end{scope}

\begin{scope}[every node/.style={inner sep=0.4mm, draw, circle, minimum size = 0pt}]
    \node[label=below:$+0$] (2p0) at (0,-3) {};
    \node[label=below:$-5$] (2m5) at (1,-3) {};
    \node[label=below:$+5$] (2p5) at (2,-3) {};
    \node[label=below:$-2$] (2m2) at (3,-3) {};
    \node[label=below:$+2$] (2p2) at (4,-3) {};
    \node[label=below:$-1$] (2m1) at (5,-3) {};
    \node[label=below:$+1$] (2p1) at (6,-3) {};
    \node[label=below:$-4$] (2m4) at (7,-3) {};
    \node[label=below:$+4$] (2p4) at (8,-3) {};
    \node[label=below:$-3$] (2m3) at (9,-3) {};
    \node[label=below:$+3$] (2p3) at (10,-3) {};
    \node[label=below:$-6$] (2m6) at (11,-3) {};
\end{scope}

\begin{scope}[dashed]
    \path [-] (2p0) edge node [black, pos=0.5, sloped, above, yshift=-0.05cm] {$\deletionElement$} (2m5);
    \path [-] (2p2) edge node [black, pos=0.5, sloped, above, yshift=-0.05cm] {$\deletionElement$} (2m1);
    \path [-] (2p4) edge node [black, pos=0.5, sloped, above, yshift=-0.05cm] {$\deletionElement$} (2m3);
\end{scope}

\begin{scope}[every edge/.style={draw=black}]
    \path [-] (2p5) edge node [black, pos=0.5, sloped, above, yshift=-0.05cm] {} (2m2);
    \path [-] (2p1) edge node [black, pos=0.5, sloped, above, yshift=-0.05cm] {} (2m4);
    \path [-] (2p3) edge node [black, pos=0.5, sloped, above, yshift=-0.05cm] {} (2m6);
\end{scope}

\begin{scope}[every edge/.style={draw=black}]
\end{scope}

\begin{scope}[every edge/.style={draw=black}]
    \path [-] (2p0) edge  [bend left=50] (2m1);
    \path [-] (2p1) edge  [bend right=50] (2m2);
    \path [-] (2p2) edge  [bend left=50] (2m3);
    \path [-] (2p3) edge  [bend right=50] (2m4);
    \path [-] (2p4) edge  [bend right=50] (2m5);
    \path [-] (2p5) edge  [bend left=40] (2m6);
\end{scope}

\begin{scope}[every edge/.style={draw=black}, every node/.style={inner sep=0pt, minimum size = 0pt}]
\node[label=\phantom{}] (bi1) at (0.7, -3.8) {};
\node[label=\phantom{}] (bi2) at (4.75, -3.8) {};
\path [{Bar}-{Bar}] (bi1) edge node [black, pos=0.5, sloped, below] {$(\deletionElement~5~2~\deletionElement)$} (bi2);

\node[label=\phantom{}] (bi3) at (4.7, -3.8) {};
\node[label=\phantom{}] (bi4) at (8.7, -3.8) {};
\path [-{Bar}] (bi3) edge node [black, pos=0.5, sloped, below] {$(1~4)$} (bi4);
\end{scope}

\begin{scope}[every node/.style={inner sep=0.4mm, draw, circle, minimum size = 0pt}]
    \node[label=below:$+0$] (3p0) at (0,-6) {};
    \node[label=below:$-1$] (3m1) at (1,-6) {};
    \node[label=below:$+1$] (3p1) at (2,-6) {};
    \node[label=below:$-4$] (3m4) at (3,-6) {};
    \node[label=below:$+4$] (3p4) at (4,-6) {};
    \node[label=below:$-5$] (3m5) at (5,-6) {};
    \node[label=below:$+5$] (3p5) at (6,-6) {};
    \node[label=below:$-2$] (3m2) at (7,-6) {};
    \node[label=below:$+2$] (3p2) at (8,-6) {};
    \node[label=below:$-3$] (3m3) at (9,-6) {};
    \node[label=below:$+3$] (3p3) at (10,-6) {};
    \node[label=below:$-6$] (3m6) at (11,-6) {};
\end{scope}

\begin{scope}[every edge/.style={draw=black}]
    \path [-] (3p0) edge node [black, pos=0.5, sloped, above, yshift=-0.05cm] {} (3m1);
    \path [-] (3p1) edge node [black, pos=0.5, sloped, above, yshift=-0.05cm] {} (3m4);
    \path [-] (3p5) edge node [black, pos=0.5, sloped, above, yshift=-0.05cm] {} (3m2);
    \path [-] (3p3) edge node [black, pos=0.5, sloped, above, yshift=-0.05cm] {} (3m6);
\end{scope}

\begin{scope}[dashed]
    \path [-] (3p4) edge node [black, pos=0.5, sloped, above, yshift=-0.05cm] {$\deletionElement$} (3m5);
    \path [-] (3p2) edge node [black, pos=0.5, sloped, above, yshift=-0.05cm] {$\deletionElement$} (3m3);
\end{scope}

\begin{scope}[every edge/.style={draw=black}]
    \path [-] (3p0) edge  [bend left=90] (3m1);
    \path [-] (3p1) edge  [bend left=50] (3m2);
    \path [-] (3p2) edge  [bend left=90] (3m3);
    \path [-] (3p3) edge  [bend right=50] (3m4);
    \path [-] (3p4) edge  [bend left=90] (3m5);
    \path [-] (3p5) edge  [bend left=40] (3m6);
\end{scope}

%

\begin{scope}[every edge/.style={draw=black}, every node/.style={inner sep=0pt, minimum size = 0pt}]
\node[label=\phantom{}] (bi1) at (2.7, -6.8) {};
\node[label=\phantom{}] (bi2) at (6.75, -6.8) {};
\path [{Bar}-{Bar}] (bi1) edge node [black, pos=0.5, sloped, below] {$(4~\deletionElement~5)$} (bi2);

\node[label=\phantom{}] (bi3) at (6.7, -6.8) {};
\node[label=\phantom{}] (bi4) at (10.7, -6.8) {};
\path [-{Bar}] (bi3) edge node [black, pos=0.5, sloped, below] {$(2~\deletionElement~\deletionElement~3)$} (bi4);
\end{scope}

%

\begin{scope}[every node/.style={inner sep=0.4mm, draw, circle, minimum size = 0pt}]
    \node[label=below:$+0$] (3p0) at (0,-8) {};
    \node[label=below:$-1$] (3m1) at (1,-8) {};
    \node[label=below:$+1$] (3p1) at (2,-8) {};
    \node[label=below:$-2$] (3m2) at (3,-8) {};
    \node[label=below:$+2$] (3p2) at (4,-8) {};
    \node[label=below:$-3$] (3m3) at (5,-8) {};
    \node[label=below:$+3$] (3p3) at (6,-8) {};
    \node[label=below:$-4$] (3m4) at (7,-8) {};
    \node[label=below:$+4$] (3p4) at (8,-8) {};
    \node[label=below:$-5$] (3m5) at (9,-8) {};
    \node[label=below:$+5$] (3p5) at (10,-8) {};
    \node[label=below:$-6$] (3m6) at (11,-8) {};
\end{scope}

\begin{scope}[every edge/.style={draw=black}]
    \path [-] (3p0) edge node [black, pos=0.5, sloped, above, yshift=-0.05cm] {} (3m1);
    \path [-] (3p1) edge node [black, pos=0.5, sloped, above, yshift=-0.05cm] {} (3m2);
    \path [-] (3p3) edge node [black, pos=0.5, sloped, above, yshift=-0.05cm] {} (3m4);
    \path [-] (3p5) edge node [black, pos=0.5, sloped, above, yshift=-0.05cm] {} (3m6);
\end{scope}

\begin{scope}[dashed]
    \path [-] (3p4) edge node [black, pos=0.5, sloped, above, yshift=-0.05cm] {$\deletionElement$} (3m5);
    \path [-] (3p2) edge node [black, pos=0.5, sloped, above, yshift=-0.05cm] {$\deletionElement$} (3m3);
\end{scope}

\begin{scope}[every edge/.style={draw=black}]
    \path [-] (3p0) edge  [bend left=90] (3m1);
    \path [-] (3p1) edge  [bend left=90] (3m2);
    \path [-] (3p2) edge  [bend left=90] (3m3);
    \path [-] (3p3) edge  [bend left=90] (3m4);
    \path [-] (3p4) edge  [bend left=90] (3m5);
    \path [-] (3p5) edge  [bend left=90] (3m6);
\end{scope}

\end{tikzpicture}

\caption{Exemplo de uma sequência de transposições que agem em dois ciclos não orientados $C = (5, 3, 1)$ e $D = (6, 4, 2)$, tal que $C$ é rotulado e $D$ é limpo. Essas operações criam três novos ciclos limpos. Neste exemplo, temos $A = (0~\deletionElement~5~4~\deletionElement~3~2~\deletionElement~1~6)$ e $\iota^n$ com $n = 5$.
\label{cap4:fig:transp_long_non_oriented_2}}
\end{figure}
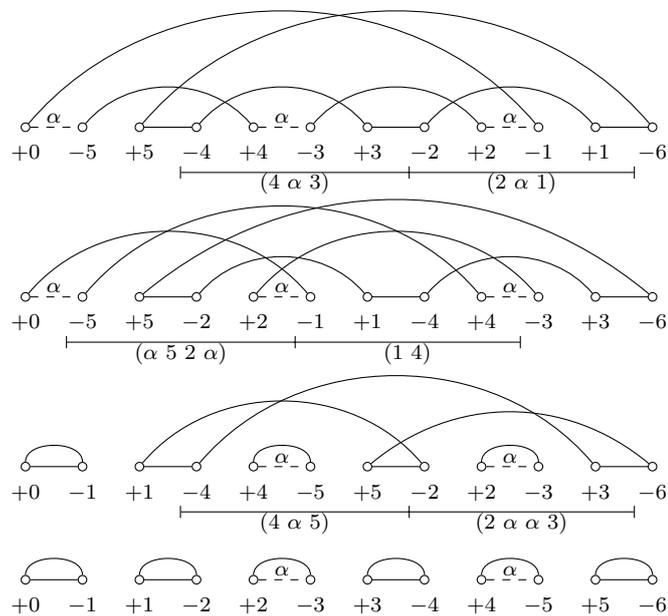
\begin{figure}[H]
\centering
\begin{tikzpicture}[scale=0.75]
\scriptsize
\begin{scope}[every node/.style={inner sep=0.4mm, draw, circle, minimum size = 0pt}]
    \node[label=below:$+0$] (p0) at (0,0) {};
    \node[label=below:$-3$] (m3) at (1,0) {};
    \node[label=below:$+3$] (p3) at (2,0) {};
    \node[label=below:$-2$] (m2) at (3,0) {};
    \node[label=below:$+2$] (p2) at (4,0) {};
    \node[label=below:$-1$] (m1) at (5,0) {};
    \node[label=below:$+1$] (p1) at (6,0) {};
    \node[label=below:$-6$] (m6) at (7,0) {};
    \node[label=below:$+6$] (p6) at (8,0) {};
    \node[label=below:$-5$] (m5) at (9,0) {};
    \node[label=below:$+5$] (p5) at (10,0) {};
    \node[label=below:$-4$] (m4) at (11,0) {};
    \node[label=below:$+4$] (p4) at (12,0) {};
    \node[label=below:$-7$] (m7) at (13,0) {};
\end{scope}

\begin{scope}[dashed]
    \path [-] (p3) edge node [black, pos=0.5, sloped, above, yshift=-0.05cm] {$\deletionElement$} (m2);
    \path [-] (p1) edge node [black, pos=0.5, sloped, above, yshift=-0.05cm] {$\deletionElement$} (m6);
    \path [-] (p5) edge node [black, pos=0.5, sloped, above, yshift=-0.05cm] {$\deletionElement$} (m4);
\end{scope}

\begin{scope}[every edge/.style={draw=black}]
    \path [-] (p0) edge node [black, pos=0.5, sloped, above, yshift=-0.05cm] {} (m3);
    \path [-] (p4) edge node [black, pos=0.5, sloped, above, yshift=-0.05cm] {} (m7);
    \path [-] (p2) edge node [black, pos=0.5, sloped, above, yshift=-0.05cm] {} (m1);
    \path [-] (p6) edge node [black, pos=0.5, sloped, above, yshift=-0.05cm] {} (m5);
\end{scope}

\begin{scope}[every edge/.style={draw=black}]
    \path [-] (p0) edge  [bend left=50] (m1);
    \path [-] (p1) edge  [bend right=50] (m2);
    \path [-] (p2) edge  [bend right=50] (m3);
    \path [-] (p3) edge  [bend left=50] (m4);
    \path [-] (p4) edge  [bend right=50] (m5);
    \path [-] (p5) edge  [bend right=50] (m6);
    \path [-] (p6) edge  [bend left=50] (m7);
\end{scope}

\begin{scope}[every edge/.style={draw=black}, every node/.style={inner sep=0pt, minimum size = 0pt}]
\node[label=\phantom{}] (bi1) at (0.5, -0.8) {};
\node[label=\phantom{}] (bi2) at (4.5, -0.8) {};
\path [{Bar}-{Bar}] (bi1) edge node [black, pos=0.5, sloped, below] {$(3~\deletionElement~2)$} (bi2);

\node[label=\phantom{}] (bi3) at (4.5, -0.8) {};
\node[label=\phantom{}] (bi4) at (8.5, -0.8) {};
\path [-{Bar}] (bi3) edge node [black, pos=0.5, sloped, below] {$(1~\deletionElement~6)$} (bi4);
\end{scope}

\begin{scope}[every node/.style={inner sep=0.4mm, draw, circle, minimum size = 0pt}]
    \node[label=below:$+0$] (p0) at (0,-3) {};
    \node[label=below:$-1$] (m1) at (1,-3) {};
    \node[label=below:$+1$] (p1) at (2,-3) {};
    \node[label=below:$-6$] (m6) at (3,-3) {};
    \node[label=below:$+6$] (p6) at (4,-3) {};
    \node[label=below:$-3$] (m3) at (5,-3) {};
    \node[label=below:$+3$] (p3) at (6,-3) {};
    \node[label=below:$-2$] (m2) at (7,-3) {};
    \node[label=below:$+2$] (p2) at (8,-3) {};
    \node[label=below:$-5$] (m5) at (9,-3) {};
    \node[label=below:$+5$] (p5) at (10,-3) {};
    \node[label=below:$-4$] (m4) at (11,-3) {};
    \node[label=below:$+4$] (p4) at (12,-3) {};
    \node[label=below:$-7$] (m7) at (13,-3) {};
\end{scope}

\begin{scope}[dashed]
    \path [-] (p3) edge node [black, pos=0.5, sloped, above, yshift=-0.05cm] {$\deletionElement$} (m2);
    \path [-] (p1) edge node [black, pos=0.5, sloped, above, yshift=-0.05cm] {$\deletionElement$} (m6);
    \path [-] (p5) edge node [black, pos=0.5, sloped, above, yshift=-0.05cm] {$\deletionElement$} (m4);
\end{scope}

\begin{scope}[every edge/.style={draw=black}]
    \path [-] (p0) edge node [black, pos=0.5, sloped, above, yshift=-0.05cm] {} (m1);
    \path [-] (p6) edge node [black, pos=0.5, sloped, above, yshift=-0.05cm] {} (m3);
    \path [-] (p2) edge node [black, pos=0.5, sloped, above, yshift=-0.05cm] {} (m5);
    \path [-] (p4) edge node [black, pos=0.5, sloped, above, yshift=-0.05cm] {} (m7);
\end{scope}

\begin{scope}[every edge/.style={draw=black}]
    \path [-] (p0) edge  [bend left=50] (m1);
    \path [-] (p1) edge  [bend left=50] (m2);
    \path [-] (p2) edge  [bend right=50] (m3);
    \path [-] (p3) edge  [bend left=50] (m4);
    \path [-] (p4) edge  [bend right=50] (m5);
    \path [-] (p5) edge  [bend right=50] (m6);
    \path [-] (p6) edge  [bend left=40] (m7);
\end{scope}

\begin{scope}[every edge/.style={draw=black}, every node/.style={inner sep=0pt, minimum size = 0pt}]
\node[label=\phantom{}] (bi1) at (2.5, -3.8) {};
\node[label=\phantom{}] (bi2) at (6.5, -3.8) {};
\path [{Bar}-{Bar}] (bi1) edge node [black, pos=0.5, sloped, below] {$(6~3)$} (bi2);

\node[label=\phantom{}] (bi3) at (6.5, -3.8) {};
\node[label=\phantom{}] (bi4) at (10.5, -3.8) {};
\path [-{Bar}] (bi3) edge node [black, pos=0.5, sloped, below] {$(\deletionElement~2~5~\deletionElement)$} (bi4);
\end{scope}

\begin{scope}[every node/.style={inner sep=0.4mm, draw, circle, minimum size = 0pt}]
    \node[label=below:$+0$] (p0) at (0,-6) {};
    \node[label=below:$-1$] (m1) at (1,-6) {};
    \node[label=below:$+1$] (p1) at (2,-6) {};
    \node[label=below:$-2$] (m2) at (3,-6) {};
    \node[label=below:$+2$] (p2) at (4,-6) {};
    \node[label=below:$-5$] (m5) at (5,-6) {};
    \node[label=below:$+5$] (p5) at (6,-6) {};
    \node[label=below:$-6$] (m6) at (7,-6) {};
    \node[label=below:$+6$] (p6) at (8,-6) {};
    \node[label=below:$-3$] (m3) at (9,-6) {};
    \node[label=below:$+3$] (p3) at (10,-6) {};
    \node[label=below:$-4$] (m4) at (11,-6) {};
    \node[label=below:$+4$] (p4) at (12,-6) {};
    \node[label=below:$-7$] (m7) at (13,-6) {};
\end{scope}

\begin{scope}[dashed]
    \path [-] (p1) edge node [black, pos=0.5, sloped, above, yshift=-0.05cm] {$\deletionElement$} (m2);
    \path [-] (p5) edge node [black, pos=0.5, sloped, above, yshift=-0.05cm] {$\deletionElement$} (m6);
\end{scope}

\begin{scope}[every edge/.style={draw=black}]
    \path [-] (p3) edge node [black, pos=0.5, sloped, above, yshift=-0.05cm] {} (m4);
    \path [-] (p0) edge node [black, pos=0.5, sloped, above, yshift=-0.05cm] {} (m1);
    \path [-] (p6) edge node [black, pos=0.5, sloped, above, yshift=-0.05cm] {} (m3);
    \path [-] (p2) edge node [black, pos=0.5, sloped, above, yshift=-0.05cm] {} (m5);
    \path [-] (p4) edge node [black, pos=0.5, sloped, above, yshift=-0.05cm] {} (m7);
\end{scope}

\begin{scope}[every edge/.style={draw=black}]
    \path [-] (p0) edge  [bend left=50] (m1);
    \path [-] (p1) edge  [bend left=80] (m2);
    \path [-] (p2) edge  [bend left=50] (m3);
    \path [-] (p3) edge  [bend left=50] (m4);
    \path [-] (p4) edge  [bend right=50] (m5);
    \path [-] (p5) edge  [bend left=80] (m6);
    \path [-] (p6) edge  [bend left=50] (m7);
\end{scope}

\begin{scope}[every edge/.style={draw=black}, every node/.style={inner sep=0pt, minimum size = 0pt}]
\node[label=\phantom{}] (bi1) at (4.5, -6.8) {};
\node[label=\phantom{}] (bi2) at (8.5, -6.8) {};
\path [{Bar}-{Bar}] (bi1) edge node [black, pos=0.5, sloped, below] {$(5~\deletionElement~6)$} (bi2);

\node[label=\phantom{}] (bi3) at (8.5, -6.8) {};
\node[label=\phantom{}] (bi4) at (12.5, -6.8) {};
\path [-{Bar}] (bi3) edge node [black, pos=0.5, sloped, below] {$(3~4)$} (bi4);
\end{scope}

\begin{scope}[every node/.style={inner sep=0.4mm, draw, circle, minimum size = 0pt}]
    \node[label=below:$+0$] (p0) at (0,-9) {};
    \node[label=below:$-1$] (m1) at (1,-9) {};
    \node[label=below:$+1$] (p1) at (2,-9) {};
    \node[label=below:$-2$] (m2) at (3,-9) {};
    \node[label=below:$+2$] (p2) at (4,-9) {};
    \node[label=below:$-3$] (m3) at (5,-9) {};
    \node[label=below:$+3$] (p3) at (6,-9) {};
    \node[label=below:$-4$] (m4) at (7,-9) {};
    \node[label=below:$+4$] (p4) at (8,-9) {};
    \node[label=below:$-5$] (m5) at (9,-9) {};
    \node[label=below:$+5$] (p5) at (10,-9) {};
    \node[label=below:$-6$] (m6) at (11,-9) {};
    \node[label=below:$+6$] (p6) at (12,-9) {};
    \node[label=below:$-7$] (m7) at (13,-9) {};
\end{scope}

\begin{scope}[dashed]
    \path [-] (p1) edge node [black, pos=0.5, sloped, above, yshift=-0.05cm] {$\deletionElement$} (m2);
    \path [-] (p5) edge node [black, pos=0.5, sloped, above, yshift=-0.05cm] {$\deletionElement$} (m6);
\end{scope}

\begin{scope}[every edge/.style={draw=black}]
    \path [-] (p3) edge node [black, pos=0.5, sloped, above, yshift=-0.05cm] {} (m4);
    \path [-] (p0) edge node [black, pos=0.5, sloped, above, yshift=-0.05cm] {} (m1);
    \path [-] (p6) edge node [black, pos=0.5, sloped, above, yshift=-0.05cm] {} (m7);
    \path [-] (p2) edge node [black, pos=0.5, sloped, above, yshift=-0.05cm] {} (m3);
    \path [-] (p4) edge node [black, pos=0.5, sloped, above, yshift=-0.05cm] {} (m5);
\end{scope}

\begin{scope}[every edge/.style={draw=black}]
    \path [-] (p0) edge  [bend left=80] (m1);
    \path [-] (p1) edge  [bend left=80] (m2);
    \path [-] (p2) edge  [bend left=80] (m3);
    \path [-] (p3) edge  [bend left=80] (m4);
    \path [-] (p4) edge  [bend left=80] (m5);
    \path [-] (p5) edge  [bend left=80] (m6);
    \path [-] (p6) edge  [bend left=80] (m7);
\end{scope}

\end{tikzpicture}

\caption{Exemplo de uma sequência de transposições que agem em três ciclos $C = (6, 4, 2)$, $D = (3, 1)$ e $E = (7, 5)$, tal que $C$ é o único ciclo rotulado. Essas operações criam três novos ciclos limpos. Neste exemplo, temos $A = (0~3~\deletionElement~2~1~\deletionElement~6~5~\deletionElement~4~7)$ e $\iota^n$ com $n = 6$.
\label{cap4:fig:transp_long_non_oriented_3}}
\end{figure}

\begin{lemma}\label{cap4:lemma:transposition_nonoriented_clean_cycles}
	Para qualquer instância $\I = (A, \iota^n)$, se $|\pi^A| + 1 - c(\I) > 0$, $G(\I)$ só possui ciclos limpos e $G(\I)$ não possui ciclos divergentes ou ciclos orientados, então existe sequência $S_\tau = (\tau_1, \tau_2)$ tal que $\Delta \cclean(\I, S_\tau) + \Delta \lambda_{\insertion}(\I, S_\tau) = 2$.
\end{lemma}

\begin{proof}
	Nesse caso, como todos os ciclos de $G(\I)$ são limpos e não existem ciclos divergentes, a string $A$ é uma permutação e podemos usar os resultados da Ordenação de Permutações por Transposições para grafo de ciclos. Considerando as condições do enunciado deste lema, Bafna e Pevzner~\cite{1998-bafna-pevzner} mostraram que sempre existe uma sequência $S_\tau = (\tau_1, \tau_2)$ que aumenta o número de ciclos em dois. Como transposições não inserem ou removem elementos, isso implica que $\Delta \cclean(\I, S_\tau) + \Delta \lambda_{\insertion}(\I, S_\tau) = 2$.
\end{proof}

Com esses lemas, podemos apresentar os algoritmos~\ref{cap4:algorithm:transp_cycle_graph} e \ref{cap4:algorithm:transp_rev_cycle_graph}. Assim como os algoritmos com {\it block interchanges}, esses algoritmos também possuem complexidade de tempo de $O(n^2)$, já que o grafo de ciclos rotulado $G(\I)$ pode ser criado em tempo linear, todos os laços de repetição executam $O(n)$ vezes, e toda operação dentro dos laços pode ser realizada em tempo linear.

No Teorema~\ref{cap4:theorem:2_approx_transp_cycle_graph} demonstramos que esses algoritmos possuem fator de aproximação igual a $2$ para os problemas de Distância de Rearranjos considerando os modelos $\Mindel_{\tau}$ e $\Mindel_{\rho,\tau}$.

\begin{theorem}\label{cap4:theorem:2_approx_transp_cycle_graph}
	Os algoritmos~\ref{cap4:algorithm:transp_cycle_graph} e \ref{cap4:algorithm:transp_rev_cycle_graph} são algoritmos de $2$-aproximação para os problemas da Distância de Transposições e Indels e da Distância de Transposições, Reversões e Indels, respectivamente.
\end{theorem}

\begin{proof}
	Considere o Algoritmo~\ref{cap4:algorithm:transp_cycle_graph}. A cada iteração, qualquer sequência $S'$ aplicada pelo algoritmo possui $k$ operações e satisfaz $\Delta \cclean(\I, S') + \Delta \lambda_{\insertion}(\I, S') \geq k$. Portanto, a sequência retornada ao fim do algoritmo possui tamanho de no máximo $|\pi^A| + 1 - \cclean(\I) + \lambda_{\insertion}(\I)$. Pelo Lema~\ref{cap4:lemma:lb_graph_transp}, esse algoritmo é uma $2$-aproximação para o problema da Distância de Transposições e Indels.
	
	A prova é similar para o Algoritmo \ref{cap4:algorithm:transp_rev_cycle_graph}.
\end{proof}

\begin{algorithm}[h]
	\caption{Algoritmo de $2$-Aproximação para a Distância de Transposições e Indels\label{cap4:algorithm:transp_cycle_graph}}
	\DontPrintSemicolon
	\Entrada{Uma instância $\I = (A, \iota^n)$}
	\Saida{Uma sequência de rearranjos que transforma $A$ em $\iota^n$}
	Seja $S \gets \emptyset$\;

	\Enqto{$|\pi^A| + 1 - \cclean(\I) + \lambda_\insertion(\I) > 0$}{
		\uSe{$G(\I)$ possui {\it runs} de inserção}{
			Seja $S' = (\phi)$, tal que $\phi$ é uma inserção com $\Delta c(\I, \insertion) + \Delta \lambda(\I, \insertion) = 1$ (Corolário~\ref{cap4:corollary:insertion_remove_run})\;
		}\uSenaoSe{$G(\I)$ possui algum ciclo unitário rotulado}{
			Seja $S' = (\psi)$, tal que $\psi$ é uma deleção que remove um {\it run} de deleção de um ciclo unitário (Corolário~\ref{cap4:corollary:deletion_remove_runs})\;
		}\uSenaoSe{$G(\I)$ possui algum ciclo orientado}{
			Seja $S' = (\tau)$, tal que $\tau$ é uma transposição com $\Delta \cclean(\I, \tau) + \Delta \lambda_{\insertion}(\I, \tau) \geq 1$ (Lema~\ref{cap4:lemma:transposition_oriented_cycles})\;
		}\uSenaoSe{$G(\I)$ possui dois ou mais ciclos rotulados}{
			Seja $S' = (\tau)$, tal que $\tau$ é uma transposição com $\Delta \cclean(\I, \tau) + \Delta \lambda_{\insertion}(\I, \tau) \geq 1$ (Lema~\ref{cap4:lemma:transposition_nonoriented_multiple_cycles})\;
		}\uSenaoSe{$G(\I)$ possui um único ciclo rotulado}{
			Seja $S'$ uma sequência de $k$ transposições tal que $\Delta \cclean(\I, S') + \Delta \lambda_{\insertion}(\I, S') = k$ (Lema~\ref{cap4:lemma:transposition_nonoriented_single_cycle})\;
		}\uSenao(){
			Seja $S' = (\tau_1, \tau_2)$, tal que $\Delta \cclean(\I, S') + \Delta \lambda_{\insertion}(\I, S') = 2$ (Lema~\ref{cap4:lemma:transposition_nonoriented_clean_cycles})\;
		}
		$A \gets A \comp S'$\;
		Adicione as operações de $S'$ na sequência $S$\;
	}

	{\bf retorne} a sequência $S$\;
\end{algorithm}

\begin{algorithm}[h]
	\caption{Algoritmo de $2$-Aproximação para a Distância de Transposições, Reversões e Indels\label{cap4:algorithm:transp_rev_cycle_graph}}
	\DontPrintSemicolon
	\Entrada{Uma instância $\I = (A, \iota^n)$}
	\Saida{Uma sequência de rearranjos que transforma $A$ em $\iota^n$}
	Seja $S \gets \emptyset$\;

	\Enqto{$|\pi^A| + 1 - \cclean(\I) + \lambda_\insertion(\I) > 0$}{
		\uSe{$G(\I)$ possui {\it runs} de inserção}{
			Seja $S' = (\phi)$, tal que $\phi$ é uma inserção com $\Delta c(\I, \insertion) + \Delta \lambda(\I, \insertion) = 1$ (Corolário~\ref{cap4:corollary:insertion_remove_run})\;
		}\uSenaoSe{$G(\I)$ possui ciclos divergentes}{
			Seja $S' = ({\rho})$, tal que $\rho$ é uma reversão com $\Delta \cclean(\I, \rho) + \Delta \lambda_{\insertion}(\I, \rho) = 1$ (Corolário~\ref{cap4:corollary:reversal_divergent})\;
		}\uSenaoSe{$G(\I)$ possui algum ciclo unitário rotulado}{
			Seja $S' = (\psi)$, tal que $\psi$ é uma deleção que remove um {\it run} de deleção de um ciclo unitário (Corolário~\ref{cap4:corollary:deletion_remove_runs})\;
		}\uSenaoSe{$G(\I)$ possui algum ciclo orientado}{
			Seja $S' = (\tau)$, tal que $\tau$ é uma transposição com $\Delta \cclean(\I, \tau) + \Delta \lambda_{\insertion}(\I, \tau) \geq 1$ (Lema~\ref{cap4:lemma:transposition_oriented_cycles})\;
		}\uSenaoSe{$G(\I)$ possui dois ou mais ciclos rotulados}{
			Seja $S' = (\tau)$, tal que $\tau$ é uma transposição com $\Delta \cclean(\I, \tau) + \Delta \lambda_{\insertion}(\I, \tau) \geq 1$ (Lema~\ref{cap4:lemma:transposition_nonoriented_multiple_cycles})\;
		}\uSenaoSe{$G(\I)$ possui um único ciclo rotulado}{
			Seja $S'$ uma sequência de $k$ transposições tal que $\Delta \cclean(\I, S') + \Delta \lambda_{\insertion}(\I, S') = k$ (Lema~\ref{cap4:lemma:transposition_nonoriented_single_cycle})\;
		}\uSenao(){
			Seja $S' = (\tau_1, \tau_2)$, tal que $\Delta \cclean(\I, S') + \Delta \lambda_{\insertion}(\I, S') = 2$ (Lema~\ref{cap4:lemma:transposition_nonoriented_clean_cycles})\;
		}
		$A \gets A \comp S'$\;
		Adicione as operações de $S'$ na sequência $S$\;
	}

	{\bf retorne} a sequência $S$\;
\end{algorithm}

\section{Conclusões}

Neste capítulo, estudamos problemas de Distância de Rearranjos em Genomas Desbalanceados utilizando a representação clássica de genomas. Para strings com e sem sinais, estudamos modelos que envolvem a combinação de {\it indels} com reversões, transposições e {\it block interchanges}. 

Na Seção~\ref{cap4:section:complexidade}, demonstramos que os seguintes problemas são NP-difíceis: a Distância de Reversões e Indels em Strings sem Sinais; a Distância de Transposições e Indels em Strings sem Sinais; e a Distância de Transposições, Reversões e Indels em Strings com ou sem Sinais.

Na Seção~\ref{cap4:section:breakpoints}, apresentamos algoritmos de aproximação que usam uma adaptação do conceito de {\it breakpoints} para genomas desbalanceados. Já na Seção~\ref{cap4:section:grafo_ciclos}, apresentamos algoritmos de $2$-aproximação usando o grafo de ciclos rotulado, que foi introduzido na Seção~\ref{cap2:sec:grafo_ciclos_strings} do Capítulo~\ref{cap:label:fundamentacao}. A Tabela~\ref{cap4:table:resumo_resultados} resume o fator de aproximação alcançado para cada problema estudado neste capítulo.

\begin{table}[H]
\caption{Resumo dos algoritmos apresentados neste capítulo para os problemas de Distância de Rearranjos em Genomas Desbalanceados.\label{cap4:table:resumo_resultados}}
\resizebox{0.99\textwidth}{!}{
\begin{tabular}{@{}lcc@{}}
\toprule
Modelo & Seção~\ref{cap4:section:breakpoints} & Seção~\ref{cap4:section:grafo_ciclos} \\ \midrule
Reversões e Indels (sem sinais) & 2-aproximação & - \\
Transposições e Indels (sem sinais) & 3-aproximação & 2-aproximação \\
Transposições, Reversões e Indels (sem sinais) & 3-aproximação & - \\
Transposições, Reversões e Indels (com sinais)& - & 2-aproximação \\
Block Interchanges e Indels (sem sinais) & - & 2-aproximação \\
Block Interchanges, Reversões e Indels (com sinais)& - & 2-aproximação \\ \bottomrule
\end{tabular}
}
\end{table}

\chapter{Distância em Genomas Desbalanceados com Regiões Intergênicas}\label{cap:label:intergenicos}

Estudos que incorporam regiões intergênicas são relativamente recentes. Esses estudos assumem que não há genes repetidos nos genomas e que inserções e deleções afetam apenas regiões intergênicas. Dessa forma, os genomas possuem o mesmo conjunto de genes e podem ser modelados usando permutações {\revisaof e uma lista de valores numéricos representando os tamanhos das regiões intergênicas}. Para permutações com sinais, Fertin e coautores~\cite{2017-fertin-etal} mostraram que o problema de Distância de Rearranjos Intergênicos para o modelo que contém apenas DCJs é NP-difícil e apresentaram uma $4/3$-aproximação, um esquema de aproximação de tempo polinomial, e uma formulação de programação linear inteira. Quando inserções e deleções de nucleotídeos (i.e., {\it indels} que afetam apenas o tamanho de regiões intergênicas, mas não inserem ou removem genes) são incorporadas ao modelo com DCJs para genomas balanceados, a distância pode ser calculada em tempo polinomial~\cite{2016b-bulteau-etal}. Note que ao considerar permutações em problemas com representação intergênica, qualquer {\it indel} intergênico deve ser um {\it indel} de nucleotídeos já que o uso de permutações implica que os genomas são balanceados.

Oliveira e coautores~\cite{2021b-oliveira-etal} apresentaram uma $2$-aproximação para reversões intergênicas em permutações com sinais, além da prova de NP-dificuldade para esse problema. Eles também apresentaram uma $2$-aproximação para uma versão do problema com reversões e {\it indels} intergênicos em permutações com sinais, que ainda possui sua complexidade em aberto. Já para o modelo com reversões e transposições intergênicas em permutações com sinais, Oliveira e coautores~\cite{2021a-oliveira-etal} apresentaram uma $3$-aproximação e uma prova de NP-dificuldade. 

Considerando permutações sem sinais, Brito e coautores~\cite{2020a-brito-etal} demonstraram que o problema de Distância de Rearranjos Intergênicos é NP-difícil para os seguintes modelos: reversões intergênicas; reversões e {\it indels} intergênicos; reversões e transposições intergênicas; e reversões, transposições e {\it indels} intergênicos. Os autores apresentaram uma $4$-aproximação para os modelos com reversões intergênicas e reversões e {\it indels} intergênicos, e apresentaram uma $4.5$-aproximação para os outros dois modelos. Para transposições intergênicas em permutações sem sinais, Oliveira e coautores~\cite{2020-oliveira-etal} desenvolveram uma $3.5$-aproximação e provaram que o problema é NP-difícil. 

Neste capítulo, estudamos problemas de Distância de Rearranjos Intergênicos em genomas desbalanceados, considerando os seguintes modelos: 
\begin{itemize}
	\item $\Mindel_{\rho}=\{\rho, \psi, \phi\}$: reversões e {\it indels} em strings com ou sem sinais;
	\item $\Mindel_{\tau}=\{\tau, \psi, \phi\}$: transposições e {\it indels} em strings sem sinais;
	\item $\Mindel_{\rho,\tau}=\{\rho, \tau, \psi, \phi\}$: reversões, transposições, e {\it indels} em strings com ou sem sinais.
\end{itemize}

\section{Complexidade dos Problemas}\label{cap5:section:complexidade}

Nesta seção, apresentamos provas de NP-dificuldade para problemas de Distância de Rearranjos Intergênicos em Genomas Desbalanceados considerando os modelos de rearranjos $\Mindel_{\rho}$, $\Mindel_{\tau}$ e $\Mindel_{\rho,\tau}$.

\begin{lemma}\label{cap5:lemma:complexidade}
	O problema de Distância de Rearranjos Intergênicos é NP-difícil para os modelos $\Mindel_{\rho}$ e $\Mindel_{\tau}$, considerando strings sem sinais, e para o modelo $\Mindel_{\rho,\tau}$, considerando strings com ou sem sinais.
\end{lemma}

\begin{proof}
	Considere o modelo $\Mindel_{\rho}$. A versão de decisão do problema da Distância de Reversões e Indels Intergênicos em Strings sem Sinais (\pname{IRID}) recebe como entrada uma instância $\Ig = (\G_o, \G_d, k)$, onde $\G_o = (A, \breve{A})$ e $\G_d = (\iota^n, \breve{\iota}^n)$, e consiste em decidir se $\G_o$ pode ser transformado em $\G_d$ usando no máximo $k$ operações de reversões ou {\it indels}.
	
	Nesta demonstração, apresentamos uma redução do problema de Ordenação de Permutações sem Sinais por Reversões (\pname{SbR}), que é um problema NP-difícil, para o problema \pname{IRID}. Dada uma instância $(\pi, k)$ para \pname{SbR}, tal que $\pi$ tem tamanho $n$, criamos a instância $(\G_o, \G_d, k)$ para o problema \pname{IRID}, onde $\breve{\gamma} = (0, \ldots, 0)$, $|\breve{\gamma}| = n+1$, $\G_o = (\pi, \breve{\gamma})$, e $\G_d = (\iota^n, \breve{\gamma})$. Mostramos a seguir que a instância $(\pi, k)$ do problema \pname{SbR} é satisfeita se, e somente se, a instância $(\G_o, \G_d, k)$ para o problema \pname{IRID} é satisfeita.

	Se existe uma sequência de $S$ reversões de tamanho no máximo $k$ que transforma $\pi$ em $\iota^n$, então uma sequência similar de mesmo tamanho pode ser usada para transformar $\G_o$ em $\G_d$, onde uma reversão em $S$ que inverte o segmento $(\pi_i, \ldots, \pi_j)$ é mapeada em $\rho^{(i,j)}_{(x,y)}$, com $x = y = 0$.

	Como $\Sigma_{\iota^n} \setminus \Sigma_{\pi} = \emptyset$ e $\breve{\gamma} = (0, \ldots, 0)$, uma sequência de rearranjos de tamanho mínimo que transforma $\G_o$ em $\G_d$ contém apenas reversões. Portanto, se existe um sequência de reversões e {\it indels} de tamanho no máximo $k$ que transforma $\G_o$ em $\G_d$, então existe uma sequência de reversões de tamanho no máximo $k$ que ordena $\pi$.

	A prova é similar para os outros modelos, já que a Ordenação de Permutações por Transposições~\cite{2012-bulteau-etal} e a Ordenação de Permutações com ou sem Sinais por Reversões e Transposições~\cite{2019b-oliveira-etal} são NP-difíceis.
\end{proof}

\section{Algoritmos de Aproximação usando Breakpoints}\label{cap5:section:breakpoints}

Nesta seção, apresentamos algoritmos de aproximação para o problema de Distância de Rearranjos Intergênicos em Strings sem Sinais considerando os modelos $\Mindel_{\rho}$, $\Mindel_{\tau}$ e $\Mindel_{\rho,\tau}$. Sempre consideramos que as strings de uma instância intergênica $\Ig = (\G_o, \G_d)$ estão nas suas versões estendidas.

Usamos o conceito de {\it breakpoints} intergênicos, apresentado na Seção~\ref{cap2:sec:breakpoints_intergenicos}, para a definição de limitantes para a distância e a criação dos algoritmos de aproximação. A seguir, apresentamos uma outra definição utilizada nos limitantes para a distância, sendo que essa definição é similar à Definição~\ref{cap4:def:delta_phi}, apresentada no Capítulo~\ref{cap:label:indel}.

\begin{definition}\label{cap5:def:delta_phi}
Dada uma operação (ou sequência de rearranjos) $\beta$ e uma instância intergênica $\Ig = (\G_o, \G_d)$, com $\G_o = (A, \breve{A})$ e $\G_d = (\iota^n, \breve{\iota}^n)$, definimos $\Delta \Phi(\Ig, \beta) = \Delta \Phi(A, \iota^n, \beta) = |\Sigma_{\iota^n} \setminus \Sigma_{A}| - |\Sigma_{\iota^n} \setminus \Sigma_{A'}|$, onde $A' = A \comp \beta$.
\end{definition}

Assim como no problema da Distância de Reversões, Transposições e Indels em Strings sem Sinais, usamos o conceito de {\it breakpoints} de reversões sem sinais para o modelo $\Mindel_{\rho,\tau}$. Os próximos lemas mostram como um rearranjo afeta o valor de $\bi_{\M}(\Ig) + |\Sigma_{\iota^n} \setminus \Sigma_{A}|$.

\begin{lemma}\label{cap5:lemma:breakpoint_bound_id}
Dada uma instância intergênica $\Ig = (\G_o, \G_d)$, com $\G_o = (A, \breve{A})$ e $\G_d = (\iota^n, \breve{\iota}^n)$, temos que $\Delta \Phi(\Ig, \beta) + \Delta \bi_{\M}(\Ig, \beta) \leq 2$, para qualquer {\it indel} $\beta$ e modelo $\M \in \{\Mindel_{\rho}$, $\Mindel_{\tau}$, $\Mindel_{\rho,\tau}$\}.
\end{lemma}

\begin{proof}
Uma inserção $\phi$ é aplicada entre dois elementos de $A$ e, portanto, pode remover apenas o {\it breakpoint} intergênico formado por esses dois elementos, caso exista. Note que modelamos o genoma de forma que qualquer par de elementos em $\Sigma_{\iota^n} \setminus \Sigma_A$ forma um {\it breakpoint}. Seja $\sigma$ a string a ser adicionada por $\phi$. Para $1 \leq i < |\sigma|$, o par $(\sigma_i, \sigma_{i+1})$ é um {\it breakpoint} intergênico. Portanto, temos que $\Delta \Phi(\Ig, \phi) + \Delta \bi_{\M}(\Ig, \phi) \leq |\sigma| + (1 - (|\sigma| - 1 )) = 2$, para $\M \in \{\Mindel_{\rho}$, $\Mindel_{\tau}$, $\Mindel_{\rho,\tau}$\}.

Uma deleção $\psi$ só pode ser aplicada em uma sequência contígua de elementos com valor $\deletionElement$. Como, por definição, não existe {\it breakpoint} intergênico entre um par de elementos com valor $\deletionElement$, apenas {\it breakpoints} intergênicos presentes nas duas extremidades da sequência afetada podem ser removidos. Além disso, para deleções, sempre temos $\Delta \Phi(\Ig, \psi) = 0$, já que apenas inserções afetam o conjunto de elementos a serem adicionados. 
\end{proof}

\begin{lemma}\label{cap5:lemma:breakpoint_bound_rev}
Dada uma instância intergênica $\Ig = (\G_o, \G_d)$, com $\G_o = (A, \breve{A})$ e $\G_d = (\iota^n, \breve{\iota}^n)$, temos que $\Delta \Phi(\Ig, \rho) + \Delta \bi_{\M}(\Ig, \rho) \leq 2$, para qualquer reversão $\rho$ e modelo $\M \in \{\Mindel_{\rho}, \Mindel_{\rho,\tau}$\}.
\end{lemma}

\begin{proof}
Considere uma reversão $\rho^{(i,j)}_{(x,y)}$. Note que essa reversão só pode remover {\it breakpoints} entre os pares de elementos $(A_{i-1}, A_i)$ e $(A_j, A_{j+1})$, caso existam. Assim como as deleções, sempre temos que $\Delta \Phi(\Ig, \rho) = 0$. Portanto, temos que $\Delta \Phi(\Ig, \rho) + \Delta \bi_{\M}(\Ig, \rho) \leq 2$, para $\M \in \{\Mindel_{\rho}, \Mindel_{\rho,\tau}$\}.
\end{proof}

\begin{lemma}\label{cap5:lemma:breakpoint_bound_transp}
Dada uma instância intergênica $\Ig = (\G_o, \G_d)$, com $\G_o = (A, \breve{A})$ e $\G_d = (\iota^n, \breve{\iota}^n)$, temos que $\Delta \Phi(\Ig, \tau) + \Delta \bi_{\M}(\Ig, \tau) \leq 3$, para qualquer transposição $\tau$ e modelo $\M \in \{\Mindel_{\tau}, \Mindel_{\rho,\tau}$\}.
\end{lemma}

\begin{proof}
Similar à prova do Lema~\ref{cap5:lemma:breakpoint_bound_rev}, mas devemos considerar que uma transposição afeta três adjacências de $A$.
\end{proof}

O próximo lema segue diretamente dos lemas~\ref{cap5:lemma:breakpoint_bound_id}, \ref{cap5:lemma:breakpoint_bound_rev}, \ref{cap5:lemma:breakpoint_bound_transp} e do fato de que $\bi_{\M}(\Ig) + |\Sigma_{\iota^n} \setminus \Sigma_{A}| = 0$ se, e somente se, $\G_o = \G_d$, para qualquer $\M \in \{\Mindel_{\rho}$, $\Mindel_{\tau}$, $\Mindel_{\rho,\tau}$\}.

\begin{lemma}\label{cap5:lemma:breakpoints_intergenicos_lower_bound}
	Para qualquer instância intergênica $\Ig = (\G_o, \G_d)$, com $\G_o = (A, \breve{A})$ e $\G_d = (\iota^n, \breve{\iota}^n)$, temos que:
	\begin{align*}
		d_{\Mindel_{\rho}}(\Ig) \geq \frac{\bi_{\Mindel_{\rho}}(\Ig) + |\Sigma_{\iota^n} \setminus \Sigma_{A}|}{2},\\
		d_{\Mindel_{\tau}}(\Ig) \geq \frac{\bi_{\Mindel_{\tau}}(\Ig) + |\Sigma_{\iota^n} \setminus \Sigma_{A}|}{3},\\
		d_{\Mindel_{\rho,\tau}}(\Ig) \geq \frac{\bi_{\Mindel_{\rho,\tau}}(\Ig) + |\Sigma_{\iota^n} \setminus \Sigma_{A}|}{3}.\\
	\end{align*}
\end{lemma}

Os próximos lemas apresentam casos em que sempre é possível achar um {\it indel} com $\Delta \Phi(\Ig, \beta) + \Delta \bi_{\M}(\Ig, \beta) > 0$. Essas operações serão úteis em todos os algoritmos de aproximação apresentados nesta seção.

\begin{lemma}\label{cap5:lemma:overchaged_undercharged}
Dada uma instância intergênica $\Ig = (\G_o, \G_d)$, com $\G_o = (A, \breve{A})$ e $\G_d = (\iota^n, \breve{\iota}^n)$, se existe algum {\it breakpoint} intergênico em $\Ig$ que é sobrecarregado ou subcarregado, então existe um {\it indel} $\beta$ com $\Delta \Phi(\Ig, \beta) + \Delta \bi_{\M}(\Ig, \beta) = 1$, para qualquer $\M \in \{\Mindel_{\rho}$, $\Mindel_{\tau}$, $\Mindel_{\rho,\tau}$\}.
\end{lemma}

\begin{proof}
Seja $(A_i, A_{i+1})$ um {\it breakpoint} intergênico sobrecarregado ou subcarregado. A região intergênica $\breve{\iota}^n_{x}$, tal que $x = \max(A_i, A_{i+1})$, é a região intergênica do genoma de destino $\G_d$ que está entre os dois elementos de $\iota^n$ correspondentes aos elementos $A_i$ e $A_{i+1}$. Se esse {\it breakpoint} é sobrecarregado, então a deleção $\psi^{(i+1,i+1)}_{(0, y)}$, onde $y = \breve{A}_{i+1} - \breve{\iota}^n_{x}$, remove esse {\it breakpoint}. Se esse {\it breakpoint} é subcarregado, então a inserção $\phi^{(i, \sigma, \breve{\sigma})}_{(0)}$, onde $\sigma = \emptyset$ e $\breve{\sigma} = (\breve{\iota}^n_{x} - \breve{A}_{i+1})$ remove esse {\it breakpoint}. Em ambos os casos, o {\it indel} descrito torna $\breve{A}_{i+1} = \breve{\iota}^n_{x}$ e remove o {\it breakpoint} intergênico entre $(A_i, A_{i+1})$. Note que $\Delta \Phi(\Ig, \beta) = 0$ em ambos os casos.
\end{proof}

\begin{lemma}\label{cap5:lemma:sets_indels}
	Dada uma instância intergênica $\Ig = (\G_o, \G_d)$, com $\G_o = (A, \breve{A})$ e $\G_d = (\iota^n, \breve{\iota}^n)$, se $|\Sigma_A \setminus \Sigma_{\iota^n}| > 0$ ou $|\Sigma_{\iota^n} \setminus \Sigma_{A}| > 0$, então existe um {\it indel} $\beta$ com $\Delta \Phi(\Ig, \beta) + \Delta \bi_{\M}(\Ig, \beta) \geq 1$, para qualquer $\M \in \{\Mindel_{\rho}$, $\Mindel_{\tau}$, $\Mindel_{\rho,\tau}$\}.
\end{lemma}

\begin{proof}
Considere que $|\Sigma_A \setminus \Sigma_{\iota^n}| > 0$. Seja $(A_i~\ldots~A_j)$ uma sequência maximal de elementos tal que $A_k = \deletionElement$, para todo $i \leq k \leq j$. Existem {\it breakpoints} intergênicos entre $(A_{i-1}, A_i)$ e $(A_j, A_{j+1})$. Portanto, uma deleção $\psi$ que age no segmento $(A_i~\ldots~A_j)$ remove esses dois {\it breakpoints} intergênicos e torna $(A_{i-1}, A_{j+1})$ adjacentes na nova string. Como $(A_{i-1}, A_{j+1})$ pode formar um {\it breakpoint} intergênico, temos que $\Delta \Phi(\Ig, \beta) + \Delta \bi_{\M}(\Ig, \beta) \geq 1$.

Considere que $|\Sigma_{\iota^n} \setminus \Sigma_{A}| > 0$. Seja $x$ um elemento que pertence ao conjunto $\Sigma_{\iota^n} \setminus \Sigma_{A}$. Seja $(A_i~\ldots~A_j)$ a {\it strip} contendo o elemento $x-1$. Note que $\breve{\iota}^n_x$ é a região intergênica entre $(x-1, x)$ em $\G_d$. Se essa {\it strip} é decrescente, então $A_i = x-1$ e a inserção $\phi^{(i-1, (x), (0,\breve{\iota}^n_x))}_{(\breve{A}_i)}$ não aumenta o número de {\it breakpoints} intergênicos e possui $\Delta \Phi(\Ig, \beta) = 1$. Se a {\it strip} é crescente, então $A_j = x-1$ e a inserção $\phi^{(j, (x), (\breve{\iota}^n_x, 0))}_{(0)}$ não aumenta o número de {\it breakpoints} intergênicos e possui $\Delta \Phi(\Ig, \beta) = 1$.
\end{proof}

\subsection{Algoritmo de Aproximação para Modelos com Reversões}

Apresentamos um algoritmo de $4$-aproximação para o problema da Distância de Reversões e Indels Intergênicos. Além disso, provamos que esse algoritmo também é uma $6$-aproximação para o problema da Distância de Reversões, Transposições e Indels Intergênicos

A seguir, mostramos casos em que sempre é possível encontrar uma sequência com no máximo duas operações (reversões ou {\it indels}) que remove {\it breakpoints} intergênicos.

\begin{lemma}\label{cap5:lemma:decreasing_strip}
	Para qualquer instância intergênica $\Ig = (\G_o, \G_d)$, com $\G_o = (A, \breve{A})$ e $\G_d = (\iota^n, \breve{\iota}^n)$, tal que $\Sigma_A = \Sigma_{\iota^n}$ e não existem {\it breakpoints} sobrecarregados ou subcarregados em $\G_o$, se a string $A$ possui pelo menos uma {\it strip} decrescente, então existe uma sequência de reversões e {\it indels} $S$ tal que $|S| \leq 2$ e $\Delta \Phi(\Ig, S) + \Delta \bi_{\M}(\Ig, S) \geq 1$.
\end{lemma}

\begin{proof}
Note que se $\Sigma_A = \Sigma_{\iota^n}$, então $\Delta \Phi(\Ig, \beta) = 0$ para qualquer reversão ou inserção $\beta$. Seja $(A_i~\ldots~A_j)$ a {\it strip} decrescente tal que $A_j$ é mínimo. Pela nossa escolha de $A_j$, a {\it strip} $(A_{i'}~\ldots~A_{j'})$ que contém o elemento $A_j - 1$ é crescente e, consequentemente, $A_{j'} = A_j - 1$.

Se $i' < i$, então a reversão $\rho^{(j'+1, j)}_{(x,y)}$ transforma $(A, \breve{A})$ em $(A', \breve{A}')$, tal que os elementos $A_j$ e $A_{j'} = A_{j} - 1$ são adjacentes em $A'$. Se $\breve{A}_{j'+1} + \breve{A}_{j+1} \geq \breve{\iota}^n_{A_j}$, então podemos escolher $x$ e $y$ de forma que $\breve{A}'_{j'+1} = \breve{\iota}^n_{A_j}$, fazendo com que um {\it breakpoint} intergênico seja removido. Caso contrário, após aplicar a reversão, o par $(A_{j'}, A_j)$ em $A'$ forma um {\it breakpoint} intergênico subcarregado. Nesse caso, existe uma inserção (Lema~\ref{cap5:lemma:overchaged_undercharged}) que remove esse {\it breakpoint} intergênico.

Se $i' > i$, então a reversão $\rho^{(j+1, j')}_{(x,y)}$ transforma $(A, \breve{A})$ em $(A', \breve{A}')$, tal que os elementos $A_j$ e $A_{j'} = A_{j} - 1$ são adjacentes em $A'$. De forma análoga ao caso anterior, talvez seja necessário usar uma inserção para remover um {\it breakpoint} intergênico subcarregado. 

Portanto, em ambos os casos, existe uma sequência de reversões e {\it indels} $S$ tal que $|S| \leq 2$ e $\Delta \Phi(\Ig, S) + \Delta \bi_{\M}(\Ig, S) \geq 1$. 
\end{proof}

\begin{lemma}\label{cap5:lemma:rev_two_breakpoints}
	Para qualquer instância intergênica $\Ig = (\G_o, \G_d)$, com $\G_o = (A, \breve{A})$ e $\G_d = (\iota^n, \breve{\iota}^n)$, tal que $\Sigma_A = \Sigma_{\iota^n}$ e não existem {\it breakpoints} sobrecarregados ou subcarregados em $\G_o$, se a string $A$ possui pelo menos uma {\it strip} decrescente e toda reversão $\rho$ que remove {\it breakpoints} (não necessariamente intergênicos) resulta na string $A \comp \rho$ que não possui {\it strips} decrescentes, então existe uma sequência de reversões e {\it indels} $S$ tal que $|S| \leq 2$ e $\Delta \Phi(\Ig, S) + \Delta \bi_{\M}(\Ig, S) = 2$.
\end{lemma}

\begin{proof}
Como $\Sigma_A = \Sigma_{\iota^n}$, temos que $A$ é uma permutação. Se qualquer reversão que remove {\it breakpoints} (não necessariamente {\it breakpoints} intergênicos) resulta na string $A \comp \rho$ que não possui {\it strips} decrescentes, então existe apenas uma reversão que remove {\it breakpoints} de $A$~(Lema \ref{cap4:lemma:reversal_breakpoint_no_decreasing}) e essa reversão remove dois {\it breakpoints} (não necessariamente intergênicos). Seja $(A_i, \ldots, A_j)$ o intervalo afetado por essa reversão que remove dois {\it breakpoints} de $A$. 

Seja $\rho^{(i,j)}_{(x,y)}$ uma reversão com valores arbitrários de $x$ e $y$. Essa reversão inverte o segmento $(A_i, \ldots, A_j)$ e remove dois {\it breakpoints}. Note que os {\it breakpoints} removidos são $(A_{i-1}, A_i)$ e $(A_j, A_{j+1})$. Essa reversão torna os pares $(A_{i-1}, A_j)$ e $(A_i, A_{j+1})$ adjacentes na string $A' = A \comp \rho^{(i,j)}_{(x,y)}$ e, como não existem {\it strips} decrescentes em $A'$, temos que $A_{i-1} + 1 = A_j$ e $A_i + 1 = A_{j+1}$. Sejam $\ell = \breve{A}_i + \breve{A}_{j+1}$ e $m= \breve{\iota}^n_{A_j} + \breve{\iota}^n_{A_{j+1}}$.

Se $\ell \geq m$, então aplicamos a reversão $\rho^{(i,j)}_{(x,y)}$ em $(A, \breve{A})$, com $x = \min(\breve{A}_i ,\breve{\iota}^n_{A_j})$ e $y = \breve{\iota}^n_{A_j} - x$, que resulta em $(A', \breve{A}')$, de forma que o par $(A'_{i-1}, A'_i) = (A_{i-1}, A_j)$ não é um {\it breakpoint} intergênico. No entanto, o par $(A'_{j}, A'_{j+1}) = (A_i, A_{j+1})$ pode ser um {\it breakpoint} intergênico, que pode ser removido usando um {\it indel} de acordo com o Lema~\ref{cap5:lemma:overchaged_undercharged}.

Considere que $\ell < m$. Nesse caso, primeiro aplicamos a inserção $\phi^{(i-1, \sigma, \breve{\sigma})}_{0}$, com $\sigma = \emptyset$ e $\breve{\sigma} = (m-\ell)$, que torna $\breve{A}'_i + \breve{A}'_{j+1} = \breve{\iota}^n_{A_j} + \breve{\iota}^n_{A_{j+1}}$, onde $(A', \breve{A}') = (A, \breve{A}) \cdot \phi^{(i-1, S, \breve{S})}_{0}$. Note que $A' = A$ e nenhum {\it breakpoint} intergênico é removido por essa inserção. Agora, aplicamos a reversão $\rho^{(i,j)}_{(x,y)}$ em $(A', \breve{A}')$, com $x = \min(\breve{A}'_i ,\breve{\iota}^n_{A_j})$ e $y = \breve{\iota}^n_{A_j} - x$, que remove dois {\it breakpoints} intergênicos.

Em ambos os casos, no máximo duas operações são aplicadas para remover dois {\it breakpoints} intergênicos.
\end{proof}

\begin{algorithm}[hb]
	\caption{\label{cap5:algorithm_reversals_breakpoints}
		Algoritmo para o problema de Distância de Rearranjos Intergênicos considerando os modelos $\Mindel_{\rho}$ e $\Mindel_{\rho,\tau}$
	}
	\DontPrintSemicolon
  \Entrada{Uma instância intergênica $\Ig = (\G_o, \G_d)$, com $\G_o = (A, \breve{A})$ e $\G_d = (\iota^n, \breve{\iota}^n)$}
  \Saida{Uma sequência de rearranjos $S$ que transforma $\G_o$ em $\G_d$}
  Seja $S \gets \emptyset$\;
	\Enqto{$|\Sigma_A \setminus \Sigma_{\iota^n}| > 0$ ou $|\Sigma_{\iota^n} \setminus \Sigma_{A}| > 0$}{
		Seja $\beta$ um {\it indel} de acordo com o Lema~\ref{cap5:lemma:sets_indels}\;
		$\G_o \gets \G_o \comp \beta$\;
		Adicione $\beta$ na sequência $S$\;
	}

	\Enqto{$\bi_{\M}(\Ig) > 0$}{
		\Se{existe {\it breakpoint} intergênico sobrecarregado ou subcarregado}{
			Seja $S' = (\beta)$, onde $\beta$ é um {\it indel} de acordo com o Lema~\ref{cap5:lemma:overchaged_undercharged}\;
		}\SenaoSe{existe {\it strip} decrescente}{
			Seja $S'$ uma sequência de reversões e {\it indels} de acordo com os lemas~\ref{cap5:lemma:decreasing_strip} e \ref{cap5:lemma:rev_two_breakpoints}.
		}\Senao(){
			Seja $(A_i, \ldots, A_j)$ uma {\it strip} tal que $i > 0$ e $i$ é mínimo\;
	  		Seja $S' = (\rho(i,j))$\;
		}
		$\G_o \gets \G_o \comp S'$\;
		Adicione as operações de $S'$ na sequência $S$\;
	}
  {\bf retorne} a sequência $S$\;
\end{algorithm}

O Algoritmo~\ref{cap5:algorithm_reversals_breakpoints} usa os lemas~\ref{cap5:lemma:overchaged_undercharged}, \ref{cap5:lemma:sets_indels}, \ref{cap5:lemma:decreasing_strip} e \ref{cap5:lemma:rev_two_breakpoints} e um passo adicional, quando nenhum dos lemas pode ser aplicado, que torna uma das {\it strips} do genoma de origem em uma {\it strip} decrescente. O Lema~\ref{cap5:lemma:upper_bound:rev_breakpoints} apresenta um limitante superior no número de rearranjos usados pelo algoritmo.

\begin{lemma}\label{cap5:lemma:upper_bound:rev_breakpoints}
Para qualquer instância intergênica $\Ig = (\G_o, \G_d)$, com $\G_o = (A, \breve{A})$ e $\G_d = (\iota^n, \breve{\iota}^n)$, o Algoritmo~\ref{cap5:algorithm_reversals_breakpoints} transforma $\G_o$ em $\G_d$ usando no máximo $2 (\bi_{\M}(\Ig) + |\Sigma_{\iota^n} \setminus \Sigma_{A}|)$ rearranjos, considerando $\M \in \{\Mindel_{\rho}, \Mindel_{\rho,\tau}$\}.
\end{lemma}

\begin{proof}
Essa demonstração é similar às provas dos lemas~\ref{cap4:lemma:reversal_breakpoint_sorting_decreasing} e \ref{cap4:lemma:reversal_breakpoint_sorting}, que foram apresentadas no Capítulo~\ref{cap:label:indel}.
\end{proof}

\begin{theorem}
O Algoritmo~\ref{cap5:algorithm_reversals_breakpoints} é uma $4$-aproximação para o problema da Distância de Reversões e Indels Intergênicos. Além disso, esse algoritmo é uma $6$-aproximação para o problema da Distância de Reversões, Transposições e Indels Intergênicos.
\end{theorem}

\begin{proof}
Diretamente dos lemas~\ref{cap5:lemma:breakpoints_intergenicos_lower_bound} e \ref{cap5:lemma:upper_bound:rev_breakpoints}.
\end{proof}

\subsection{Algoritmo de 4.5-Aproximação para Transposições e Indels}

Nesta seção, apresentamos um algoritmo com fator de aproximação igual a $4.5$ para o problema da Distância de Transposições e Indels Intergênicos. Assim como o algoritmo apresentado anteriormente, esse algoritmo usa os lemas~\ref{cap5:lemma:overchaged_undercharged} e \ref{cap5:lemma:sets_indels} para tornar $\Sigma_A = \Sigma_{\iota^n}$ e remover {\it breakpoints} intergênicos sobrecarregados e subcarregados. Para remover os outros tipos de {\it breakpoints} intergênicos usando transposições e {\it indels}, esse algoritmo usa o resultado do Lema~\ref{cap5:lemma:breakpoints_transpositions}, que é apresentado a seguir. {\revisaof Para simplificar a notação, definimos $\bi_{\tau}(\Ig) = \bi_{\Mindel_{\tau}}(\Ig)$.}

\begin{lemma}\label{cap5:lemma:breakpoints_transpositions}
Dada uma instância intergênica $\Ig = (\G_o, \G_d)$, com $\G_o = (A, \breve{A})$ e $\G_d = (\iota^n, \breve{\iota}^n)$, se $\Sigma_A = \Sigma_{\iota^n}$ e não existem {\it breakpoints} intergênicos sobrecarregados ou subcarregados, então existe uma transposição $\tau$ que remove um {\it breakpoint} intergênico ou existe uma sequência $S$, com $|S| \leq 3$, que remove dois {\it breakpoints} intergênicos.
\end{lemma}

\begin{proof}
Note que se $A \neq \iota^n$, então $\bi_{\tau}(\Ig) \geq 2$, já que para cada {\it breakpoint} $(A_i, A_{i+1})$ existe outro {\it breakpoint} com o elemento $A_{i} + 1$. 

Sejam $(A_{i-1}, A_{i})$ e $(A_{j}, A_{j+1})$ dois {\it breakpoints} intergênicos tal que $i < j$ e $j$ é mínimo. Lembre-se que para o problema da Distância de Transposições e Indels Intergênicos consideramos a definição de {\it breakpoints} de transposição (Definição~\ref{cap2:def:breakpoint_indel_t}). Portanto, existem apenas {\it strips} crescentes em $A$. 

Considere a {\it strip} $(A_i, \ldots, A_j)$ e sejam $A_x = A_{j} + 1$ e $A_y = A_{j+1} - 1$. Pela nossa escolha de $A_j$, temos que $x > j$. 

Se $\breve{A}_{j+1} + \breve{A}_x \geq \breve{\iota}^n_{A_x}$, então a transposição que troca a posição relativa dos segmentos adjacentes $(A_i~\ldots~A_j)$ e $(A_{j+1}~\ldots~A_{x-1})$ torna os elementos $A_j$ e $A_x = A_{j} + 1$ adjacentes. Como $\breve{A}_{j+1} + \breve{A}_x \geq \breve{\iota}^n_{A_x}$, podemos mover o número de nucleotídeos necessários de $\breve{A}_{j+1}$ para tornar $\breve{A}'_x = \breve{\iota}^n_{A_x}$ (ou mover o número de nucleotídeos excedentes de $\breve{A}_x$), onde $(A', \breve{A}')$ é o genoma resultante após aplicar a transposição. Dessa forma, essa transposição remove um {\it breakpoint} intergênico. 

Se $\breve{A}_{j+1} + \breve{A}_x < \breve{\iota}^n_{A_x}$, então a seguinte sequência remove dois {\it breakpoints} intergênicos: 
\begin{enumerate}
	\item Uma inserção de $\breve{\iota}^n_{A_x} - (\breve{A}_{j+1} + \breve{A}_x) + max(0, \breve{\iota}^n_{A_{j+1}} - \breve{A}_{y+1})$  nucleotídeos em $\breve{A}_{j+1}$; 
	\item Uma transposição que troca a posição relativa dos segmentos adjacentes $(A_i~\ldots~A_j)$ e $(A_{j+1}~\ldots~A_{x-1})$ torna os elementos $A_j$ e $A_x = A_{j} + 1$ adjacentes. Como descrito no caso anterior, o número de nucleotídeos necessários (ou excedentes) podem ser movidos de $\breve{A}_{j+1}$ (ou movidos de $\breve{A}_x$), removendo um {\it breakpoint} intergênico. Além disso, essa transposição torna $A_{i-1}$ e $A_{j+1}$ adjacentes;
	\item Uma transposição que torna $A_y = A_{j+1} - 1$ e $A_{j+1}$ adjacentes e move nucleotídeos necessários ou excedentes para remover o {\it breakpoint} intergênico entre esses elementos. 
\end{enumerate}

Note que, após as duas primeiras operações dessa sequência serem aplicadas, o tamanho da região intergênica entre $(A_{i-1}, A_{j+1})$ somado ao tamanho da região intergênica entre $(A_y, A_{y+1})$ é maior ou igual ao valor de $\breve{\iota}^n_{A_{j+1}}$.
\end{proof}

\begin{algorithm}[h]
	\caption{\label{cap5:algorithm_transpositions_breakpoints}
		Algoritmo para o problema da Distância de Transposições e Indels Intergênicos
	}
	\DontPrintSemicolon
  \Entrada{Uma instância intergênica $\Ig = (\G_o, \G_d)$, com $\G_o = (A, \breve{A})$ e $\G_d = (\iota^n, \breve{\iota}^n)$}
  \Saida{Uma sequência de rearranjos $S$ que transforma $\G_o$ em $\G_d$}
  	Seja $S \gets \emptyset$\;
	\Enqto{$\bi_{\tau}(\Ig) + |\Sigma_{\iota^n} \setminus \Sigma_{A}| > 0$}{
		\Se{$|\Sigma_A \setminus \Sigma_{\iota^n}| > 0$ ou $|\Sigma_{\iota^n} \setminus \Sigma_{A}| > 0$}{
			Seja $S' = (\beta)$, onde $\beta$ é um {\it indel} de acordo com o Lema~\ref{cap5:lemma:sets_indels}\;
		}
		\SenaoSe{existe {\it breakpoint} intergênico sobrecarregado ou subcarregado}{
			Seja $S' = (\beta)$, onde $\beta$ é um {\it indel} de acordo com o Lema~\ref{cap5:lemma:overchaged_undercharged}\;
		}\Senao(){
			Seja $S'$ uma sequência de operações de acordo com o Lema~\ref{cap5:lemma:breakpoints_transpositions}\;
		}
		$\G_o \gets \G_o \comp S'$\;
		Adicione as operações de $S'$ na sequência $S$\;
	}
  {\bf retorne} a sequência $S$\;
\end{algorithm}

O Algoritmo~\ref{cap5:algorithm_transpositions_breakpoints} usa os lemas~\ref{cap5:lemma:overchaged_undercharged}, \ref{cap5:lemma:sets_indels} e \ref{cap5:lemma:breakpoints_transpositions} para construir uma sequência de operações que transforma $\G_o$ em $\G_d$, considerando o modelo $\Mindel_{\tau}$. O Lema~\ref{cap5:lemma:upper_bound:transp_breakpoints} apresenta um limitante superior no número de rearranjos usados pelo algoritmo.

\begin{lemma}\label{cap5:lemma:upper_bound:transp_breakpoints}
Para qualquer instância intergênica $\Ig = (\G_o, \G_d)$, com $\G_o = (A, \breve{A})$ e $\G_d = (\iota^n, \breve{\iota}^n)$, o Algoritmo~\ref{cap5:algorithm_transpositions_breakpoints} transforma $\G_o$ em $\G_d$ usando no máximo $\frac{3}{2} (\bi_{\tau}(\Ig) + |\Sigma_{\iota^n} \setminus \Sigma_{A}|)$ rearranjos.
\end{lemma}

\begin{proof}
A cada iteração, o algoritmo usa uma operação $\beta$ com $\Delta \Phi(\Ig, \beta) + \Delta \bi_{\tau}(\Ig, \beta) = 1$ ou uma sequência $S'$ com três operações e $\Delta \Phi(\Ig, S') + \Delta \bi_{\tau}(\Ig, S') = 2$. Dessa forma, no pior caso, o algoritmo usa em média $\frac{3}{2}$ operações para diminuir o valor de $\bi_{\tau}(\Ig) + |\Sigma_{\iota^n} \setminus \Sigma_{A}|$ em uma unidade. Consequentemente, a sequência encontrada pelo algoritmo transforma $\G_o$ em $\G_d$ usando no máximo $\frac{3}{2} (\bi_{\tau}(\Ig) + |\Sigma_{\iota^n} \setminus \Sigma_{A}|)$ operações.
\end{proof}

\begin{theorem}
O Algoritmo~\ref{cap5:algorithm_transpositions_breakpoints} é uma $4.5$-aproximação para o problema da Distância de Transposições e Indels Intergênicos.
\end{theorem}

\begin{proof}
Diretamente dos lemas~\ref{cap5:lemma:breakpoints_intergenicos_lower_bound} e \ref{cap5:lemma:upper_bound:transp_breakpoints}.
\end{proof}

Note que para o problema da Distância de Reversões, Transposições e Indels Intergênicos, podemos melhorar o Algoritmo~\ref{cap5:algorithm_reversals_breakpoints} ao usar as operações descritas no Lema~\ref{cap5:lemma:breakpoints_transpositions} quando existem apenas {\it strips} crescentes na string do genoma de origem. No entanto, no pior caso, o fator de aproximação continua o mesmo para essa nova versão do algoritmo. 

Os algoritmos~\ref{cap5:algorithm_reversals_breakpoints} e \ref{cap5:algorithm_transpositions_breakpoints} possuem complexidade de tempo de $O(n^2)$. Note que o algoritmo executa por no máximo $O(n)$ iterações, já que $\bi_{\M}(\Ig) + |\Sigma_{\iota^n} \setminus \Sigma_{A}|$ é $O(n)$, e cada operação pode ser encontrada e aplicada em tempo linear.

\section{Algoritmos de Aproximação usando Grafo de Ciclos}\label{cap5:section:grafo_ciclos}

Nesta seção, apresentamos os seguintes algoritmos de aproximação:

\begin{itemize}
	\item Um algoritmo de $2.5$-aproximação para o problema da Distância de Reversões e Indels Intergênicos em Strings com Sinais (modelo $\Mindel_{\rho}$);
	\item Um algoritmo de $4$-aproximação para o problema da Distância de Transposições e Indels Intergênicos em Strings sem Sinais (modelo $\Mindel_{\tau}$);
	\item Um algoritmo de $4$-aproximação para o problema da Distância de Reversões, Transposições e Indels Intergênicos em Strings com Sinais (modelo $\Mindel_{\rho,\tau}$).
\end{itemize}

Esses algoritmos e os limitantes inferiores apresentados nesta seção utilizam os conceitos relacionados ao grafo de ciclos rotulado e ponderado (Seção~\ref{cap2:sec:grafo_ciclos_intergenico}). Sempre consideramos que as strings ($A$ e $\iota^n$) de uma instância $\Ig = (\G_o, \G_d)$ e suas formas simplificadas ($\pi^A$ e $\pi^\iota$) estão nas suas versões estendidas.

Agora, apresentamos limitantes para o valor de $\Delta \cgood(\Ig, \beta)$ dependendo do tipo do rearranjo $\beta$.

\begin{lemma}\label{cap5:lemma:var_good_cycles_indels}
Para qualquer {\it indel} $\beta$ e instância intergênica $\Ig = (\G_o, \G_d)$, com $\G_o = (A, \breve{A})$ e $\G_d = (\iota^n, \breve{\iota}^n)$, temos que $\Delta \cgood(\Ig, \beta) \leq 1$.
\end{lemma}

\begin{proof}
Uma deleção afeta uma única aresta de origem do grafo, sendo que apenas o custo e o rótulo dessa aresta podem ser alterados. Portanto, apenas um ciclo é afetado por uma deleção. Considere uma deleção $\psi$ e seja $C$ o ciclo afetado por $\psi$. Quando $C$ é rotulado ou desbalanceado, no melhor cenário, a deleção $\psi$ transforma $C$ em um ciclo bom. Quando $C$ é um ciclo bom, a deleção não aumenta a quantidade de ciclos bons no grafo. Portanto, $\Delta \cgood(\Ig, \psi) \leq 1$.

Considere uma inserção $\phi$ que insere $k$ elementos em $A$. Essa inserção adiciona $2k$ vértices no grafo e substitui uma aresta de origem de um ciclo $C$ por $k+1$ arestas de origem, também adicionando $k$ arestas de destino no grafo. Como o novo grafo possui $k+1$ arestas de origem distintas em relação ao grafo $G(\Ig)$ e pelo menos uma das novas arestas pertence a $C$, no melhor cenário, $k$ ciclos são criados no grafo. Portanto, 

\begin{align*}
&(|\pi^A| + 1 - c(\G_o, \G_d)) - (|\pi^{A \comp \phi}| + 1 - c(\G_o \comp \phi, \G_d)) \\
=& (|\pi^A| - |\pi^{A \comp \phi}|) - (c(\G_o, \G_d) - c(\G_o \comp \phi, \G_d)) \\
=& -k - (-k) = 0.
\end{align*}

No melhor cenário, o ciclo $C$ se torna um ciclo bom e todos os $k$ ciclos adicionados também são ciclos bons, o que resulta em $\Delta \cgood(\Ig, \phi) = 1$. 
Note que se qualquer outro ciclo $C'$ de $G(\Ig)$ se torna balanceado e limpo após aplicar $\phi$, então pelo menos uma das novas arestas de origem foi adicionada em $C'$ e a inserção não pode ter adicionado $k$ novos ciclos no grafo. De forma similar, se $x$ ciclos distintos de $C$ se tornam ciclos bons, então no máximo $k-x$ novos ciclos foram adicionados no grafo por $\phi$, o que também resulta em $\Delta \cgood(\Ig, \phi) \leq 1$.
\end{proof}

\begin{lemma}\label{cap5:lemma:var_good_cycles_rev}
Para qualquer reversão $\rho$ e instância intergênica $\Ig = (\G_o, \G_d)$, com $\G_o = (A, \breve{A})$ e $\G_d = (\iota^n, \breve{\iota}^n)$, temos que $\Delta \cgood(\Ig, \rho) \leq 1$.
\end{lemma}

\begin{proof}
Bafna e Pevzner~\cite{1996-bafna-pevzner} demonstraram que $\Delta c(\Ig, \rho)$ $\in \{-1,0,1\}$, para qualquer reversão $\rho$. Além disso, como uma reversão não adiciona elementos, sabemos que a quantidade de arestas no grafo permanece a mesma. 

Se $\Delta c(\Ig, \rho) = -1$, então $\rho$ afeta dois ciclos $C$ e $D$, juntando esses dois ciclos em um único ciclo $C'$. No melhor cenário, tanto $C$ quanto $D$ são desbalanceados e limpos, mas $C'$ é um ciclo balanceado e limpo, fazendo com que $\Delta \cgood(\Ig, \rho) = 1$. Se $C$ ou $D$ é rotulado, então $C'$ também é rotulado e nenhum ciclo bom é adicionado ao grafo. 

Se $\Delta c(\Ig, \rho) = 0$, então $\rho$ afeta um único ciclo $C$, transformando-o em um ciclo $C'$ com os mesmos vértices e arestas de destino que o ciclo $C$. Note que se $C$ é desbalanceado ou rotulado, então $C'$ também é desbalanceado ou rotulado, já que nenhum rótulo é removido e a soma dos custos das arestas de origem permanece a mesma. Portanto, $\Delta \cgood(\Ig, \rho) = 0$.

Se $\Delta c(\Ig, \rho) = 1$, então $\rho$ afeta um único ciclo $C$, transformando-o em dois ciclos $C'$ e $D'$. {\revisaof Se $C$ é um ciclo bom, então temos $\Delta \cgood(\Ig, \rho) = 1$, mesmo que os ciclos $C'$ e $D'$ sejam bons.} Se $C$ é um ciclo desbalanceado ou rotulado, então pelo menos um dos ciclos $C'$ e $D'$ também deve ser desbalanceado ou rotulado. Portanto, no melhor cenário, $\Delta \cgood(\Ig, \rho) = 1$.
\end{proof}

\begin{lemma}\label{cap5:lemma:var_good_cycles_transp}
Para qualquer transposição $\tau$ e instância intergênica $\Ig = (\G_o, \G_d)$, com $\G_o = (A, \breve{A})$ e $\G_d = (\iota^n, \breve{\iota}^n)$, temos que $\Delta \cgood(\Ig, \tau) \leq 2$.
\end{lemma}

\begin{proof}

Uma transposição não remove rótulos de arestas de destino. Além disso, uma transposição só pode remover um rótulo de uma arestas de origem ao transferir um elemento $\deletionElement$ para outra aresta de origem. Dividimos essa prova de acordo com a quantidade de ciclos afetados por uma transposição $\tau$~\cite{1998-bafna-pevzner}.

Se $\tau$ afeta três ciclos $C_1$, $C_2$, e $C_3$, então os vértices desses ciclos são unidos em um único ciclo $C'$. No melhor cenário, $C_1$, $C_2$, e $C_3$ são ciclos desbalanceados e limpos, e $C'$ é um ciclo bom. Portanto, $\Delta \cgood(\Ig, \tau) = 1$.

Se $\tau$ afeta dois ciclos $C_1$ e $C_2$, então os vértices desses ciclos são rearranjados em dois ciclos $C'_1$ e $C'_2$. No melhor cenário, $C_1$ e $C_2$ são ciclos desbalanceados e limpos, e $C'_1$ e $C'_2$ são ciclos bons. Portanto, $\Delta \cgood(\Ig, \tau) = 2$.

Se $\tau$ afeta um único ciclo $C$, então os vértices de $C$ são rearranjados em um ciclo $C'$ ($\Delta c(\Ig, \tau) = 0$) ou $C$ é transformado em três novos ciclos ($\Delta c(\Ig, \tau) = 2$). Se $\Delta c(\Ig, \tau) = 0$, então o número de ciclos bons não é alterado. Se $\Delta c(\Ig, \tau) = 2$, então dividimos a prova em dois casos. Se $C$ é um ciclo bom, então o melhor cenário ocorre quando todos os três novos ciclos também são bons e, consequentemente, $\Delta \cgood(\Ig, \tau) = 2$. Se $C$ é desbalanceado ou rotulado, então pelo menos um dos novos ciclos também é um ciclo desbalanceado ou rotulado. Consequentemente, no máximo dois desses novos ciclos são bons e, consequentemente,  temos $\Delta \cgood(\Ig, \tau) = 2$, no melhor cenário.
\end{proof}

\begin{lemma}\label{cap5:lemma:grafo_ciclos_intergenicos_lower_bound}
	Para qualquer instância intergênica $\Ig = (\G_o, \G_d)$, com $\G_o = (A, \breve{A})$ e $\G_d = (\iota^n, \breve{\iota}^n)$, temos que:
	\begin{align*}
		d_{\Mindel_{\rho}}(\Ig) \geq |\pi^A| + 1 - \cgood(\Ig), \\
		d_{\Mindel_{\tau}}(\Ig) \geq \frac{|\pi^A| + 1 - \cgood(\Ig)}{2}, \\
		d_{\Mindel_{\rho,\tau}}(\Ig) \geq \frac{|\pi^A| + 1 - \cgood(\Ig)}{2}. \\
	\end{align*}
\end{lemma}

\begin{proof}
Considere o modelo $\Mindel_{\rho, \tau}$. Como $|\pi^A| + 1 - \cgood(\Ig) = 0$ se, e somente se, $\G_o = \G_d$, qualquer sequência de rearranjos $S$ que transforma $\G_o$ em $\G_d$ deve tornar o valor de $|\pi^A| + 1 - \cgood(\Ig)$ igual a zero. Pelos lemas~\ref{cap5:lemma:var_good_cycles_indels}, \ref{cap5:lemma:var_good_cycles_rev} e \ref{cap5:lemma:var_good_cycles_transp}, qualquer rearranjo $\beta$ em $\Mindel_{\rho, \tau}$ satisfaz $\Delta \cgood(\Ig, \beta) \leq 2$ e, portanto, $|S| \geq \frac{|\pi^A| + 1 - \cgood(\Ig)}{2}$.

A prova é similar para os outros modelos. Note que qualquer rearranjo de $\Mindel_{\rho}$ satisfaz $\Delta \cgood(\Ig, \beta) \leq 1$ (lemas~\ref{cap5:lemma:var_good_cycles_indels} e \ref{cap5:lemma:var_good_cycles_rev}).
\end{proof}

\subsection{Uma 2.5-Aproximação para Reversões e Indels}\label{cap5:subsection:2_5_reversal}

A ideia principal da $2.5$-aproximação é criar novos ciclos bons a cada iteração, até que o grafo contenha apenas ciclos unitários bons. Os próximos lemas mostram como aumentar o número de ciclos bons no grafo usando reversões e/ou indels.

\begin{lemma}\label{cap5:lemma:trivial_cycles_one_indel}
Para qualquer grafo $G(\Ig)$ que contém ciclo unitário $C=(o_1, d_1)$, se (i) $C$ é um ciclo limpo ou (ii) $C$ é um ciclo não negativo e $\ell(e_{o_1}) = \emptyset$, então existe um {\it indel} $\beta$ com $\Delta \cgood(\Ig, \beta) = 1$.
\end{lemma}

\begin{proof}
	Como $\ell(e_{o_1}) = \emptyset$, as arestas $e_{o_1}$ e $e'_{d_1}$ são incidentes a vértices que correspondem a elementos adjacentes em $A$. Sejam $+A_{i-1}$ e $-A_{i}$ os vértices do ciclo unitário $C$. Dividimos nossa análise em dois casos.

	Considere que $C$ é um ciclo limpo. Se $w(e'_{d_1}) > w(e_{o_1})$, então uma inserção de $w(e'_{d_1}) - w(e_{o_1})$ nucleotídeos na região intergênica $\breve{A}_i$ torna o ciclo $C$ em um ciclo balanceado e limpo. Se $w(e'_{d_1}) < w(e_{o_1})$, então uma deleção de $w(e_{o_1}) - w(e'_{d_1})$ nucleotídeos na região intergênica $\breve{A}_i$ torna o ciclo $C$ em um ciclo balanceado e limpo. Como nenhum elemento é adicionado na string e $C$ é transformado em um ciclo bom, então existe {\it indel} $\beta$ com $\Delta \cgood(\Ig, \beta) = 1$.

	Agora, considere que $C$ é um ciclo não negativo e $\ell(e_{o_1}) = \emptyset$, mas $\ell(e_{d_1}) \neq \emptyset$. Note que se $\ell(e_{d_1}) = \emptyset$, então $C$ é um ciclo limpo e podemos usar um {\it indel} como descrito no caso anterior. Pela nossa representação de genomas e pela definição do grafo $G(\Ig)$, o rótulo de $e'_{d_1}$ é igual a $z = |A_{i-1} + 1|$. Seja $x$ a região intergênica entre $A_{i-1}$ e $z$ em $\G_d$, e seja $y$ a região intergênica entre $A_{i}$ e $z$ em $\G_d$. Note que $A_{i-1}$ e $A_{i}$ possuem o mesmo sinal, já que eles formam um ciclo unitário. 

	A inserção $\phi_{\min(x, w(e_{o_1}))}^{(i-1, \sigma, \breve{\sigma})}$, tal que $\sigma = (A_{i-1} + 1)$, $\breve{\sigma} = (x', y')$, $x' = x - \min(x, w(e_{o_1}))$ e $y' = y - (w(e_{o_1}) - \min(x, w(e_{o_1})))$, transforma $C$ em dois ciclos unitários $C'$ e $C''$ que são balanceados e limpos, como mostrado na Figura~\ref{cap5:fig:insertion_element_trivial_cycle}. Já que um elemento é adicionado na string e dois ciclos bons são criados, essa inserção $\phi$ satisfaz $\Delta \cgood(\Ig, \phi) = 1$.
\end{proof}

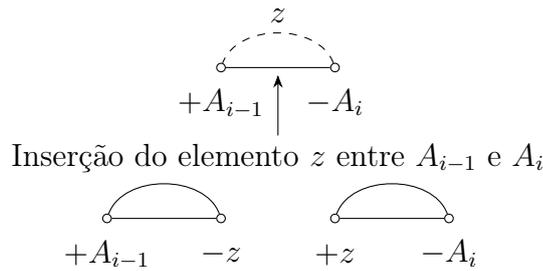
\begin{figure}[tb]
\centering
\begin{tikzpicture}[scale=1]
\begin{scope}[every node/.style={inner sep=0.4mm, draw, circle, minimum size = 0pt}]
  \node[label={[yshift=-1.2cm]$+{A}_{i-1}$}] (v0) at (0,0) {};
  \node[label={[yshift=-1.0cm]$-{A}_{i}$}] (v2) at (1.5,0) {};
\end{scope}

\begin{scope}[>={Stealth[black]},
              every edge/.style={draw=black}, every node/.style={inner sep=0pt, minimum size = 0pt}]
  \path[->] (0.75, -0.9) edge (0.75, -0.1);
  \node[label=\phantom{}] (ins) at (0.75, -1.2) {{Inserção do elemento $z$ entre ${A}_{i-1}$ e ${A}_{i}$}};
\end{scope}

\begin{scope}[>={Stealth[black]},
              every edge/.style={draw=black}]
    \path [-] (v0) edge (v2);
\end{scope}

\begin{scope}[>={Stealth[black]},
              every edge/.style={draw=black}]
    \path [-, dashed] (v2) edge [bend right=70] node [black, pos=0.5, sloped, above] {$z$} (v0);
\end{scope}

\begin{scope}[every node/.style={inner sep=0.4mm, draw, circle, minimum size = 0pt}]
  \node[label={[yshift=-1.2cm]$+{A}_{i-1}$}] (v0) at (-1.5,-2) {};
  \node[label={[yshift=-0.9cm]$-z$}] (mx) at (-0,-2) {};
  \node[label={[yshift=-0.9cm]$+z$}] (px) at (1.5,-2) {};
  \node[label={[yshift=-1.0cm]$-{A}_{i}$}] (v2) at (3,-2) {};
\end{scope}

\begin{scope}[>={Stealth[black]},
              every edge/.style={draw=black}]
    \path [-] (v0) edge (mx);
    \path [-] (px) edge (v2);
\end{scope}

\begin{scope}[>={Stealth[black]},
              every edge/.style={draw=black}]
    \path [-] (mx) edge [bend right=70] (v0);
    \path [-] (v2) edge [bend right=70] (px);
\end{scope}

\end{tikzpicture}
\caption{\label{cap5:fig:insertion_element_trivial_cycle}
Exemplo de inserção que transforma um ciclo unitário não negativo, que possui aresta de origem limpa e aresta de destino rotulada, em dois ciclos bons. Assumimos que ${A}_{i-1}$ tem sinal ``$+$''.
}
\end{figure}

\begin{lemma}\label{cap5:lemma:trivial_bad_cycles_two_indels}
	Para qualquer grafo $G(\Ig)$ que contém um ciclo unitário $C=(o_1, d_1)$, se $C$ é desbalanceado ou rotulado, então existe uma sequência $S$, tal que $|S| \leq 2$ e $\Delta \cgood(\Ig, S) = 1$.
\end{lemma}

\begin{proof}
	Se a aresta de origem $e_{o_1}$ é rotulada ou se $C$ é um ciclo negativo, então uma deleção $\psi$ pode remover o rótulo dessa aresta e tornar $C$ em um ciclo não negativo. Essa deleção torna $C$ em um ciclo unitário bom ou em um ciclo unitário não negativo com $\ell(e_{o_1}) = \emptyset$. No primeiro caso, temos que $S = (\psi)$ e $\Delta \cgood(\Ig, S) = 1$. No segundo caso, o novo ciclo unitário satisfaz as condições do Lema~\ref{cap5:lemma:trivial_cycles_one_indel} e, portanto, existe {\it indel} $\beta$ com $\Delta \cgood(\Ig, \beta) = 1$. Nesse caso, temos que $S = (\psi, \beta)$ e $\Delta \cgood(\Ig, S) = 1$.

	Caso contrário, o ciclo unitário $C$ é não negativo e $\ell(e_{o_1}) = \emptyset$. Pelo Lema~\ref{cap5:lemma:trivial_cycles_one_indel}, existe {\it indel} $\beta$ com $\Delta \cgood(\Ig, \beta) = 1$ e, portanto, temos que $S = (\beta)$ e $\Delta \cgood(\Ig, S) = 1$.
\end{proof}

\begin{lemma}\label{cap5:lemma:divergent_labeled_cycle}
	Para qualquer grafo $G(\Ig)$, se existe um ciclo rotulado $C$ divergente, então existe uma sequência $S$, tal que $|S| \leq 2$ e $\Delta \cgood(\Ig, S) = 1$.
\end{lemma}

\begin{proof}
	Seja $C = (o_1, d_1, \ldots, o_k, d_k)$ um ciclo divergente rotulado. Seja $(e_{o_x}, e_{o_{x+1}})$ um par divergente de $C$ tal que $x$ é mínimo. Uma reversão $\rho$ aplicada nessas duas arestas de origem $(e_{o_x}, e_{o_{x+1}})$ transforma $C$ em um ciclo unitário $C'$ e um $(k-1)$-ciclo $C''$, como mostrado na Figura~\ref{cap5:fig:trivial_cycle_from_divergent_cycle}. Assuma, sem perda de generalidade, que o ciclo unitário $C'$ possui a aresta de destino com índice $o_{x+1}$ no novo grafo. Se a aresta de origem $e_{o_{x+1}}$ é rotulada, então a reversão $\rho$ move qualquer elemento $\deletionElement$ entre os vértices que são extremidades de $e_{o_{x+1}}$, de forma que os elementos $\deletionElement$ são acumulados na aresta com índice $o_x$ no novo grafo. Além disso, se $w(e'_{d_x}) < w(e_{o_{x+1}})$, a reversão move o custo excedente de $e_{o_{x+1}}$ para $e_{o_{x}}$, fazendo com que o ciclo unitário $C'$ seja um ciclo balanceado. Caso contrário, a reversão não move nucleotídeos das arestas $e_{o_x}$ e $e_{o_{x+1}}$, fazendo com que o ciclo unitário $C'$ seja não negativo. Se $C'$ é balanceado e limpo, então temos $S = (\rho)$ com $\Delta \cgood(\Ig, S) = 1$. Caso contrário, $C'$ satisfaz as condições do Lema~\ref{cap5:lemma:trivial_cycles_one_indel} e, portanto, existe {\it indel} $\beta$ com $\Delta \cgood(\Ig, \beta) = 1$. Nesse caso, temos que $S = (\rho, \beta)$ e $\Delta \cgood(\Ig, S) = 1$.
\end{proof}

\begin{figure}[tb]
\centering
\begin{tikzpicture}[scale=0.85]
\tiny
\begin{scope}[every node/.style={inner sep=0.4mm, draw, circle, minimum size = 0pt}]
  \node[label=below:$+\pi^{A}_{i'}$] (v0) at (0,0) {};
  \node[label=below:$-\pi^A_{i}$] (v1) at (1.5,0) {};
  \node[label=below:$+\pi^A_{j}$] (v2) at (4.5,0) {};
  \node[label=below:$-\pi^A_{j'}$] (v3) at (6,0) {};
\end{scope}

\begin{scope}
  \node[] (aux) at (3,0) {\small ...};
\end{scope}

\begin{scope}[>={Stealth[black]},
              every edge/.style={draw=black}]
    \path [-, dashed] (v0) edge node [black, pos=0.5, sloped, above, yshift=+0.00cm] {$\ell = \deletionElement$} (v1);
    \path [-] (v2) edge (v3);
\end{scope}

\begin{scope}[>={Stealth[black]},
              every edge/.style={draw=black}]
    \path [-, dashed] (v2) edge [bend right=40] node [black, pos=0.5, sloped, above] {$z$} (v0);
    \path [-] (v1) edge [bend right=20] (-1,1);
    \path [-] (v3) edge [bend right=20] (3.5,1);
\end{scope}

\begin{scope}[>={Stealth[black]},
              every edge/.style={draw=black}, every node/.style={inner sep=0pt, minimum size = 0pt}]
  \node[label=\phantom{}] (bi1) at (0.75, -0.8) {};
  \node[label=\phantom{}] (bi2) at (5.25, -0.8) {};
  \path [{Bar}-{Bar}] (bi1) edge node [black, pos=0.5, sloped, below] {} (bi2);
\end{scope}

\end{tikzpicture}

\begin{tikzpicture}[scale=0.85]
\tiny
\begin{scope}[every node/.style={inner sep=0.4mm, draw, circle, minimum size = 0pt}]
  \node[label=below:$+\pi^{A}_{i'}$] (v0) at (0,0) {};
  \node[label=below:$-\pi^A_{i}$] (v1) at (4.5,0) {};
  \node[label=below:$+\pi^A_{j}$] (v2) at (1.5,0) {};
  \node[label=below:$-\pi^A_{j'}$] (v3) at (6,0) {};
\end{scope}

\begin{scope}
  \node[] (aux) at (3,0) {\small ...};
\end{scope}

\begin{scope}[>={Stealth[black]},
              every edge/.style={draw=black}]
    \path [-] (v0) edge (v2);
    \path [-, dashed] (v1) edge node [black, pos=0.5, sloped, above, yshift=+0.05cm] {$\ell = \deletionElement$} (v3);
\end{scope}

\begin{scope}[>={Stealth[black]},
              every edge/.style={draw=black}]
    \path [-, dashed] (v2) edge [bend right=50] node [black, pos=0.5, sloped, above] {$z$} (v0);
    \path [-] (v1) edge [bend right=20] (-1,1);
    \path [-] (v3) edge [bend right=20] (3.5,1);
\end{scope}

\end{tikzpicture}

\caption{\label{cap5:fig:trivial_cycle_from_divergent_cycle}
Reversão aplicada em um par divergente de um ciclo $C$ que transforma esse ciclo em um ciclo unitário $C'$ e um $(k-1)$-ciclo $C''$.
}

\end{figure}
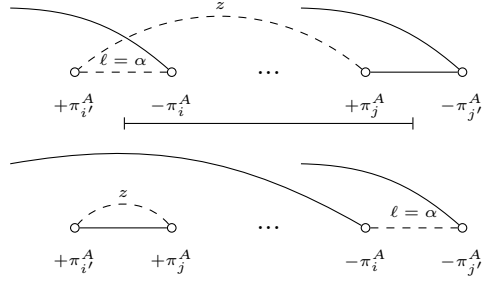

\begin{lemma}\label{cap5:lemma:divergent_cycle}
	Para qualquer grafo $G(\Ig)$, se existe um ciclo limpo $C$ divergente, então existe rearranjo $\beta$ com $\Delta \cgood(\Ig, \beta) = 1$.
\end{lemma}

\begin{proof}
	Se $C = (o_1, d_1, \ldots, o_k, d_k)$ é um ciclo positivo, então uma inserção $\phi$ que aumenta o custo de qualquer aresta de origem de $C$ em $\sum_{i=1}^{k}{w(e'_{d_i})} - \sum_{i=1}^{k}{w(e_{o_i})}$ unidades (i.e., a inserção de $\sum_{i=1}^{k}{w(e'_{d_i})} - \sum_{i=1}^{k}{w(e_{o_i})}$ nucleotídeos na região intergênica que corresponde a $w(e_{o_1})$) transforma $C$ em um ciclo balanceado. Dessa forma, temos que $\Delta \cgood(\Ig, \phi) = 1$. 

	Caso contrário, o ciclo divergente $C$ é negativo ou balanceado. Oliveira e coautores~\cite{2021b-oliveira-etal} demonstraram que existe reversão $\rho$ que transforma $C$ em dois ciclos $C'$ e $C''$, aumentando o número de ciclos balanceados no grafo. Se $C$ é negativo, então ou $C'$ ou $C''$ é balanceado. Caso contrário, ambos $C'$ e $C''$ são balanceados. Como $C$ é um ciclo limpo e a reversão é aplicada em duas arestas de origem de $C$, temos que os ciclos $C'$ e $C''$ são ciclos limpos. Portanto, temos que $\Delta \cgood(\Ig, \rho) = 1$.
\end{proof}

Até agora, apresentamos lemas que tratam de ciclos unitários ou ciclos divergentes. O próximo lema mostra como encontrar uma sequência de operações que aumentam o número de ciclos bons quando existem apenas ciclos convergentes no grafo. 

\begin{lemma}\label{cap5:lemma:convergent_cycles}
	Para qualquer grafo $G(\Ig)$, se $G(\Ig)$ possui apenas ciclos convergentes, então existe uma  sequência $S$, tal que $|S| \leq 5$ e $\Delta \cgood(\Ig, S) = 2$.
\end{lemma}

\begin{proof}
Seja $C = (o_1, d_1, \ldots, o_k, d_k)$ um ciclo convergente de $G(\Ig)$. Oliveira e coautores~\cite{2021b-oliveira-etal} mostraram que um dos dois casos sempre ocorre:
\begin{itemize}
	\item Se $C$ é um ciclo orientado, então existe reversão $\rho_1$ que age em duas arestas de origem de $C$, transformando o ciclo $C$ em um ciclo divergente $C'$.
	\item Se $C$ é não orientado, então existe reversão $\rho_1$ que age em duas arestas de origem de um outro ciclo convergente $D$, transformando o ciclo $C$ em um ciclo divergente $C'$.
\end{itemize}

Além disso, Oliveira e coautores~\cite{2021b-oliveira-etal} mostraram que o número de ciclos do grafo permanece o mesmo após a aplicação da reversão $\rho_1$.

Pelos lemas~\ref{cap5:lemma:divergent_labeled_cycle} e \ref{cap5:lemma:divergent_cycle}, existe uma sequência $S_1$ que age no ciclo $C'$, tal que $|S_1| \leq 2$ e $\Delta \cgood(\G_o \comp \rho, \G_d, S_1) = 1$. Note que a sequência $S_1$ possui apenas uma reversão e, caso $|S_1| = 2$, o segundo rearranjo de $S_1$ é um indel. Seja $\rho_2$ a reversão de $S_1$. Seja $\G^{'}_o = (A', A) = (\G_o \comp \rho_1) \comp \rho_2$, ou seja, o genoma resultante após a aplicação da reversão $\rho_1$, que cria um ciclo divergente no grafo, e da primeira operação da sequência $S_1$. 

Note que o grafo $G(\Ig)$ possui apenas ciclos convergentes se, e somente se, não existe elemento $A_i$ em $A$ com sinal ``$-$''. Suponha, por contradição, que o grafo $G(\G^{'}_o, \G_d)$ possui apenas ciclos convergentes. Isso implica que após aplicar as reversões $\rho_1$ e $\rho_2$, não existe elemento $A'_i$ em $A'$ com sinal ``$-$''. Isso implica que tanto $\rho_1$ quanto $\rho_2$ agem no mesmo intervalo e, portanto, temos que $A = A'$. Isso contradiz o fato de que a reversão $\rho_1$ não altera o número de ciclos do grafo e a reversão $\rho_2$ transforma $C'$ em dois ciclos. Portanto, o grafo $G(\G^{'}_o, \G_d)$ possui um ciclo divergente. 

Como um {\it indel} não consegue transformar um ciclo divergente em convergente, o grafo $G(\G^{''}_o, \G_d)$, onde $\G^{''}_o = (\G_o \comp \rho_1) \comp S_1$, também possui um ciclo divergente $D'$. Pelos lemas ~\ref{cap5:lemma:divergent_labeled_cycle} e \ref{cap5:lemma:divergent_cycle}, existe uma  sequência $S_2$ que age no ciclo $D'$, tal que $|S_2| \leq 2$ e $\Delta \cgood(\G^{''}_o, \G_d, S_2) = 1$. Portanto, a sequência $S$, que é formada a partir de $\rho_1$, $S_1$ e $S_2$, satisfaz as condições do enunciado deste lema. As figuras~\ref{cap5:fig:reversals_operations_on_oriented_cycle1} e \ref{cap5:fig:reversals_operations_on_oriented_cycle2} apresentam exemplos das operações de $S$.
\end{proof}

\begin{figure}[tb]
\color{black}
\centering

\resizebox{0.5\textwidth}{!}{
\begin{tikzpicture}[]
\scriptsize
\begin{scope}
    \node at (0,1) {};
\end{scope}

\begin{scope}[every node/.style={inner sep=0.4mm, draw, circle, minimum size = 0pt}]
  \node[label=below:\phantom{+}] (v0) at (0,0) {};
  \node[label=below:\phantom{+}] (v1) at (1.5,0) {};
  \node[label=below:\phantom{+}] (v2) at (3,0) {};
  \node[label=below:\phantom{+}] (v3) at (4.5,0) {};
  \node[label=below:\phantom{+}] (v4) at (6,0) {};
  \node[label=below:\phantom{+}] (v5) at (7.5,0) {};
\end{scope}

\begin{scope}
  \node[] (aux) at (2.25,0) {\small ...};
  \node[] (aux) at (5.25,0) {\small ...};
\end{scope}

\begin{scope}[>={Stealth[black]},
              every edge/.style={draw=black}]
    \path [-, dashed] (v0) edge node [black, pos=0.5, sloped, above, yshift=+0.05cm] {$\ell = \deletionElement$} (v1);
    \path [-, dashed] (v2) edge node [black, pos=0.5, sloped, above, yshift=+0.05cm] {$\ell = \deletionElement$} (v3);
    \path [-, dashed] (v4) edge node [black, pos=0.5, sloped, above, yshift=+0.05cm] {$\ell = \deletionElement$} (v5);
\end{scope}

\begin{scope}[>={Stealth[black]},
              every edge/.style={draw=black}]
    \path [-] (v1) edge [bend left=50] (v4);
    \path [-] (v0) edge [bend left=50] (v3);
    \path [-] (v2) edge [bend left=50] (v5);
\end{scope}

\begin{scope}[>={Stealth[black]},
              every edge/.style={draw=black}, every node/.style={inner sep=0pt, minimum size = 0pt}]
  \node[label=\phantom{}] (bi1) at (3.75, -0.5) {};
  \node[label=\phantom{}] (bi2) at (6.75, -0.5) {};
  \path [{Bar}-{Bar}] (bi1) edge node [black, pos=0.5, sloped, below] {} (bi2);
\end{scope}

\end{tikzpicture}
}

\resizebox{0.5\textwidth}{!}{
\begin{tikzpicture}[]
\scriptsize
\begin{scope}
    \node at (0,1) {};
\end{scope}

\begin{scope}[every node/.style={inner sep=0.4mm, draw, circle, minimum size = 0pt}]
  \node[label=below:\phantom{+}] (v0) at (0,0) {};
  \node[label=below:\phantom{+}] (v1) at (1.5,0) {};
  \node[label=below:\phantom{+}] (v2) at (3,0) {};
  \node[label=below:\phantom{+}] (v3) at (6,0) {};
  \node[label=below:\phantom{+}] (v4) at (4.5,0) {};
  \node[label=below:\phantom{+}] (v5) at (7.5,0) {};
\end{scope}

\begin{scope}
  \node[] (aux) at (2.25,0) {\small ...};
  \node[] (aux) at (5.25,0) {\small ...};
\end{scope}

\begin{scope}[>={Stealth[black]},
              every edge/.style={draw=black}]
    \path [-, dashed] (v0) edge node [black, pos=0.5, sloped, above, yshift=+0.05cm] {$\ell = \deletionElement$} (v1);
    \path [-, dashed] (v2) edge node [black, pos=0.5, sloped, above, yshift=+0.05cm] {$\ell = \deletionElement$} (v4);
    \path [-, dashed] (v3) edge node [black, pos=0.5, sloped, above, yshift=+0.05cm] {$\ell = \deletionElement$} (v5);
\end{scope}

\begin{scope}[>={Stealth[black]},
              every edge/.style={draw=black}]
    \path [-] (v1) edge [bend left=50] (v4);
    \path [-] (v0) edge [bend left=40] (v3);
    \path [-] (v2) edge [bend left=50] (v5);
\end{scope}

\begin{scope}[>={Stealth[black]},
              every edge/.style={draw=black}, every node/.style={inner sep=0pt, minimum size = 0pt}]
  \node[label=\phantom{}] (bi1) at (0.75, -0.5) {};
  \node[label=\phantom{}] (bi2) at (6.75, -0.5) {};
  \path [{Bar}-{Bar}] (bi1) edge node [black, pos=0.5, sloped, below] {} (bi2);
\end{scope}

\end{tikzpicture}
}

\resizebox{0.5\textwidth}{!}{
\begin{tikzpicture}[]
\scriptsize
\begin{scope}
    \node at (0,1) {};
\end{scope}

\begin{scope}[every node/.style={inner sep=0.4mm, draw, circle, minimum size = 0pt}]
  \node[label=below:\phantom{+}] (v0) at (0,0) {};
  \node[label=below:\phantom{+}] (v3) at (1.5,0) {};
  \node[label=below:\phantom{+}] (v4) at (3,0) {};
  \node[label=below:\phantom{+}] (v2) at (4.5,0) {};
  \node[label=below:\phantom{+}] (v1) at (6,0) {};
  \node[label=below:\phantom{+}] (v5) at (7.5,0) {};
\end{scope}

\begin{scope}
  \node[] (aux) at (2.25,0) {\small ...};
  \node[] (aux) at (5.25,0) {\small ...};
\end{scope}

\begin{scope}[>={Stealth[black]},
              every edge/.style={draw=black}]
    \path [-] (v0) edge (v3);
    \path [-, dashed] (v4) edge node [black, pos=0.5, sloped, above, yshift=+0.05cm] {$\ell = \deletionElement$} (v2);
    \path [-, dashed] (v1) edge node [black, pos=0.5, sloped, above, yshift=+0.05cm] {$\ell = \deletionElement$} (v5);
\end{scope}

\begin{scope}[>={Stealth[black]},
              every edge/.style={draw=black}]
    \path [-] (v1) edge [bend right=50] (v4);
    \path [-] (v0) edge [bend left=60] (v3);
    \path [-] (v2) edge [bend left=50] (v5);

\end{scope}

\begin{scope}[>={Stealth[black]},
              every edge/.style={draw=black}, every node/.style={inner sep=0pt, minimum size = 0pt}]
  \node[label=\phantom{}] (bi1) at (3.75, -0.5) {};
  \node[label=\phantom{}] (bi2) at (6.75, -0.5) {};
  \path [{Bar}-{Bar}] (bi1) edge node [black, pos=0.5, sloped, below] {} (bi2);
\end{scope}

\end{tikzpicture}
}

\resizebox{0.5\textwidth}{!}{
\begin{tikzpicture}[]
\scriptsize
\begin{scope}
    \node at (0,1) {};
\end{scope}

\begin{scope}[every node/.style={inner sep=0.4mm, draw, circle, minimum size = 0pt}]
  \node[label=below:\phantom{+}] (v0) at (0,0) {};
  \node[label=below:\phantom{+}] (v3) at (1.5,0) {};
  \node[label=below:\phantom{+}] (v4) at (3,0) {};
  \node[label=below:\phantom{+}] (v1) at (4.5,0) {};
  \node[label=below:\phantom{+}] (v2) at (6,0) {};
  \node[label=below:\phantom{+}] (v5) at (7.5,0) {};
\end{scope}

\begin{scope}
  \node[] (aux) at (2.25,0) {\small ...};
  \node[] (aux) at (5.25,0) {\small ...};
\end{scope}

\begin{scope}[>={Stealth[black]},
              every edge/.style={draw=black}]
    \path [-] (v0) edge (v3);
    \path [-, dashed] (v4) edge node [black, pos=0.5, sloped, above, yshift=+0.05cm] {$\ell = \deletionElement$} (v1);
    \path [-] (v2) edge (v5);
\end{scope}

\begin{scope}[>={Stealth[black]},
              every edge/.style={draw=black}]
    \path [-] (v1) edge [bend right=60] (v4);
    \path [-] (v0) edge [bend left=60] (v3);
    \path [-] (v2) edge [bend left=60] (v5);
\end{scope}

\end{tikzpicture}
}

\caption{\label{cap5:fig:reversals_operations_on_oriented_cycle1}
Exemplo das operações aplicadas pelo Lema~\ref{cap5:lemma:convergent_cycles} quando o ciclo $C$ é orientado.
}
\end{figure}
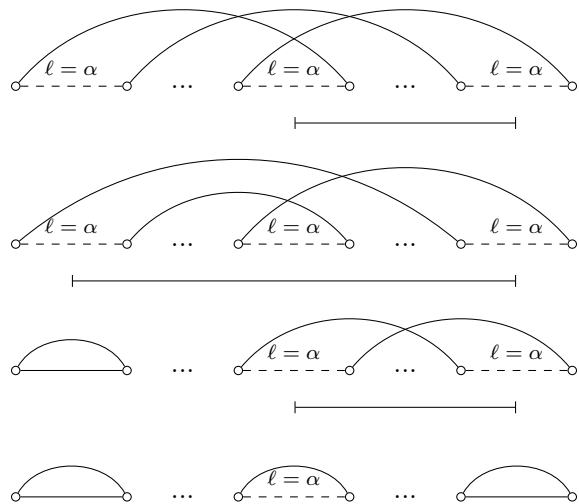
\begin{figure}[tb]
\color{black}
\centering
\resizebox{0.5\textwidth}{!}{
\begin{tikzpicture}
\scriptsize
\begin{scope}
    \node at (0,1) {};
\end{scope}

\begin{scope}[every node/.style={inner sep=0.4mm, draw, circle, minimum size = 0pt}]
  \node[label=below:\phantom{+}] (v0) at (0,0) {};
  \node[label=below:\phantom{+}] (v1) at (1.5,0) {};
  \node[label=below:\phantom{+}] (v2) at (3,0) {};
  \node[label=below:\phantom{+}] (v3) at (4.5,0) {};
  \node[label=below:\phantom{+}] (v4) at (6,0) {};
  \node[label=below:\phantom{+}] (v5) at (7.5,0) {};
  \node[label=below:\phantom{+}] (v6) at (9,0) {};
  \node[label=below:\phantom{+}] (v7) at (10.5,0) {};
\end{scope}

\begin{scope}
  \node[] (aux) at (2.25,0) {\small ...};
  \node[] (aux2) at (5.25,0) {\small ...};
  \node[] (aux3) at (8.25,0) {\small ...};
\end{scope}

\begin{scope}[>={Stealth[black]},
              every edge/.style={draw=black}]
    \path [-, dashed] (v0) edge node [black, pos=0.5, sloped, above, yshift=+0.05cm] {$\ell = \deletionElement$} (v1);
    \path [-, dashed] (v2) edge node [black, pos=0.5, sloped, above, yshift=+0.05cm] {$\ell = \deletionElement$} (v3);
    \path [-, dashed] (v4) edge node [black, pos=0.5, sloped, above, yshift=+0.05cm] {$\ell = \deletionElement$} (v5);
    \path [-, dashed] (v6) edge node [black, pos=0.5, sloped, above, yshift=+0.05cm] {$\ell = \deletionElement$} (v7);
\end{scope}

\begin{scope}[>={Stealth[black]},
              every edge/.style={draw=black}]
    \path [-] (v1) edge [bend left=50] (v4);
    \path [-] (v0) edge [bend left=50] (v5);
    \path [-] (v2) edge [bend left=50] (v7);
    \path [-] (v3) edge [bend left=50] (v6);
\end{scope}

\begin{scope}[>={Stealth[black]},
              every edge/.style={draw=black}, every node/.style={inner sep=0pt, minimum size = 0pt}]
  \node[label=\phantom{}] (bi1) at (3.75, -0.5) {};
  \node[label=\phantom{}] (bi2) at (9.75, -0.5) {};
  \path [{Bar}-{Bar}] (bi1) edge node [black, pos=0.5, sloped, below] {} (bi2);
\end{scope}

\end{tikzpicture}
}

\resizebox{0.5\textwidth}{!}{
\begin{tikzpicture}
\scriptsize
\begin{scope}
    \node at (0,1) {};
\end{scope}

\begin{scope}[every node/.style={inner sep=0.4mm, draw, circle, minimum size = 0pt}]
  \node[label=below:\phantom{+}] (v0) at (0,0) {};
  \node[label=below:\phantom{+}] (v1) at (1.5,0) {};
  \node[label=below:\phantom{+}] (v2) at (3,0) {};
  \node[label=below:\phantom{+}] (v6) at (4.5,0) {};
  \node[label=below:\phantom{+}] (v5) at (6,0) {};
  \node[label=below:\phantom{+}] (v4) at (7.5,0) {};
  \node[label=below:\phantom{+}] (v3) at (9,0) {};
  \node[label=below:\phantom{+}] (v7) at (10.5,0) {};
\end{scope}

\begin{scope}
  \node[] (aux) at (2.25,0) {\small ...};
  \node[] (aux2) at (5.25,0) {\small ...};
  \node[] (aux3) at (8.25,0) {\small ...};
\end{scope}

\begin{scope}[>={Stealth[black]},
              every edge/.style={draw=black}]
    \path [-, dashed] (v0) edge node [black, pos=0.5, sloped, above, yshift=+0.05cm] {$\ell = \deletionElement$} (v1);
    \path [-, dashed] (v2) edge node [black, pos=0.5, sloped, above, yshift=+0.05cm] {$\ell = \deletionElement$} (v6);
    \path [-, dashed] (v5) edge node [black, pos=0.5, sloped, above, yshift=+0.05cm] {$\ell = \deletionElement$} (v4);
    \path [-, dashed] (v3) edge node [black, pos=0.5, sloped, above, yshift=+0.05cm] {$\ell = \deletionElement$} (v7);
\end{scope}

\begin{scope}[>={Stealth[black]},
              every edge/.style={draw=black}]
    \path [-] (v1) edge [bend left=50] (v4);
    \path [-] (v0) edge [bend left=50] (v5);
    \path [-] (v2) edge [bend left=50] (v7);
    \path [-] (v3) edge [bend right=50] (v6);
\end{scope}

\begin{scope}[>={Stealth[black]},
              every edge/.style={draw=black}, every node/.style={inner sep=0pt, minimum size = 0pt}]
  \node[label=\phantom{}] (bi1) at (0.75, -0.5) {};
  \node[label=\phantom{}] (bi2) at (6.75, -0.5) {};
  \path [{Bar}-{Bar}] (bi1) edge node [black, pos=0.5, sloped, below] {} (bi2);
\end{scope}

\end{tikzpicture}
}

\resizebox{0.5\textwidth}{!}{
\begin{tikzpicture}
\scriptsize
\begin{scope}
    \node at (0,1) {};
\end{scope}

\begin{scope}[every node/.style={inner sep=0.4mm, draw, circle, minimum size = 0pt}]
  \node[label=below:\phantom{+}] (v0) at (0,0) {};
  \node[label=below:\phantom{+}] (v5) at (1.5,0) {};
  \node[label=below:\phantom{+}] (v6) at (3,0) {};
  \node[label=below:\phantom{+}] (v2) at (4.5,0) {};
  \node[label=below:\phantom{+}] (v1) at (6,0) {};
  \node[label=below:\phantom{+}] (v4) at (7.5,0) {};
  \node[label=below:\phantom{+}] (v3) at (9,0) {};
  \node[label=below:\phantom{+}] (v7) at (10.5,0) {};
\end{scope}

\begin{scope}
  \node[] (aux) at (2.25,0) {\small ...};
  \node[] (aux2) at (5.25,0) {\small ...};
  \node[] (aux3) at (8.25,0) {\small ...};
\end{scope}

\begin{scope}[>={Stealth[black]},
              every edge/.style={draw=black}]
    \path [-] (v0) edge (v5);
    \path [-, dashed] (v2) edge node [black, pos=0.5, sloped, above, yshift=+0.05cm] {$\ell = \deletionElement$} (v6);
    \path [-, dashed] (v1) edge node [black, pos=0.5, sloped, above, yshift=+0.05cm] {$\ell = \deletionElement$} (v4);
    \path [-, dashed] (v3) edge node [black, pos=0.5, sloped, above, yshift=+0.05cm] {$\ell = \deletionElement$} (v7);
\end{scope}

\begin{scope}[>={Stealth[black]},
              every edge/.style={draw=black}]
    \path [-] (v1) edge [bend left=70] (v4);
    \path [-] (v0) edge [bend left=70] (v5);
    \path [-] (v2) edge [bend left=50] (v7);
    \path [-] (v3) edge [bend right=50] (v6);
\end{scope}

\begin{scope}[>={Stealth[black]},
              every edge/.style={draw=black}, every node/.style={inner sep=0pt, minimum size = 0pt}]
  \node[label=\phantom{}] (bi1) at (3.75, -0.5) {};
  \node[label=\phantom{}] (bi2) at (9.75, -0.5) {};
  \path [{Bar}-{Bar}] (bi1) edge node [black, pos=0.5, sloped, below] {} (bi2);
\end{scope}

\end{tikzpicture}
}

\resizebox{0.5\textwidth}{!}{
\begin{tikzpicture}
\scriptsize
\begin{scope}
    \node at (0,1) {};
\end{scope}

\begin{scope}[every node/.style={inner sep=0.4mm, draw, circle, minimum size = 0pt}]
  \node[label=below:\phantom{+}] (v0) at (0,0) {};
  \node[label=below:\phantom{+}] (v5) at (1.5,0) {};
  \node[label=below:\phantom{+}] (v6) at (3,0) {};
  \node[label=below:\phantom{+}] (v3) at (4.5,0) {};
  \node[label=below:\phantom{+}] (v4) at (6,0) {};
  \node[label=below:\phantom{+}] (v1) at (7.5,0) {};
  \node[label=below:\phantom{+}] (v2) at (9,0) {};
  \node[label=below:\phantom{+}] (v7) at (10.5,0) {};
\end{scope}

\begin{scope}
  \node[] (aux) at (2.25,0) {\small ...};
  \node[] (aux2) at (5.25,0) {\small ...};
  \node[] (aux3) at (8.25,0) {\small ...};
\end{scope}

\begin{scope}[>={Stealth[black]},
              every edge/.style={draw=black}]
    \path [-] (v0) edge (v5);
    \path [-] (v3) edge (v6);
    \path [-, dashed] (v1) edge node [black, pos=0.5, sloped, above, yshift=+0.05cm] {$\ell = \deletionElement$} (v4);
    \path [-, dashed] (v2) edge node [black, pos=0.5, sloped, above, yshift=+0.05cm] {$\ell = \deletionElement$} (v7);
\end{scope}

\begin{scope}[>={Stealth[black]},
              every edge/.style={draw=black}]
    \path [-] (v1) edge [bend right=70] (v4);
    \path [-] (v0) edge [bend left=70] (v5);
    \path [-] (v2) edge [bend left=70] (v7);
    \path [-] (v3) edge [bend right=70] (v6);
\end{scope}

\end{tikzpicture}
}
\caption{\label{cap5:fig:reversals_operations_on_oriented_cycle2}
Exemplo das operações aplicadas pelo Lema~\ref{cap5:lemma:convergent_cycles} quando o ciclo $C$ é não orientado.
}
\end{figure}
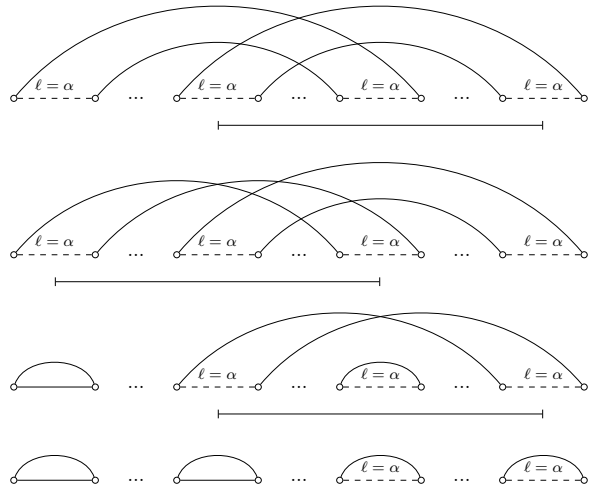

Enquanto $\G_o \neq \G_d$, o Algoritmo~\ref{cap5:alg:2_5-approx} usa uma sequência $S'$ de acordo com os lemas~\ref{cap5:lemma:trivial_cycles_one_indel}, \ref{cap5:lemma:trivial_bad_cycles_two_indels}, \ref{cap5:lemma:divergent_labeled_cycle}, \ref{cap5:lemma:divergent_cycle}, e \ref{cap5:lemma:convergent_cycles}. Como as operações desses lemas garantem que $\Delta \cgood(\Ig, S') \geq 1$, eventualmente chegamos em um grafo de ciclos rotulado e ponderado que possui apenas ciclos unitários bons e, consequentemente, $\G_o = \G_d$. O Algoritmo~\ref{cap5:alg:2_5-approx} é uma $2.5$-aproximação como mostrado no Teorema~\ref{cap5:thr:2_5-approximation}.

\begin{algorithm}[h]
	\caption{\label{cap5:alg:2_5-approx}
		Algoritmo de $2.5$-aproximação para o problema da Distância de Reversões e Indels Intergênicos em Strings com Sinais
	}
	\DontPrintSemicolon
	\Entrada{Uma instância intergênica $\Ig = (\G_o, \G_d)$, com $\G_o = (A, \breve{A})$ e $\G_d = (\iota^n, \breve{\iota}^n)$}
	\Saida{Uma sequência de rearranjos $S$ que transforma $\G_o$ em $\G_d$}
	Seja $S \gets \emptyset$\;
	\Enqto{$\G_o \neq \G_d$}{
		\Se{$G(\Ig)$ possui ciclo unitário que é rotulado ou desbalanceado}{
			Seja $S'$ uma sequência de acordo com os lemas~\ref{cap5:lemma:trivial_cycles_one_indel} ou \ref{cap5:lemma:trivial_bad_cycles_two_indels}\;
		}\SenaoSe{$G(\Ig)$ possui ciclo divergente}{
			Seja $S'$ uma sequência de acordo com os lemas~\ref{cap5:lemma:divergent_labeled_cycle} ou \ref{cap5:lemma:divergent_cycle}\;
		}\Senao(){
			Seja $S'$ uma sequência de acordo com o Lema~\ref{cap5:lemma:convergent_cycles}\;
		}
		$\G_o \gets \G_o \comp S'$\;
		Adicione as operações de $S'$ na sequência $S$\;
	}
 	{\bf retorne} a sequência $S$\;
\end{algorithm}

\begin{theorem}\label{cap5:thr:2_5-approximation}
O Algoritmo~\ref{cap5:alg:2_5-approx} é uma $2.5$-aproximação para o o problema da Distância de Reversões e Indels Intergênicos em Strings com Sinais. 
\end{theorem}

\begin{proof}
	A cada iteração, o algoritmo utiliza uma sequência $S'$ de acordo com um dos lemas~\ref{cap5:lemma:trivial_cycles_one_indel}, \ref{cap5:lemma:trivial_bad_cycles_two_indels}, \ref{cap5:lemma:divergent_labeled_cycle}, \ref{cap5:lemma:divergent_cycle}, ou \ref{cap5:lemma:convergent_cycles}. No pior caso, a sequência $S'$ satisfaz $|S'|/\Delta \cgood(\Ig, S') = 5/2$. Portanto, para transformar $\G_o$ em $\G_d$, o que é equivalente a tornar $|\pi^A| + 1 - \cgood(\Ig)$ igual a zero, o algoritmo usa no máximo $5/2 (|\pi^A| + 1 - \cgood(\Ig))$ operações. Pelo Lema~\ref{cap5:lemma:grafo_ciclos_intergenicos_lower_bound}, esse algoritmo é uma $2.5$-aproximação.
\end{proof}

Para qualquer grafo de ciclos rotulado e ponderado $G(\Ig)$, tanto o número de vértices quanto o número de arestas são $O(n)$ e, portanto, temos que $|\pi^A| + 1 - \cgood(\Ig) \in O(n)$. Consequentemente, o número de iterações do laço principal do algoritmo é limitado por $O(n)$. Além disso, cada iteração possui complexidade de tempo linear, já que podemos listar todos os ciclos do grafo, achar o caso apropriado, e aplicar a sequência $S'$, que possui tamanho de no máximo cinco operações, em tempo linear. Dessa forma, concluímos que o algoritmo possui complexidade de tempo de $O(n^2)$.

\subsection{Uma 4-Aproximação para Modelos com Transposições}\label{cap5:subsection:4_transp}

Nesta seção, apresentamos algoritmos com fator de aproximação igual a $4$ para o problema da Distância de Transposições e Indels Intergênicos em Strings sem Sinais e o problema da Distância de Reversões, Transposições e Indels Intergênicos em Strings com Sinais. Assim como na $2.5$-aproximação apresentada na seção anterior, a ideia principal desses algoritmos é criar novos ciclos bons a cada iteração até que o grafo contenha apenas ciclos unitários bons. 

Na seção anterior, quando só existem ciclos convergentes no grafo, demonstramos como encontrar uma sequência $S$ tal que $|S|/\Delta \cgood(\Ig, S) \leq 5/2$ usando reversões e {\it indels}. Nos próximos lemas, mostramos como tratar ciclos convergentes usando uma sequência $S$, que possui apenas transposições e indels, tal que $|S|/\Delta \cgood(\Ig, S) \leq 2$. Lembre-se que para strings sem sinais, o grafo de ciclos rotulado e ponderado $G(\Ig)$ só possui ciclos convergentes.

Quando consideramos apenas operações que agem em arestas de origem de ciclos bons, os resultados de problemas intergênicos que não consideram {\it indels} podem ser diretamente aplicados. A seguir, listamos alguns desses resultados.

\begin{lemma}[Oliveira e coautores~\cite{2021a-oliveira-etal}, Lema 4.6]\label{cap5:lemma:breaking_good_oriented_cycle}
Para qualquer grafo $G(\Ig)$, se $G(\Ig)$ possui um ciclo orientado bom $C$, então existe uma sequência $S$ tal que $|S| \leq 3$ e $\Delta \cgood(\Ig, S) = 2$.
\end{lemma}

\begin{lemma}[Oliveira e coautores~\cite{2021a-oliveira-etal}, Lema 4.7]\label{cap5:lemma:breaking_good_non_oriented_cycle}
Para qualquer grafo $G(\Ig)$, se existe $k$-ciclo não orientado $C$ em $G(\Ig)$, com $k > 2$, não existem ciclos orientados em $G(\Ig)$, e todos os ciclos em $G(\Ig)$ são bons, então existe uma sequência $S$ tal que $|S| \leq 7$ e $\Delta \cgood(\Ig, S) = 4$.
\end{lemma}

\begin{lemma}[Oliveira e coautores~\cite{2021a-oliveira-etal}, Lema 4.8]\label{cap5:lemma:breaking_short_cycle}
Para qualquer grafo $G(\Ig)$, se todo $k$-ciclo $C$ em $G(\Ig)$ é um ciclo bom e $k \leq 2$, então existe uma sequência $S$ tal que $|S| \leq 2$ e $\Delta \cgood(\Ig, S) = 2$.
\end{lemma}

Os próximos lemas apresentam sequências de transposições ou {\it indels} que agem em ciclos convergentes rotulados ou desbalanceados.

\begin{lemma}\label{cap5:lemma:oriented_bad_cycle_transp}
Para qualquer grafo $G(\Ig)$, se $G(\Ig)$ possui um ciclo orientado $C$ rotulado ou desbalanceado, então existe transposição $\tau$ que transforma $C$ em três novos ciclos, tal que um deles é um ciclo unitário não negativo que possui aresta de origem limpa.
\end{lemma}

\begin{proof}
	Como $C$ é um ciclo orientado, sempre existe tripla orientada $o_i$, $o_j$ e $o_k$, com $i < j < k$, tal que $o_i > o_k > o_j$ e $k = j+1$~\cite{1998-bafna-pevzner}. Uma transposição aplicada nessas três arestas de origem cria três novos ciclos, de forma que a aresta de origem com índice $o_j$ no novo grafo pertence a um ciclo unitário, já que $k = j+1$. Sempre podemos escolher a transposição de forma que o ciclo unitário é não negativo e possui aresta de origem limpa. Para isso, a transposição precisa mover qualquer elemento $\deletionElement$ e nucleotídeos excedentes, caso existam, da aresta de origem com índice $o_j$ para um dos outros dois ciclos.
\end{proof}

\begin{lemma}\label{cap5:lemma:oriented_cycle_transp}
Para qualquer grafo $G(\Ig)$, se $G(\Ig)$ possui um ciclo orientado $C$ rotulado ou desbalanceado, então existe uma sequência $S$ tal que $|S| \leq 2$ e $\Delta \cgood(\Ig, S) = 1$.
\end{lemma}

\begin{proof}
	Diretamente dos lemas~\ref{cap5:lemma:oriented_bad_cycle_transp} e \ref{cap5:lemma:trivial_cycles_one_indel}. Note que, após aplicar a transposição do Lema~\ref{cap5:lemma:oriented_bad_cycle_transp}, existe um ciclo unitário $C'$ que é um ciclo bom ou que satisfaz as condições do Lema~\ref{cap5:lemma:trivial_cycles_one_indel}.
\end{proof}

\begin{lemma}\label{cap5:lemma:two_labeled_cycles}
Para qualquer grafo $G(\Ig)$, se $G(\Ig)$ possui pelo menos dois ciclos rotulados ou desbalanceados $C$ e $D$, tal que $C$ e $D$ são não orientados e não unitários, então existe uma transposição $\tau$ que transforma $C$ e $D$ nos ciclos $C'$ e $D'$, tal que ou $C'$ ou $D'$ é um ciclo unitário não negativo que possui aresta de origem limpa.
\end{lemma}

\begin{proof}
	Sejam $C = (o_1, d_1, o_2, d_2 \ldots, o_x, d_x)$ e $D = (o'_1, d'_1, o'_2, d'_2, \ldots, o'_y, d'_y)$. Suponha, sem perda de generalidade, que $o_1 < o'_1$. Uma transposição aplicada nas arestas de origem $o_1$, $o_x$ e $o'_1$ transforma esses ciclos em dois ciclos $C'$ e $D'$, tal que $C'$ é um ciclo unitário que contém a aresta de destino $d_x$ e $D'$ é um $(y + x - 1)$-ciclo. Podemos mover qualquer elemento $\deletionElement$ e nucleotídeos excedentes da aresta de origem $o_x$, onde o ciclo unitário é formado, para o ciclo $D'$. Dessa forma, podemos garantir que $C'$ é um ciclo unitário não negativo que possui aresta de origem limpa.
\end{proof}

\begin{lemma}\label{cap5:lemma:two_labeled_cycles_sequence}
Para qualquer grafo $G(\Ig)$, se $G(\Ig)$ possui pelo menos dois ciclos rotulados ou desbalanceados $C$ e $D$, tal que $C$ e $D$ são não orientados e não unitários, então existe uma sequência $S$ tal que $|S| \leq 2$ e $\Delta \cgood(\Ig, S) = 1$.
\end{lemma}

\begin{proof}
	Diretamente dos lemas~\ref{cap5:lemma:two_labeled_cycles} e \ref{cap5:lemma:trivial_cycles_one_indel}. Note que, após aplicar a transposição do Lema~\ref{cap5:lemma:two_labeled_cycles}, existe um ciclo unitário $C'$ que é um ciclo bom ou satisfaz as condições do Lema~\ref{cap5:lemma:trivial_cycles_one_indel}.
\end{proof}

A Figura~\ref{cap5:fig:oriented_or_two_labeled} mostra exemplos das transposições aplicadas de acordo com os lemas~\ref{cap5:lemma:oriented_bad_cycle_transp} e \ref{cap5:lemma:two_labeled_cycles}.

\begin{figure}[tb]
    \centering
    \includegraphics[width=0.7\textwidth]{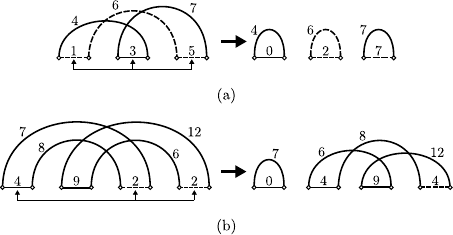}
    \caption{\label{cap5:fig:oriented_or_two_labeled}
    Exemplos de transposições aplicadas pelos lemas~\ref{cap5:lemma:oriented_bad_cycle_transp} e \ref{cap5:lemma:two_labeled_cycles}. {\bf (a)} Nesse caso, existe um ciclo orientado $C$ rotulado e desbalanceado. Existe uma transposição que transforma esse ciclo em três novos ciclos, tal que um deles é um ciclo unitário não negativo que possui aresta de origem limpa. {\bf (b)} Nesse caso, existem dois ciclos não unitários $C$ e $D$ rotulados e desbalanceados. Existe uma transposição aplicada em três arestas de origem desses dois ciclos que transforma $C$ e $D$ nos ciclos $C'$ e $D'$, tal que um desses novos ciclos é um ciclo unitário não negativo que possui aresta de origem limpa.
    }
\end{figure}

\begin{figure}[tb]
	\centering
    \includegraphics[width=0.8\textwidth]{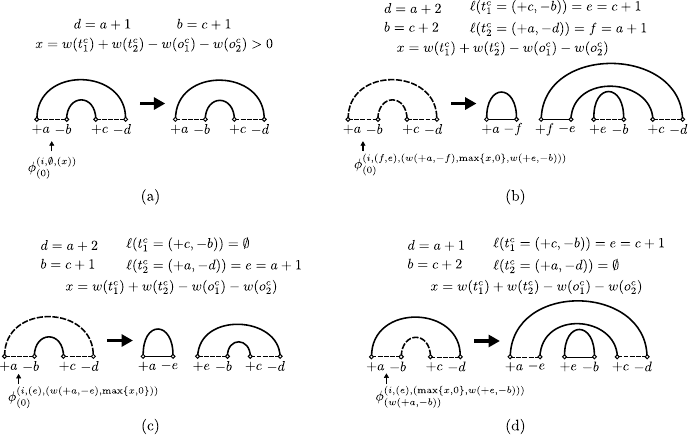}
    \caption{\label{cap5:fig:one_bad_cycle_short_indels}
    	Quatro possíveis casos de um inserção aplicada em um $2$-ciclo $C$. {\bf (a)} Nesse caso, $C$ é positivo e ambas arestas de destino de $C$ são limpas, então usamos um {\it indel} que adiciona a quantidade necessária de nucleotídeos para tornar $C$ balanceado. {\bf (b)} Nesse caso, ambas arestas de destino de $C$ são rotuladas, então usamos um {\it indel} que adiciona dois elementos, gerando um ciclo unitário bom e um ciclo unitário não positivo, além de adicionar nucleotídeos suficientes para tornar o $2$-ciclo em um ciclo não positivo. {\bf (c-d)} Nesses casos, apenas uma das arestas de destino de $C$ é rotulada, então usamos um {\it indel} que adiciona um elemento, gerando um ciclo unitário bom, além de adicionar nucleotídeos suficientes para tornar o $2$-ciclo em um ciclo não positivo.
    }
\end{figure}

\begin{lemma}\label{cap5:lemma:one_bad_cycle_short}
	Para qualquer grafo $G(\Ig)$, se $G(\Ig)$ possui apenas um ciclo rotulado ou desbalanceado $C$, tal que $C$ é um $2$-ciclo não orientado, e todos os outros ciclos de $G(\Ig)$ são bons e não orientados, então existe uma sequência $S$ tal que $|S|/\Delta \cgood(\Ig, S) \leq 2$.
\end{lemma}

\begin{proof}
	Primeiramente, precisamos aplicar {\it indels} de acordo com a Figura~\ref{cap5:fig:one_bad_cycle_short_indels} para transformar $C$ em ciclos bons. No pior caso, todas as arestas de $C$ são rotuladas. Precisamos aplicar uma inserção de acordo com a Figura~\ref{cap5:fig:one_bad_cycle_short_indels}b para remover os rótulos das duas arestas de origem de $C$. Note que o ciclo unitário mais à esquerda sempre pode ser um ciclo bom, já que podemos realizar a inserção antes de qualquer elemento $\deletionElement$ e atribuir qualquer custo à aresta de origem desse ciclo unitário. Nesse caso, a inserção cria ciclos $C'$ (unitário), $C''$ ($2$-ciclo) e $C'''$ (unitário), tal que apenas $C'$ é um ciclo bom, mas todos eles são não positivos, já que essa inserção pode adicionar nucleotídeos na aresta de origem mais à esquerda de $C''$ e na única aresta de origem de $C'''$. Note que a aresta mais à esquerda de $C''$ foi criada pela inserção e, consequentemente, essa aresta de origem é limpa e podemos atribuir qualquer custo a ela. Agora, apenas precisamos aplicar uma deleção para cada um das duas arestas de origem rotuladas, tornando $C''$ e $C'''$ em ciclos bons.

	Até agora, no pior caso, temos uma sequência $S_1$ com três {\it indels} e $\Delta \cgood(\Ig, S_1) = 1$. Podemos aplicar os lemas~\ref{cap5:lemma:breaking_good_non_oriented_cycle} ou \ref{cap5:lemma:breaking_short_cycle} no novo grafo $G(\G^{'}_o, \G_d)$, onde $\G^{'}_o = \G_o \comp S_1$. Usando a sequência $S_1$ e o Lema~\ref{cap5:lemma:breaking_good_non_oriented_cycle}, podemos construir uma sequência $S$, com no máximo dez operações, tal que $\Delta \cgood(\Ig, S) = 5$. Já ao usar a sequência $S_1$ e o Lema~\ref{cap5:lemma:breaking_short_cycle}, podemos construir uma sequência $S$, com no máximo cinco operações, tal que $\Delta \cgood(\Ig, S) = 3$.
\end{proof}

\begin{lemma}\label{cap5:lemma:one_bad_cycle_long}
	Para qualquer grafo $G(\Ig)$, se $G(\Ig)$ possui apenas um ciclo rotulado ou desbalanceado $C$, tal que $C$ é um $x$-ciclo não orientado com $x > 2$, e todos os outros ciclos de $G(\Ig)$ são bons e não orientados, então existe uma sequência $S$ tal que $|S|/\Delta \cgood(\Ig, S) \leq 2$.
\end{lemma}

\begin{proof}
	Seja $C = (o_1, d_1, o_2, d_2 \ldots, o_x, d_x)$. Suponha que existe outro ciclo não orientado $D = (o'_1, d'_1, o'_2, d'_2, \ldots, o'_y, d'_y)$, que é um ciclo bom de acordo com o enunciado deste lema, tal que existem triplas $(o_i, o_j, o_k)$ e $(o'_r, o'_s, o'_t)$ que satisfazem uma das seguintes condições: $o_i < o'_r < o_j < o'_s < o_k < o'_t$ ou $o_i > o'_r > o_j > o'_s > o_k > o'_t$. Bafna e Pevzner~\cite{1998-bafna-pevzner} mostraram que uma transposição $\tau$ aplicada nas arestas de origem $(o'_r, o'_s, o'_t)$ cria novos ciclos $C'$ e $D'$ com os mesmos conjuntos de vértices de $C$ e $D$, respectivamente, tal que $C'$ é um ciclo orientado. Além disso, Bafna e Pevzner~\cite{1998-bafna-pevzner} mostraram que uma transposição aplicada na tripla orientada de $C'$ torna $D'$ orientado no novo grafo. Pelo Lema~\ref{cap5:lemma:oriented_cycle_transp}, existe uma sequência $S_1$, tal que $|S_1| \leq 2$ e $\Delta \cgood(\G_o \comp \tau, \G_d, S_1) = 1$, aplicada na tripla orientada de $C'$ que transforma $D'$ em um ciclo orientado $D''$. Como $D''$ possui os mesmos vértices de $D$, sabemos que $D''$ é um ciclo bom. Dessa forma, usando o Lema~\ref{cap5:lemma:breaking_good_oriented_cycle}, existe uma sequência $S_2$ que é aplicada no ciclo $D''$, tal que $S_2$ possui no máximo três operações e adiciona dois ciclos bons no grafo. Portanto, podemos construir uma sequência $S$ tal que $|S| \leq 6$ e $\Delta \cgood(\Ig, S) \leq 3$.

	Se a condição anterior não é verdadeira, Bafna e Pevzner~\cite{1998-bafna-pevzner} demonstraram que existem ciclos $D$ e $E$, que são ciclos bons de acordo com o enunciado deste lema, tal que existe uma transposição $\tau$ aplicada nas arestas de $D$ e $E$ que cria dois novos ciclos $D'$ e $E'$, além de transformar $C$ em um ciclo orientado $C'$. Oliveira e coautores~\cite{2021a-oliveira-etal} demonstraram como escolher a transposição de forma que os novos ciclos $D'$ e $E'$ sejam ciclos balanceados. Além disso, $D'$ e $E'$ são ciclos limpos, pois são formados pelos mesmos vértices dos ciclos bons $D$ e $E$. Bafna e Pevzner~\cite{1998-bafna-pevzner} também demostraram que, após aplicar uma transposição na tripla orientada de $C'$, existe um ciclo orientado $F'$ no grafo, tal que $F'$ tem o mesmo conjunto de vértices de $D'$ ou $E'$. 

	De forma similar ao caso anterior, podemos usar uma sequência $S_1$ (Lema~\ref{cap5:lemma:breaking_good_oriented_cycle}) na tripla orientada de $C'$, tal que $|S_1| \leq 2$ e $\Delta \cgood(\G_o \comp \tau, \G_d, S_1) = 1$, e uma sequência $S_2$ (Lema~\ref{cap5:lemma:breaking_good_oriented_cycle}) no ciclo $F'$, tal que $S_2$ possui no máximo três operações e adiciona dois ciclos bons no grafo.  
\end{proof}

As transposições descritas na prova do Lema~\ref{cap5:lemma:one_bad_cycle_long}, que criam novos ciclos orientados, são similares às transposições dos exemplos das figuras~\ref{cap4:fig:transp_long_non_oriented_2} e \ref{cap4:fig:transp_long_non_oriented_3}, que foram apresentadas no Capítulo~\ref{cap:label:indel}. Com esses lemas, podemos apresentar os algoritmos~\ref{cap5:alg:4_transp_cycle_graph} e \ref{cap5:alg:4_transp_rev_cycle_graph}. Assim como o Algoritmo~\ref{cap5:alg:2_5-approx}, esses algoritmos possuem complexidade de tempo de $O(n^2)$. 

\begin{algorithm}[th]
	\caption{\label{cap5:alg:4_transp_cycle_graph}
		Algoritmo de $4$-aproximação para o problema da Distância de Transposições e Indels Intergênicos em Strings sem Sinais
	}
	\DontPrintSemicolon
	\Entrada{Uma instância intergênica $\Ig = (\G_o, \G_d)$, com $\G_o = (A, \breve{A})$ e $\G_d = (\iota^n, \breve{\iota}^n)$}
	\Saida{Uma sequência de rearranjos $S$ que transforma $\G_o$ em $\G_d$}
	Seja $S \gets \emptyset$\;
	\Enqto{$\G_o \neq \G_d$}{
		\Se{$G(\Ig)$ possui ciclo unitário que é rotulado ou desbalanceado}{
			Seja $S'$ uma sequência de acordo com os lemas~\ref{cap5:lemma:trivial_cycles_one_indel} ou \ref{cap5:lemma:trivial_bad_cycles_two_indels}\;
		}\SenaoSe{$G(\Ig)$ possui ciclo orientado}{
			Seja $S'$ uma sequência de acordo com os lemas~\ref{cap5:lemma:breaking_good_oriented_cycle} ou \ref{cap5:lemma:oriented_cycle_transp}\;
		}\Senao(){
			Seja $S'$ uma sequência de acordo com os lemas~\ref{cap5:lemma:breaking_good_non_oriented_cycle}, \ref{cap5:lemma:breaking_short_cycle}, \ref{cap5:lemma:two_labeled_cycles_sequence}, \ref{cap5:lemma:one_bad_cycle_short}, ou \ref{cap5:lemma:one_bad_cycle_long}\;
		}
		$\G_o \gets \G_o \comp S'$\;
		Adicione as operações de $S'$ na sequência $S$\;
	}
 	{\bf retorne} a sequência $S$\;
\end{algorithm}

\begin{algorithm}[th]
	\caption{\label{cap5:alg:4_transp_rev_cycle_graph}
		Algoritmo de $4$-aproximação para o problema da Distância de Reversões, Transposições e Indels Intergênicos em Strings com Sinais
	}
	\DontPrintSemicolon
	\Entrada{Uma instância intergênica $\Ig = (\G_o, \G_d)$, com $\G_o = (A, \breve{A})$ e $\G_d = (\iota^n, \breve{\iota}^n)$}
	\Saida{Uma sequência de rearranjos $S$ que transforma $\G_o$ em $\G_d$}
	Seja $S \gets \emptyset$\;
	\Enqto{$\G_o \neq \G_d$}{
		\Se{$G(\Ig)$ possui ciclo unitário que é rotulado ou desbalanceado}{
			Seja $S'$ uma sequência de acordo com os lemas~\ref{cap5:lemma:trivial_cycles_one_indel} ou \ref{cap5:lemma:trivial_bad_cycles_two_indels}\;	
		}\SenaoSe{$G(\Ig)$ possui ciclo divergente}{
			Seja $S'$ uma sequência de acordo com os lemas~\ref{cap5:lemma:divergent_labeled_cycle} ou \ref{cap5:lemma:divergent_cycle}\;
		}\SenaoSe{$G(\Ig)$ possui ciclo orientado}{
			Seja $S'$ uma sequência de acordo com os lemas~\ref{cap5:lemma:breaking_good_oriented_cycle} ou \ref{cap5:lemma:oriented_cycle_transp}\;
		}\Senao(){
			Seja $S'$ uma sequência de acordo com os lemas~\ref{cap5:lemma:breaking_good_non_oriented_cycle}, \ref{cap5:lemma:breaking_short_cycle}, \ref{cap5:lemma:two_labeled_cycles_sequence}, \ref{cap5:lemma:one_bad_cycle_short}, ou \ref{cap5:lemma:one_bad_cycle_long}\;
		}
		$\G_o \gets \G_o \comp S'$\;
		Adicione as operações de $S'$ na sequência $S$\;
	}
 	{\bf retorne} a sequência $S$\;
\end{algorithm}

No Teorema~\ref{cap5:theorem:4_approx_transp_cycle_graph}, demonstramos que esses algoritmos possuem fator de aproximação igual a $4$ para os problemas de Distância de Rearranjos Intergênicos considerando os modelos $\Mindel_{\tau}$ e $\Mindel_{\rho,\tau}$.

\begin{theorem}\label{cap5:theorem:4_approx_transp_cycle_graph}
Os algoritmos~\ref{cap5:alg:4_transp_cycle_graph} e \ref{cap5:alg:4_transp_rev_cycle_graph} são algoritmos de $4$-aproximação para o problema da Distância de Transposições e Indels Intergênicos em Strings sem Sinais e o problema da Distância de Reversões, Transposições e Indels Intergênicos em Strings com Sinais, respectivamente.
\end{theorem}

\begin{proof}
Similar à prova do Teorema~\ref{cap5:thr:2_5-approximation}, mas usando o fato de que a cada iteração o algoritmo usa uma sequência $S'$ tal que $|S'|/\Delta \cgood(\Ig, S') \leq 2$. 
\end{proof}

\section{Conclusões}

Neste capítulo, estudamos os problemas de Distância de Rearranjos Intergênicos em Genomas Desbalanceados. Consideramos modelos que contém a combinação de {\it indels} com reversões e/ou transposições, para strings com ou sem sinais.

Na Seção~\ref{cap5:section:complexidade}, demonstramos que os seguintes problemas são NP-difíceis: a Distância de Reversões e Indels Intergênicos em Strings sem Sinais; a Distância de Transposições e Indels Intergênicos em Strings sem Sinais; e a Distância de Transposições, Reversões e Indels Intergênicos em Strings com ou sem Sinais.

Na Seção~\ref{cap5:section:breakpoints}, apresentamos algoritmos de aproximação que usam o conceito de {\it breakpoints} intergênicos, considerando strings sem sinais. Já na Seção~\ref{cap5:section:grafo_ciclos}, apresentamos algoritmos de aproximação usando o grafo de ciclos rotulado e ponderado, uma nova estrutura que representa uma instância intergênica de genomas desbalanceados. A Tabela~\ref{cap5:table:resumo_resultados} resume o fator de aproximação alcançado para cada problema estudado neste capítulo.

\begin{table}[tbh]
\centering
\caption{Resumo dos algoritmos apresentados neste capítulo para os problemas de Distância de Rearranjos Intergênicos em Genomas Desbalanceados.\label{cap5:table:resumo_resultados}}
\begin{tabular}{@{}lcc@{}}
\toprule
Modelo & Seção~\ref{cap5:section:breakpoints} & Seção~\ref{cap5:section:grafo_ciclos} \\ \midrule
Reversões e Indels (com sinais) & - & $2.5$-aproximação \\
Reversões e Indels (sem sinais) & $4$-aproximação & - \\
Transposições e Indels (sem sinais) & $4.5$-aproximação & $4$-aproximação \\
Transposições, Reversões e Indels (com sinais) & - & $4$-aproximação \\
Transposições, Reversões e Indels (sem sinais) & $6$-aproximação & - \\
\bottomrule
\end{tabular}
\end{table}

\chapter{Resultados Experimentais}\label{cap:label:experimentos}

Neste capítulo, apresentamos experimentos computacionais com genomas sintéticos e genomas reais, considerando os algoritmos desenvolvidos nos capítulos~\ref{cap:label:indel} e \ref{cap:label:intergenicos}. As implementações de todos os algoritmos deste capítulo estão disponíveis em um repositório público\footnote{\href{https://github.com/compbiogroup}{https://github.com/compbiogroup}}. Este capítulo está dividido da seguinte forma. Na Seção~\ref{cap6:section:instancias}, descrevemos o procedimento de criação das bases de dados de genomas sintéticos. Na Seção~\ref{cap6:section:indels}, apresentamos os resultados dos experimentos considerando genomas sintéticos e os algoritmos do Capítulo~\ref{cap:label:indel}. Já na Seção~\ref{cap6:section:intergenicos}, apresentamos os resultados dos experimentos considerando genomas sintéticos e os algoritmos do Capítulo~\ref{cap:label:intergenicos}. Por fim, na Seção~\ref{cap6:section:genomas_reais}, apresentamos experimentos usando genomas reais de cianobactérias da base de dados Cyanorak 2.1~\cite{2020:garczarek-etal}.

\section{Criação de Instâncias Sintéticas}\label{cap6:section:instancias}

Dados os parâmetros $n$, $k$, $\M$, e $t$, onde $n$ é a quantidade de genes no genoma de destino, $k$ é o número de rearranjos conservativos (reversões ou transposições) e o número de rearranjos não conservativos a serem aplicados (inserções ou deleções), $\M$ é um modelo de rearranjo, e $t$ indica o tipo das strings (\texttt{signed} ou \texttt{unsigned}), criamos uma instância intergênica sintética da seguinte forma:

\begin{enumerate}
	\item Criamos uma lista $\breve{\iota}^n$ de $n+1$ números inteiros escolhidos de forma aleatória, considerando uma distribuição uniforme dos valores no intervalo $[0,100]$;
	\item Criamos a string identidade $\iota^n$, tal que $\iota^n$ é uma string com sinais, se $t = \texttt{signed}$, ou $\iota^n$ é uma string sem sinais, se $t = \texttt{unsigned}$;
	\item Criamos o genoma de destino $\G_d = (\iota^n, \breve{\iota}^n)$;
	\item Inicializamos o genoma de origem $\G_o$ como uma cópia de $\G_d$;
	\item Aplicamos $k$ rearranjos conservativos em $\G_o$. As operações são definidas de acordo com as operações de $\M$:
	\begin{itemize}
		\item $k$ reversões, se $\M$ contém reversões e não contém transposições (i.e., $\M = \Mindel_{\rho}$);
		\item $k$ transposições, se $\M$ contém transposições e não contém reversões (i.e., $\M = \Mindel_{\tau}$);
		\item $k$ reversões ou transposições, se $\M$ contém reversões e transposições (i.e., $\M = \Mindel_{\rho, \tau}$). Nesse caso, cada operação possui $50\%$ de chance de ser uma reversão ou uma transposição.
	\end{itemize}
	\item Aplicamos $k/2$ deleções em $\G_o$;
	\item Por fim, aplicamos $k/2$ inserções em $\G_o$.
\end{enumerate}

O nosso procedimento aplica deleções antes das inserções a fim de evitar que um elemento seja adicionado e, posteriormente, removido no processo de criação da instância. 
Cada rearranjo usado possui parâmetros escolhidos de maneira aleatória, considerando uma distribuição uniforme do conjunto de todos os valores válidos para os parâmetros daquela operação. {\revisaof Esse procedimento usa inserções que adicionam, numa posição escolhida aleatoriamente, um segmento que contém um único caractere, que ainda não pertence a string, e duas regiões intergênicas, sendo que os tamanhos dessas duas regiões intergênicas são escolhidos de forma aleatória, considerando uma distribuição uniforme dos valores no intervalo $[0,100]$}. Na criação de uma instância clássica sintética, seguimos o mesmo processo e desconsideramos as regiões intergênicas. 

Nossa base de dados é dividida em conjuntos agrupados pelos parâmetros $n$, $k$, $\M$, e $t$. Para $n \in \{50, 100, \ldots, 500\}$, $k \in \{n/2, n\}$, $\M \in \{\Mindel_{\rho}, \Mindel_{\tau}, \Mindel_{\rho,\tau}\}$ e $t \in \{\texttt{signed}, \texttt{unsigned}\}$, cada conjunto $DS_{n, k, t}^{\M}$ possui $1000$ instâncias clássicas sintéticas criadas usando os parâmetros $n$, $k$, $\M$ e $t$. Note que quando $\M = \Mindel_{\tau}$, consideramos que $t = \texttt{unsigned}$. Além disso, usamos as mesmas combinações de parâmetros para a criação de cada conjunto $DSI_{n, k, t}^{\M}$ com $1000$ instâncias intergênicas sintéticas.

\section{Experimentos com Instâncias Clássicas}\label{cap6:section:indels}

Os algoritmos gulosos baseados em {\it breakpoints} do Capítulo~\ref{cap:label:indel} possuem complexidade de tempo de $O(n^3)$ ou $O(n^4)$. Por isso, decidimos implementar uma versão desses algoritmos que possui complexidade de tempo de $O(n^2)$. Nessa implementação, ao invés de encontrar a operação $\beta$ com melhor valor de $\Delta \Phi(\I, \beta) + \Delta b_{\M}(\I, \beta)$, apenas encontramos uma operação que garante o fator de aproximação, a cada iteração. Essas operações já são descritas nas provas dos lemas do Capítulo~\ref{cap:label:indel}. Dessa forma, podemos compará-los de forma mais justa com os algoritmos baseados em grafos de ciclos, que também possuem complexidade de tempo de $O(n^2)$. 

Cada tabela desta seção apresenta resultados para um algoritmo do Capítulo~\ref{cap:label:indel} e um valor de $k \in \{n/2, n\}$. 
Cada linha dessas tabelas possui dados referentes a execução do algoritmo usando todas as instâncias do conjunto $DS_{n, k, t}^{\M}$, sendo que o valor de $n$ é descrito na primeira coluna da tabela, o valor de $k$ é descrito no cabeçalho da tabela, e os valores de $\M$ e $t$ são inferidos a partir do algoritmo considerado. {\revisaof A segunda, terceira e quarta colunas apresentam os valores de mínimo, média e máximo, respectivamente, para o tamanho da sequência de rearranjos das soluções encontradas pelo algoritmo. A quinta, sexta e sétima colunas apresentam os valores de mínimo, média e máximo, respectivamente, para o fator de aproximação prático das soluções encontradas pelo algoritmo.} O fator de aproximação prático para uma instância é calculado usando o tamanho da solução encontrada pelo algoritmo para essa instância e o seu limitante inferior, sendo que usamos o limitante inferior que corresponde ao fator de aproximação teórico do algoritmo considerado. Note que cada problema tem um limitante inferior específico. Como esperado, em todas as tabelas, os valores referentes à quantidade de operações usadas na solução aumentam à medida que o valor de $n$ aumenta. 

As tabelas~\ref{cap6:table:classic_b_ur_50} e \ref{cap6:table:classic_b_ur_100} apresentam os resultados do algoritmo de $2$-aproximação para a Distância de Reversões e Indels em Strings sem Sinais (Algoritmo~\ref{cap4:algorithm_reversals}). Esse algoritmo usa o conceito de {\it breakpoints}. Apesar do fato de que as instâncias sintéticas da Tabela~\ref{cap6:table:classic_b_ur_100} foram criadas usando o dobro de operações em relação às instâncias da Tabela~\ref{cap6:table:classic_b_ur_50}, a razão entre o tamanho das soluções do experimento da Tabela~\ref{cap6:table:classic_b_ur_100} e o tamanho das soluções do experimento da Tabela~\ref{cap6:table:classic_b_ur_50} é menor que $1.5$ para todas as medidas usadas. 
Para a Tabela~\ref{cap6:table:classic_b_ur_50}, a média do fator de aproximação prático está entre $1.68$ e $1.81$. Já para a Tabela~\ref{cap6:table:classic_b_ur_100}, os valores para a média do fator de aproximação prático estão entre $1.81$ e $1.94$. Em ambas as tabelas, os valores dessa coluna formam uma sequência não decrescente. Na Tabela~\ref{cap6:table:classic_b_ur_50}, os valores de mínimo e máximo para o fator de aproximação prático foram de $1.40$ e $1.92$, respectivamente, e ocorreram em instâncias com $n = 50$. Na Tabela~\ref{cap6:table:classic_b_ur_100}, os valores de mínimo e máximo para o fator de aproximação prático foram de $1.64$ e $2.00$, respectivamente, e ocorreram em instâncias com $n = 50$.

As tabelas~\ref{cap6:table:classic_b_t_50} e \ref{cap6:table:classic_b_t_100} apresentam os resultados do algoritmo de $3$-aproximação para a Distância de Transposições e Indels em Strings sem Sinais (Algoritmo~\ref{cap4:algorithm_transpositions}), que usa o conceito de {\it breakpoints}. Já as tabelas~\ref{cap6:table:classic_cg_t_50} e \ref{cap6:table:classic_cg_t_100} apresentam os resultados do algoritmo de $2$-aproximação (Algoritmo~\ref{cap4:algorithm:transp_cycle_graph}), que usa o conceito de grafo de ciclos rotulado, para o mesmo problema. Como esperado, o algoritmo de $2$-aproximação possui resultados práticos muito melhores do que o algoritmo de $3$-aproximação. Agora, discutiremos com mais detalhes os resultados para cada algoritmo. 

A razão entre o tamanho das soluções do experimento da Tabela~\ref{cap6:table:classic_b_t_100} e o tamanho das soluções do experimento da Tabela~\ref{cap6:table:classic_b_t_50} é menor que $1.4$ para todas as medidas usadas. Para a Tabela~\ref{cap6:table:classic_b_t_50}, a média do fator de aproximação prático foi menor que $2.8$ para todos os valores de $n$. Já na Tabela~\ref{cap6:table:classic_b_t_100}, a média do fator de aproximação prático foi menor que $2.8$ apenas para $n \in \{{50, 100}\}$. Os valores dessa coluna, em ambas as tabelas, formam uma sequência não decrescente, sendo que os maiores valores foram de $2.79$ e $2.91$ nas tabelas \ref{cap6:table:classic_b_t_50} e \ref{cap6:table:classic_b_t_100}, respectivamente. Para a Tabela~\ref{cap6:table:classic_b_t_100}, os valores de mínimo e máximo para o fator de aproximação prático foram de $2.24$ e $2.96$, respectivamente. Já na Tabela~\ref{cap6:table:classic_b_t_50}, os valores de mínimo e máximo para o fator de aproximação prático foram de $2.00$ e $2.89$, respectivamente.

A razão entre o tamanho das soluções do experimento da Tabela~\ref{cap6:table:classic_cg_t_100} e o tamanho das soluções do experimento da Tabela~\ref{cap6:table:classic_cg_t_50} é menor que $1.28$ para todas as medidas usadas. {\revisaof Em ambas as tabelas, os valores de mínimo, média e máximo para o fator de aproximação prático são consideravelmente menores do que o fator de aproximação teórico.} Todos os valores dessas colunas, em ambas as tabelas, estão entre $1.18$ e $1.33$, o que mostra uma boa estabilidade do algoritmo. Em quase todos os casos, o algoritmo de $3$-aproximação encontra soluções com tamanho maior que o dobro do tamanho das soluções encontradas pelo algoritmo de $2$-aproximação.

As tabelas~\ref{cap6:table:classic_b_urt_50} e \ref{cap6:table:classic_b_urt_100} apresentam os resultados do algoritmo de $3$-aproximação para a Distância de Reversões, Transposições e Indels em Strings sem Sinais (Algoritmo~\ref{cap4:algorithm_reversals_transpositions}), que usa o conceito de {\it breakpoints}. A razão entre o tamanho das soluções do experimento da Tabela~\ref{cap6:table:classic_b_urt_100} e o tamanho das soluções do experimento da Tabela~\ref{cap6:table:classic_b_urt_50} é menor que $1.42$ para todas as medidas usadas. Para a Tabela~\ref{cap6:table:classic_b_urt_50}, {\revisaof os valores da sexta coluna (média do fator de aproximação prático) estão entre $2.53$ e $2.80$}. Já para a Tabela~\ref{cap6:table:classic_b_urt_100}, os valores para a média do fator de aproximação prático estão entre $2.70$ e $2.93$. Em ambas as tabelas, os valores dessa coluna formam uma sequência não decrescente. Na Tabela~\ref{cap6:table:classic_b_urt_50}, os valores de mínimo e máximo para o fator de aproximação prático foram de $2.17$ e $2.94$, respectivamente. Já na Tabela~\ref{cap6:table:classic_b_urt_100}, os valores de mínimo e máximo para o fator de aproximação prático foram de $2.32$ e $2.98$, respectivamente.

Por fim, as tabelas~\ref{cap6:table:classic_cg_srt_50} e \ref{cap6:table:classic_cg_srt_100} apresentam os resultados do algoritmo de $2$-aproximação para a Distância de Transposições, Reversões e Indels em Strings com Sinais (Algoritmo~\ref{cap4:algorithm:transp_rev_cycle_graph}), que usa o conceito de grafo de ciclos rotulado. Assim como os resultados do algoritmo de $2$-aproximação para a Distância de Transposições e Indels, essas tabelas mostram que o fator de aproximação prático é consideravelmente menor que o teórico. A razão entre o tamanho das soluções do experimento da Tabela~\ref{cap6:table:classic_cg_srt_100} e o tamanho das soluções do experimento da Tabela~\ref{cap6:table:classic_cg_srt_50} é menor que $1.43$ para todas as medidas usadas. Para a Tabela~\ref{cap6:table:classic_cg_srt_50}, os valores para a média do fator de aproximação prático estão entre $1.18$ e $1.22$. Para a Tabela~\ref{cap6:table:classic_cg_srt_100}, os valores para a média do fator de aproximação prático estão entre $1.17$ e $1.22$. Na Tabela~\ref{cap6:table:classic_cg_srt_50}, os valores de mínimo e máximo para o fator de aproximação prático foram de $1.16$ e $1.31$, respectivamente. Já na Tabela~\ref{cap6:table:classic_cg_srt_100}, os valores de mínimo e máximo para o fator de aproximação prático foram de $1.15$ e $1.24$, respectivamente. Esses valores mostram que o algoritmo possui boa estabilidade ao considerar o fator de aproximação prático.

\begin{table}[H]
\caption{\label{cap6:table:classic_b_ur_50} Resultados experimentais do algoritmo de $2$-aproximação para a Distância de Reversões e Indels em Strings sem Sinais (Algoritmo~\ref{cap4:algorithm_reversals}), considerando instâncias criadas com $k = \frac{n}{2}$.}
\centering
\begin{tabular}{@{}lrrrrrr@{}}
\toprule
 & \multicolumn{3}{c}{Operações} & \multicolumn{3}{c}{Aproximação} \\ \midrule
Tamanho & Mínimo & Média & Máximo & Mínimo & Média & Máximo \\ \midrule
50 & 34 & 42.62 & 51 & 1.40 & 1.68 & 1.92 \\
100 & 77 & 88.84 & 102 & 1.50 & 1.74 & 1.90 \\
150 & 122 & 134.69 & 147 & 1.60 & 1.76 & 1.89 \\
200 & 166 & 181.83 & 197 & 1.68 & 1.78 & 1.91 \\
250 & 207 & 227.88 & 249 & 1.69 & 1.79 & 1.88 \\
300 & 253 & 275.46 & 294 & 1.69 & 1.80 & 1.88 \\
350 & 296 & 321.41 & 342 & 1.72 & 1.80 & 1.88 \\
400 & 344 & 368.55 & 392 & 1.71 & 1.81 & 1.88 \\
450 & 388 & 415.28 & 451 & 1.74 & 1.81 & 1.88 \\
500 & 433 & 462.67 & 488 & 1.74 & 1.81 & 1.88 \\ \bottomrule
\end{tabular}
\end{table}

\begin{table}[H]
\caption{\label{cap6:table:classic_b_ur_100} Resultados experimentais do algoritmo de $2$-aproximação para a Distância de Reversões e Indels em Strings sem Sinais (Algoritmo~\ref{cap4:algorithm_reversals}), considerando instâncias criadas com $k = n$.}
\centering
\begin{tabular}{@{}lrrrrrr@{}}
\toprule
 & \multicolumn{3}{c}{Operações} & \multicolumn{3}{c}{Aproximação} \\ \midrule
Tamanho & Mínimo & Média & Máximo & Mínimo & Média & Máximo \\ \midrule
50 & 50 & 58.74 & 66 & 1.64 & 1.81 & 2.00 \\
100 & 107 & 120.18 & 130 & 1.75 & 1.88 & 1.97 \\
150 & 169 & 182.69 & 197 & 1.80 & 1.90 & 1.97 \\
200 & 229 & 244.79 & 260 & 1.83 & 1.91 & 1.98 \\
250 & 290 & 307.00 & 326 & 1.86 & 1.92 & 1.97 \\
300 & 350 & 369.37 & 389 & 1.88 & 1.93 & 1.97 \\
350 & 410 & 431.73 & 454 & 1.87 & 1.93 & 1.98 \\
400 & 472 & 494.47 & 517 & 1.89 & 1.94 & 1.98 \\
450 & 537 & 557.16 & 579 & 1.89 & 1.94 & 1.98 \\
500 & 597 & 620.05 & 642 & 1.90 & 1.94 & 1.97 \\ \bottomrule
\end{tabular}
\end{table}

\begin{table}[H]
\caption{\label{cap6:table:classic_b_t_50} Resultados experimentais do algoritmo de $3$-aproximação para a Distância de Transposições e Indels em Strings sem Sinais (Algoritmo~\ref{cap4:algorithm_transpositions}), considerando instâncias criadas com $k = \frac{n}{2}$.}
\centering
\begin{tabular}{@{}lrrrrrr@{}}
\toprule
 & \multicolumn{3}{c}{Operações} & \multicolumn{3}{c}{Aproximação} \\ \midrule
Tamanho & Mínimo & Média & Máximo & Mínimo & Média & Máximo \\ \midrule
50 & 34 & 43.08 & 53 & 2.00 & 2.39 & 2.79 \\
100 & 80 & 93.13 & 107 & 2.31 & 2.57 & 2.81 \\
150 & 126 & 143.17 & 158 & 2.44 & 2.66 & 2.85 \\
200 & 175 & 194.60 & 213 & 2.52 & 2.70 & 2.84 \\
250 & 225 & 245.86 & 266 & 2.56 & 2.73 & 2.87 \\
300 & 277 & 297.38 & 322 & 2.63 & 2.75 & 2.85 \\
350 & 322 & 348.58 & 373 & 2.64 & 2.76 & 2.89 \\
400 & 379 & 400.88 & 425 & 2.67 & 2.77 & 2.89 \\
450 & 424 & 451.99 & 480 & 2.69 & 2.79 & 2.87 \\
500 & 472 & 504.65 & 532 & 2.70 & 2.79 & 2.89 \\ \bottomrule
\end{tabular}
\end{table}

\begin{table}[H]
\caption{\label{cap6:table:classic_b_t_100} Resultados experimentais do algoritmo de $3$-aproximação para a Distância de Transposições e Indels em Strings sem Sinais (Algoritmo~\ref{cap4:algorithm_transpositions}), considerando instâncias criadas com $k = n$.}
\centering
\begin{tabular}{@{}lrrrrrr@{}}
\toprule
 & \multicolumn{3}{c}{Operações} & \multicolumn{3}{c}{Aproximação} \\ \midrule
Tamanho & Mínimo & Média & Máximo & Mínimo & Média & Máximo \\ \midrule
50 & 47 & 56.96 & 67 & 2.24 & 2.58 & 2.86 \\
100 & 108 & 119.61 & 130 & 2.50 & 2.75 & 2.91 \\
150 & 171 & 182.95 & 197 & 2.63 & 2.81 & 2.93 \\
200 & 232 & 246.11 & 262 & 2.68 & 2.84 & 2.93 \\
250 & 289 & 309.89 & 328 & 2.75 & 2.86 & 2.94 \\
300 & 354 & 373.76 & 389 & 2.78 & 2.88 & 2.95 \\
350 & 418 & 437.47 & 456 & 2.79 & 2.89 & 2.95 \\
400 & 481 & 501.26 & 520 & 2.83 & 2.90 & 2.96 \\
450 & 545 & 565.42 & 585 & 2.84 & 2.91 & 2.96 \\
500 & 609 & 629.45 & 649 & 2.86 & 2.91 & 2.96 \\ \bottomrule
\end{tabular}
\end{table}

\begin{table}[H]
\caption{\label{cap6:table:classic_cg_t_50} Resultados experimentais do algoritmo de $2$-aproximação para a Distância de Transposições e Indels em Strings sem Sinais (Algoritmo~\ref{cap4:algorithm:transp_cycle_graph}), considerando instâncias criadas com $k = \frac{n}{2}$.}
\centering
\begin{tabular}{@{}lrrrrrr@{}}
\toprule
 & \multicolumn{3}{c}{Operações} & \multicolumn{3}{c}{Aproximação} \\ \midrule
Tamanho & Mínimo & Média & Máximo & Mínimo & Média & Máximo \\ \midrule
50 & 22 & 23.66 & 25 & 1.19 & 1.23 & 1.28 \\
100 & 47 & 49.19 & 52 & 1.24 & 1.26 & 1.27 \\
150 & 69 & 71.53 & 75 & 1.21 & 1.22 & 1.23 \\
200 & 87 & 90.08 & 94 & 1.19 & 1.21 & 1.23 \\
250 & 116 & 119.30 & 124 & 1.19 & 1.21 & 1.22 \\
300 & 139 & 143.85 & 149 & 1.21 & 1.23 & 1.24 \\
350 & 163 & 167.17 & 171 & 1.20 & 1.21 & 1.21 \\
400 & 196 & 202.08 & 207 & 1.22 & 1.24 & 1.25 \\
450 & 207 & 212.93 & 217 & 1.18 & 1.18 & 1.19 \\
500 & 233 & 239.49 & 246 & 1.21 & 1.21 & 1.22 \\ \bottomrule
\end{tabular}
\end{table}

\begin{table}[H]
\caption{\label{cap6:table:classic_cg_t_100} Resultados experimentais do algoritmo de $2$-aproximação para a Distância de Transposições e Indels em Strings sem Sinais (Algoritmo~\ref{cap4:algorithm:transp_cycle_graph}), considerando instâncias criadas com $k = n$.}
\centering
\begin{tabular}{@{}lrrrrrr@{}}
\toprule
 & \multicolumn{3}{c}{Operações} & \multicolumn{3}{c}{Aproximação} \\ \midrule
Tamanho & Mínimo & Média & Máximo & Mínimo & Média & Máximo \\ \midrule
50 & 28 & 29.97 & 32 & 1.28 & 1.31 & 1.33 \\
100 & 53 & 56.87 & 61 & 1.25 & 1.28 & 1.32 \\
150 & 78 & 81.45 & 85 & 1.25 & 1.27 & 1.30 \\
200 & 104 & 109.31 & 114 & 1.26 & 1.27 & 1.29 \\
250 & 123 & 128.52 & 133 & 1.21 & 1.22 & 1.24 \\
300 & 161 & 166.38 & 172 & 1.26 & 1.27 & 1.29 \\
350 & 178 & 185.29 & 192 & 1.22 & 1.23 & 1.24 \\
400 & 214 & 219.87 & 226 & 1.23 & 1.24 & 1.25 \\
450 & 234 & 241.21 & 248 & 1.22 & 1.23 & 1.24 \\
500 & 267 & 273.73 & 282 & 1.24 & 1.25 & 1.26 \\ \bottomrule
\end{tabular}
\end{table}

\begin{table}[H]
\caption{\label{cap6:table:classic_b_urt_50} Resultados experimentais do algoritmo de $3$-aproximação para a Distância de Reversões, Transposições e Indels em Strings sem Sinais (Algoritmo~\ref{cap4:algorithm_reversals_transpositions}), considerando instâncias criadas com $k = \frac{n}{2}$.}
\centering
\begin{tabular}{@{}lrrrrrr@{}}
\toprule
 & \multicolumn{3}{c}{Operações} & \multicolumn{3}{c}{Aproximação} \\ \midrule
Tamanho & Mínimo & Média & Máximo & Mínimo & Média & Máximo \\ \midrule
50 & 36 & 44.63 & 55 & 2.17 & 2.53 & 2.94 \\
100 & 82 & 93.74 & 106 & 2.40 & 2.66 & 2.89 \\
150 & 127 & 142.47 & 158 & 2.51 & 2.71 & 2.90 \\
200 & 172 & 192.76 & 212 & 2.55 & 2.74 & 2.92 \\
250 & 218 & 241.93 & 262 & 2.61 & 2.76 & 2.88 \\
300 & 271 & 293.10 & 315 & 2.61 & 2.77 & 2.88 \\
350 & 317 & 342.29 & 363 & 2.67 & 2.78 & 2.89 \\
400 & 365 & 393.17 & 419 & 2.67 & 2.79 & 2.89 \\
450 & 418 & 441.94 & 466 & 2.68 & 2.79 & 2.89 \\
500 & 466 & 492.81 & 519 & 2.71 & 2.80 & 2.90 \\ \bottomrule
\end{tabular}
\end{table}

\begin{table}[H]
\caption{\label{cap6:table:classic_b_urt_100} Resultados experimentais do algoritmo de $3$-aproximação para a Distância de Reversões, Transposições e Indels em Strings sem Sinais (Algoritmo~\ref{cap4:algorithm_reversals_transpositions}), considerando instâncias criadas com $k = n$.}
\centering
\begin{tabular}{@{}lrrrrrr@{}}
\toprule
 & \multicolumn{3}{c}{Operações} & \multicolumn{3}{c}{Aproximação} \\ \midrule
Tamanho & Mínimo & Média & Máximo & Mínimo & Média & Máximo \\ \midrule
50 & 51 & 58.96 & 68 & 2.32 & 2.70 & 2.95 \\
100 & 109 & 121.74 & 134 & 2.65 & 2.81 & 2.95 \\
150 & 169 & 185.06 & 198 & 2.73 & 2.86 & 2.96 \\
200 & 229 & 248.19 & 260 & 2.76 & 2.88 & 2.97 \\
250 & 295 & 311.71 & 326 & 2.80 & 2.90 & 2.97 \\
300 & 354 & 375.19 & 389 & 2.82 & 2.91 & 2.96 \\
350 & 418 & 438.20 & 458 & 2.84 & 2.91 & 2.98 \\
400 & 484 & 502.15 & 521 & 2.86 & 2.92 & 2.97 \\
450 & 536 & 565.53 & 588 & 2.86 & 2.92 & 2.97 \\
500 & 606 & 629.18 & 648 & 2.88 & 2.93 & 2.98 \\ \bottomrule
\end{tabular}
\end{table}

\begin{table}[H]
\caption{\label{cap6:table:classic_cg_srt_50} Resultados experimentais do algoritmo de $2$-aproximação para a Distância de Transposições, Reversões e Indels em Strings com Sinais (Algoritmo~\ref{cap4:algorithm:transp_rev_cycle_graph}), considerando instâncias criadas com $k = \frac{n}{2}$.}
\centering
\begin{tabular}{@{}lrrrrrr@{}}
\toprule
 & \multicolumn{3}{c}{Operações} & \multicolumn{3}{c}{Aproximação} \\ \midrule
Tamanho & Mínimo & Média & Máximo & Mínimo & Média & Máximo \\ \midrule
50 & 17 & 19.33 & 22 & 1.17 & 1.22 & 1.31 \\
100 & 45 & 48.37 & 52 & 1.16 & 1.18 & 1.22 \\
150 & 65 & 69.34 & 74 & 1.20 & 1.22 & 1.25 \\
200 & 90 & 95.24 & 100 & 1.19 & 1.21 & 1.23 \\
250 & 99 & 105.86 & 112 & 1.17 & 1.18 & 1.19 \\
300 & 128 & 134.78 & 140 & 1.19 & 1.20 & 1.21 \\
350 & 161 & 168.54 & 176 & 1.21 & 1.22 & 1.23 \\
400 & 183 & 191.17 & 199 & 1.21 & 1.22 & 1.24 \\
450 & 194 & 202.28 & 211 & 1.18 & 1.19 & 1.20 \\
500 & 230 & 238.14 & 246 & 1.19 & 1.20 & 1.20 \\ \bottomrule
\end{tabular}
\end{table}

\begin{table}[H]
\caption{\label{cap6:table:classic_cg_srt_100} Resultados experimentais do algoritmo de $2$-aproximação para a Distância de Transposições, Reversões e Indels em Strings com Sinais (Algoritmo~\ref{cap4:algorithm:transp_rev_cycle_graph}), considerando instâncias criadas com $k = n$.}
\centering
\begin{tabular}{@{}lrrrrrr@{}}
\toprule
 & \multicolumn{3}{c}{Operações} & \multicolumn{3}{c}{Aproximação} \\ \midrule
Tamanho & Mínimo & Média & Máximo & Mínimo & Média & Máximo \\ \midrule
50 & 23 & 27.59 & 30 & 1.15 & 1.17 & 1.21 \\
100 & 51 & 55.41 & 60 & 1.16 & 1.18 & 1.21 \\
150 & 82 & 88.08 & 94 & 1.21 & 1.22 & 1.24 \\
200 & 107 & 112.87 & 120 & 1.17 & 1.18 & 1.19 \\
250 & 132 & 139.17 & 146 & 1.18 & 1.19 & 1.20 \\
300 & 159 & 165.79 & 173 & 1.18 & 1.20 & 1.21 \\
350 & 188 & 195.62 & 205 & 1.18 & 1.19 & 1.21 \\
400 & 210 & 219.26 & 227 & 1.18 & 1.19 & 1.20 \\
450 & 236 & 245.37 & 255 & 1.18 & 1.19 & 1.20 \\
500 & 260 & 269.07 & 279 & 1.16 & 1.17 & 1.18 \\ \bottomrule
\end{tabular}
\end{table}

\section{Experimentos com Instâncias Intergênicas}\label{cap6:section:intergenicos}

As tabelas apresentadas nesta seção seguem um padrão similar às tabelas da seção anterior, mas apresentam resultados para os algoritmos do Capítulo~\ref{cap:label:intergenicos} e utilizam os conjuntos de instâncias intergênicas $DSI_{n, k, t}^{\M}$. Como esperado, em todas as tabelas, os valores referentes à quantidade de operações usadas na solução aumentam à medida que o valor de $n$ aumenta.

As tabelas~\ref{cap6:table:inter_b_ur_50} e \ref{cap6:table:inter_b_ur_100} apresentam os resultados do algoritmo de $4$-aproximação para o problema da Distância de Reversões e Indels Intergênicos em Strings sem Sinais (Algoritmo~\ref{cap5:algorithm_reversals_breakpoints}), que usa o conceito de {\it breakpoints} intergênicos. 
Apesar do fato de que as instâncias sintéticas da Tabela~\ref{cap6:table:inter_b_ur_100} foram criadas usando o dobro de operações em relação às instâncias da Tabela~\ref{cap6:table:inter_b_ur_50}, a razão entre o tamanho das soluções do experimento da Tabela~\ref{cap6:table:inter_b_ur_100} e o tamanho das soluções do experimento da Tabela~\ref{cap6:table:inter_b_ur_50} é menor que $1.36$ para todas as medidas usadas. 
Para a Tabela~\ref{cap6:table:inter_b_ur_50}, a média do fator de aproximação prático está entre $2.00$ e $2.03$. Já para a Tabela~\ref{cap6:table:inter_b_ur_100}, os valores dessa coluna estão entre $2.01$ e $2.02$. Na Tabela~\ref{cap6:table:inter_b_ur_50}, os valores de mínimo e máximo para o fator de aproximação prático foram de $1.85$ e $2.20$, respectivamente. Já na Tabela~\ref{cap6:table:inter_b_ur_100}, os valores de mínimo e máximo para o fator de aproximação prático foram de $1.94$ e $2.13$, respectivamente. Essas colunas mostram que o fator de aproximação prático é um pouco maior que $2$ no pior cenário, enquanto o fator de aproximação teórico é $4$.

As tabelas~\ref{cap6:table:inter_b_urt_50} e \ref{cap6:table:inter_b_urt_100} apresentam os resultados do algoritmo de $6$-aproximação para o problema da Distância de Reversões, Transposições e Indels Intergênicos em Strings sem Sinais (Algoritmo~\ref{cap5:algorithm_reversals_breakpoints}), que usa o conceito de {\it breakpoints} intergênicos. Como o Algoritmo~\ref{cap5:algorithm_reversals_breakpoints} usa apenas reversões e {\it indels}, implementamos uma modificação desse algoritmo que aplica, sempre que possível, uma transposição de acordo com o Lema~\ref{cap5:lemma:breakpoints_transpositions}, mantendo o fator de aproximação igual a $6$. {\revisaof Apesar do fator de aproximação permanecer o mesmo}, essa modificação garante que o algoritmo use menos operações quando o genoma de origem possui apenas {\it strips} crescentes.

A razão entre o tamanho das soluções do experimento da Tabela~\ref{cap6:table:inter_b_urt_100} e o tamanho das soluções do experimento da Tabela~\ref{cap6:table:inter_b_urt_50} é menor que $1.35$ para todas as medidas usadas. Para a Tabela~\ref{cap6:table:inter_b_urt_50}, a média do fator de aproximação prático está entre $2.99$ e $3.06$. Já para a Tabela~\ref{cap6:table:inter_b_urt_100}, os valores dessa coluna estão entre $2.98$ e $3.03$. Na Tabela~\ref{cap6:table:inter_b_urt_50}, os valores de mínimo e máximo para o fator de aproximação prático foram de $2.81$ e $3.25$, respectivamente. Já na Tabela~\ref{cap6:table:inter_b_urt_100}, os valores de mínimo e máximo para o fator de aproximação prático foram de $2.82$ e $3.19$, respectivamente. Essas colunas mostram que o fator de aproximação prático é um pouco maior que $3$ no pior cenário, enquanto o fator de aproximação teórico é $6$.

As tabelas~\ref{cap6:table:inter_b_t_50} e \ref{cap6:table:inter_b_t_100} apresentam os resultados do algoritmo de $4.5$-aproximação para o problema da Distância de Transposições e Indels Intergênicos em Strings sem Sinais (Algoritmo~\ref{cap5:algorithm_transpositions_breakpoints}), que usa {\it breakpoints} intergênicos, enquanto as tabelas~\ref{cap6:table:inter_cg_t_50} e \ref{cap6:table:inter_cg_t_100} apresentam os resultados do algoritmo de $4$-aproximação para o mesmo problema (Algoritmo~\ref{cap5:alg:4_transp_cycle_graph}), que usa o conceito de grafo de ciclos rotulado e ponderado. Essas tabelas mostram que o algoritmo de $4$-aproximação, que usa grafo de ciclos rotulado e ponderado, possui resultados práticos melhores que o algoritmo de $4.5$-aproximação, que usa {\it breakpoints} intergênicos. Agora, discutiremos com mais detalhes os resultados para cada algoritmo.

A razão entre o tamanho das soluções do experimento da Tabela~\ref{cap6:table:inter_b_t_100} e o tamanho das soluções do experimento da Tabela~\ref{cap6:table:inter_b_t_50} é menor que $1.29$ para todas as medidas usadas. Para a Tabela~\ref{cap6:table:inter_b_t_50}, a média do fator de aproximação prático está entre $3.18$ e $3.35$. Já para a Tabela~\ref{cap6:table:inter_b_t_100}, os valores dessa coluna estão entre $3.13$ e $3.25$. Na Tabela~\ref{cap6:table:inter_b_t_50}, os valores de mínimo e máximo para o fator de aproximação prático foram de $2.88$ e $3.50$, respectivamente. Já na Tabela~\ref{cap6:table:inter_b_t_100}, os valores de mínimo e máximo para o fator de aproximação prático foram de $2.91$ e $3.38$, respectivamente. Essas colunas mostram que o fator de aproximação prático foi sempre menor que $3.4$, enquanto o fator de aproximação teórico é $4.5$.

A razão entre o tamanho das soluções do experimento da Tabela~\ref{cap6:table:inter_cg_t_100} e o tamanho das soluções do experimento da Tabela~\ref{cap6:table:inter_cg_t_50} é menor que $1.22$ para todas as medidas usadas. Para a Tabela~\ref{cap6:table:inter_cg_t_50}, a média do fator de aproximação prático está entre $2.31$ e $2.38$. Já para a Tabela~\ref{cap6:table:inter_cg_t_100}, os valores dessa coluna estão entre $2.66$ e $2.72$. Na Tabela~\ref{cap6:table:inter_cg_t_50}, os valores de mínimo e máximo para o fator de aproximação prático foram de $2.00$ e $2.69$, respectivamente. Já na Tabela~\ref{cap6:table:inter_cg_t_100}, os valores de mínimo e máximo para o fator de aproximação prático foram de $2.26$ e $3.09$, respectivamente. Essas colunas mostram que o fator de aproximação prático foi sempre menor que $3.1$, enquanto o fator de aproximação teórico é $4$.

As tabelas~\ref{cap6:table:inter_cg_sr_50} e \ref{cap6:table:inter_cg_sr_100} apresentam os resultados do algoritmo de $2.5$-aproximação para o problema da Distância de Reversões e Indels Intergênicos em Strings com Sinais (Algoritmo~\ref{cap5:alg:2_5-approx}), que usa o conceito de grafo de ciclos rotulado e ponderado. 
A razão entre o tamanho das soluções do experimento da Tabela~\ref{cap6:table:inter_cg_sr_100} e o tamanho das soluções do experimento da Tabela~\ref{cap6:table:inter_cg_sr_50} é menor que $1.46$ para todas as medidas usadas. Para a Tabela~\ref{cap6:table:inter_cg_sr_50}, a média do fator de aproximação prático está entre $1.28$ e $1.30$. Já para a Tabela~\ref{cap6:table:inter_cg_sr_100}, os valores dessa coluna estão entre $1.42$ e $1.45$. Na Tabela~\ref{cap6:table:inter_cg_sr_50}, os valores de mínimo e máximo para o fator de aproximação prático foram de $1.03$ e $1.60$, respectivamente. Já na Tabela~\ref{cap6:table:inter_cg_sr_100}, os valores de mínimo e máximo para o fator de aproximação prático foram de $1.19$ e $1.71$, respectivamente. Essas colunas mostram que o fator de aproximação prático é sempre menor que $1.8$, enquanto o fator de aproximação teórico é $2.5$.

Por fim, as tabelas~\ref{cap6:table:inter_cg_srt_50} e \ref{cap6:table:inter_cg_srt_100} apresentam os resultados do algoritmo de $4$-aproximação para o problema da Distância de Reversões, Transposições e Indels Intergênicos em Strings com Sinais (Algoritmo~\ref{cap5:alg:4_transp_rev_cycle_graph}), que usa o conceito de grafo de ciclos ponderado e rotulado. 
A razão entre o tamanho das soluções do experimento da Tabela~\ref{cap6:table:inter_cg_srt_100} e o tamanho das soluções do experimento da Tabela~\ref{cap6:table:inter_cg_srt_50} é menor que $1.38$ para todas as medidas usadas. Para a Tabela~\ref{cap6:table:inter_cg_srt_50}, a média do fator de aproximação prático está entre $2.63$ e $2.73$. Já para a Tabela~\ref{cap6:table:inter_cg_srt_100}, os valores dessa coluna estão entre $2.84$ e $2.92$. Na Tabela~\ref{cap6:table:inter_cg_srt_50}, os valores de mínimo e máximo para o fator de aproximação prático foram de $2.11$ e $3.27$, respectivamente. Já na Tabela~\ref{cap6:table:inter_cg_srt_100}, os valores de mínimo e máximo para o fator de aproximação prático foram de $2.27$ e $3.35$, respectivamente. Essas colunas mostram que o fator de aproximação prático é sempre menor ou igual a $3.4$, enquanto o fator de aproximação teórico é $4$.

Com esses experimentos, podemos concluir que os algoritmos que usam grafo de ciclos são melhores tanto no fator de aproximação teórico quanto no fator de aproximação prático.

\begin{table}[H]
\caption{\label{cap6:table:inter_b_ur_50} Resultados experimentais do algoritmo de $4$-aproximação para o problema da Distância de Reversões e Indels Intergênicos em Strings sem Sinais (Algoritmo~\ref{cap5:algorithm_reversals_breakpoints}), considerando instâncias criadas com $k = \frac{n}{2}$.}
\centering
\begin{tabular}{@{}lrrrrrr@{}}
\toprule
 & \multicolumn{3}{c}{Operações} & \multicolumn{3}{c}{Aproximação} \\ \midrule
Tamanho & Mínimo & Média & Máximo & Mínimo & Média & Máximo \\ \midrule
50 & 42 & 51.41 & 61 & 1.85 & 2.00 & 2.20 \\
100 & 91 & 103.47 & 116 & 1.89 & 2.02 & 2.14 \\
150 & 141 & 154.62 & 168 & 1.93 & 2.02 & 2.11 \\
200 & 191 & 206.93 & 225 & 1.95 & 2.02 & 2.10 \\
250 & 238 & 257.85 & 274 & 1.96 & 2.02 & 2.10 \\
300 & 290 & 310.60 & 334 & 1.97 & 2.02 & 2.08 \\
350 & 329 & 361.58 & 386 & 1.97 & 2.02 & 2.09 \\
400 & 388 & 413.69 & 435 & 1.98 & 2.02 & 2.09 \\
450 & 441 & 465.22 & 492 & 1.97 & 2.02 & 2.08 \\
500 & 488 & 517.61 & 545 & 1.98 & 2.03 & 2.08 \\ \bottomrule
\end{tabular}
\end{table}

\begin{table}[H]
\caption{\label{cap6:table:inter_b_ur_100} Resultados experimentais do algoritmo de $4$-aproximação para o problema da Distância de Reversões e Indels Intergênicos em Strings sem Sinais (Algoritmo~\ref{cap5:algorithm_reversals_breakpoints}), considerando instâncias criadas com $k = n$.}
\centering
\begin{tabular}{@{}lrrrrrr@{}}
\toprule
 & \multicolumn{3}{c}{Operações} & \multicolumn{3}{c}{Aproximação} \\ \midrule
Tamanho & Mínimo & Média & Máximo & Mínimo & Média & Máximo \\ \midrule
50 & 57 & 65.45 & 72 & 1.94 & 2.01 & 2.13 \\
100 & 118 & 129.47 & 138 & 1.97 & 2.01 & 2.08 \\
150 & 183 & 194.14 & 206 & 1.98 & 2.02 & 2.09 \\
200 & 245 & 258.71 & 273 & 1.98 & 2.02 & 2.07 \\
250 & 306 & 322.95 & 343 & 1.99 & 2.02 & 2.06 \\
300 & 367 & 387.27 & 403 & 1.99 & 2.02 & 2.06 \\
350 & 435 & 451.84 & 473 & 1.99 & 2.02 & 2.05 \\
400 & 496 & 516.31 & 534 & 2.00 & 2.02 & 2.06 \\
450 & 561 & 580.83 & 601 & 2.00 & 2.02 & 2.05 \\
500 & 622 & 645.82 & 666 & 2.00 & 2.02 & 2.05 \\ \bottomrule
\end{tabular}
\end{table}

\begin{table}[H]
\caption{\label{cap6:table:inter_b_urt_50} Resultados experimentais do algoritmo de $6$-aproximação para o problema da Distância de Reversões, Transposições e Indels Intergênicos em Strings sem Sinais (Algoritmo~\ref{cap5:algorithm_reversals_breakpoints}), considerando instâncias criadas com $k = \frac{n}{2}$.}
\centering
\begin{tabular}{@{}lrrrrrr@{}}
\toprule
 & \multicolumn{3}{c}{Operações} & \multicolumn{3}{c}{Aproximação} \\ \midrule
Tamanho & Mínimo & Média & Máximo & Mínimo & Média & Máximo \\ \midrule
50 & 43 & 53.32 & 63 & 2.81 & 2.99 & 3.25 \\
100 & 96 & 107.25 & 121 & 2.86 & 3.03 & 3.18 \\
150 & 148 & 160.38 & 176 & 2.90 & 3.04 & 3.16 \\
200 & 194 & 215.15 & 235 & 2.94 & 3.04 & 3.17 \\
250 & 246 & 268.05 & 291 & 2.97 & 3.05 & 3.14 \\
300 & 300 & 323.33 & 343 & 2.98 & 3.05 & 3.16 \\
350 & 345 & 376.30 & 399 & 2.98 & 3.05 & 3.13 \\
400 & 408 & 431.22 & 455 & 2.99 & 3.05 & 3.12 \\
450 & 459 & 483.79 & 509 & 2.99 & 3.05 & 3.12 \\
500 & 507 & 538.87 & 570 & 2.99 & 3.06 & 3.13 \\ \bottomrule
\end{tabular}
\end{table}

\begin{table}[H]
\caption{\label{cap6:table:inter_b_urt_100} Resultados experimentais do algoritmo de $6$-aproximação para o problema da Distância de Reversões, Transposições e Indels Intergênicos em Strings sem Sinais (Algoritmo~\ref{cap5:algorithm_reversals_breakpoints}), considerando instâncias criadas com $k = n$.}
\centering
\begin{tabular}{@{}lrrrrrr@{}}
\toprule
 & \multicolumn{3}{c}{Operações} & \multicolumn{3}{c}{Aproximação} \\ \midrule
Tamanho & Mínimo & Média & Máximo & Mínimo & Média & Máximo \\ \midrule
50 & 58 & 65.72 & 71 & 2.82 & 2.98 & 3.19 \\
100 & 120 & 130.73 & 140 & 2.95 & 3.01 & 3.14 \\
150 & 184 & 195.89 & 208 & 2.97 & 3.02 & 3.11 \\
200 & 247 & 260.96 & 273 & 2.97 & 3.02 & 3.13 \\
250 & 312 & 325.98 & 341 & 2.98 & 3.02 & 3.13 \\
300 & 372 & 391.13 & 406 & 2.98 & 3.03 & 3.08 \\
350 & 439 & 455.98 & 475 & 2.99 & 3.03 & 3.08 \\
400 & 501 & 521.34 & 539 & 2.99 & 3.03 & 3.10 \\
450 & 561 & 586.37 & 603 & 3.00 & 3.03 & 3.08 \\
500 & 627 & 651.29 & 672 & 2.99 & 3.03 & 3.07 \\ \bottomrule
\end{tabular}
\end{table}

\begin{table}[H]
\caption{\label{cap6:table:inter_b_t_50} Resultados experimentais do algoritmo de $4.5$-aproximação para o problema da Distância de Transposições e Indels Intergênicos em Strings sem Sinais (Algoritmo~\ref{cap5:algorithm_transpositions_breakpoints}), considerando instâncias criadas com $k = \frac{n}{2}$.}
\centering
\begin{tabular}{@{}lrrrrrr@{}}
\toprule
 & \multicolumn{3}{c}{Operações} & \multicolumn{3}{c}{Aproximação} \\ \midrule
Tamanho & Mínimo & Média & Máximo & Mínimo & Média & Máximo \\ \midrule
50 & 46 & 57.95 & 67 & 2.88 & 3.18 & 3.50 \\
100 & 102 & 118.29 & 133 & 3.06 & 3.26 & 3.47 \\
150 & 160 & 178.31 & 196 & 3.11 & 3.30 & 3.46 \\
200 & 220 & 239.26 & 259 & 3.17 & 3.31 & 3.44 \\
250 & 280 & 300.18 & 326 & 3.20 & 3.33 & 3.48 \\
300 & 332 & 361.69 & 391 & 3.23 & 3.33 & 3.46 \\
350 & 395 & 422.07 & 446 & 3.22 & 3.34 & 3.45 \\
400 & 456 & 483.99 & 507 & 3.26 & 3.35 & 3.43 \\
450 & 516 & 544.27 & 569 & 3.26 & 3.35 & 3.46 \\
500 & 570 & 606.44 & 637 & 3.26 & 3.35 & 3.45 \\ \bottomrule
\end{tabular}
\end{table}

\begin{table}[H]
\caption{\label{cap6:table:inter_b_t_100} Resultados experimentais do algoritmo de $4.5$-aproximação para o problema da Distância de Transposições e Indels Intergênicos em Strings sem Sinais (Algoritmo~\ref{cap5:algorithm_transpositions_breakpoints}), considerando instâncias criadas com $k = n$.}
\centering
\begin{tabular}{@{}lrrrrrr@{}}
\toprule
 & \multicolumn{3}{c}{Operações} & \multicolumn{3}{c}{Aproximação} \\ \midrule
Tamanho & Mínimo & Média & Máximo & Mínimo & Média & Máximo \\ \midrule
50 & 59 & 69.30 & 77 & 2.91 & 3.13 & 3.38 \\
100 & 129 & 139.10 & 151 & 3.00 & 3.19 & 3.37 \\
150 & 196 & 209.36 & 220 & 3.07 & 3.21 & 3.36 \\
200 & 266 & 279.44 & 293 & 3.10 & 3.22 & 3.38 \\
250 & 333 & 349.88 & 367 & 3.11 & 3.23 & 3.33 \\
300 & 401 & 420.03 & 436 & 3.14 & 3.23 & 3.32 \\
350 & 470 & 490.50 & 509 & 3.15 & 3.24 & 3.32 \\
400 & 543 & 560.49 & 581 & 3.17 & 3.24 & 3.32 \\
450 & 607 & 631.33 & 655 & 3.16 & 3.24 & 3.33 \\
500 & 680 & 701.51 & 728 & 3.17 & 3.25 & 3.32 \\ \bottomrule
\end{tabular}
\end{table}

\begin{table}[H]
\caption{\label{cap6:table:inter_cg_t_50} Resultados experimentais do algoritmo de $4$-aproximação para o problema da Distância de Transposições e Indels Intergênicos em Strings sem Sinais (Algoritmo~\ref{cap5:alg:4_transp_cycle_graph}), considerando instâncias criadas com $k = \frac{n}{2}$.}
\centering
\begin{tabular}{@{}lrrrrrr@{}}
\toprule
 & \multicolumn{3}{c}{Operações} & \multicolumn{3}{c}{Aproximação} \\ \midrule
Tamanho & Mínimo & Média & Máximo & Mínimo & Média & Máximo \\ \midrule
50 & 32 & 40.23 & 48 & 2.00 & 2.31 & 2.69 \\
100 & 70 & 81.44 & 94 & 2.11 & 2.36 & 2.69 \\
150 & 107 & 122.32 & 140 & 2.17 & 2.36 & 2.59 \\
200 & 148 & 163.50 & 180 & 2.16 & 2.37 & 2.55 \\
250 & 187 & 205.35 & 226 & 2.22 & 2.37 & 2.55 \\
300 & 226 & 246.71 & 269 & 2.24 & 2.37 & 2.55 \\
350 & 264 & 287.71 & 318 & 2.23 & 2.37 & 2.55 \\
400 & 305 & 330.16 & 352 & 2.26 & 2.38 & 2.51 \\
450 & 347 & 370.77 & 395 & 2.23 & 2.38 & 2.51 \\
500 & 386 & 412.84 & 440 & 2.24 & 2.38 & 2.51 \\ \bottomrule
\end{tabular}
\end{table}

\begin{table}[H]
\caption{\label{cap6:table:inter_cg_t_100} Resultados experimentais do algoritmo de $4$-aproximação para o problema da Distância de Transposições e Indels Intergênicos em Strings sem Sinais (Algoritmo~\ref{cap5:alg:4_transp_cycle_graph}), considerando instâncias criadas com $k = n$.}
\centering
\begin{tabular}{@{}lrrrrrr@{}}
\toprule
 & \multicolumn{3}{c}{Operações} & \multicolumn{3}{c}{Aproximação} \\ \midrule
Tamanho & Mínimo & Média & Máximo & Mínimo & Média & Máximo \\ \midrule
50 & 39 & 46.98 & 54 & 2.26 & 2.66 & 3.09 \\
100 & 84 & 93.74 & 102 & 2.41 & 2.70 & 2.98 \\
150 & 130 & 140.80 & 155 & 2.46 & 2.71 & 2.99 \\
200 & 176 & 187.63 & 199 & 2.52 & 2.72 & 2.94 \\
250 & 217 & 234.31 & 247 & 2.53 & 2.72 & 2.93 \\
300 & 266 & 281.43 & 300 & 2.54 & 2.72 & 2.91 \\
350 & 312 & 328.17 & 347 & 2.53 & 2.72 & 2.89 \\
400 & 357 & 374.88 & 399 & 2.57 & 2.72 & 2.90 \\
450 & 399 & 421.64 & 445 & 2.57 & 2.72 & 2.89 \\
500 & 450 & 468.51 & 486 & 2.60 & 2.72 & 2.89 \\ \bottomrule
\end{tabular}
\end{table}

\begin{table}[H]
\caption{\label{cap6:table:inter_cg_sr_50} Resultados experimentais do algoritmo de $2.5$-aproximação para o problema da Distância de Reversões e Indels Intergênicos em Strings com Sinais (Algoritmo~\ref{cap5:alg:2_5-approx}), considerando instâncias criadas com $k = \frac{n}{2}$.}
\centering
\begin{tabular}{@{}lrrrrrr@{}}
\toprule
 & \multicolumn{3}{c}{Operações} & \multicolumn{3}{c}{Aproximação} \\ \midrule
Tamanho & Mínimo & Média & Máximo & Mínimo & Média & Máximo \\ \midrule
50 & 31 & 40.52 & 51 & 1.03 & 1.28 & 1.60 \\
100 & 69 & 81.77 & 98 & 1.11 & 1.29 & 1.50 \\
150 & 105 & 122.56 & 141 & 1.14 & 1.30 & 1.47 \\
200 & 145 & 164.21 & 185 & 1.17 & 1.30 & 1.43 \\
250 & 180 & 205.22 & 227 & 1.17 & 1.30 & 1.43 \\
300 & 224 & 247.02 & 271 & 1.20 & 1.30 & 1.42 \\
350 & 254 & 288.14 & 311 & 1.22 & 1.30 & 1.44 \\
400 & 305 & 329.81 & 360 & 1.20 & 1.30 & 1.43 \\
450 & 340 & 370.39 & 401 & 1.22 & 1.30 & 1.38 \\
500 & 385 & 411.94 & 446 & 1.21 & 1.30 & 1.38 \\ \bottomrule
\end{tabular}
\end{table}

\begin{table}[H]
\caption{\label{cap6:table:inter_cg_sr_100} Resultados experimentais do algoritmo de $2.5$-aproximação para o problema da Distância de Reversões e Indels Intergênicos em Strings com Sinais (Algoritmo~\ref{cap5:alg:2_5-approx}), considerando instâncias criadas com $k = n$.}
\centering
\begin{tabular}{@{}lrrrrrr@{}}
\toprule
 & \multicolumn{3}{c}{Operações} & \multicolumn{3}{c}{Aproximação} \\ \midrule
Tamanho & Mínimo & Média & Máximo & Mínimo & Média & Máximo \\ \midrule
50 & 45 & 55.57 & 67 & 1.19 & 1.42 & 1.71 \\
100 & 95 & 110.99 & 126 & 1.26 & 1.43 & 1.67 \\
150 & 150 & 167.06 & 184 & 1.29 & 1.44 & 1.60 \\
200 & 199 & 222.75 & 244 & 1.30 & 1.44 & 1.59 \\
250 & 255 & 278.39 & 299 & 1.33 & 1.44 & 1.58 \\
300 & 311 & 334.73 & 359 & 1.33 & 1.45 & 1.56 \\
350 & 361 & 390.63 & 415 & 1.36 & 1.45 & 1.55 \\
400 & 418 & 446.35 & 484 & 1.36 & 1.45 & 1.55 \\
450 & 475 & 502.24 & 537 & 1.38 & 1.45 & 1.55 \\
500 & 529 & 558.20 & 588 & 1.36 & 1.45 & 1.53 \\ \bottomrule
\end{tabular}
\end{table}

\begin{table}[H]
\caption{\label{cap6:table:inter_cg_srt_50} Resultados experimentais do algoritmo de $4$-aproximação para o problema da Distância de Reversões, Transposições e Indels Intergênicos em Strings com Sinais (Algoritmo~\ref{cap5:alg:4_transp_rev_cycle_graph}), considerando instâncias criadas com $k = \frac{n}{2}$.}
\centering
\begin{tabular}{lrrrrrr}
\toprule
 & \multicolumn{3}{c}{Operações} & \multicolumn{3}{c}{Aproximação} \\ \midrule
Tamanho & Mínimo & Média & Máximo & Mínimo & Média & Máximo \\ \midrule
50 & 35 & 46.05 & 58 & 2.11 & 2.63 & 3.27 \\
100 & 77 & 93.57 & 113 & 2.33 & 2.67 & 3.02 \\
150 & 121 & 141.41 & 166 & 2.41 & 2.70 & 3.01 \\
200 & 165 & 189.05 & 213 & 2.43 & 2.71 & 3.03 \\
250 & 212 & 236.93 & 263 & 2.47 & 2.71 & 2.94 \\
300 & 259 & 285.49 & 313 & 2.51 & 2.72 & 2.93 \\
350 & 302 & 333.40 & 365 & 2.55 & 2.72 & 2.92 \\
400 & 349 & 381.91 & 425 & 2.52 & 2.73 & 2.93 \\
450 & 393 & 431.13 & 465 & 2.54 & 2.73 & 2.91 \\
500 & 437 & 478.93 & 520 & 2.58 & 2.73 & 2.92 \\ \bottomrule
\end{tabular}
\end{table}

\begin{table}[H]
\caption{\label{cap6:table:inter_cg_srt_100} Resultados experimentais do algoritmo de $4$-aproximação para o problema da Distância de Reversões, Transposições e Indels Intergênicos em Strings com Sinais (Algoritmo~\ref{cap5:alg:4_transp_rev_cycle_graph}), considerando instâncias criadas com $k = n$.}
\centering
\begin{tabular}{@{}lrrrrrr@{}}
\toprule
 & \multicolumn{3}{c}{Operações} & \multicolumn{3}{c}{Aproximação} \\ \midrule
Tamanho & Mínimo & Média & Máximo & Mínimo & Média & Máximo \\ \midrule
50 & 48 & 57.01 & 69 & 2.27 & 2.84 & 3.35 \\
100 & 99 & 114.47 & 127 & 2.55 & 2.88 & 3.29 \\
150 & 156 & 172.17 & 191 & 2.59 & 2.90 & 3.21 \\
200 & 212 & 230.07 & 250 & 2.66 & 2.90 & 3.18 \\
250 & 268 & 288.32 & 309 & 2.67 & 2.91 & 3.09 \\
300 & 320 & 346.40 & 371 & 2.69 & 2.91 & 3.12 \\
350 & 381 & 404.21 & 431 & 2.64 & 2.91 & 3.10 \\
400 & 435 & 462.18 & 494 & 2.76 & 2.92 & 3.15 \\
450 & 495 & 520.33 & 554 & 2.77 & 2.92 & 3.08 \\
500 & 542 & 579.01 & 618 & 2.76 & 2.92 & 3.12 \\ \bottomrule
\end{tabular}
\end{table}

\section{Experimentos com Genomas Reais}\label{cap6:section:genomas_reais}

Nesta seção, evidenciamos a aplicabilidade de um dos nossos algoritmos usando genomas reais de cianobactérias da base de dados Cyanorak 2.1~\cite{2020:garczarek-etal}. Como as reversões são os eventos mais comuns em cianobactérias~\cite{2002-dalevi-etal,2003-lefebvre-etal} e a base de dados utilizada possui informações sobre a orientação dos genes, utilizamos o algoritmo de $2.5$-aproximação para o problema da Distância de Reversões e Indels Intergênicos em Strings com Sinais (Algoritmo~\ref{cap5:alg:2_5-approx}) e o algoritmo polinomial para a Ordenação de Permutações por Reversões~\cite{1999-hannenhalli-pevzner} (Algoritmo \texttt{HP}). 

A base de dados Cyanorak 2.1 possui $94$ genomas. O número mínimo e o número máximo de genes nesses genomas são $1834$ e $4391$, respectivamente. Em média, cada genoma possui um pouco mais de $95\%$ de genes únicos. Para ajustar os dados à entrada do Algoritmo~\ref{cap5:alg:2_5-approx}, realizamos o seguinte pré-processamento em cada par de genomas para construir uma instância intergênica:

\begin{enumerate}
	\item {\revisaof Mantemos apenas a primeira ocorrência de genes repetidos em cada genoma, marcando as outras cópias como genes a serem inseridos ou removidos};
	\item Mapeamos o genoma de destino em $\G_d = (\iota^n, \breve{\iota}^n)$: neste passo, cada sequência contígua de genes que está presente apenas no genoma de destino é representada por um único elemento na string $\iota^n$. Após isso, para as extremidades e para cada par de elementos consecutivos em $\iota^n$, computamos o tamanho das regiões intergênicas e construímos $\breve{\iota}^n$;
	\item Mapeamos o genoma de origem em $\G_o = (A, \breve{A})$: neste passo, cada sequência contígua de genes que está presente apenas no genoma de origem é representada por um único elemento $\deletionElement$ na string $A$. Após isso, para as extremidades e para cada par de elementos consecutivos em $A$, computamos o tamanho das regiões intergênicas e construímos $\breve{A}$.
\end{enumerate}

Como o Algoritmo \texttt{HP}~\cite{1999-hannenhalli-pevzner} não considera regiões intergênicas e só pode ser aplicado em genomas balanceados, para cada par de genoma, aplicamos o seguinte pré-processamento para construir uma instância:

\begin{enumerate}
	\item {\revisaof Mantemos apenas a primeira ocorrência de genes repetidos em cada genoma e marcamos as outras cópias como genes a serem inseridos ou removidos;}
	\item Marcamos como $\deletionElement$ qualquer sequência contígua de genes presente apenas no genoma de origem ou no genoma de destino;
	\item Removemos cada elemento $\deletionElement$ nos genomas de origem e destino, simulando {\it indels};
	\item Mapeamos a sequência de genes do genoma de destino em $\iota^n$ e a sequência de genes do genoma de origem em $\pi$. 
\end{enumerate}

Para cada par de genomas, o tamanho da solução final é igual à soma do número de {\it indels} simulados no passo $3$ e da distância $d_\rho(\pi)$, encontrada pelo Algoritmo \texttt{HP}~\cite{1999-hannenhalli-pevzner}.

Para cada um desses dois algoritmos, construímos uma matriz de distâncias usando todos os pares de genomas da base de dados Cyanorak 2.1~\cite{2020:garczarek-etal}. Cada uma dessas matrizes é usada para a criação de árvores filogenéticas. Após isso, comparamos cada uma das árvores criadas com a árvore filogenética apresentada por Laurence e coautores~\cite{2020:garczarek-etal}. Usamos como métrica o número de folhas da subárvore de concordância máxima (MAST)~\cite{de2007congruence}, que é construída com um par de árvores. {\revisaof Nessa métrica, quanto maior o número de folhas da MAST, maior o nível de congruência topológica entre as duas árvores comparadas.} A Tabela~\ref{cap6:table:mast} mostra os resultados obtidos ao comparar essas árvores filogenéticas.

\begin{table}[bt]
\caption{\label{cap6:table:mast}
Resultados da comparação entre a árvore filogenética apresentada por Laurence e coautores~\cite{2020:garczarek-etal} e as árvores filogenéticas construídas com as matrizes de distâncias obtidas pelo Algoritmo~\ref{cap5:alg:2_5-approx} e pelo Algoritmo \texttt{HP}. Foram usados três diferentes métodos de reconstrução para a criação das árvores filogenéticas. A métrica usada é a quantidade de folhas na subárvore de concordância máxima (MAST).}
\centering
  {
    \renewcommand{\arraystretch}{1.2}
    \begin{tabular}{@{}llc@{}}
    \toprule
    Método de Reconstrução                                                         & Algoritmo & MAST \\ \hline
    \multirow{2}{*}{{\it Neighbor Joining}~\cite{saitou1987neighbor}}                   & Algoritmo \texttt{HP}        & 53   \\ \cline{2-3} 
                                                                                  & Algoritmo~\ref{cap5:alg:2_5-approx}      & 56   \\ \hline
    \multirow{2}{*}{{\it Unweighted Neighbor Joining}~\cite{gascuel1997concerning}}     & Algoritmo \texttt{HP}        & 52   \\ \cline{2-3} 
                                                                                  & Algoritmo~\ref{cap5:alg:2_5-approx}      & 56   \\ \hline
    \multirow{2}{*}{{\it Circular Order Reconstruction}~\cite{1997:makarenkov-leclerc}} & Algoritmo \texttt{HP}        & 55   \\ \cline{2-3} 
                                                                                  & Algoritmo~\ref{cap5:alg:2_5-approx}      & 57   \\ \hline
    \end{tabular}
  }
\end{table}

Com os resultados da Tabela~\ref{cap6:table:mast}, concluímos que, para todos os métodos de reconstrução usados, as árvores filogenéticas criadas com os resultados do Algoritmo~\ref{cap5:alg:2_5-approx} obtiveram maior nível de congruência topológica com a árvore filogenética apresentada por Laurence e coautores~\cite{2020:garczarek-etal}. A Figura~\ref{cap6:fig:tree} mostra a representação visual da árvore filogenética criada a partir dos resultados do Algoritmo~\ref{cap5:alg:2_5-approx} e do método de reconstrução {\it Circular Order Reconstruction}~\cite{1997:makarenkov-leclerc}, que foi a configuração que obteve maior valor de MAST. Podemos observar que a árvore filogenética da Figura~\ref{cap6:fig:tree} apresenta boa separação das espécies e dos grupos de organismos.

\begin{figure*}[thb]
    \centering
    \includegraphics[width=\columnwidth]{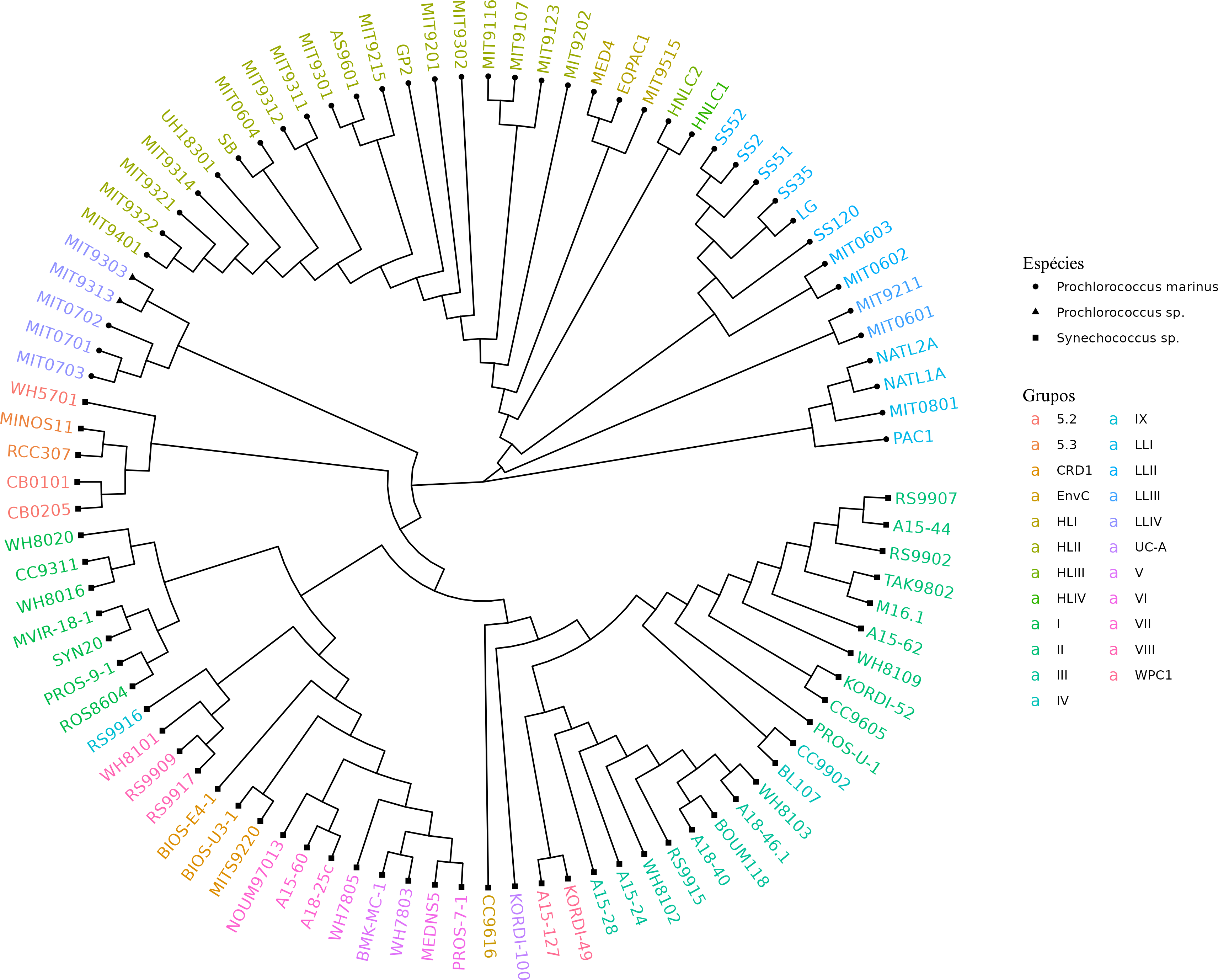}
    \caption{\label{cap6:fig:tree} Árvore filogenética baseada em rearranjos de genomas criada a partir do Algoritmo~\ref{cap5:alg:2_5-approx} e do método de reconstrução {\it Circular Order Reconstruction}~\cite{1997:makarenkov-leclerc} usando genomas da base de dados Cyanorak 2.1~\cite{2020:garczarek-etal}. Utilizamos o pacote \texttt{treeio} da linguagem R~\cite{2020:wang-etal} para a construção desta imagem.}
\end{figure*}

\chapter{Considerações Finais}\label{cap:label:conclusao}

Esta tese apresentou os resultados mais importantes obtidos durante o período de doutorado. Para todos os problemas considerados nesta tese, foi realizado um estudo extensivo com foco na prova da complexidade desses problemas e no desenvolvimento de algoritmos de aproximação. 

Iniciamos esta tese tratando de problemas de Distância de Rearranjos em genomas balanceados, mais conhecidos como problemas de Ordenação de Permutações por Rearranjos. Apresentamos uma nova versão de um dos algoritmos mais conhecidos na literatura da área, a $1.375$-aproximação de Elias e Hartman~\cite{2006-elias-hartman}, {\revisaof que corrige um problema do algoritmo original com uma complexidade de tempo menor que a de outros algoritmos propostos para correção desse mesmo problema}. Além disso, provamos que a Ordenação de Permutações por Rearranjos é NP-difícil para seis modelos de rearranjos que incluem transposições junto com a combinação de reversões, transposições inversas e revrevs, sendo que alguns desses modelos são estudados desde o século passado, mas não possuíam complexidade conhecida.

A cada capítulo, incorporamos mais características genômicas aos problemas com o objetivo de torná-los mais relevantes do ponto de vista biológico. Estudamos a Distância de Rearranjos e a Distância de Rearranjos Intergênicos em genomas desbalanceados, adaptando e criando novos conceitos e estruturas que possibilitaram a criação de algoritmos de aproximação. Também provamos que a maioria dos problemas investigados são NP-Difíceis. 

{\revisaodois Por fim, apresentamos experimentos em genomas sintéticos e em genomas reais. Os experimentos em genomas sintéticos mostram que o fator de aproximação prático de alguns dos algoritmos foi muito menor que o fator de aproximação teórico. Já no experimento com dados reais, usamos genomas de cianobactérias da base de dados Cyanorak 2.1~\cite{2020:garczarek-etal}. Nesse experimento, criamos árvores filogenéticas usando o nosso algoritmo para a Distância de Reversões e Indels Intergênicos e o algoritmo polinomial exato para a Ordenação de Reversões, criado por Hannenhalli e Pevzner~\cite{1999-hannenhalli-pevzner}. Quando comparadas com a árvore filogenética criada pelos autores que publicaram a base de dados Cyanorak 2.1, a árvore filogenética criada com o nosso algoritmo obteve um nível de congruência topológica maior do que a árvore filogenética criada com o algoritmo de Hannenhalli e Pevzner~\cite{1999-hannenhalli-pevzner}}.

Todos os resultados desta tese foram publicados em congressos e revistas internacionais, nos seguintes artigos:

\begin{itemize}
	\item ``A 1.375-Approximation Algorithm for Sorting by Transpositions with
  Faster Running Time'', apresentado na conferência {\it Brazilian Symposium on Bioinformatics} (BSB) em 2022~\cite{2022d-alexandrino-etal} (Capítulo~\ref{cap:label:transposicao}, Seção~\ref{cap3:secao_novo_algoritmo});
	\item ``On the Complexity of Some Variations of Sorting by Transpositions'', publicado na revista {\it Journal of Universal Computer Science} em 2020~\cite{2020c-alexandrino-etal} (Capítulo~\ref{cap:label:transposicao}, Seção~\ref{cap3:secao_complexidade});
	\item ``Genome Rearrangement Distance with Reversals, Transpositions, and
  Indels'', publicado na revista {\it Journal of Computational Biology} em 2021~\cite{2021a-alexandrino-etal} (Capítulo~\ref{cap:label:indel}, Seção~\ref{cap4:section:breakpoints});
	\item ``Labeled Cycle Graph for Transposition and Indel Distance'', publicado na revista {\it Journal of Computational Biology} em 2022~\cite{2022a-alexandrino-etal} (Capítulo~\ref{cap:label:indel}, Seção~\ref{cap4:section:grafo_ciclos});
	\item ``Block Interchange and Reversal Distance on Unbalanced Genomes'', apresentado na conferência {\it Brazilian Symposium on Bioinformatics} (BSB) em 2023~\cite{alexandrino2023block} (Capítulo~\ref{cap:label:indel}, Seção~\ref{cap4:section:grafo_ciclos});
	\item ``Incorporating Intergenic Regions into Reversal and Transposition
  Distances with Indels'', apresentado na conferência RECOMB {\it Comparative Genomics} em 2021. Uma versão estendida foi publicada na revista {\it {Journal of Bioinformatics and Computational Biology}} em 2021~\cite{2021c-alexandrino-etal} (Capítulo~\ref{cap:label:intergenicos}, Seção~\ref{cap5:section:breakpoints});	
	\item ``Reversal Distance on Genomes with Different Gene Content and
  Intergenic Regions Information'', apresentado na conferência {\it Algorithms for Computational Biology} (AlCoB) em 2021~\cite{2021b-alexandrino-etal} (Capítulo~\ref{cap:label:intergenicos}, Seção~\ref{cap5:subsection:2_5_reversal});
	\item ``Reversal and Indel Distance with Intergenic Region Information'', publicado na revista {\it IEEE/ACM Transactions on Computational Biology and
  Bioinformatics} em 2023~\cite{2022b-alexandrino-etal} (Capítulo~\ref{cap:label:intergenicos}, Seção~\ref{cap5:subsection:2_5_reversal});
	\item ``Transposition Distance Considering Intergenic Regions for Unbalanced
  Genomes'', apresentado na conferência {\it International Symposium on
  Bioinformatics Research and Applications} (ISBRA) em 2022~\cite{2022c-alexandrino-etal} (Capítulo~\ref{cap:label:intergenicos}, Seção~\ref{cap5:subsection:4_transp});
	\item ``Reversal and Transposition Distance on Unbalanced Genomes Using
  Intergenic Information'', publicado na revista {\it Journal of Computational Biology} em 2023~\cite{2023-alexandrino-jcb} (Capítulo~\ref{cap:label:intergenicos}, Seção~\ref{cap5:subsection:4_transp}).
  \item ``Rearrangement Distance Problems: An updated survey'', aceito para publicação na revista {\it ACM Computing Surveys} em 2024~\cite{oliveira2024rearrangement} (Revisão bibliográfica).
\end{itemize}

Além dos artigos acima, que são diretamente relacionados a resultados específicos desta tese, outras contribuições na área de rearranjos de genomas foram publicadas durante o período de doutorado nos seguintes artigos:

\begin{itemize}
  \item ``Approximation Algorithms for Sorting Permutations by Length-Weighted
  Short Rearrangements'', apresentado na conferência {\it Latin and American Algorithms, Graphs and Optimization Symposium} (LAGOS) em 2019~\cite{2019-alexandrino-etal};
  \item ``Sorting Permutations by Fragmentation-Weighted Operations'', publicado na revista {\it Journal of Bioinformatics and Computational Biology} em 2020~\cite{2020a-alexandrino-etal};
  \item ``Length-Weighted $\lambda$-Rearrangement Distance'', publicado na revista {\it Journal of Combinatorial Optimization} em 2020~\cite{2020b-alexandrino-etal};
  \item ``Heuristics for Breakpoint Graph Decomposition with Applications in Genome Rearrangement Problems'', apresentado na conferência {\it Brazilian Symposium on
  Bioinformatics} (BSB) em 2020~\cite{2020-pinheiro-etal};
  \item ``Sorting by Reversals and Transpositions with Proportion Restriction'', apresentado na conferência {\it Brazilian Symposium on
  Bioinformatics} (BSB) em 2020~\cite{2020c-brito-etal};
  \item ``Reversals Distance Considering Flexible Intergenic Regions Sizes'', apresentado na conferência {\it Algorithms for Computational Biology} (AlCoB) em 2021~\cite{2021c-brito-etal};
  \item ``Approximation Algorithms for Sorting $\lambda$-Permutations by $\lambda$-Operations'', publicado na revista {\it Algorithms} em 2021~\cite{2021-miranda-etal};
  \item ``Reversals and Transpositions Distance with Proportion Restriction'', publicado na revista {\it Journal of Bioinformatics and Computational Biology} em 2021~\cite{2021a-brito-etal};
  \item ``Approximation Algorithm for Rearrangement Distances Considering Repeated Genes and Intergenic Regions'', publicado na revista {\it Algorithms for Molecular Biology} em 2021~\cite{2021b-siqueira-etal};
  \item ``Algorithms for the Maximum Eulerian Cycle Decomposition Problem'', apresentado na conferência Simp{\'{o}}sio Brasileiro de Pesquisa Operacional (SBPO) em 2021~\cite{2021-pinheiro-etal};
  \item ``Heuristics for Cycle Packing of Adjacency Graphs for Genomes with Repeated Genes'', apresentado na conferência {\it Brazilian Symposium on Bioinformatics} (BSB) em 2021~\cite{2021c-siqueira-etal};
  \item ``Reversal and Transposition Distance of Genomes Considering Flexible Intergenic Regions'', apresentado na conferência {\it Latin and American Algorithms, Graphs and Optimization Symposium} (LAGOS) em 2021~\cite{2021d-brito-etal}; 
  \item ``An Improved Approximation Algorithm for the Reversal and Transposition Distance Considering Gene Order and Intergenic Sizes'', publicado na revista {\it Algorithms for Molecular Biology} em 2021~\cite{2021b-brito-etal};
  \item ``Signed Rearrangement Distances Considering Repeated Genes and Intergenic Regions'', apresentado na conferência {\it Bioinformatics and Computational Biology} (BICoB) em 2022~\cite{2022-siqueira-etal};
  \item ``A New Approach for the Reversal Distance with Indels and Moves in Intergenic Regions'', apresentado na conferência RECOMB {\it Comparative Genomics} em 2022~\cite{2022b-brito-etal};
  \item ``Sorting by $k$-Cuts on Signed Permutations'',  apresentado na conferência RECOMB {\it Comparative Genomics} em 2022~\cite{2022-oliveira-etal};
  \item ``Genome Rearrangement Distance with a Flexible Intergenic Regions Aspect'', publicado na revista {\it IEEE/ACM Transactions on Computational Biology and Bioinformatics} em 2023~\cite{2022a-brito-etal};
  \item ``Approximation Algorithms for Sorting by k-Cuts on Signed Permutations'', publicado na revista {\it Journal of Combinatorial Optimization} em 2023~\cite{2023-oliveira-etal};  
  \item ``Rearrangement Distance with Reversals, Indels, and Moves in Intergenic Regions on Signed and Unsigned Permutations'', publicado na revista {\it Journal of Bioinformatics and Computational Biology} em 2023~\cite{brito2023rearrangement};
  \item ``Signed Rearrangement Distances Considering Repeated Genes, Intergenic Regions, and Indels'', publicado na revista {\it Journal of Combinatorial Optimization} em 2023~\cite{siqueira2023joco};
  \item ``Approximating Rearrangement Distances with Replicas and Flexible Intergenic Regions'', apresentado na conferência {\it International Symposium on Bioinformatics Research and Applications} (ISBRA) em 2023~\cite{siqueira2023isbra};
  \item ``Maximum Alternating Balanced Cycle Decomposition and Applications in Sorting by Intergenic Operations Problems'', aceito para apresentação na conferência RECOMB {\it Comparative Genomics} em 2024~\cite{2024-brito-etal}.
\end{itemize}

O primeiro trabalho futuro importante deixado por esta tese é a determinação da complexidade dos problemas de Ordenação de Permutações por Transposições e Outros Rearranjos quando $w_{\tau}/w_{\rho} > 1.5$, onde $w_{\rho}$ é o custo de reversões e $w_{\tau}$ é o custo de transposições e rearranjos similares.

Para os problemas de Distância de Rearranjos e Distância de Rearranjos Intergênicos em genomas desbalanceados, consideramos apenas a abordagem não ponderada nesta tese. Assim, um trabalho futuro interessante é o estudo desses problemas considerando uma abordagem ponderada. Além disso, os seguintes problemas continuam com complexidade em aberto: a Distância de Block Interchanges e Indels em Strings sem Sinais; a Distância de Block Interchanges, Reversões e Indels em Strings com Sinais; e a Distância de Reversões e Indels Intergênicos em Strings com Sinais.

\bibliographystyle{plain}
\bibliography{bibfile}


\end{document}